\documentclass[10pt,twoside,openright]{report} 

\usepackage{geometry,enumerate,latexsym,amsmath,amssymb,booktabs,setspace,graphicx,multirow,pst-all,subcaption,pdflscape,mathtools,microtype}
\usepackage[colorinlistoftodos]{todonotes}
\usepackage[nottoc]{tocbibind}
\usepackage{amsthm,dsfont}
\newtheorem{Theorem}{Theorem}
\newtheorem{Lemma}{Lemma}
\newtheorem{Corollary}{Corollary}%[section]

\numberwithin{Theorem}{section}
\numberwithin{Definition}{section}
\numberwithin{Lemma}{section}
\numberwithin{Corollary}{section}
\numberwithin{Algorithm}{section}
\numberwithin{equation}{section}
\DeclarePairedDelimiter{\ceil}{\lceil}{\rceil}
\DeclarePairedDelimiter{\floor}{\lfloor}{\rfloor}

\title{The Stackelberg Game: \\responses to regular strategies}
\author{Thomas Byrne}
\date{2020}

\usepackage[phd]{edmaths}

\usepackage[hidelinks]{hyperref}
\usepackage{apacite}

\begin{document}
%\nobibliography*{}
%\pagenumbering{roman}

%\maketitle
\begin{titlepage}
%\setstretch{1.3}
\vspace*{.5em}
\center
\vspace*{1em}
%\begin{figure}[!h]
%\centering
%\includegraphics[width=150pt]{Logo.png}
%\end{figure}
\vspace{2em}
\textbf{\Huge{The Stackelberg Game: \\responses to regular strategies}}\par%\\
\vspace{2em}
\textbf{\normalsize{ }}\par%\\
\textbf{\LARGE{Thomas Byrne}}\\
\vspace{1em}
\text{tbyrne@ed.ac.uk}\\
\vspace{1em}

%\flushleft{\hspace{5.9cm}Strathclyde Business School\\
%
%\hspace{5.9cm}University of Strathclyde\\
%
%\hspace{5.9cm}199 Cathedral Street\\
%
%\hspace{5.9cm}Glasgow\\
%
%\hspace{5.9cm}G4 0QU\\
%
%\hspace{5.9cm}United Kingdom}

\flushleft{}

\parbox{\linewidth}{This document contains exclusively Chapter 5 of the author's PhD thesis titled \textit{Facility Location Problems and Games}. Besides the page numbering (where page~1 in this document corresponds to page~75 in the thesis), the chapter is reproduced here in its entirety and appears identically to how it appears in the thesis. This document serves to act as a reference chapter. Interested readers, and those desiring a more complete introduction to the notation used and the material explored in this chapter, are encouraged to read the preceding and subsequent chapters in the thesis. Additionally, the related paper \href{https://doi.org/10.48550/arXiv.2011.13275}{\textit{Competitive Location Problems: Balanced Facility Location and the One-Round
Manhattan Voronoi Game}} may be of interest.}

\vspace{1em}

\parbox{\linewidth}{The following results from \textit{Facility Location Problems and Games} are referenced in the chapter:}

\setcounter{chapter}{4}
\setcounter{section}{1}
\setcounter{subsection}{1}
\begin{Lemma}
\label{steal}
For any arena $\mathcal{P}$, Black can place a point $b$ within any bounded Voronoi cell $V^\circ(w)$ of White's in order to steal at least $50(1-\varepsilon)\%$ of $V^\circ(w)$ for any $\varepsilon>0$.
\end{Lemma}

\begin{Lemma}
\label{equalarea}
For any arena $\mathcal{P}$, if the Voronoi cells $V^\circ(w)$ have unequal area then Black can win.
\end{Lemma}

\begin{Lemma}
\label{symmetricarea}
For any arena $\mathcal{P}$, if any Voronoi cell $V^\circ(w)$ has unequal area either side of the horizontal or vertical through $w$ then Black can win.
\end{Lemma}

\setcounter{section}{2}
\begin{Theorem}
The only winning arrangement for White is the $1 \times n$ arrangement for $\frac{p}{q} \geq n$. Otherwise Black wins.
\label{TheoremVoronoiGame}
\end{Theorem}

\begin{Corollary}
For any Voronoi cell $V^\circ(w)$ in a winning arrangement of White's, if $V^\circ(w)$ does not touch opposite sides of $\mathcal{P}$ then its arms are all equal length.
\label{armsequal}
\end{Corollary}

\setcounter{Lemma}{1}
\begin{Lemma}
For any Voronoi cell $V^\circ(w)$ in a winning arrangement of White's, if one of the arms does not touch the boundary of $\mathcal{P}$ then the opposite arm parallel to this one is no shorter than the arms perpendicular to these arms.
\label{armsperpendicular}
\end{Lemma}

\begin{Lemma}
For any Voronoi cell $V^\circ(w)$ in a winning arrangement of White's, if one of the arms does not touch the boundary of $\mathcal{P}$ then the arms perpendicular to this arm are equal.
\label{armsparallel}
\end{Lemma}

\vfill
%\text{\textit{Corresponding author}: Thomas Byrne}
\end{titlepage}

\thispagestyle{empty}

\setcounter{chapter}{4}
\setcounter{page}{0}
\chapter{The Stackelberg Game: keeping regular}
\label{StackelbergGameGrid}

%\section{Introduction}

Following the solution to the One-Round Voronoi Game we naturally may want to consider similar games based upon the competitive locating of points and subsequent dividing of territories. In order to appease White's tears after they have potentially been tricked into going first in a game of point-placement, an alternative game (or rather, an extension of the previous game) is the Stackelberg game where all is not lost if Black gains over half of the contested area.

The set-up is identical to that of the Voronoi game. We consider the Voronoi game as before with two players, White and Black, who take turns to place a total of $n$ points into the playing arena (without the ability to place atop or move an existing point) before it is partitioned into the Voronoi diagram of these points. Each player gains a score equal to the area of the Voronoi cells generated by their points $W$ and $B$ respectively and each player's objective is to maximise this score not to be more than their competitor's score, but to have the largest score. That is, White and Black wish to maximise their respective scores
\begin{equation*}
\begin{split}
\mathcal{W}=\sum_{w\in W} {Area(V(w))} \\
\mathcal{B}=\sum_{b\in B} {Area(V(b))}
\end{split}
\end{equation*}
where, as before, we have the notation scheme:
\begin{equation*}
    \begin{split}
        \mathcal{VD}(W) &= \{V^\circ(x) : x\in W \} \\
        \mathcal{VD}(W \cup b) &= \{V^+(x) : x\in W \cup b \} \\
        \mathcal{VD}(W \cup B) &= \{V(x) : x\in W \cup B \} \,.
    \end{split}
\end{equation*}

This is subtly different to the Voronoi game wherein each player cared solely about controlling more than the other player (or over half of the playing arena) and so did not present an arrangement in such cases where they could not win over half of the playing area. Because of this, the Stackelberg game is the obvious extension to the Voronoi game.

Stackelberg games (generally defined to be a game in which a leader and a follower compete for certain quantities) present themselves in a wide range of applications so, perhaps unsurprisingly, there is substantial literature on a diverse range of interpretations. For a full classification of these competitive facility location problems and their many variations see the survey \citeA{Pla01}, and the detailed \citeA{EisLap97} for a study focused upon the more sequential problems.

Many bi-level Stackelberg location models make use of an attractiveness measure for each facility, the most popular of which is the gravity-based model proposed by \citeA{Rei31} wherein the patronage of each customer is decided (deterministically or randomly) based upon a function proportional to the attractiveness score of the facility and inversely proportional to the distance between the facility and customer. Both the location and attractiveness of new facilities is allowed to be optimised in \citeA{KucAraKub12} where the leader locates new facilities within a market containing the follower's existing facilities in order to maximise captured demand, before the follower is allowed the opportunity to adjust their facilities.

Given one existing facility and a number of demand points, \citeA{Dre82} located a new facility in order to maximise its attracted buying power both in the situation where the existing facility is fixed, and where the follower is allowed to open a new facility. A centroid model is proposed in the presence of continuous demand in \citeA{BhaEisJar03} which gives the follower the opportunity to respond to the leader's facility placement with placements of their own.

\citeA{SerRev94} introduced a model wherein both players locate the same number of facilities in a network with customers patronising only the closest facility, and two accompanying heuristic algorithms are presented therein. However, in a cruel twist the objective of each player is to minimise the score of the other player. Nevertheless this may not be a surprising sentiment of each player since, as \citeA{MooBar90} indicated, the players' objectives almost always conflict with one another in the Stackelberg game.

In the hope of some level of benevolence between warring players White and Black, again we shall focus on the One-Round Stackelberg Game over a rectangular playing arena $\mathcal{P}$ with length $p$ and height $q$. Just as in the Voronoi game, this is impossible to write in a closed form since the objectives rely entirely on the relative locations of the other points and so we approach the problem from a geometrical standpoint.

Firstly we shall note that the winning arrangement found for White for the Voronoi game carries over to this game since, if $\frac{p}{q} \geq n$, it was shown that deviating from this arrangement in any way would give Black more than $\frac{pq}{2n}$ and so decrease White's score. We also found the optimal strategy for Black in response to this arrangement given that the condition $\frac{p}{q} \geq n$ held in the proof of Theorem \ref{TheoremVoronoiGame}. %4.2.1.
The supremum of all areas of $V^+(b_1)$ in Sections $III$ and $IV$ was found to be $\frac{pq}{2n}$, achieved when $b_1$ lay atop one of White's existing points. Therefore Black's optimal strategy would be that described in Lemma \ref{steal}, %4.1.1,
placing each separate point as close as possible to one of White's points and thereby securing a score of $\frac{(1-\varepsilon)pq}{2}$.

What remains to be explored for the Stackelberg game is how best White can mitigate the damage of Black's placements when $\frac{1}{n} < \frac{p}{q} < n$. [Note that since we will not enforce that $p \geq q$ within this chapter we must ensure that $\frac{p}{q} < n$ holds upon reflection in $y=x$ (i.e. for $p$ and $q$ swapped giving $\frac{q}{p} < n$), thus providing the $\frac{1}{n} < \frac{p}{q}$ condition.]
%we must enforce $\frac{p}{q} \geq 1$ (allowing a reflection in $y=x$ in order to swap $p$ and $q$ if necessary) since the results in Chapter~\ref{VoronoiGame} relied upon $p \geq q$ and dropping this constraint in this chapter would allow $\frac{q}{p} \geq n$ for which the results in Chapter~\ref{VoronoiGame} carry over albeit reflected in $y=x$. We shall not necessarily enforce that $p \geq q$ in this chapter ]

Since the Lemmas \ref{equalarea}, \ref{symmetricarea}, \ref{armsperpendicular}, and \ref{armsparallel} %4.1.2, 4.1.3, 4.2.2, and 4.2.3
outlined significant weaknesses in certain arrangements, we shall first consider arrangements that still satisfy these results and explore how Black can best exploit these positions. This investigation begins in Section \ref{sec:StackelbergGrid} wherein an early result shows that White must play a certain grid arrangement. From there we consider Black's possible responses, exploring their best positions for stealing area from White and then their best overall strategy for when White plays a row (in Sections~\ref{sec:WhiteRow} and~\ref{sec:BlackRow}) or a grid (in Sections~\ref{sec:WhiteGrid} and~\ref{sec:BlackGrid}). %Following this, White may wonder how they fare by relaxing any of the conditions provided in the aforementioned lemmas. Since the remaining case (i.e. White does not play a grid) is incredibly non-restricted we keep the sensible requirement that White's cells are balanced, thereby investigating area-symmetric $V^\circ(w)$ of equal areas, studied in Section \ref{sec:StackelbergNonGrid} along with Black's optimal responses. Concluding thoughts are summarised in Section \ref{sec:StackelbergConc}.

\section{White's optimal strategy: a grid}
\label{sec:StackelbergGrid}

%Thus, White needs to ensure that $VD(W)$ contains only equally-sized cells, with equal area each side horizontally and vertically through each $w\in W$, otherwise Black can certainly win. Since $\mathcal{P}$ is a rectangle, and bisectors consist of purely horizontal, vertical, and diagonal lines, any bisector intersecting the perimeter of $\mathcal{P}$ must be perpendicular to the edge of $\mathcal{P}$ -- else it is a diagonal, requiring White to place points on the perimeter of $\mathcal{P}$, contradicting the area symmetry requirement. This means that the bottom left corner of $\mathcal{P}$ will be contained in one Voronoi cell. \todo[inline]{Can we prove from these that White must play a grid?} %Since the bottom left corner of $\mathcal{P}$ is contained in only one Voronoi cell, by symmetry this cell must be a rectangle. If the top edge of the cell is not an edge of $\mathcal{P}$ then it must be a bisector. A horizontal bisector requires the points to be on the same vertical. This requires the generator of the neighbouring cell to be directly above the bottom left generator, and by symmetry the height of this neighbouring cell is the same as the bottom left cell, and therefore the same width to be of equal size. The same argument is applied to horizontal neighbours to require that White produces a Voronoi diagram of identical rectangles with their generators evenly distributed regularly across $\mathcal{P}$.

It was proven in Chapter 4 %\ref{VoronoiGame}
that any winning arrangement of White's points in the Voronoi game must have cells $V^\circ(w)$ of equal area (Lemma \ref{equalarea}), %4.1.2),
each with every horizontal and vertical half of the cell equal (Lemma \ref{symmetricarea}), %4.1.3),
and that if any arm does not touch the boundary of $\mathcal{P}$ then the opposite arm is not shorter than the perpendicular arms (Lemma \ref{armsperpendicular}) %4.2.2)
and these perpendicular arms are of equal length (Lemma \ref{armsparallel}). %4.2.3).
It is natural to wonder what forms an arrangement can take if it adheres to all of these results, and this is summarised in Lemma \ref{lem:grids}.

Firstly, let us define a regular orthogonal grid. A set of $n$ points is a \emph{regular orthogonal $a \times b$ grid} within $\mathcal{P}$ ($n=ab$ and $a, b \geq  1$) if, without loss of generality locating the origin at the bottom left vertex of $\mathcal{P}$, for every point $w \in W$ there exists $i,j\in \mathbb{Z}$, $0\leq i< a$ and $0\leq j< b$, such that $w=(\frac{p}{2a}+\frac{p}{a}i,\frac{q}{2b}+\frac{q}{b}j)$. Additionally, a regular orthogonal $a \times b$ grid is a \emph{square regular orthogonal $a \times b$ grid} if $\frac{p}{a}=\frac{q}{b}$. From this point onwards, unless explicitly stated otherwise, we shall simply use the term \emph{grid} to mean a regular orthogonal grid, and \emph{square grid} to mean a square regular orthogonal grid.

\pagebreak The following result establishes the properties of an arrangement which satisfies Lemmas~\ref{equalarea}, \ref{symmetricarea}, \ref{armsperpendicular}, and \ref{armsparallel} %4.1.2, 4.1.3, 4.2.2, and 4.2.3
\cite{ByrFekKalKle20}.

\begin{Lemma}\label{lem:grids}
For any arrangement $W$ satisfying Lemmas~\ref{equalarea}, \ref{symmetricarea}, \ref{armsperpendicular}, and \ref{armsparallel}, %4.1.2, 4.1.3, 4.2.2, and 4.2.3,
if $\frac{p}{q} \geq n$ then $W$ is a $1 \times n$ grid; otherwise, $W$ is a square grid or no such arrangement exists.
\end{Lemma}
\begin{proof}
Firstly let us clarify that, from Theorem~\ref{TheoremVoronoiGame}, %4.2.1, 
if $\frac{p}{q} \geq n$ then the only winning strategy for White in the Voronoi game is a $1 \times n$ row. This, however, does not provide us with our required result here since we no longer restrict $W$ to being a winning arrangement.

If $\frac{p}{q} \geq n$ then, by Lemmas \ref{equalarea} and \ref{symmetricarea}, %4.1.2 and 4.1.3,
the area of every half cell of $\mathcal{VD}(W)$ is $\frac{pq}{2n}= \frac{p}{q} \times \frac{q^2}{2n} \geq \frac{q^2}{2}$. In order to achieve this area, since the height of every cell is bounded above by $q$, the left and right arms of every cell must be at least $\frac{q}{2}$. If any cell were not to touch opposite sides of $\mathcal{P}$ then, by Corollary \ref{armsequal}, %4.2.1,
its arms must be of equal length and so would be of length no less than $\frac{q}{2}$ which would make it touch the horizontal sides of $\mathcal{P}$. Therefore every cell touches opposite sides of $\mathcal{P}$. If a cell were to touch both vertical sides of $\mathcal{P}$ then, by Lemma \ref{armsperpendicular}, %4.2.2,
at least one of the vertical arms would have to be longer than the horizontal arms, the minimum length therefore being $\frac{p}{2}$. If this vertical arm did not touch the boundary of $\mathcal{P}$ then the same logic would apply to the other vertical arm, forcing it to have length at least $\frac{p}{2}$, which would create two vertical arms with lengths summing to $p$ ($>q$). However, if this arm did touch the boundary of $\mathcal{P}$ then the half cell containing the arm, split along the horizontal arms, would have area $p\times \frac{p}{2} = \frac{p^2}{2} > \frac{pq}{2n}$. Thus every cell must touch each horizontal edge of $\mathcal{P}$.

By Lemma \ref{armsparallel} %4.2.3
the vertical arms of every cell are therefore $\frac{q}{2}$, i.e. every point of $W$ is placed on the horizontal centre line of $\mathcal{P}$. Noting that all bisectors are now vertical lines, the only way to distribute these across $\mathcal{P}$ in order to divide $\mathcal{P}$ into equal areas (of $\frac{pq}{n}$) satisfying Lemma \ref{equalarea} %4.1.2
is to place them at intervals of $\frac{p}{n}$. This corresponds to the $1 \times n$ grid.
	
For the $\frac{p}{q} < n$ case, let us consider the point $w$ whose cell $V^\circ(w)$ contains the bottom left corner of $\mathcal{P}$. If $V^\circ(w)$ were to touch both horizontal edges of $\mathcal{P}$ then by Lemma \ref{armsperpendicular} %4.2.2
its left arm would be of length no less than $\frac{q}{2}$, causing the left half of $V^\circ(w)$ to have area at least $\frac{q^2}{2} > \frac{p}{qn} \times \frac{q^2}{2} = \frac{pq}{2n}$. $V^\circ(w)$ also cannot touch both vertical sides of $\mathcal{P}$ by the same argument presented in the $\frac{p}{q} \geq n$ case. Therefore we can apply the result from Corollary \ref{armsequal} %4.2.1
and all arms of the cell are of equal length, $d$ say.

Since the bottom left vertex of $\mathcal{P}$ is contained in $V^\circ(w)$ there are no $\mathcal{CC}^4(w)$, $\mathcal{CC}^5(w)$, $\mathcal{CC}^6(w)$, or $\mathcal{CC}^7(w)$ bisectors, so the entire third quadrant of $w$ contained in $\mathcal{P}$ is also contained in $V^\circ(w)$. Therefore the bottom left quadrant of $V^\circ(w)$ is a square of area $d^2$. By Lemmas \ref{equalarea} and \ref{symmetricarea}, %4.1.2 and 4.1.3,
the top right quadrant of $V^\circ(w)$ must also have area $d^2$, and with arms $u=r=d$ this top right quadrant must also be a square.

Considering the bisector which contributes the vertical segment bounding the top right quadrant of $V^\circ(w)$, the other point, $w'=(x,y)$, in this bisector must lie on the line $y=x-2d$ for $0 \leq y \leq d$ (as shown in Figure \ref{fig:GridProof}) and no other point may lie between $w$ and this line. Since $B(w,w')$ is a bound on the advancement of $V^\circ(w')$ and no other point can be closer than $w$ or $w'$ to the lower breakpoint of $B(w,w')$ (else this would contradict the shape of the top right quadrant of $V^\circ(w)$), the left arm of $w'$ must also be of length $d$.
	 
\begin{figure}[!ht]
\centering
\includegraphics[width=0.9\textwidth]{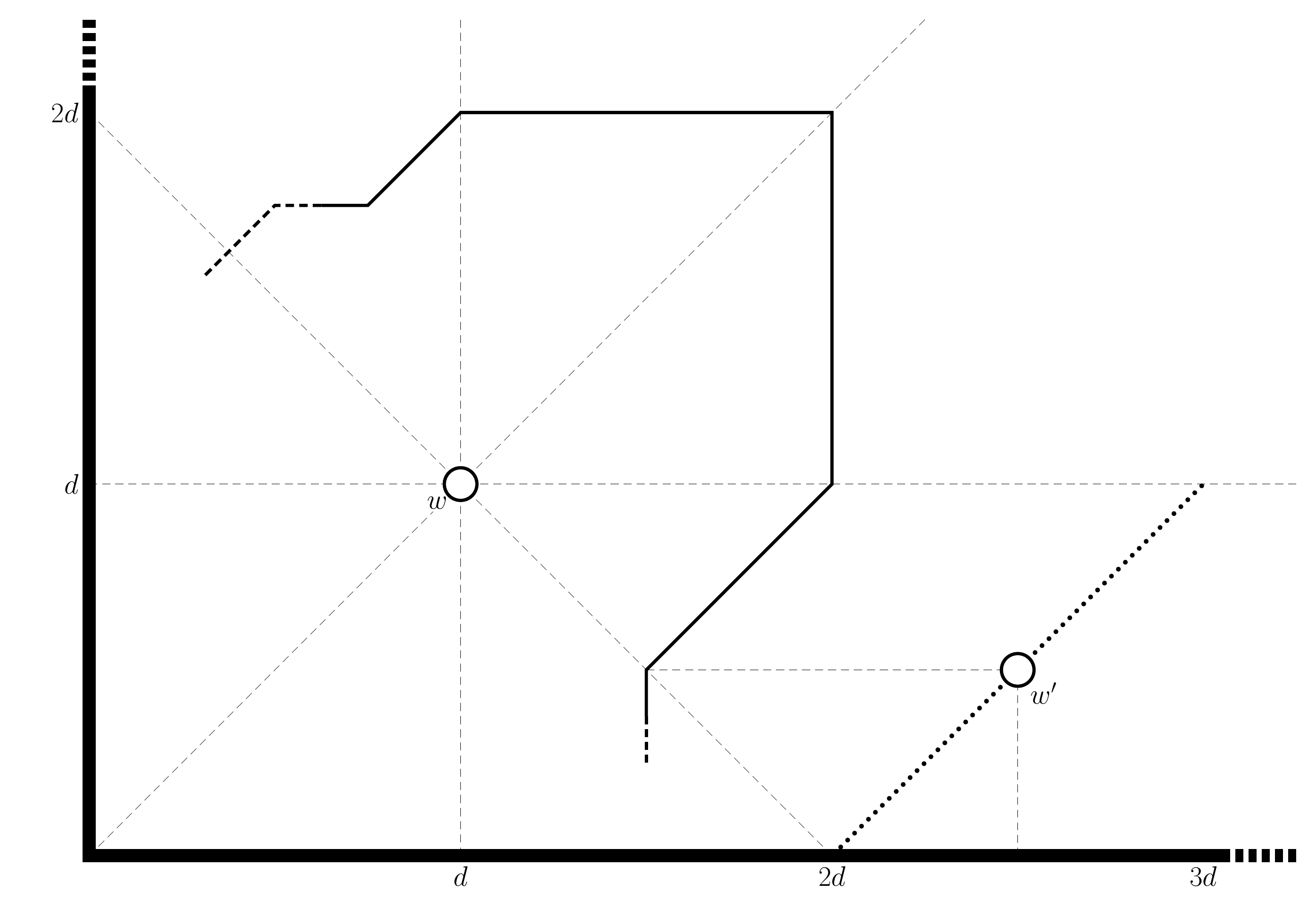}
\caption{$V^\circ(w)$ and a neighbouring point $w'$.}
\label{fig:GridProof}
\end{figure}

We can easily show that $V^\circ(w')$ cannot touch opposite sides of $\mathcal{P}$ since $w$ is blocking it from touching both vertical sides of $\mathcal{P}$ and to touch both horizontal sides of $\mathcal{P}$ would mean, by Lemma~\ref{armsparallel}, %4.2.3,
that its upper and lower arms are equal and so $\frac{q}{2}$, contradicting $0 \leq y \leq d$. Therefore, utilising Corollary \ref{armsequal}, %4.2.1,
all arms of $V^\circ(w')$ have length $d$. This places $w'=(3d,d)$, on the same horizontal as $w$, so the bisector $B(w,w')$ is vertical, and the bottom right quadrant of $V^\circ(w)$ is also square. This forces the top left quadrant to also be square in order to have area $d^2$.

Analogously this argument can be applied to the right-hand boundary of $V^\circ(w')$ (since the bottom left quadrant of $V^\circ(w')$ is now seen to be a $d \times d$ square) to establish that its unique neighbour $w''$ has arms of length $d$ and is situated at $(5d,d)$, and can be continued to give a row of points $w^{(i)}=((2i+1)d,d)$ for $i\in\mathbb{Z}^+$, giving $2d \times 2d$ square Voronoi cells up until the right-hand boundary of $\mathcal{P}$. By symmetry the argument is identical for the points above this row (starting from $w$ we get a column of $2d \times 2d$ square Voronoi cells, and then identically upwards from $w'$, and $w''$ and so on).

\pagebreak The iterative use of this argument gives a square grid arrangement. We should be careful to note, however, that in order for this to work we require that the dimensions of $\mathcal{P}$ allow $n$ squares of length $2d$ to fit within it. That is, there exists $a,b \in \mathbb{N}$ such that $a=\frac{p}{2d}$ and $b=\frac{q}{2d}$ where $a \times b = n$ for some $d \in \mathbb{R}^+$ (alternatively, $4d^2 = \frac{pq}{n}$ gives $2d=\sqrt{\frac{pq}{n}}$ so our conditions are $a=\sqrt{\frac{pn}{q}}\in\mathbb{N}$ and $b=\sqrt{\frac{qn}{p}}\in\mathbb{N}$).
\end{proof}

Since adherence to Lemmas \ref{equalarea}, \ref{symmetricarea}, \ref{armsperpendicular}, and \ref{armsparallel} %4.1.2, 4.1.3, 4.2.2, and 4.2.3
lends an obvious advantage to White in the Voronoi game, it may also be considered sensible to implement the strategies suggested by these results in the Stackelberg game. Therefore we shall explore such arrangements in the Stackelberg setting. Though Lemma \ref{lem:grids} provided constraints on the aspect ratio of $\mathcal{P}$, we shall explore the $a \times b$ grid where $a,b>1$, and $1 \times n$ row strategies (even relaxing the square grid constraint) for any aspect ratio to test the relationship between the games and outline how best White's positions can be exploited by Black.

\section{White plays a \texorpdfstring{$1 \times n$}{1 x n} row}
\label{sec:WhiteRow}

Firstly we shall explore the placement of Black's point $b_1$ assuming that White plays their points in a row. Without loss of generality let this be horizontally (a rotation of $\mathcal{P}$ can easily fix this -- note that, whichever rotation we choose, we are only required to explore $\frac{p}{n} < q$) and label the vertices of $W$ running from left to right as $w_1$ through to $w_n$. Since White's arrangement is repetitive and has such symmetry, our search for Black's optimal location is greatly simplified as we need only consider the placement of $b_1$ within a small selection of areas of $\mathcal{P}$.

We want to investigate the possible Voronoi diagrams $\mathcal{VD}(W\cup b_1)$ (in order to find the placement of $b_1$ so as to maximise $Area(V^+(b_1))$ which should give us an idea of how Black should play all of their points). To do this we aim to partition the arena into subsets within which the Voronoi diagram is structurally identical; that is, the vertices and line segments of the Voronoi diagram have the same algebraic representation in terms of the coordinates of $b_1$. We require this so that, once the algebraic representation of the area of $V^+(b_1)$ is found, we can maximise this over the partition to find the optimal placement of $b_1$ within that partition, thereby reducing Black's problem into many smaller, more manageable subproblems. Since $\mathcal{P}$ is rectangular and all of White's bisectors are vertical then, from \citeA{AveBerKalKra15}, the partitioning lines are simply the configuration lines of each of White's points.

The partition of the top right quadrant of a Voronoi cell of a general point $w_i \in W$ is shown in Figure \ref{fig:RowPartition}. Ignoring the bounding above and below of $\mathcal{P}$ (taking $q$ to be sufficiently large), notice that this partition is made up of configuration lines $\mathcal{CL}^1(w_j)$ for every $j \leq i$ and $\mathcal{CL}^3(w_k)$ for every $k>i$, creating exactly $n+1$ partition cells, irrespective of the value of $i$. For ease of computation we shall say $w_i=(0,0)$ and $b_1=(x,y)$.

\begin{figure}[!ht]
\centering
\includegraphics[width=1.0\textwidth]{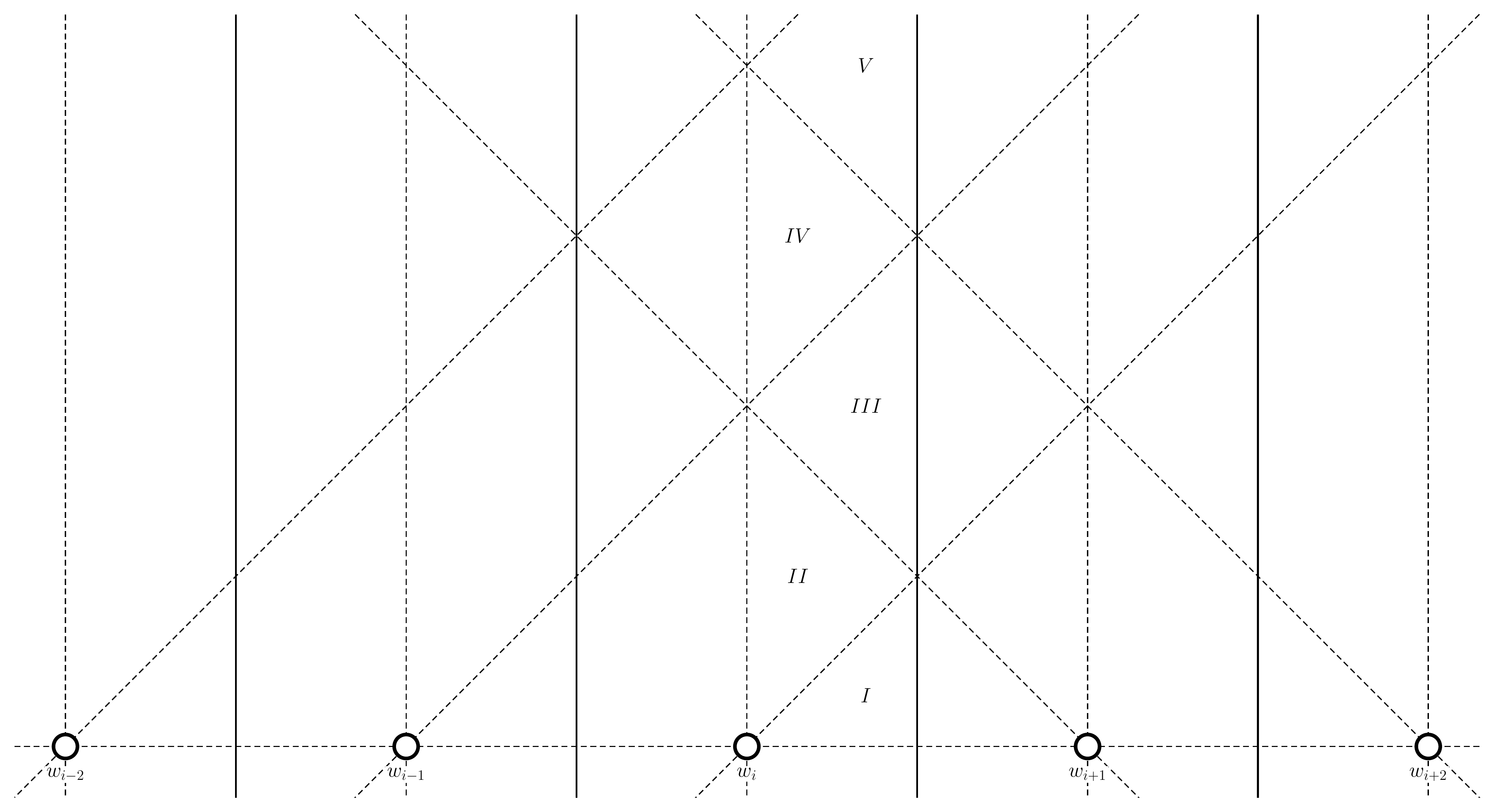}
\caption{The partition of $w_i$ in $\mathcal{P}$ in a boundless $\mathcal{P}$ of fixed width.}
\label{fig:RowPartition}
\end{figure}

Observing the cell structures of $V^+(b_1)$ for $b_1$ in the first few sections, as shown in Figure \ref{fig:RowCellExamples}, we can see the repetitive nature of these structures as each configuration line is crossed.

\begin{figure}[!ht]
\begin{subfigure}{.5\textwidth}
  \centering
  \includegraphics[width=0.9\textwidth]{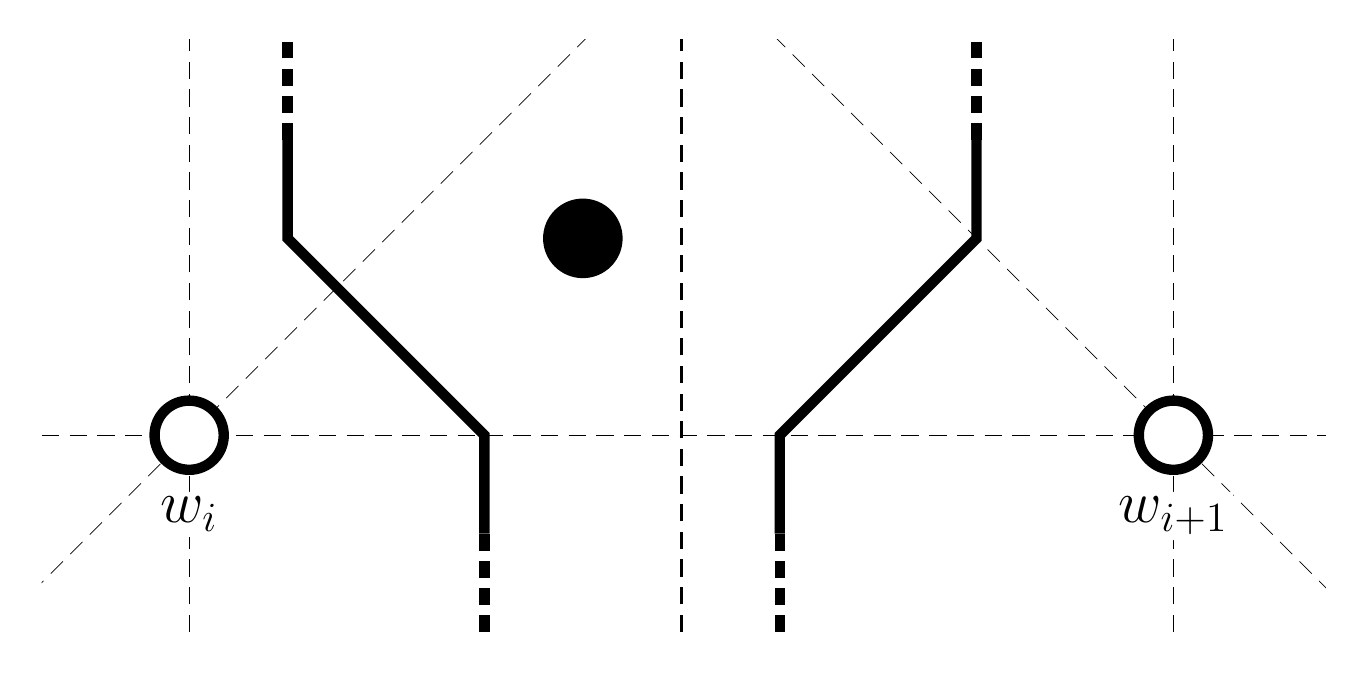}
  \caption{Voronoi cell $V^+(b_1)$ for $b_1$ in Section $I$.}
  \label{fig:RowCellI}
\end{subfigure}%
\begin{subfigure}{.5\textwidth}
  \centering
  \includegraphics[width=0.9\textwidth]{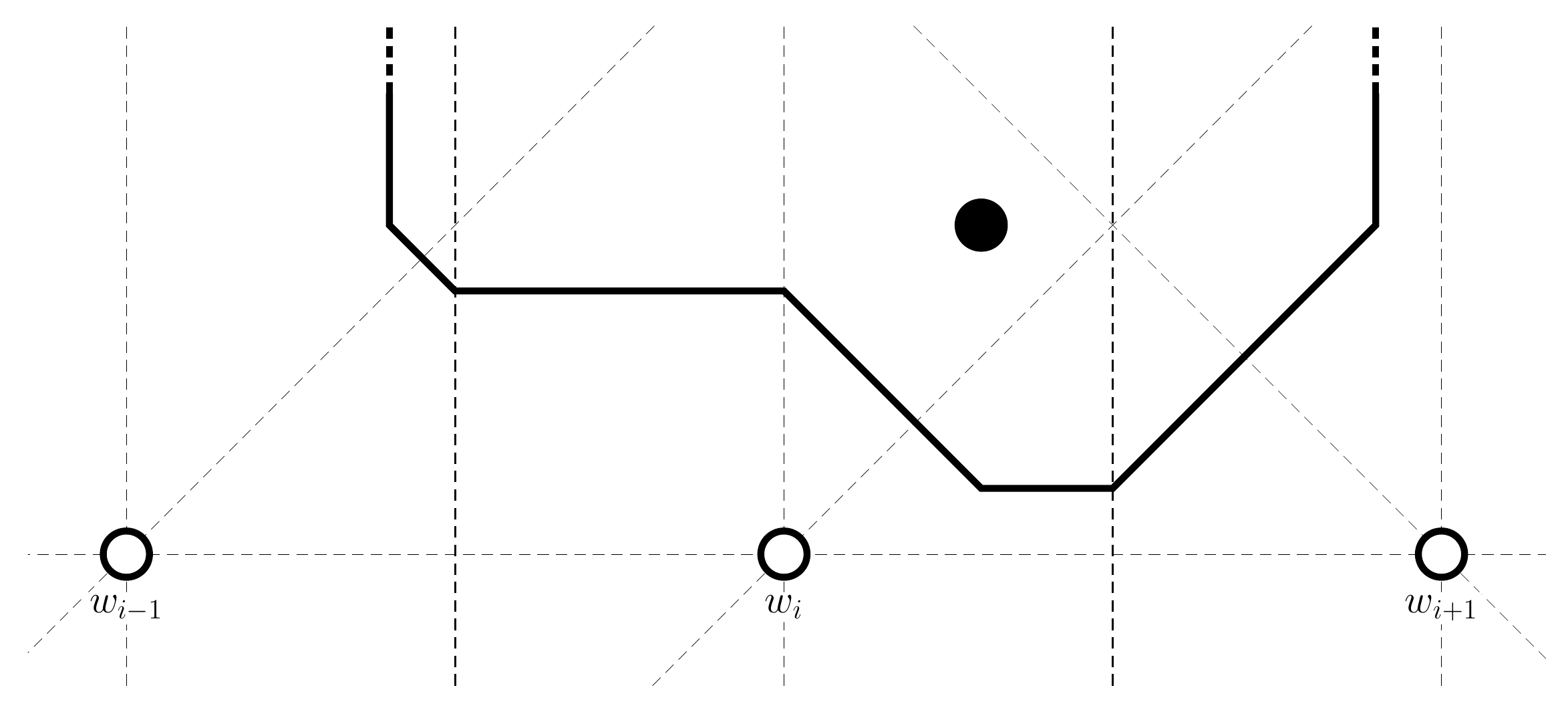}
  \caption{Voronoi cell $V^+(b_1)$ for $b_1$ in Section $II$.}
  \label{fig:RowCellII}
\end{subfigure}

\begin{subfigure}{.5\textwidth}
  \centering
  \includegraphics[width=0.9\textwidth]{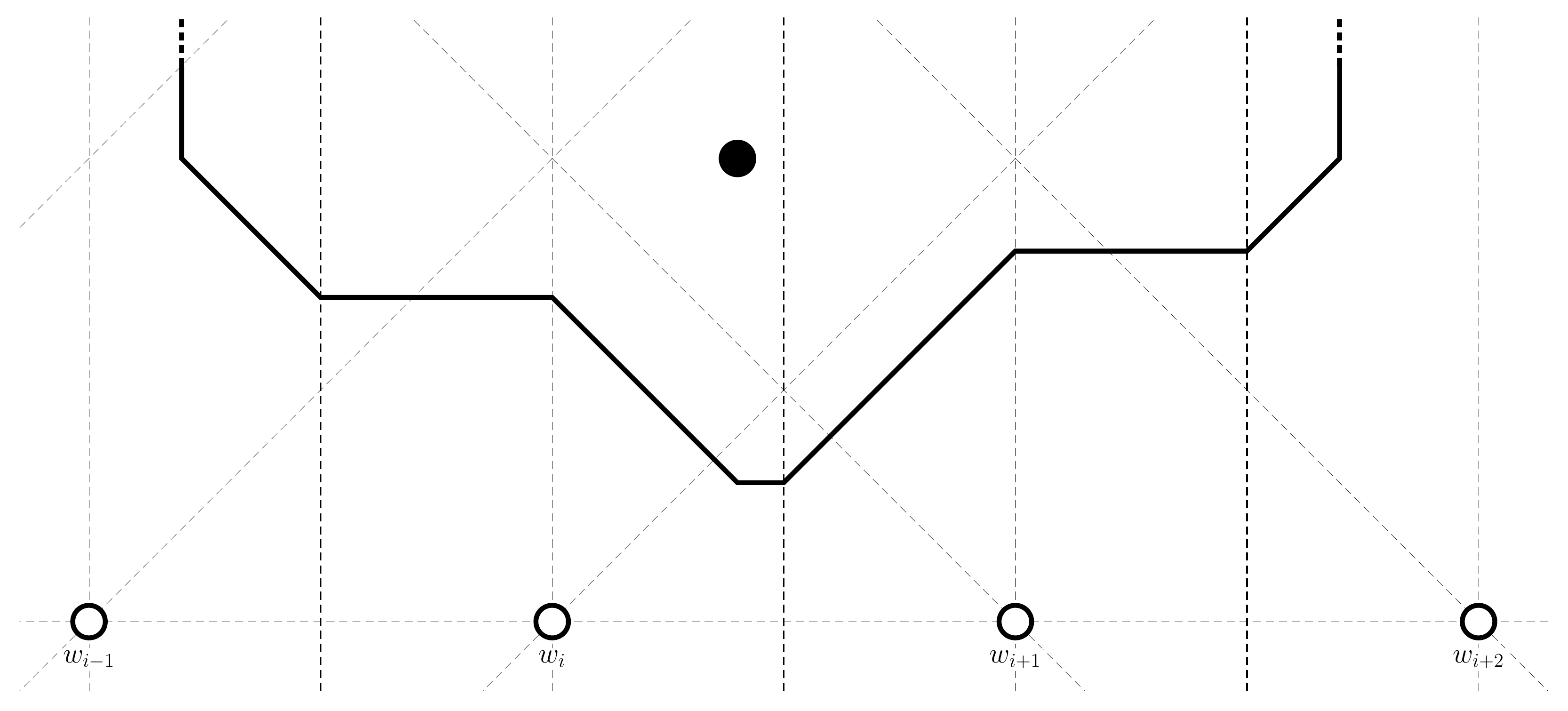}
  \caption{Voronoi cell $V^+(b_1)$ for $b_1$ in Section $III$.}
  \label{fig:RowCellIII}
\end{subfigure}%
\begin{subfigure}{.5\textwidth}
  \centering
  \includegraphics[width=0.9\textwidth]{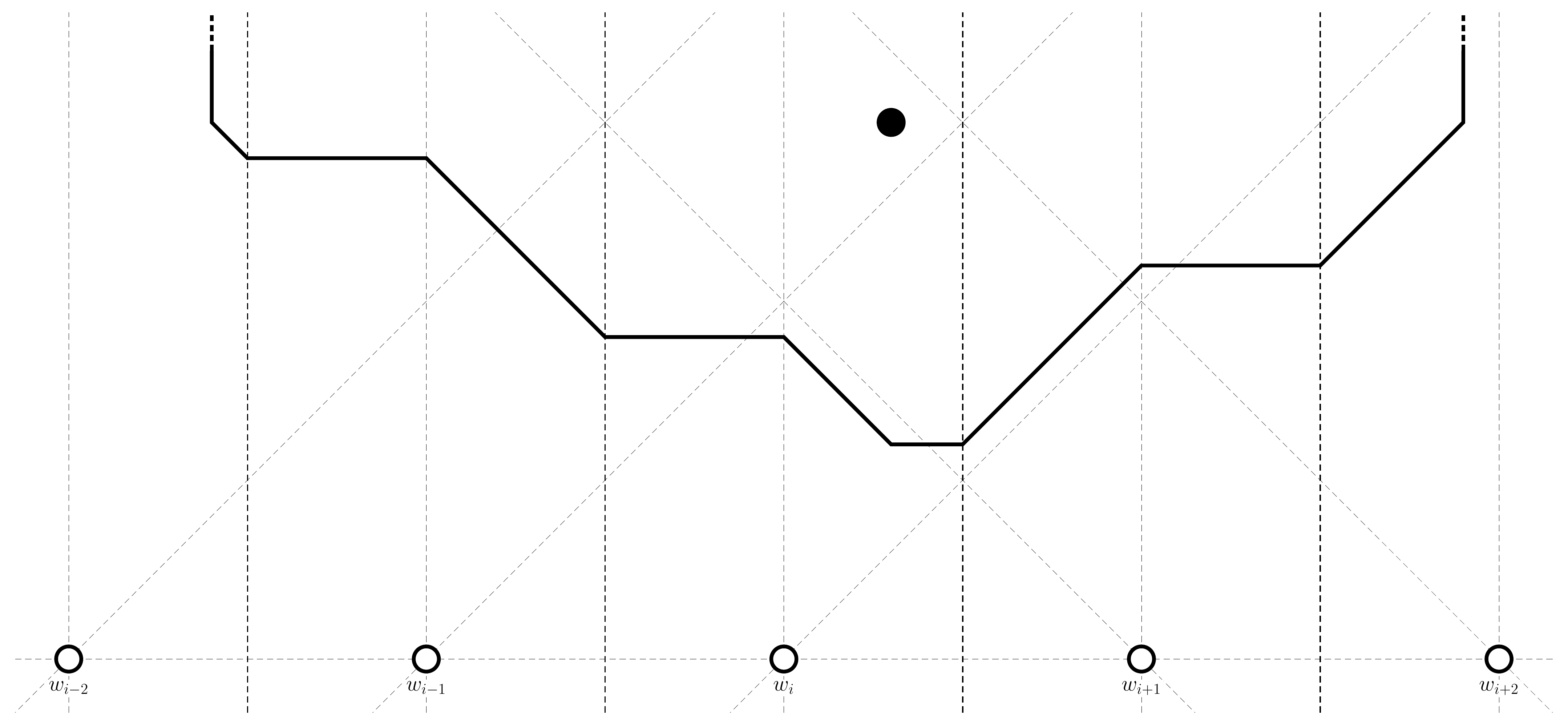}
  \caption{Voronoi cell $V^+(b_1)$ for $b_1$ in Section $IV$.}
  \label{fig:RowCellIV}
\end{subfigure}
\caption{Voronoi cells $V^+(b_1)$ for $b_1$ in respective sections according to Figure \ref{fig:RowPartition}.}
\label{fig:RowCellExamples}
\end{figure}

\pagebreak Let us first observe $V^+(b_1)$ when $b_1$ is in Section $I$ of Figure \ref{fig:RowPartition}. Assuming $i < n$, $V^+(b_1)$ has vertices $(\frac{x-y}{2},\frac{q}{2})$, $(\frac{x-y}{2},y)$, $(\frac{x+y}{2},0)$, $(\frac{x+y}{2},-\frac{q}{2})$, $(\frac{p}{2n}+\frac{x-y}{2},-\frac{q}{2})$, $(\frac{p}{2n}+\frac{x-y}{2},0)$, $(\frac{x+y}{2}+\frac{p}{2n},y)$, and $(\frac{x+y}{2}+\frac{p}{2n},\frac{q}{2})$ and so an area
\begin{equation*}
    \begin{split}
        Area(V^+(b_1)) &= (\frac{p}{2n}+\frac{x-y}{2} - \frac{x+y}{2}) \times q + 2 \times (y \times \frac{q}{2} - \frac{1}{2}y^2) \\
        &= \frac{pq}{2n} - y^2
    \end{split}
\end{equation*}
which is simply maximised at $(0,0)$. If $i=n$ then this instance of $V^+(b_1)$ will be maximised by also playing as close to $w_n$ as possible, stealing a total area bounded above by $\frac{pq}{2n}$.

\subsection{The encroachment of \texorpdfstring{$V^+(b_1)$}{V+(b1)} into \texorpdfstring{$V^\circ(w_j)$}{Vo(wj)}}

It may now seem a daunting task to work out the area of $V^+(b_1)$ for every possible placement of $b_1$. Instead, one more appealing approach would be to calculate the area that $b_1$ can steal from each Voronoi cell in $\mathcal{VD}(W)$, obtaining a formula based on the generator's location in relation to $w_i$ and to $b_1$. With this information we may be able to piece such areas together in order to obtain a general formula for the area of $V^+(b_1)$.

\paragraph{Theft from $V^\circ(w_i)$} Firstly, the area stolen from $V^\circ(w_i)$ takes two different forms depending on whether $b_1 \in \mathcal{CC}^1(w_i)$ or $b_1 \in \mathcal{CC}^2(w_i)$. If $b_1 \in \mathcal{CC}^1(w_i)$ then the area can only take the form that we have already explored in Section $I$ so we need not continue further along this avenue.
\iffalse
has vertices $(\frac{x+y}{2},-\frac{q}{2})$, $(\frac{x+y}{2},0)$, $(\frac{x-y}{2},y)$, $(\frac{x-y}{2},\frac{q}{2})$, $(\frac{p}{2n},\frac{q}{2})$, and $(\frac{p}{2n},-\frac{q}{2})$ and totals
\begin{equation*}
    \begin{split}
        Area(V^+(b_1) \cap V^\circ(w_i)) &= (\frac{p}{2n} - \frac{x+y}{2}) \times q + y \times  \frac{q}{2} - \frac{1}{2}y^2 \\
        &= - \frac{1}{2}y^2 - \frac{q}{2}x + \frac{pq}{2n}
    \end{split}
\end{equation*}
\fi
The area stolen from $V^\circ(w_i)$ if $b_1 \in \mathcal{CC}^2(w_i)$ has vertices $(-\frac{p}{2n},\frac{x+y}{2})$, $(0,\frac{x+y}{2})$, $(x,\frac{y-x}{2})$, $(\frac{p}{2n},\frac{y-x}{2})$, $(\frac{p}{2n},\frac{q}{2})$, and $(-\frac{p}{2n},\frac{q}{2})$ and totals
\begin{equation*}
    \begin{split}
        Area(V^+(b_1) \cap V^\circ(w_i)) &= \frac{p}{n} \times (\frac{q}{2} - \frac{x+y}{2}) + \frac{p}{2n} \times x - \frac{1}{2}x^2 \\
        &= - \frac{x^2}{2} - \frac{p}{2n}y + \frac{pq}{2n} \, .\\
    \end{split}
\end{equation*}

\paragraph{Theft from $V^\circ(w_j)$ for $j<i$} Next let us investigate what occurs when $b_1 \in \mathcal{CC}^1(w_j)$ for $j < i$. If there exists a $j$ such that $b_1 \in \mathcal{CC}^1(w_j) \setminus \mathcal{CC}^1(w_{j+1})$ then $V^+(b_1)$ enters $V^\circ(w_j)$. Therefore, writing $w_j$ as $(\frac{(j-i)p}{n},0)$, the area stolen from $V^\circ(w_j)$ if $b_1 \in \mathcal{CC}^1(w_j) \setminus \mathcal{CC}^1(w_{j+1})$ has vertices $(\frac{x-y}{2}+\frac{(j-i)p}{2n},y)$, $(\frac{(2(j-i)+1)p}{2n},\frac{x+y}{2}-\frac{(j-i+1)p}{2n})$, $(\frac{(2(j-i)+1)p}{2n},\frac{q}{2})$, and $(\frac{x-y}{2}+\frac{(j-i)p}{2n},\frac{q}{2})$ and totals
\begin{equation*}
    \begin{split}
        Area(V^+(b_1) \cap V^\circ(w_j)) &= \frac{1}{2} \times (\frac{(2(j-i)+1)p}{2n} - (\frac{x-y}{2}+\frac{(j-i)p}{2n})) \\ &\qquad\times (\frac{q}{2}-y + \frac{q}{2}- (\frac{x+y}{2}-\frac{(j-i+1)p}{2n})) \\
        &= \frac{1}{2} \times (\frac{(j-i+1)p}{2n} - \frac{x-y}{2}) \times (\frac{(j-i+1)p}{2n}+q - \frac{x+3y}{2}) \\
        &= \frac{1}{2} \times ((\frac{(j-i+1)p}{2n})^2 + \frac{(j-i+1)p}{2n} (q - \frac{x+3y}{2} - \frac{x-y}{2}) \\ &\qquad- \frac{x-y}{2}(q - \frac{x+3y}{2})) \\
        &= \frac{x^2}{8} - \frac{3y^2}{8} + \frac{xy}{4} + (\frac{(i-j-1)p}{4n} - \frac{q}{4})x + (\frac{(i-j-1)p}{4n} + \frac{q}{4})y \\ 
        &\qquad - \frac{(i-j-1)pq}{4n} + \frac{(i-j-1)^2p^2}{8n^2} \, . \\
        %(i^2 p^2)/(8 n^2) - (i j p^2)/(4 n^2) - (i p^2)/(4 n^2) - (i p q)/(4 n) + (i p x)/(4 n) + (i p y)/(4 n) + (j^2 p^2)/(8 n^2) + (j p^2)/(4 n^2) + (j p q)/(4 n) - (j p x)/(4 n) - (j p y)/(4 n) + p^2/(8 n^2) + (p q)/(4 n) - (p x)/(4 n) - (p y)/(4 n) - (q x)/4 + (q y)/4 + x^2/8 + (x y)/4 - (3 y^2)/8
        %(i^2 p^2 - 2 i j p^2 - 2 i n p q + 2 i n p x + 2 i n p y - 2 i p^2 + j^2 p^2 + 2 j n p q - 2 j n p x - 2 j n p y + 2 j p^2 - 2 n^2 q x + 2 n^2 q y + n^2 x^2 + 2 n^2 x y - 3 n^2 y^2 + 2 n p q - 2 n p x - 2 n p y + p^2)/(8 n^2)
        %CHECKED FOR x=0 and i-j-1=0 and y=p/n or =1 and =2p/n and x = y for i-j-1=0 and =1
    \end{split}
\end{equation*}

If $b_1 \in \mathcal{CC}^2(w_j)$ then $V^+(b_1)$ always steals from $V^\circ(w_j)$. This area stolen has vertices $(\frac{(2(j-i)-1)p}{2n},\frac{x+y}{2}-\frac{(j-i)p}{2n})$, $(\frac{(j-i)p}{n},\frac{x+y}{2}-\frac{(j-i)p}{2n})$, $(\frac{(2(j-i)+1)p}{2n},\frac{x+y}{2}-\frac{(j-i+1)p}{2n})$, $(\frac{(2(j-i)+1)p}{2n},\frac{q}{2})$, and $(\frac{(2(j-i)-1)p}{2n},\frac{q}{2})$ and totals

\begin{equation*}
    \begin{split}
        Area(V^+(b_1) \cap V^\circ(w_j)) &= \frac{p}{n} \times (\frac{q}{2} - (\frac{x+y}{2}-\frac{(j-i)p}{2n})) + \frac{1}{2}(\frac{p}{2n})^2 \\
        &= - \frac{p}{2n}x - \frac{p}{2n}y + \frac{pq}{2n} - \frac{(4(i-j)-1)p^2}{8n^2} \, .
        %VERIFIED TO BE TRUE
    \end{split}
\end{equation*}

\paragraph{Theft from $V^\circ(w_k)$ for $k>i$} Now moving our focus over to the Voronoi cells of $w_k$ for $k > i$, we have the analogous situations explored above. If there exists a $k$ such that $b_1 \in \mathcal{CC}^4(w_k) \setminus \mathcal{CC}^4(w_{k-1})$ then $V^+(b_1)$ enters $V^\circ(w_k)$. Therefore, writing $w_k$ as $(\frac{(k-i)p}{n},0)$, the area stolen from $V^\circ(w_k)$ if $b_1 \in \mathcal{CC}^4(w_k) \setminus \mathcal{CC}^4(w_{k-1})$ has vertices $(\frac{(2(k-i)-1)p}{2n},\frac{y-x}{2} + \frac{(k-i-1)p}{2n})$, $(\frac{x+y}{2} + \frac{(k-i)p}{2n},y)$, $(\frac{x+y}{2} + \frac{(k-i)p}{2n},\frac{q}{2})$, and $(\frac{(2(k-i)-1)p}{2n},\frac{q}{2})$ and totals
\begin{equation*}
    \begin{split}
        Area(V^+(b_1) \cap V^\circ(w_k)) &= \frac{1}{2} \times (\frac{x+y}{2} + \frac{(k-i)p}{2n} - \frac{(2(k-i)-1)p}{2n}) \\ &\qquad\times (\frac{q}{2} - (\frac{y-x}{2} + \frac{(k-i-1)p}{2n}) + \frac{q}{2} - y) \\
        &= \frac{1}{2} \times (\frac{x+y}{2} - \frac{(k-i-1)p}{2n}) \times (\frac{x-3y}{2} - \frac{(k-i-1)p}{2n} + q) \\
        &= \frac{x^2}{8} -\frac{3y^2}{8} -\frac{xy}{4} + (\frac{q}{4} - \frac{(k-i-1)p}{4n})x + (\frac{(k-i-1)p}{4n} + \frac{q}{4})y \\
        &\qquad - \frac{(k-i-1)pq}{4n} + \frac{(k-i-1)^2p^2}{8n^2} \, .
    \end{split}
\end{equation*}
We are comforted to see that this area is in fact identical, up to a reflection, to that for $b_1 \in \mathcal{CC}^2(w_j)$ where the axes have been reflecting in the $y$-axis (i.e. $x$ becomes $-x$ and $i-j$ becomes $k-i$).

And as before, if $b_1 \in \mathcal{CC}^3(w_k)$ then $V^+(b_1)$ always steals from $V^\circ(w_k)$. This area stolen has vertices $(\frac{(2(k-i)-1)p}{2n},\frac{y-x}{2} + \frac{(k-i-1)p}{2n})$, $(\frac{(k-i)p}{n},\frac{y-x}{2} + \frac{(k-i)p}{2n})$, $(\frac{(2(k-i)+1)p}{2n},\frac{y-x}{2} + \frac{(k-i)p}{2n})$, $(\frac{(2(k-i)+1)p}{2n},\frac{q}{2})$, and $(\frac{(2(k-i)-1)p}{2n},\frac{q}{2})$ and totals
\begin{equation*}
    \begin{split}
        Area(V^+(b_1) \cap V^\circ(w_k)) &= \frac{p}{n} \times (\frac{q}{2} - (\frac{y-x}{2} + \frac{(k-i)p}{2n})) + \frac{1}{2}(\frac{p}{2n})^2 \\
        &= \frac{p}{2n}x - \frac{p}{2n}y + \frac{pq}{2n} - \frac{(4(k-i)-1)p^2}{8n^2} \, .
    \end{split}
\end{equation*}
This is again identical to the area found for $b_1 \in \mathcal{CC}^2(w_j)$ after the reflection described previously.

We have now found all formulae for the area of $V^+(b_1)$ contained in each Voronoi cell of $\mathcal{VD}(W)$ when White plays a row. From these we can derive the area for a general cell $V^+(b_1)$ where $b_1 \in V^\circ(w_i)$ for some $i$, and find the optimal solution within each of the partition cells that produce such a structure of $V^+(b_1)$. Figure \ref{fig:RowOptimals} will depict all optimal locations of $b_1$ within each section under the certain circumstances we will discuss below unless optima have location $(0,0)$, a placement already described in Lemma \ref{steal}; %4.1.1;
Section $IV$ and Section $III$ are depicted as the poster children for the general Section $2l$ and Section $2l+1$ results respectively, and for clarity these respective sections will be shaded in each figure.

\subsection{\texorpdfstring{$V^+(b_1)$}{V+(b1)} not touching the vertical edges of \texorpdfstring{$\mathcal{P}$}{P}}

Since we have already explored Section $I$, we will look only at $b_1 \in \mathcal{CC}^2(w_i)$. Firstly, ignoring intersections with the vertical boundaries of $\mathcal{P}$, we can see from Figure \ref{fig:RowCellExamples} that the left and right ends of $V^+(b_1)$ always have the same structure. This is because there is always a $j$ such that $b_1 \in \mathcal{CC}^1(w_j) \setminus \mathcal{CC}^1(w_{j+1})$ and similarly always a $k$ such that $b_1 \in \mathcal{CC}^4(w_k) \setminus \mathcal{CC}^4(w_{k-1})$. Furthermore, viewing $j$ and $k$ as points $l$ away from $i$ we can write each partition cell in Figure~\ref{fig:RowPartition} as either $(\mathcal{CC}^1(w_{i-l}) \setminus \mathcal{CC}^1(w_{i-l+1})) \cap (\mathcal{CC}^4(w_{i+l}) \setminus \mathcal{CC}^4(w_{i+l-1}))$ or  $(\mathcal{CC}^1(w_{i-l}) \setminus \mathcal{CC}^1(w_{i-l+1})) \cap (\mathcal{CC}^4(w_{i+l+1}) \setminus \mathcal{CC}^4(w_{i+l}))$ (exploring the top right quadrant of $V^\circ(w_i)$ means we may interact with $w_j$ either for all $j=i-l, \dots, i+l$ or for all $j=i-l, \dots, i+l+1$).

\paragraph{Section $2l$} Therefore we have one of two area formulae, the first being for $b_1 \in (\mathcal{CC}^1(w_{i-l}) \setminus \mathcal{CC}^1(w_{i-l+1})) \cap (\mathcal{CC}^4(w_{i+l}) \setminus \mathcal{CC}^4(w_{i+l-1})) = \mathcal{CC}^1(w_{i-l}) \bigcap_{j=i-l+1}^{i}{\mathcal{CC}^2(w_j)} \bigcap_{j=i+1}^{i+l-1}{\mathcal{CC}^3(w_j)} \cap \mathcal{CC}^4(w_{i+l})$ for $l \in \mathbb{N}$ (this would be Section $2l$ in Figure \ref{fig:RowPartition}) with
\begin{equation*}
    \begin{split}
        Area(V^+(b_1)) &= Area(V^+(b_1) \cap V^\circ(w_{i-l})) + \sum_{j=i-l+1}^{i-1}Area(V^+(b_1) \cap V^\circ(w_j)) \\ &\qquad+ Area(V^+(b_1) \cap V^\circ(w_i)) + \sum_{j=i+1}^{i+l-1}Area(V^+(b_1) \cap V^\circ(w_j)) \\ &\qquad+ Area(V^+(b_1) \cap V^\circ(w_{i+l})) \\
        &= \frac{x^2}{8} - \frac{3y^2}{8} + \frac{xy}{4} + (\frac{(i-(i-l)-1)p}{4n} - \frac{q}{4})x + (\frac{(i-(i-l)-1)p}{4n} + \frac{q}{4})y \\ 
        &\qquad - \frac{(i-(i-l)-1)pq}{4n} + \frac{(i-(i-l)-1)^2p^2}{8n^2} \\ %CC1(w_{(i-l)}):: ((l - 1)^2 p^2)/(8 n^2) + x (((l - 1) p)/(4 n) - q/4) + y (((l - 1) p)/(4 n) + q/4) - ((l - 1) p q)/(4 n) + x^2/8 + (x y)/4 - (3 y^2)/8
        &\qquad+ \sum_{j=i-l+1}^{i-1}(- \frac{p}{2n}x - \frac{p}{2n}y + \frac{pq}{2n} - \frac{(4(i-j)-1)p^2}{8n^2}) %CC2(w_j)
        - \frac{x^2}{2} - \frac{p}{2n}y + \frac{pq}{2n} \\ %CC2(w_i):: (p q)/(2 n) - (p y)/(2 n) - x^2/2
        &\qquad+ \sum_{j=i+1}^{i+l-1}(\frac{p}{2n}x - \frac{p}{2n}y + \frac{pq}{2n} - \frac{(4(j-i)-1)p^2}{8n^2}) %CC3(w-j)
        + \frac{x^2}{8} -\frac{xy}{4} -\frac{3y^2}{8} \\
        &\qquad+ (\frac{q}{4} - \frac{((i+l)-i-1)p}{4n})x + (\frac{((i+l)-i-1)p}{4n} + \frac{q}{4})y \\
        &\qquad- \frac{((i+l)-i-1)pq}{4n} + \frac{((i+l)-i-1)^2p^2}{8n^2} \\ %CC4(w_{(i+l)}):: ((l - 1)^2 p^2)/(8 n^2) + x (q/4 - ((l - 1) p)/(4 n)) + y (((l - 1) p)/(4 n) + q/4) - ((l - 1) p q)/(4 n) + x^2/8 - (x y)/4 - (3 y^2)/8
        % w_{i-l} + w_{i+l} = (l^2 p^2)/(8 n^2) + ((l - 1)^2 p^2)/(8 n^2) - (l p^2)/(4 n^2) + (p q)/(2 n) - (l p q)/(2 n) + (l p y)/(2 n) + p^2/(8 n^2) - (p y)/(2 n) + (q y)/2 + x^2/4 - (3 y^2)/4
        %all outside sums = ((l - 1)^2 p^2)/(4 n^2) - (l p q)/(2 n) + (l p y)/(2 n) + (p q)/n - (p y)/n + (q y)/2 - x^2/4 - (3 y^2)/4
        &= - \frac{x^2}{4} - \frac{3y^2}{4} + (\frac{(l-2)p}{2n} + \frac{q}{2})y - \frac{(l-2)pq}{2n} + \frac{(l-1)^2p^2}{4n^2} \\
        &\qquad+ (- \frac{p}{2n}x - \frac{p}{2n}y + \frac{pq}{2n} + \frac{p^2}{8n^2})\times((i-1) - (i-l+1-1)) - \sum_{j=i-l+1}^{i-1}(\frac{4(i-j)p^2}{8n^2}) \\
        &\qquad+ (\frac{p}{2n}x - \frac{p}{2n}y + \frac{pq}{2n} + \frac{p^2}{8n^2})\times((i+l-1)-(i+1-1)) - \sum_{j=i+1}^{i+l-1}(\frac{4(j-i)p^2}{8n^2}) \\
        &= - \frac{x^2}{4} - \frac{3y^2}{4} + (\frac{(l-2)p}{2n} + \frac{q}{2})y - \frac{(l-2)pq}{2n} + \frac{(l-1)^2p^2}{4n^2} \\
        &\qquad+ (l-1)(- \frac{p}{n}y + \frac{pq}{n} + \frac{p^2}{4n^2}) - \frac{p^2}{2n^2}(\sum_{i-j=1}^{l-1}(i-j) + \sum_{j-i=1}^{l-1}(j-i)) \\
        %without sum bits (l^2 p^2)/(4 n^2) - (l p^2)/(4 n^2) + (l p q)/(2 n) - (l p y)/(2 n) + (q y)/2 - x^2/4 - (3 y^2)/4
        &= - \frac{x^2}{4} - \frac{3y^2}{4} + (-\frac{lp}{2n} + \frac{q}{2})y + \frac{lpq}{2n} - \frac{(l-1)lp^2}{4n^2} \, .
        %-(l^2 p^2)/(4 n^2) + (l p^2)/(4 n^2) + (l p q)/(2 n) - (l p y)/(2 n) + (q y)/2 - x^2/4 - (3 y^2)/4
    \end{split}
\end{equation*}
This area has partial derivatives
\begin{equation*}
\begin{split}
\frac{\delta A}{\delta x}&= - \frac{x}{2} \\
\frac{\delta A}{\delta y}&= - \frac{3y}{2} -\frac{lp}{2n} + \frac{q}{2}
\end{split}
\end{equation*}
giving the optimal value $b^*_1=(0,\frac{q}{3}-\frac{lp}{3n})$ with $Area(V^+(b_1^*))= \frac{lpq}{3n} + \frac{(3-2l)lp^2}{12n^2} + \frac{q^2}{12}$. This is depicted in Figure \ref{fig:RowOptimalNotTouchingIVb} for $l=2$. For $b^*_1$ to lie within Section $2l$ we must have $x^* + (l-1)\frac{p}{n} \leq y^* \leq l\frac{p}{n} - x^*$ so it must be the case that $\frac{(4l-3)p}{n} \leq q \leq \frac{4lp}{n}$.

If $\frac{4lp}{n} \leq q$ then the optimum must lie at the intersection of $x=0$ and $y=l\frac{p}{n} - x$ (since $\frac{\delta A}{\delta x}= - \frac{x}{2}$, the area will always increase as $x$ moves towards $0$ and since $\frac{4lp}{n} \leq q$ the global optimum lies above Section $2l$). Therefore the optimum in this section is $b^*_1=(0,l\frac{p}{n})$ achieving $Area(V^+(b_1))=\frac{lpq}{n} + \frac{(1-6l)lp^2}{4n^2}$. This is depicted in Figure~\ref{fig:RowOptimalNotTouchingIVa}.

Alternatively, if $\frac{(4l-3)p}{n} \geq q$ then the optimum must lie at the intersection of $x=0$ and $y=x+(l-1)\frac{p}{n}$ (since $\frac{\delta A}{\delta x}= - \frac{x}{2}$, the area will always increase as $x$ moves towards $0$ and since $\frac{(4l-3)p}{n} \geq q$ the global optimum lies below Section $2l$). Therefore the optimum in this section is $b^*_1=(0,(l-1)\frac{p}{n})$ achieving $Area(V^+(b_1))=\frac{(2l-1)pq}{2n} - \frac{3(2l-1)(l-1)p^2}{4n^2}$. This is depicted in Figure~\ref{fig:RowOptimalNotTouchingIVc}.

\begin{figure}[!ht]
\begin{subfigure}{.5\textwidth}
  \centering
  \includegraphics[width=0.9\textwidth]{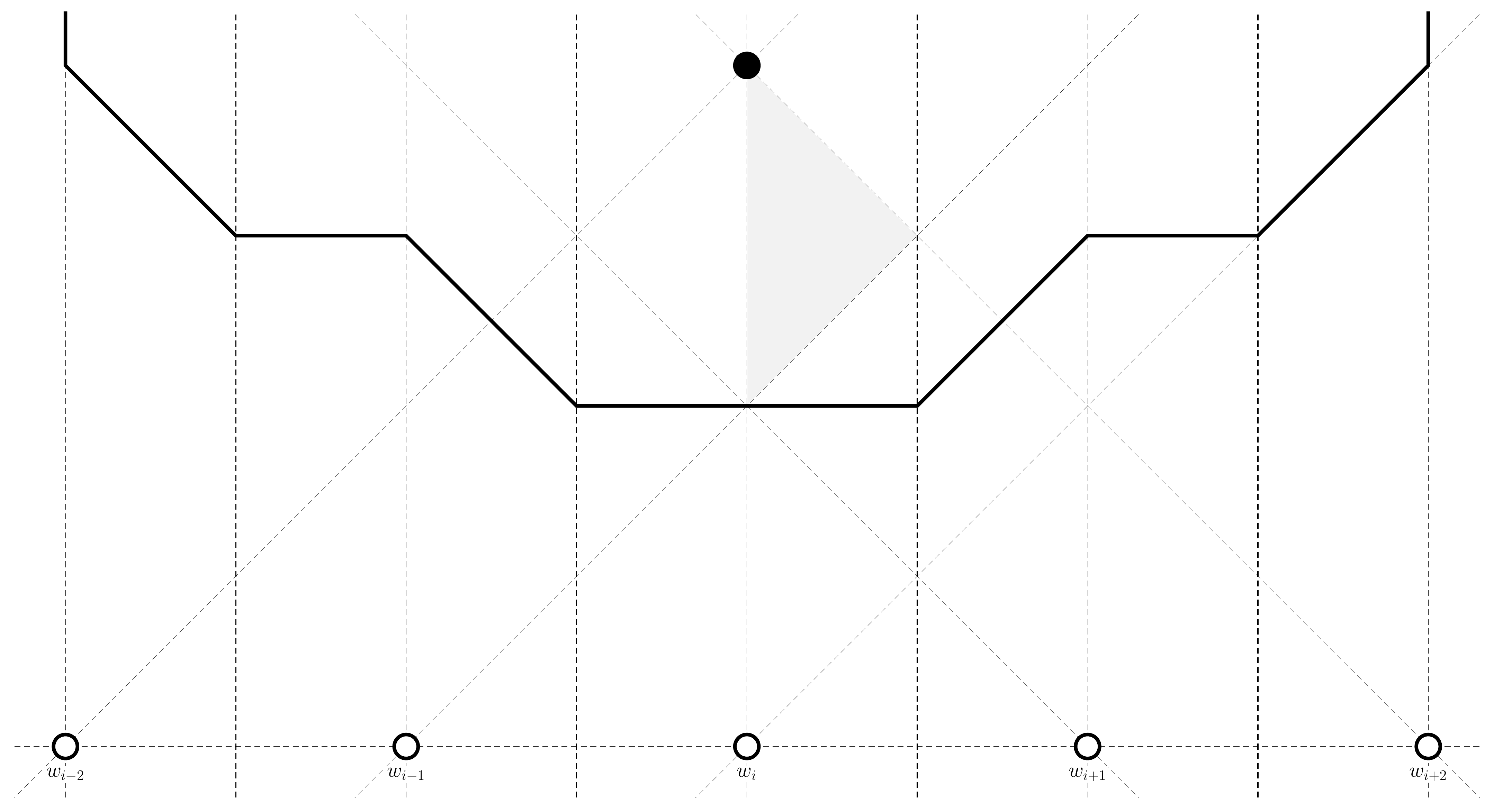}
  \caption{$b_1=(0,l\frac{p}{n})$ only if $\frac{4lp}{n} \leq q$.}
%  \caption{$Area(V^+((0,l\frac{p}{n})))=\frac{lpq}{n} + \frac{(1-6l)lp^2}{4n^2}$ \\only if $\frac{4lp}{n} \leq q$.}
  \label{fig:RowOptimalNotTouchingIVa}
\end{subfigure}%
\begin{subfigure}{.5\textwidth}
  \centering
  \includegraphics[width=0.9\textwidth]{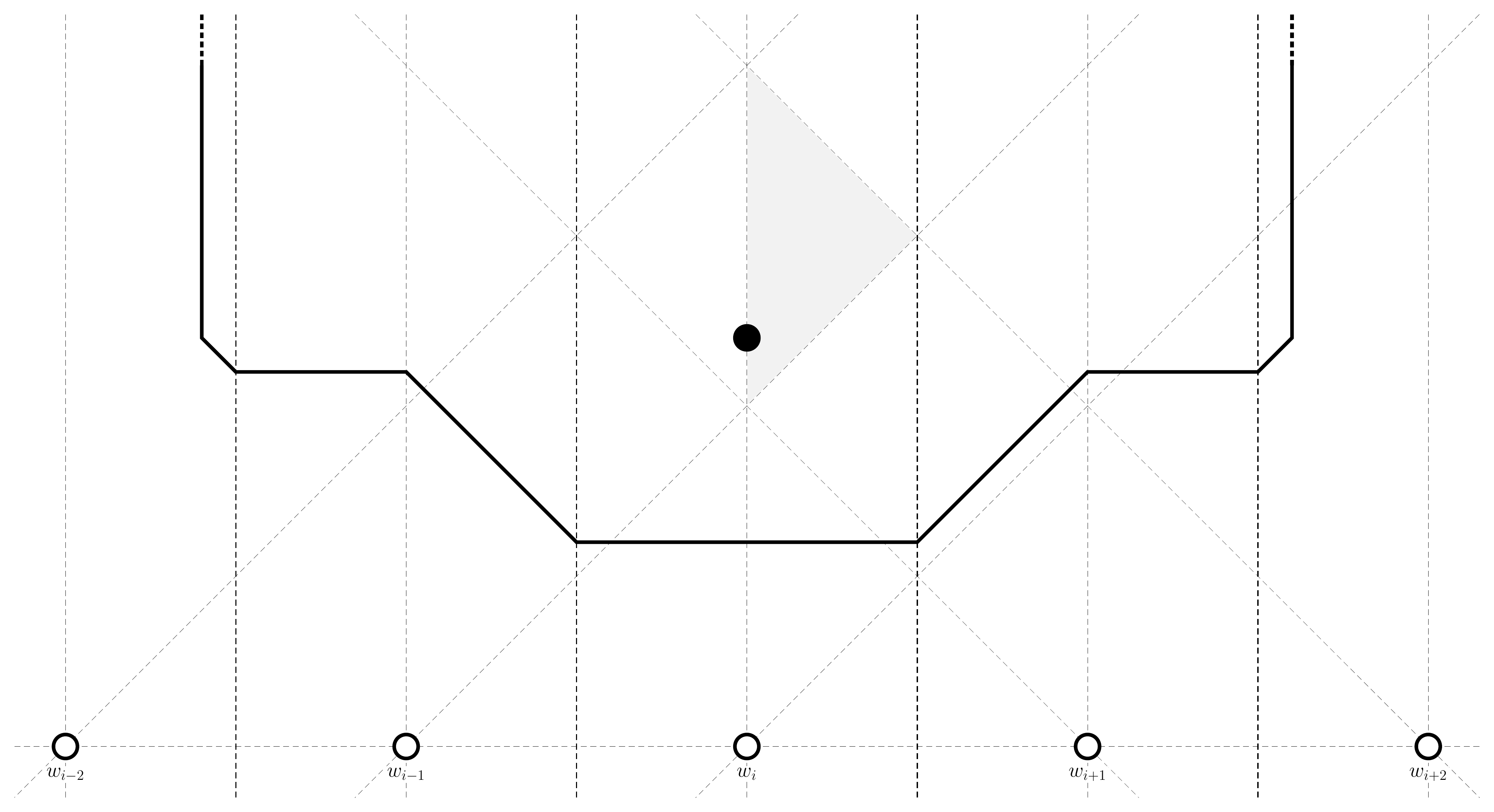}
  \caption{$b_1=(0,\frac{q}{3}-\frac{lp}{3n})$ only if $\frac{(4l-3)p}{n} \leq q \leq \frac{4lp}{n}$.}
%  \caption{$Area(V^+((0,\frac{q}{3}-\frac{lp}{3n})))=\frac{lpq}{3n} + \frac{(3-2l)lp^2}{12n^2} + \frac{q^2}{12}$ \\only if $\frac{(4l-3)p}{n} \leq q \leq \frac{4lp}{n}$.}
  \label{fig:RowOptimalNotTouchingIVb}
\end{subfigure}

\begin{subfigure}{1.0\textwidth}
  \centering
  \includegraphics[width=0.45\textwidth]{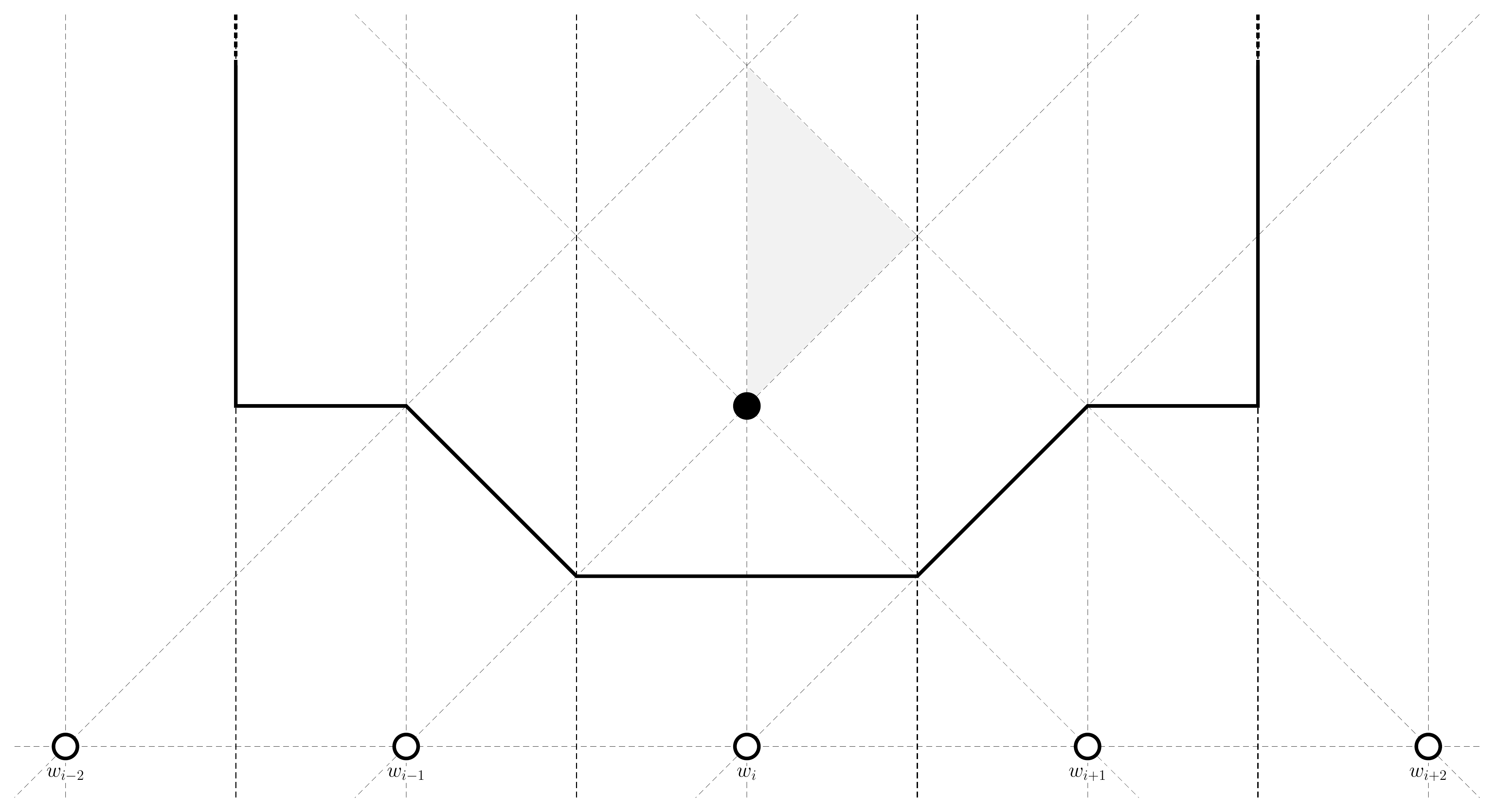}
  \caption{$b_1=(0,(l-1)\frac{p}{n})$ only if $\frac{(4l-3)p}{n} \geq q$.}
%  \caption{$Area(V^+((0,(l-1)\frac{p}{n}))) = \frac{(2l-1)pq}{2n} - \frac{3(2l-1)(l-1)p^2}{4n^2}$ \\only if $\frac{(4l-3)p}{n} \geq q$.}
  \label{fig:RowOptimalNotTouchingIVc}
\end{subfigure}
\caption{Maximal area Voronoi cells $V^+(b_1)$ for $b_1$ within Section $2l$ not touching the vertical edges of $\mathcal{P}$.}
\end{figure}

\paragraph{Section $2l+1$} The other area formula for $b_1 \in (\mathcal{CC}^1(w_{i-l}) \setminus \mathcal{CC}^1(w_{i-l+1})) \cap (\mathcal{CC}^4(w_{i+l+1}) \setminus \mathcal{CC}^4(w_{i+l})) = \mathcal{CC}^1(w_{i-l}) \bigcap_{j=i-l+1}^{i}{\mathcal{CC}^2(w_j)} \bigcap_{j=i+1}^{i+l}{\mathcal{CC}^3(w_j)} \cap \mathcal{CC}^4(w_{i+l+1})$ for $l \in \mathbb{N}$ (this would be Section $2l+1$ in Figure \ref{fig:RowPartition}) is, adapting from the formula found for Section $2l$,
\begin{equation*}
    \begin{split}
        Area(V^+(b_1)) &= Area(V^+(b_1) \cap V^\circ(w_{i-l})) + \sum_{j=i-l+1}^{i-1}Area(V^+(b_1) \cap V^\circ(w_j)) \\ &\qquad+ Area(V^+(b_1) \cap V^\circ(w_i)) + \sum_{j=i+1}^{i+l}Area(V^+(b_1) \cap V^\circ(w_j)) \\ &\qquad+ Area(V^+(b_1) \cap V^\circ(w_{i+l+1})) \\
        &=  - \frac{x^2}{4} - \frac{3y^2}{4} + (-\frac{lp}{2n} + \frac{q}{2})y + \frac{lpq}{2n} - \frac{(l-1)lp^2}{4n^2} %Previous calculation
        - \left(\frac{x^2}{8} -\frac{3y^2}{8} -\frac{xy}{4} \right. \\
        &\qquad+ \left. (\frac{q}{4} - \frac{((i+l)-i-1)p}{4n})x + (\frac{((i+l)-i-1)p}{4n} + \frac{q}{4})y - \frac{((i+l)-i-1)pq}{4n} \right. \\
        &\qquad+ \left. \frac{((i+l)-i-1)^2p^2}{8n^2}\right) %CC4(w_{i+l})
        + \frac{p}{2n}x - \frac{p}{2n}y + \frac{pq}{2n} - \frac{(4((i+l)-i)-1)p^2}{8n^2} \\ %CC3(w_{i+l})
        &\qquad+ \frac{x^2}{8} -\frac{3y^2}{8} -\frac{xy}{4} + (\frac{q}{4} - \frac{((i+l+1)-i-1)p}{4n})x + (\frac{((i+l+1)-i-1)p}{4n} + \frac{q}{4})y \\
        &\qquad- \frac{((i+l+1)-i-1)pq}{4n} + \frac{((i+l+1)-i-1)^2p^2}{8n^2} \\ %CC4(w_{i+l+1})
    \end{split}
\end{equation*}

\begin{equation*}
\begin{split}
      \qquad  &=  - \frac{x^2}{4} - \frac{3y^2}{4} + (-\frac{lp}{2n} + \frac{q}{2})y + \frac{lpq}{2n} - \frac{(l-1)lp^2}{4n^2} \\
      &\qquad- \left(\frac{p}{4n}x - \frac{p}{4n}y + \frac{pq}{4n} + \frac{(-2l+1)p^2}{8n^2}\right) + \frac{p}{2n}x - \frac{p}{2n}y + \frac{pq}{2n} - \frac{(4l-1)p^2}{8n^2} \\
        &= - \frac{x^2}{4} - \frac{3y^2}{4} + \frac{p}{4n}x + (-\frac{(2l+1)p}{4n} + \frac{q}{2})y + \frac{(2l+1)pq}{4n} - \frac{l(l+1)p^2}{4n^2} \, .
        %-(l (l + 1) p^2)/(4 n^2) + y (q/2 - ((2 l + 1) p)/(4 n)) + ((2 l + 1) p q)/(4 n) + (p x)/(4 n) - x^2/4 - (3 y^2)/4
\end{split}
\end{equation*}
This area has partial derivatives
\begin{equation*}
\begin{split}
\frac{\delta A}{\delta x}&= - \frac{x}{2} + \frac{p}{4n} \\
\frac{\delta A}{\delta y}&= - \frac{3y}{2} -\frac{(2l+1)p}{4n} + \frac{q}{2}
\end{split}
\end{equation*}
giving the optimal value $b^*_1=(\frac{p}{2n},\frac{q}{3}-\frac{(2l+1)p}{6n})$ and $Area(V^+((\frac{p}{2n},\frac{q}{3}-\frac{(2l+1)p}{6n})))=\frac{(2l+1)pq}{6n} - \frac{(2l^2+2l-1)p^2}{12n^2} + \frac{q^2}{12}$. For $b^*_1$ to lie within Section $2l$ we must have $l\frac{p}{n} - x^* \leq y^* \leq x^* + l\frac{p}{n}$ so it must be the case that $\frac{(4l-1)p}{n} \leq q \leq \frac{(4l+2)p}{n}$. This is depicted in Figure~\ref{fig:RowOptimalNotTouchingIIIb}.

If $\frac{(4l+2)p}{n} \leq q$ then the optimum must lie at the intersection of $x=\frac{p}{2n}$ and $y=x + l\frac{p}{n}$ (since $x^*=\frac{p}{2n}$ does not restrict the values of $y$ over Section $2l+1$ and since $\frac{(4l-1)p}{n} \geq q$ the global optimum lies above Section $2l+1$). Therefore the optimum in this section is $b^*_1=(\frac{p}{2n},\frac{(2l+1)p}{2n})$ achieving $Area(V^+(b_1))= \frac{(2l+1)pq}{2n}-\frac{(6l^2+6l+1)p^2}{4n^2}$. This is depicted in Figure~\ref{fig:RowOptimalNotTouchingIIIa}.

Alternatively, if $\frac{(4l-1)p}{n} \geq q$ then the optimum must lie at the intersection of $x=\frac{p}{2n}$ and $y=l\frac{p}{n} - x$ (since $x^*=\frac{p}{2n}$ does not restrict the values of $y$ over Section $2l+1$ and since $\frac{(4l-1)p}{n} \geq q$ the global optimum lies below Section $2l+1$). Therefore the optimum in this section is $b^*_1=(\frac{p}{2n},\frac{(2l-1)p}{2n})$ achieving $Area(V^+(b_1))=\frac{lpq}{n} - \frac{(3l-1)lp^2}{2n^2}$. This is depicted in Figure~\ref{fig:RowOptimalNotTouchingIIIc}.

\begin{figure}[!ht]\ContinuedFloat
\begin{subfigure}{.5\textwidth}
  \centering
  \includegraphics[width=0.9\textwidth]{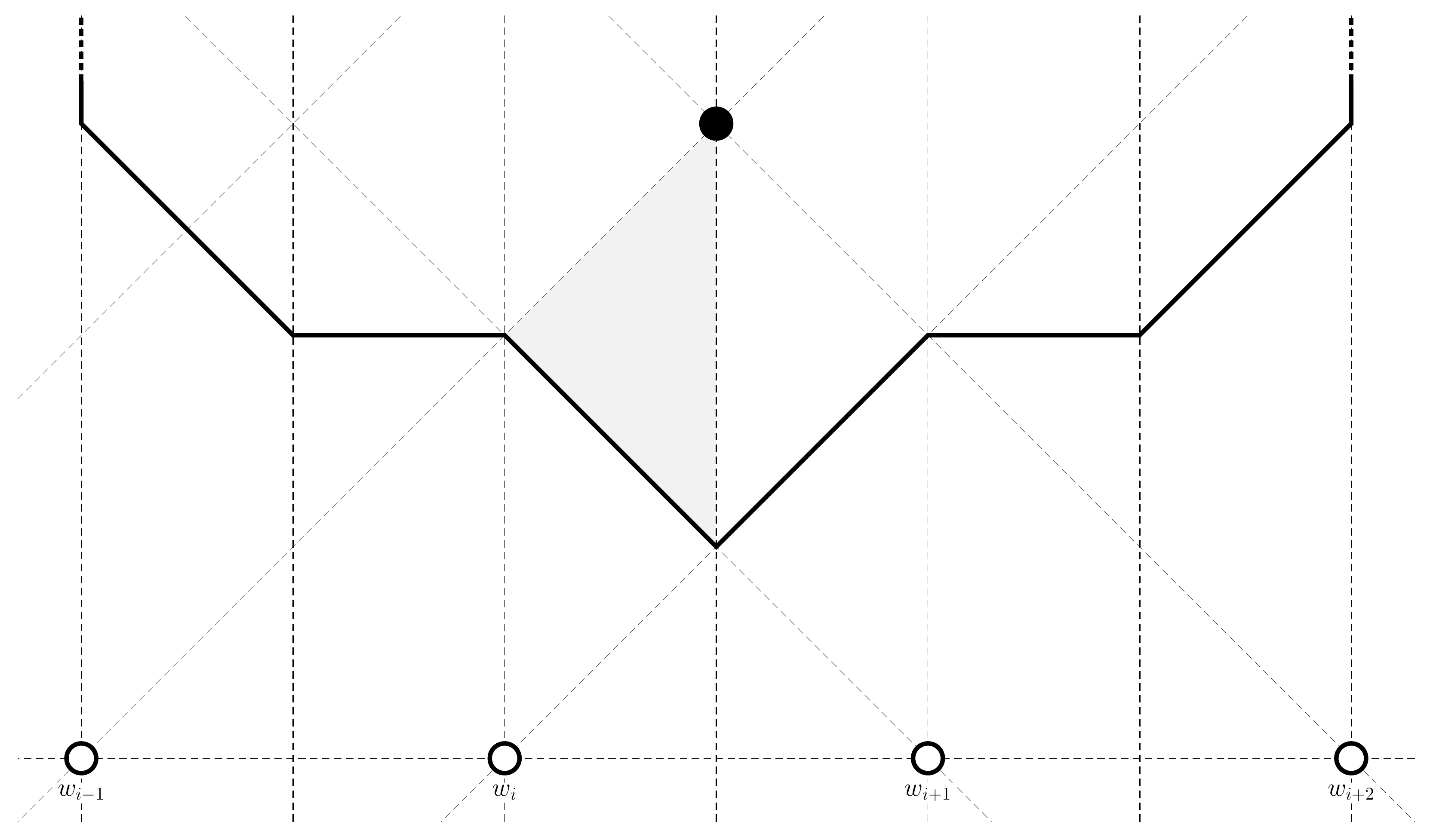}
  \caption{$b_1=(\frac{p}{2n},\frac{(2l+1)p}{2n})$ only if $\frac{(4l+2)p}{n} \leq q$. \\ \,}
%  \caption{$Area(V^+((\frac{p}{2n},\frac{(2l+1)p}{2n})))=\frac{(2l+1)pq}{2n}-\frac{(6l^2+6l+1)p^2}{4n^2}$ \\only if $\frac{(4l+2)p}{n} \leq q$.}
  \label{fig:RowOptimalNotTouchingIIIa}
\end{subfigure}%
\begin{subfigure}{.5\textwidth}
  \centering
  \includegraphics[width=0.9\textwidth]{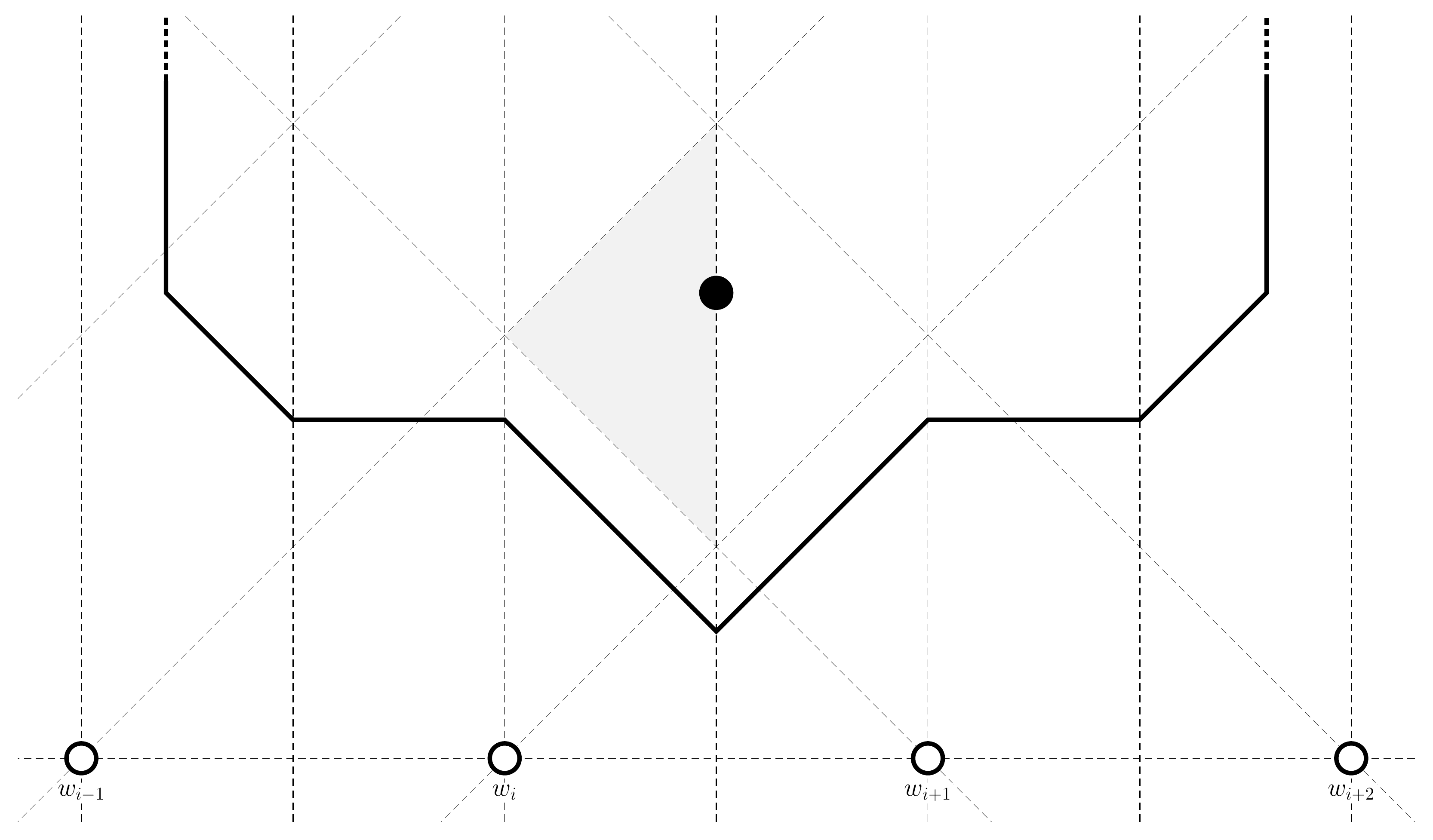}
  \caption{$b_1=(\frac{p}{2n},\frac{q}{3}-\frac{(2l+1)p}{6n})$ only if  $\frac{(4l-1)p}{n} \leq q \leq \frac{(4l+2)p}{n}$.}
%  \caption{$Area(V^+((\frac{p}{2n},\frac{q}{3}-\frac{(2l+1)p}{6n})))=\frac{(2l+1)pq}{6n} - \frac{(2l^2+2l-1)p^2}{12n^2} + \frac{q^2}{12}$ \\only if $\frac{(4l-1)p}{n} \leq q \leq \frac{(4l+2)p}{n}$.}
  \label{fig:RowOptimalNotTouchingIIIb}
\end{subfigure}

\begin{subfigure}{1.0\textwidth}
  \centering
  \includegraphics[width=0.45\textwidth]{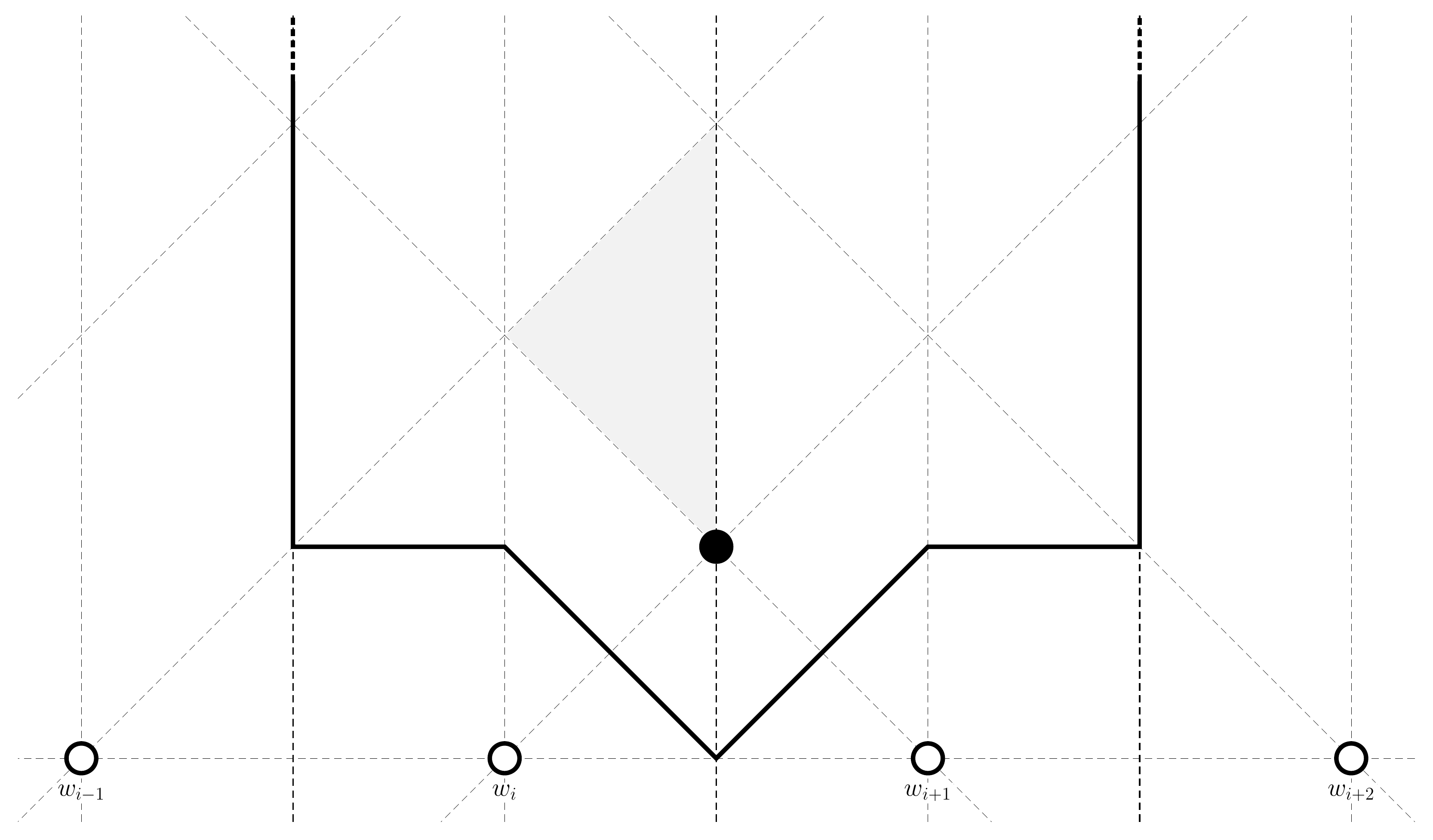}
  \caption{$b_1=(\frac{p}{2n},\frac{(2l-1)p}{2n})$ only if $\frac{(4l-1)p}{n} \geq q$.}
%  \caption{$Area(V^+((\frac{p}{2n},\frac{(2l-1)p}{2n})))=\frac{lpq}{n} - \frac{(3l-1)lp^2}{2n^2}$ \\only if $\frac{(4l-1)p}{n} \geq q$.}
  \label{fig:RowOptimalNotTouchingIIIc}
\end{subfigure}
\caption{Maximal area Voronoi cells $V^+(b_1)$ for $b_1$ within Section $2l+1$ not touching the vertical edges of $\mathcal{P}$.}
\end{figure}

\pagebreak
\subsection{\texorpdfstring{$V^+(b_1)$}{V+(b1)} touching only the leftmost vertical edge of \texorpdfstring{$\mathcal{P}$}{P}}

Now the only areas not yet calculated are those that intersect the vertical boundaries of $\mathcal{P}$. Since placing $b_1$ in Section $2l$ will cause $V^+(b_1)$ to steal from $V^\circ(w_j)$ for $j=i-l,...,i+l$ and placing $b_1$ in Section $2l+1$ will cause $V^+(b_1)$ to steal from $V^\circ(w_j)$ for $j=i-l,...,i+l+1$, $V^+(b_1)$ intersects a vertical boundary of $\mathcal{P}$ if $b_1$ is in Section $2l$ and $i-l \leq 0$ or $i+l > n$, or if $b_1$ is in Section $2l+1$ and $i-l \leq 0$ or $i+l+1 > n$. That is, $V^+(b_1)$ will intersect the leftmost boundary of $\mathcal{P}$ if $b_1$ is placed in Section $2i$ or above, and the rightmost boundary of $\mathcal{P}$ if placed in Section $2(n-i)+1$ or above.

\paragraph{Section $2l$} If $i \leq \frac{n}{2}$ then it can be the case that $V^+(b_1)$ intersects only the leftmost boundary (and not the rightmost boundary) of $\mathcal{P}$. For this, $b_1$ will have to be contained in Sections $2i$ to $2(n-i)$. In order to compute the area of $V^+(b_1)$ for $b_1$ within these sections, we can take the area calculated previously for $b_1$ in Section $2l$ where $V^+(b_1)$ does not touch either vertical edge of $\mathcal{P}$ and remove the extra areas included in the previous calculation which do not exist in the set-up studied here (i.e. the areas entering $V^\circ(w_j)$ for $j<1$); in calculations presented henceforth, whenever we use a previously formulated area expression and wish to remove an area from the original calculation, we shall display the foreign (or phantom) area $A$ being removed within quotation marks: $``A"$. Thus, if $b_1$ is in Section $2l$ for $l=i,...,n-i$ then, for $i>1$,
\begin{equation*}
    \begin{split}
        Area(V^+(b_1)) &= \sum_{j=1}^{i-1}Area(V^+(b_1) \cap V^\circ(w_j)) + Area(V^+(b_1) \cap V^\circ(w_i)) \\ &\qquad+ \sum_{j=i+1}^{i+l-1}Area(V^+(b_1) \cap V^\circ(w_j)) + Area(V^+(b_1) \cap V^\circ(w_{i+l})) \\
        &= - \frac{x^2}{4} - \frac{3y^2}{4} + (-\frac{lp}{2n} + \frac{q}{2})y + \frac{lpq}{2n} - \frac{(l-1)lp^2}{4n^2} %original calculation:: -(l^2 p^2)/(4 n^2) + (l p^2)/(4 n^2) + (l p q)/(2 n) - (l p y)/(2 n) + (q y)/2 - x^2/4 - (3 y^2)/4
        - ``Area(V^+(b_1) \cap V^\circ(w_{i-l}))" \\ &\qquad- \sum_{j=i-l+1}^{0}``Area(V^+(b_1) \cap V^\circ(w_j))" \\
        &= - \frac{x^2}{4} - \frac{3y^2}{4} + (-\frac{lp}{2n} + \frac{q}{2})y + \frac{lpq}{2n} - \frac{(l-1)lp^2}{4n^2}
        - \left(\frac{x^2}{8} - \frac{3y^2}{8} + \frac{xy}{4} \right.\\
        &\qquad+ (\frac{(i-(i-l)-1)p}{4n} - \frac{q}{4})x + (\frac{(i-(i-l)-1)p}{4n} + \frac{q}{4})y \\ 
        &\qquad \left. - \frac{(i-(i-l)-1)pq}{4n} + \frac{(i-(i-l)-1)^2p^2}{8n^2}\right) \\ %CC1(w_{(i-l)}):: ((l - 1)^2 p^2)/(8 n^2) + x (((l - 1) p)/(4 n) - q/4) + y (((l - 1) p)/(4 n) + q/4) - ((l - 1) p q)/(4 n) + x^2/8 + (x y)/4 - (3 y^2)/8
        &\qquad- \sum_{j=i-l+1}^{0}(- \frac{p}{2n}x - \frac{p}{2n}y + \frac{pq}{2n} - \frac{(4(i-j)-1)p^2}{8n^2}) \\ %CC2(w_j):: ((i - l) p ((-3 + 2 i + 2 l) p + 4 n (-q + x + y)))/(8 n^2)
        &= - \frac{3x^2}{8} - \frac{3y^2}{8} - \frac{xy}{4} - (\frac{(l-1)p}{4n} - \frac{q}{4})x + (-\frac{(3l-1)p}{4n} + \frac{q}{4})y + \frac{(3l-1)pq}{4n} \\
        &\qquad- \frac{(l-1)(3l-1)p^2}{8n^2} %SUM NON-SUM -(3 l^2 p^2)/(8 n^2) + (l p^2)/(2 n^2) + (3 l p q)/(4 n) - (3 l p y)/(4 n) - p^2/(8 n^2) - (p q)/(4 n) + (p y)/(4 n) + (q y)/4 - (3 x^2)/8 - (x y)/4 - (3 y^2)/8 - x (((l - 1) p)/(4 n) - q/4)
        - (- \frac{p}{2n}x - \frac{p}{2n}y + \frac{pq}{2n} + \frac{p^2}{8n^2})\times(0-(i-l+1-1)) \\
        &\qquad+ \frac{p^2}{2n^2}\sum_{j=i-l+1}^{0}(i-j) \\
        &= - \frac{3x^2}{8} - \frac{3y^2}{8} - \frac{xy}{4} + (\frac{(l-2i+1)p}{4n} + \frac{q}{4})x + (- \frac{(l+2i-1)p}{4n} + \frac{q}{4})y \\
        &\qquad+ \frac{(l+2i-1)pq}{4n} - \frac{(3l^2-3l-i+1)p^2}{8n^2} + \frac{p^2}{2n^2}\frac{(l-1+i)(l-1-(i-1))}{2} \\
        &= - \frac{3x^2}{8} - \frac{3y^2}{8} - \frac{xy}{4} + (\frac{(l-2i+1)p}{4n} + \frac{q}{4})x + (- \frac{(l+2i-1)p}{4n} + \frac{q}{4})y \\
        &\qquad+ \frac{(l+2i-1)pq}{4n} - \frac{(l^2 - l + 1 - 3i + 2i^2)p^2}{8n^2}
        %- (3 x^2)/8 - (x y)/4 - (3 y^2)/8 + (l p x)/(4 n) - (i p x)/(2 n) + (p x)/(4 n) + (q x)/4 - (i p y)/(2 n) - (l p y)/(4 n) + (p y)/(4 n) + (q y)/4 - (i^2 p^2)/(4 n^2) + (3 i p^2)/(8 n^2) + (i p q)/(2 n) - (l^2 p^2)/(8 n^2) + (l p^2)/(8 n^2) + (l p q)/(4 n) - p^2/(8 n^2) - (p q)/(4 n)
    \end{split}
\end{equation*}
or, for $i=1$,
\begin{equation*}
    \begin{split}
        Area(V^+(b_1)) &= Area(V^+(b_1) \cap V^\circ(w_1)) + \sum_{j=2}^{l}Area(V^+(b_1) \cap V^\circ(w_j)) + Area(V^+(b_1) \cap V^\circ(w_{l+1})) \\
        &= - \frac{x^2}{2} - \frac{p}{2n}y + \frac{pq}{2n} + \sum_{j=2}^{l}(\frac{p}{2n}x - \frac{p}{2n}y + \frac{pq}{2n} - \frac{(4(j-1)-1)p^2}{8n^2}) \\
        &\qquad+ \frac{x^2}{8} -\frac{3y^2}{8} -\frac{xy}{4} + (\frac{q}{4} - \frac{(l-1)p}{4n})x + (\frac{(l-1)p}{4n} + \frac{q}{4})y \\
        &\qquad- \frac{(l-1)pq}{4n} + \frac{(l-1)^2p^2}{8n^2} \\
        &= - \frac{3x^2}{8} -\frac{3y^2}{8} - \frac{xy}{4} + (\frac{q}{4} - \frac{(l-1)p}{4n})x + (\frac{(l-3)p}{4n} + \frac{q}{4})y - \frac{(l-3)pq}{4n} \\
        &\qquad+ \frac{(l-1)^2p^2}{8n^2} + (l-1)(\frac{p}{2n}x - \frac{p}{2n}y + \frac{pq}{2n} + \frac{5p^2}{8n^2}) - \frac{p^2}{2n^2}\sum_{j=2}^{l}(j) \\
        &= - \frac{3x^2}{8} -\frac{3y^2}{8} - \frac{xy}{4} + (\frac{q}{4} + \frac{(l-1)p}{4n})x + (-\frac{(l+1)p}{4n} + \frac{q}{4})y + \frac{(l+1)pq}{4n} \\
        &\qquad- \frac{l(l-1)p^2}{8n^2}
    \end{split}
\end{equation*}
(identical to the previous area upon a substitution of $i=1$, so we need only use this former representation).

This area has partial derivatives
\begin{equation*}
\begin{split}
\frac{\delta A}{\delta x}&= - \frac{3x}{4} - \frac{y}{4} + \frac{(l-2i+1))p}{4n} + \frac{q}{4} \\
\frac{\delta A}{\delta y}&= - \frac{3y}{4} - \frac{x}{4} - \frac{(l+2i-1))p}{4n} + \frac{q}{4} \\
&\Rightarrow  - 2x^* + \frac{(2l-2i+1)p}{2n} + \frac{q}{2} = 0 \Rightarrow x^* = \frac{(2(l-i)+1)p}{4n} + \frac{q}{4} \\
\text{and} &\Rightarrow 2y^* + \frac{(2l+2i-1)p}{2n} - \frac{q}{2} = 0 \Rightarrow y^* = -\frac{(2(l+i)-1)p}{4n} + \frac{q}{4} \\
\end{split}
\end{equation*}
but $x^* > \frac{p}{2n}$ (since $\frac{p}{4n}<\frac{q}{4}$) so we are required to investigate when $b_1$ is placed on the boundary of Section $2l$. Note that since the global optimum lies to the right of Section $2l$ we will not find the optimum on the boundary $x=0$ outside its endpoints.
\begin{itemize}
    \item Upon boundary $y=x+(l-1)\frac{p}{n}$ we have $Area(V^+((x,x+\frac{(l-1)p}{n}))) = - x^2 + (-\frac{(2(l+i)-3)p}{2n} + \frac{q}{2})x + \frac{(l+i-1)pq}{2n} - \frac{(6l^2 + 4il + 2i^2 - 11l - 7i + 6)p^2}{8n^2}$, maximised by $x^*= -\frac{(2(l+i)-3)p}{4n} + \frac{q}{4}$ giving $Area(V^+((-\frac{(2(l+i)-3)p}{4n} + \frac{q}{4}, \frac{(2(l-i)-1)p}{4n} + \frac{q}{4})))=\frac{(2(l+i)-1)pq}{8n} - \frac{(8l^2-10l-2i+3)p^2}{16n^2} + \frac{q^2}{16}$. However, for $0 \leq -\frac{(2(l+i)-3)p}{4n} + \frac{q}{4} \leq \frac{p}{2n}$ to be true we require $\frac{(2(l+i)-3)p}{n} \leq q \leq \frac{(2(l+i)-1)p}{n}$. If $\frac{(2(l+i)-3)p}{n} \geq q$ then the optimum lies on the endpoint $(0,\frac{(l-1)p}{n})$ giving $Area(V^+((0,\frac{(l-1)p}{n})))\\=\frac{(l+i-1)pq}{2n} - \frac{(6l^2 + 4il + 2i^2 - 11l - 7i + 6)p^2}{8n^2}$, and if $\frac{(2(l+i)-1)p}{n} \leq q$ then the optimum lies on the endpoint $(\frac{p}{2n},\frac{(2l-1)p}{2n})$ giving $Area(V^+((\frac{p}{2n},\frac{(2l-1)p}{2n})))= \frac{(2(l+i)-1)pq}{4n} - \frac{(6l^2+2i^2+4il-7l-3i+2)p^2}{8n^2}$.
    
    \item Upon boundary $y=l\frac{p}{n}-x$ we have $Area(V^+((x,\frac{lp}{n}-x))) = - \frac{x^2}{2} + \frac{lp}{n}x + \frac{(2(l+i)-1)pq}{4n} - \frac{(6l^2+2i^2+4il-3l-3i+1)p^2}{8n^2}$, maximised by $x^*=\frac{lp}{2n}$ which is only in Section $2l$ for $l=1$. So if $l>1$ the optimum on this boundary will lie on the endpoint $(\frac{p}{2n},\frac{(2l-1)p}{2n})$ giving $Area(V^+((\frac{p}{2n},\frac{(2l-1)p}{2n})))=\frac{(2(l+i)-1)pq}{4n} - \frac{(6l^2+2i^2+4il-7l-3i+2)p^2}{8n^2}$.
\end{itemize}

Since the optimum over boundary $y=l\frac{p}{n}-x$ is found at the endpoint of the boundary $y=x+(l-1)\frac{p}{n}$, we need only take the results from the latter for the optimal placement over all of Section $2l$. This means that our optimal areas are: if $\frac{(2(l+i)-1)p}{n} \leq q$ then $Area(V^+((\frac{p}{2n},\frac{(2l-1)p}{2n})))= \frac{(2(l+i)-1)pq}{4n} - \frac{(6l^2+2i^2+4il-7l-3i+2)p^2}{8n^2}$; if $\frac{(2(l+i)-3)p}{n} \leq q \leq \frac{(2(l+i)-1)p}{n}$ then $Area(V^+((-\frac{(2(l+i)-3)p}{4n} + \frac{q}{4}, \frac{(2(l-i)-1)p}{4n} + \frac{q}{4})))=\frac{(2(l+i)-1)pq}{8n} - \frac{(8l^2-10l-2i+3)p^2}{16n^2} + \frac{q^2}{16}$; and if $\frac{(2(l+i)-3)p}{n} \geq q$ then $Area(V^+((0,\frac{(l-1)p}{n})))=\frac{(l+i-1)pq}{2n} - \frac{(6l^2 + 4li + 2i^2 - 11l - 7i + 6)p^2}{8n^2}$.

\paragraph{Section $2l+1$} Alternatively, if $V^+(b_1)$ hits the leftmost, and not the rightmost, boundary (so $i< \frac{n}{2}$) and $b_1$ is in Section $2l+1$ for $l=i,...,n-i-1$ then, for $i>1$,
\begin{equation*}
    \begin{split}
        Area(V^+(b_1)) &= \sum_{j=1}^{i-1}Area(V^+(b_1) \cap V^\circ(w_j)) + Area(V^+(b_1) \cap V^\circ(w_i)) \\ &\qquad+ \sum_{j=i+1}^{i+l}Area(V^+(b_1) \cap V^\circ(w_j)) + Area(V^+(b_1) \cap V^\circ(w_{i+l+1})) \\
        &= - \frac{3x^2}{8} - \frac{3y^2}{8} - \frac{xy}{4} + (\frac{(l-2i+1))p}{4n} + \frac{q}{4})x + (- \frac{(l+2i-1))p}{4n} + \frac{q}{4})y \\
        &\qquad+ \frac{(l+2i-1)pq}{4n} - \frac{(l^2 - l + 1 - 3i + 2i^2)p^2}{8n^2} %previous calculation
        - \left(\frac{x^2}{8} -\frac{3y^2}{8} -\frac{xy}{4} \right. \\
        &\qquad+ \left. (\frac{q}{4} - \frac{((i+l)-i-1)p}{4n})x + (\frac{((i+l)-i-1)p}{4n} + \frac{q}{4})y - \frac{((i+l)-i-1)pq}{4n} \right. \\
        &\qquad+ \left. \frac{((i+l)-i-1)^2p^2}{8n^2}\right) %CC4(w_{i+l})
        + \frac{p}{2n}x - \frac{p}{2n}y + \frac{pq}{2n} - \frac{(4((i+l)-i)-1)p^2}{8n^2} \\ %CC3(w_{i+l})
        &\qquad+ \frac{x^2}{8} -\frac{3y^2}{8} -\frac{xy}{4} + (\frac{q}{4} - \frac{((i+l+1)-i-1)p}{4n})x + (\frac{((i+l+1)-i-1)p}{4n} + \frac{q}{4})y \\
        &\qquad- \frac{((i+l+1)-i-1)pq}{4n} + \frac{((i+l+1)-i-1)^2p^2}{8n^2} \\ %CC4(w_{i+l+1})
        &= - \frac{3x^2}{8} - \frac{3y^2}{8} - \frac{xy}{4} + (\frac{(l-2i+1))p}{4n} + \frac{q}{4})x + (- \frac{(l+2i-1))p}{4n} + \frac{q}{4})y \\
        &\qquad+ \frac{(l+2i-1)pq}{4n} - \frac{(l^2 - l + 1 - 3i + 2i^2)p^2}{8n^2}
        - \frac{(l-1)^2p^2}{8n^2} + \frac{l^2p^2}{8n^2} \\
        &\qquad+ \frac{p}{4n}x - \frac{p}{4n}y + \frac{pq}{4n} - \frac{(4l-1)p^2}{8n^2} \\
        &= - \frac{3x^2}{8} - \frac{3y^2}{8} - \frac{xy}{4} + (\frac{(l-2i+2)p}{4n} + \frac{q}{4})x + (- \frac{(l+2i)p}{4n} + \frac{q}{4})y \\
        &\qquad+ \frac{(l+2i)pq}{4n} - \frac{(l^2 + l + 2i^2 - 3i + 1)p^2}{8n^2} %-(p^2 (2 i^2 - 3 i + l^2 + l + 1))/(8 n^2) + x ((p (-2 i + l + 2))/(4 n) + q/4) + y (q/4 - (p (2 i + l))/(4 n)) + (p q (2 i + l))/(4 n) - (3 x^2)/8 - (x y)/4 - (3 y^2)/8
    \end{split}
\end{equation*}
or, for $i=1$,
\begin{equation*}
    \begin{split}
        Area(V^+(b_1)) &= Area(V^+(b_1) \cap V^\circ(w_i)) + \sum_{j=2}^{l+1}Area(V^+(b_1) \cap V^\circ(w_j)) + Area(V^+(b_1) \cap V^\circ(w_{l+2})) \\
        &= - \frac{3x^2}{8} -\frac{3y^2}{8} - \frac{xy}{4} + (\frac{q}{4} + \frac{(l-1)p}{4n})x + (-\frac{(l+1)p}{4n} + \frac{q}{4})y + \frac{(l+1)pq}{4n} \\
        &\qquad- \frac{l(l-1)p^2}{8n^2} + (\frac{p}{2n}x - \frac{p}{2n}y + \frac{pq}{2n} - \frac{(4l-1)p^2}{8n^2}) \\
        &\qquad+ \frac{x^2}{8} -\frac{3y^2}{8} -\frac{xy}{4} + (\frac{q}{4} - \frac{lp}{4n})x + (\frac{lp}{4n} + \frac{q}{4})y - \frac{lpq}{4n} + \frac{l^2p^2}{8n^2} \\
        &\qquad- (\frac{x^2}{8} -\frac{3y^2}{8} -\frac{xy}{4} + (\frac{q}{4} - \frac{(l-1)p}{4n})x + (\frac{(l-1)p}{4n} + \frac{q}{4})y - \frac{(l-1)pq}{4n} + \frac{(l-1)^2p^2}{8n^2}) \\
        &= - \frac{3x^2}{8} -\frac{3y^2}{8} - \frac{xy}{4} + (\frac{lp}{4n} + \frac{q}{4})x + (-\frac{(l+2)p}{4n} + \frac{q}{4})y + \frac{(l+2)pq}{4n} \\
        &\qquad- \frac{l(l+1)p^2}{8n^2} \\
    \end{split}
\end{equation*}
(which, again, we check is identical to the representation found for $i>1$ so we shall proceed to use the former formulation).

This area has partial derivatives
\begin{equation*}
\begin{split}
\frac{\delta A}{\delta x}&= - \frac{3x}{4} - \frac{y}{4} + \frac{(l-2i+2)p}{4n} + \frac{q}{4} \\
\frac{\delta A}{\delta y}&= - \frac{3y}{4} - \frac{x}{4} - \frac{(l+2i)p}{4n} + \frac{q}{4} \\
&\Rightarrow  - 2x^* + \frac{(2(l-i)+3)p}{2n} + \frac{q}{2} = 0 \Rightarrow x^* = \frac{(2(l-i)+3)p}{4n} + \frac{q}{4} \\
and &\Rightarrow  - 2y^* - \frac{(2l+2i+1)p}{2n} + \frac{q}{2} = 0 \Rightarrow y^* = - \frac{(2(l+i)+1)p}{4n} + \frac{q}{4} \, .\\
\end{split}
\end{equation*}

Using identical logic to that in Section $2l$, $x^*>\frac{p}{2n}$ so we explore the boundaries (all boundaries this time).
\begin{itemize}
    \item Upon boundary $x=\frac{p}{2n}$ we have $Area(V^+((\frac{p}{2n},y))) = - \frac{3y^2}{8}  + (- \frac{(2l+4i+1)p}{8n} + \frac{q}{4})y + \frac{(2l+4i+1)pq}{8n} - \frac{(4l^2+8i^2-4i-1)p^2}{32n^2}$, maximised by $y^*=- \frac{(2l+4i+1)p}{6n} + \frac{q}{3}$ giving $Area(V^+((\frac{p}{2n},\\ - \frac{(2l+4i+1)p}{6n} + \frac{q}{3})))= \frac{(2l+4i+1)pq}{12n} + \frac{(-2l^2-2i^2+4il+l+5i+1)p^2}{24n^2} + \frac{q^2}{24}$. However, for $\frac{(2l-1)p}{2n} \leq y^* \leq \frac{(2l+1)p}{2n}$ we require $\frac{(4(l+i)-5)p}{2n} \leq q \leq \frac{(8l+4i+7)p}{2n}$. If $\frac{(4(l+i)-5)p}{2n} \geq q$ then the optimum lies on the endpoint $(\frac{p}{2n},\frac{(2l-1)p}{2n})$ giving $Area(V^+((\frac{p}{2n},\frac{(2l-1)p}{2n}))) =\frac{(l+i)pq}{2n} - \frac{(6l^2 + 2i^2 + 4il - 3l - 3i)p^2}{8n^2}$, and if $\frac{(8l+4i+7)p}{2n} \leq q$ then the optimum lies on the endpoint $(\frac{p}{2n},\frac{(2l+1)p}{2n})$ giving $Area(V^+(\\(\frac{p}{2n},\frac{(2l+1)p}{2n}))) =\frac{(2(l+i)+1)pq}{4n} - \frac{(6l^2 + 2i^2 + 4il + 5l + i + 1)p^2}{8n^2}$.
    \item Upon boundary $y=\frac{lp}{n}-x$ we have $Area(V^+((x,\frac{lp}{n}-x))) = - \frac{x^2}{2} + \frac{(2l+1)p}{2n}x + \frac{(l+i)pq}{2n} - \frac{(6l^2 + 2i^2 + 4il - 3i + l + 1)p^2}{8n^2}$, maximised by $x^*=\frac{(2l+1)p}{2n} > \frac{p}{2n}$, so the optimum is achieved at $x=\frac{p}{2n}$, the value of which has been found above.
    \item Upon boundary $y=x+\frac{lp}{n}$ we have $Area(V^+((x,x+\frac{lp}{n})))= - x^2 + (-\frac{(2(l+i)-1)p}{2n} + \frac{q}{2})x + \frac{(l+i)pq}{2n} \, - \, \frac{(6l^2 + 2i^2 + 4il + l - 3i + 1)p^2}{8n^2}$,$\, \,$ maximised by $\, \, \, x^* \, = \, -\frac{(2(l+i)-1)p}{4n} \, + \, \frac{q}{4} \, \, \,$ giving $\, \, \, Area(V^+\\((-\frac{(2(l+i)-1)p}{4n} + \frac{q}{4},\frac{(2(l-i)+1)p}{4n}+\frac{q}{4})))= \frac{(2l + 2i + 1)pq}{8n} - \frac{(8l^2 + 6l - 2i + 1))p^2}{16n^2} + \frac{q^2}{16}$. However, for $0 \leq x^* \leq \frac{p}{2n}$ we require $\frac{(2(l+i)-1)p}{n} \leq q \leq \frac{(2(l+i)+1)p}{n}$. If $\frac{(2(l+i)-1)p}{n} \geq q$ then the optimum lies on the endpoint $(0,\frac{lp}{n})$: this lies on the boundary $y=\frac{lp}{n}-x$ upon which it was found never to be optimal. Alternatively, if $\frac{(2(l+i)+1)p}{n} \leq q$ then the optimum lies on the endpoint $(\frac{p}{2n},\frac{(2l+1)p}{2n})$ upon the boundary $x=\frac{p}{2n}$ for which all optimal values have been found.
\end{itemize}

Following this it is clear that, from the exploration of the boundaries $y=x+\frac{lp}{n}$ and $y=\frac{lp}{n}-x$, if $\frac{(2(l+i)+1)p}{n} \leq q$ or $\frac{(2(l+i)-1)p}{n} \geq q$ then the optimum lies on the boundary $x=\frac{p}{2n}$. What remains to be seen is whether it is optimal to place on the boundary $y=x+\frac{lp}{n}$ or $x=\frac{p}{2n}$ when $\frac{(2(l+i)-1)p}{n} \leq q \leq \frac{(2(l+i)+1)p}{n}$, so we are required to compare the optimal values found within boundaries $x=\frac{p}{2n}$ and $y=x+\frac{lp}{n}$ (so not the endpoints). Firstly, if $\frac{(2(l+i)-1)p}{n} \leq q \leq \frac{(2(l+i)+1)p}{n}$ then
\begin{alignat*}{2}
    \frac{(4(l+i)-5)p}{2n} &= \frac{(2(l+i)-2.5)p}{n} &&\, \\
    &< \frac{(2(l+i)-1)p}{n} \leq q &&\leq \frac{(2(l+i)+1)p}{n} \\
    & \, &&< \frac{(2(l+i) + 2l +3.5)p}{n} = \frac{(8l+4i+7)p}{2n}
\end{alignat*}
so the optimum upon $x=\frac{p}{2n}$ for these values of $p$ and $q$ is within the boundary (not an endpoint). Therefore we must check
\begin{equation*}
\begin{split}
    \frac{(2l+4i+1)pq}{12n} &+ \frac{(-2l^2-2i^2+4il+l+5i+1)p^2}{24n^2} + \frac{q^2}{24} \\
    &\qquad- \left(\frac{(2l + 2i + 1)pq}{8n} - \frac{(8l^2 + 6l - 2i + 1))p^2}{16n^2} + \frac{q^2}{16}\right) \\
    &= \frac{(-2l+2i-1)pq}{24n} + \frac{(20l^2-4i^2+8il+20l+4i+5)p^2}{48n^2} - \frac{q^2}{48} \\
    %&\geq 0 \Leftrightarrow \frac{(-4l+4i-2)pq}{n} + \frac{(20l^2-4i^2+8il+20l+4i+5)p^2}{n^2} - q^2 \geq 0 \\
    &=  \frac{1}{48}\left(\frac{(20l^2-4i^2+8il+20l+4i+5)p^2}{n^2} - \frac{(4(l-i)+2)pq}{n} - q^2\right) \\
    &\geq \frac{p^2}{48n^2}((20l^2-4i^2+8il+20l+4i+5) - (4(l-i)+2)(2(l+i)+1) - (2(l+i)+1)^2) \\
    %&= \frac{p^2}{48n^2}(20l^2-4i^2+8il+20l+4i+5 - 8l^2 + 8i^2 - 4l - 4i - 4l + 4i - 2 - 4l^2 - 8il - 4i^2 - 4l - 4i - 1)
    &= \frac{p^2}{48n^2}(8l^2+8l+2) > 0 \, .
\end{split}
\end{equation*}
This settles all concerns and proves that the optimum is always located on the boundary $x=\frac{p}{2n}$. This means that our optimal areas are: if $\frac{(8l+4i+7)p}{2n} \leq q$ then $Area(V^+((\frac{p}{2n},\frac{(2l+1)p}{2n}))) =\frac{(2(l+i)+1)pq}{4n} - \frac{(6l^2 + 2i^2 + 4il + 5l + i + 1)p^2}{8n^2}$ as depicted in Figure \ref{fig:RowOptimalLeftTouchingIIIc}; if $\frac{(4(l+i)-5)p}{2n} \leq q \leq \frac{(8l+4i+7)p}{2n}$ then $Area(V^+((\frac{p}{2n},- \frac{(2l+4i+1)p}{6n} + \frac{q}{3})))= \frac{(2l+4i+1)pq}{12n} + \frac{(-2l^2-2i^2+4il+l+5i+1)p^2}{24n^2} + \frac{q^2}{24}$ as depicted in Figure \ref{fig:RowOptimalLeftTouchingIIIb}; and if $\frac{(4(l+i)-5)p}{2n} \geq q$ then $Area(V^+((\frac{p}{2n},\frac{(2l-1)p}{2n}))) =\frac{(l+i)pq}{2n} - \frac{(6l^2 + 2i^2 + 4il - 3l - 3i)p^2}{8n^2}$ as depicted in Figure \ref{fig:RowOptimalLeftTouchingIIIa}.

\begin{figure}[!ht]\ContinuedFloat
\begin{subfigure}{.5\textwidth}
  \centering
  \includegraphics[width=0.9\textwidth]{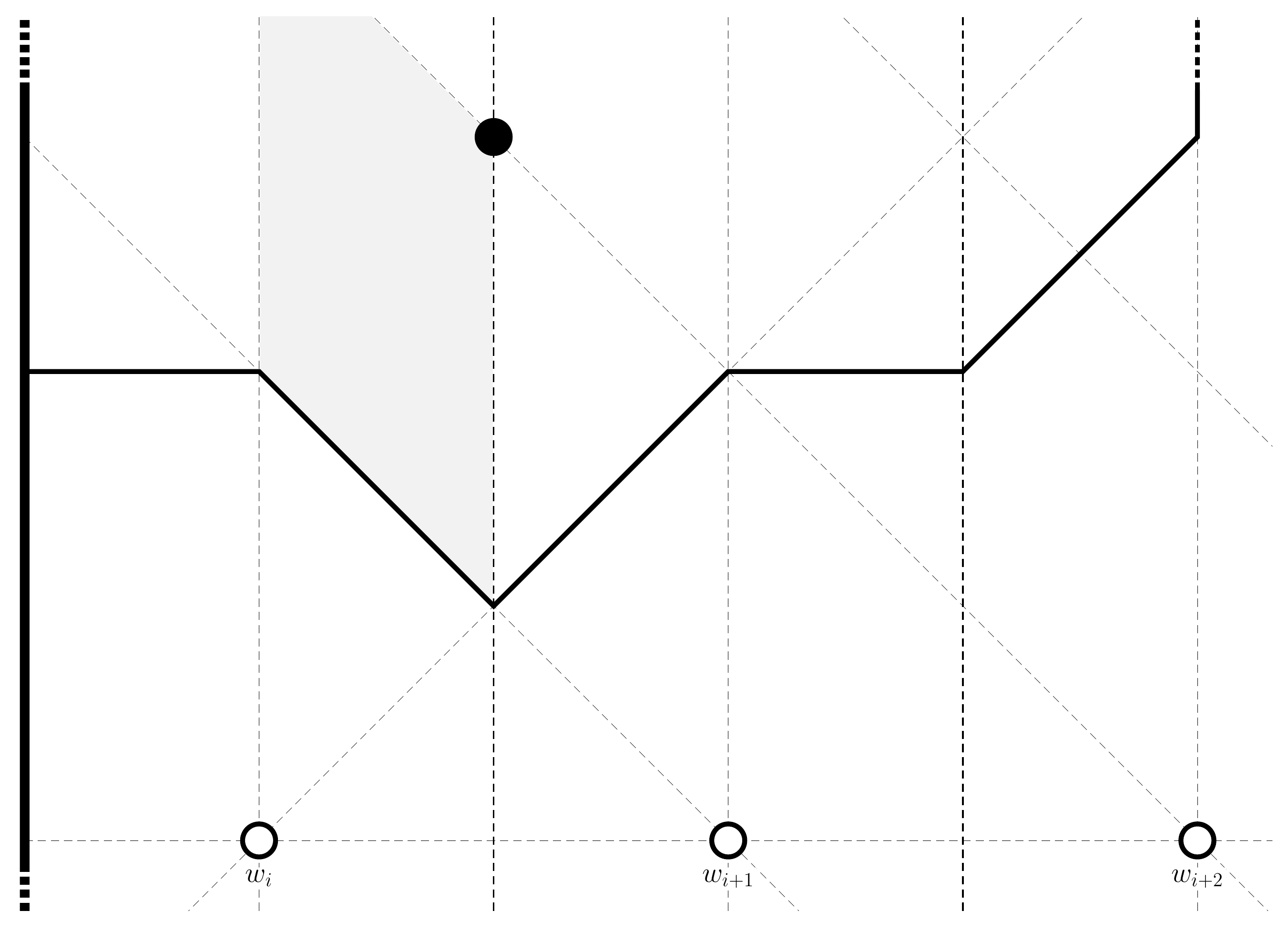}
  \caption{$b_1=(\frac{p}{2n},\frac{(2l+1)p}{2n})$ only if $\frac{(8l+4i+7)p}{2n} \leq q$. \\ \,}
%  \caption{$Area(V^+((\frac{p}{2n},\frac{(2l+1)p}{2n}))) =\frac{(2(l+i)+1)pq}{4n} - \frac{(6l^2 + 2i^2 + 4il + 5l + i + 1)p^2}{8n^2}$ only if $\frac{(8l+4i+7)p}{2n} \leq q$.}
  \label{fig:RowOptimalLeftTouchingIIIc}
\end{subfigure}%
\begin{subfigure}{.5\textwidth}
  \centering
  \includegraphics[width=0.9\textwidth]{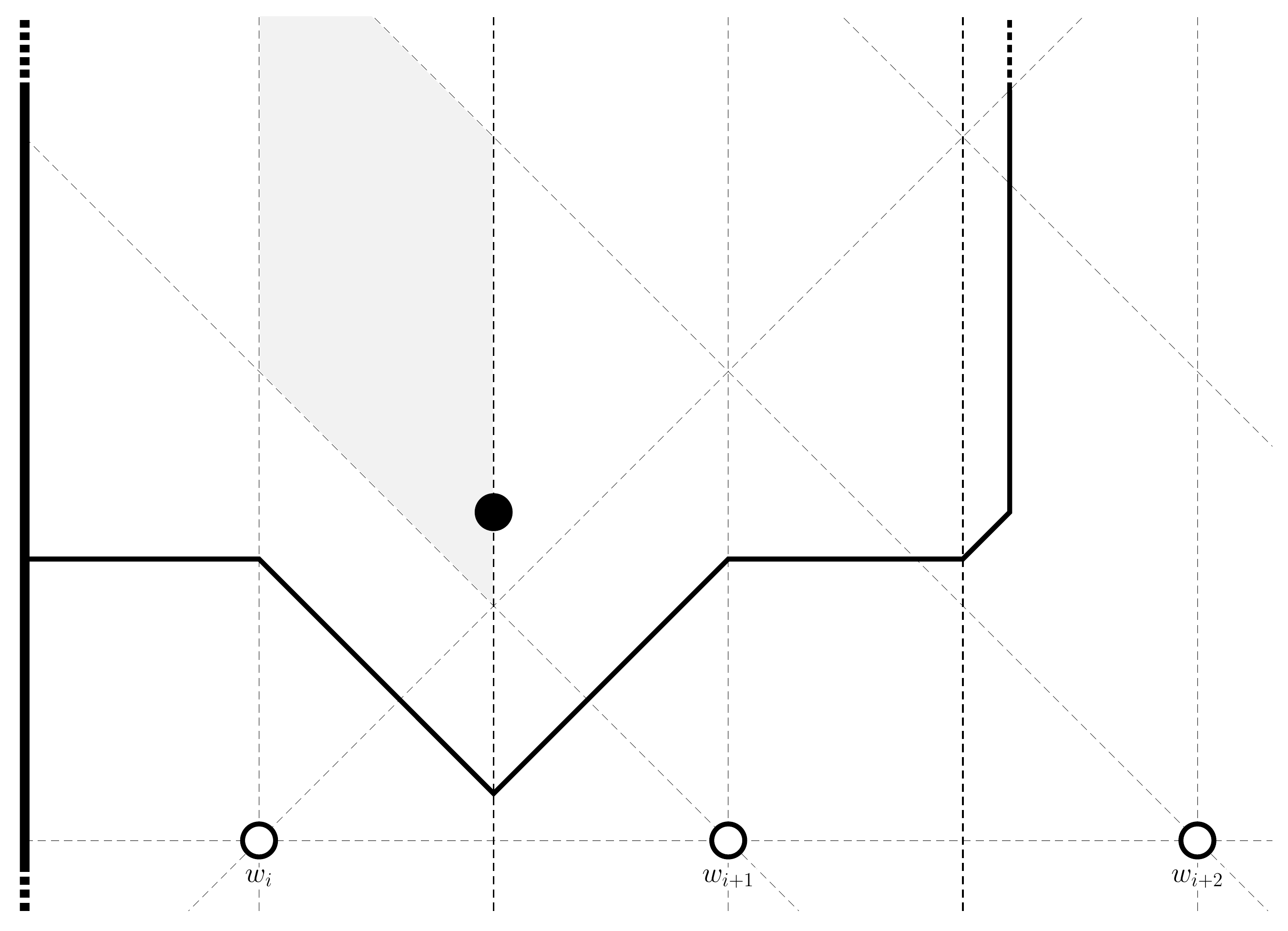}
  \caption{$b_1=(\frac{p}{2n},- \frac{(2l+4i+1)p}{6n} + \frac{q}{3})$ only if $\frac{(4(l+i)-5)p}{2n} \leq q \leq \frac{(8l+4i+7)p}{2n}$.}
%  \caption{$Area(V^+((\frac{p}{2n},- \frac{(2l+4i+1)p}{6n} + \frac{q}{3})))= \frac{(2l+4i+1)pq}{12n} + \frac{(-2l^2-2i^2+4il+l+5i+1)p^2}{24n^2} + \frac{q^2}{24}$ only if $\frac{(4(l+i)-5)p}{2n} \leq q \leq \frac{(8l+4i+7)p}{2n}$.}
  \label{fig:RowOptimalLeftTouchingIIIb}
\end{subfigure}

\begin{subfigure}{1.0\textwidth}
  \centering
  \includegraphics[width=0.45\textwidth]{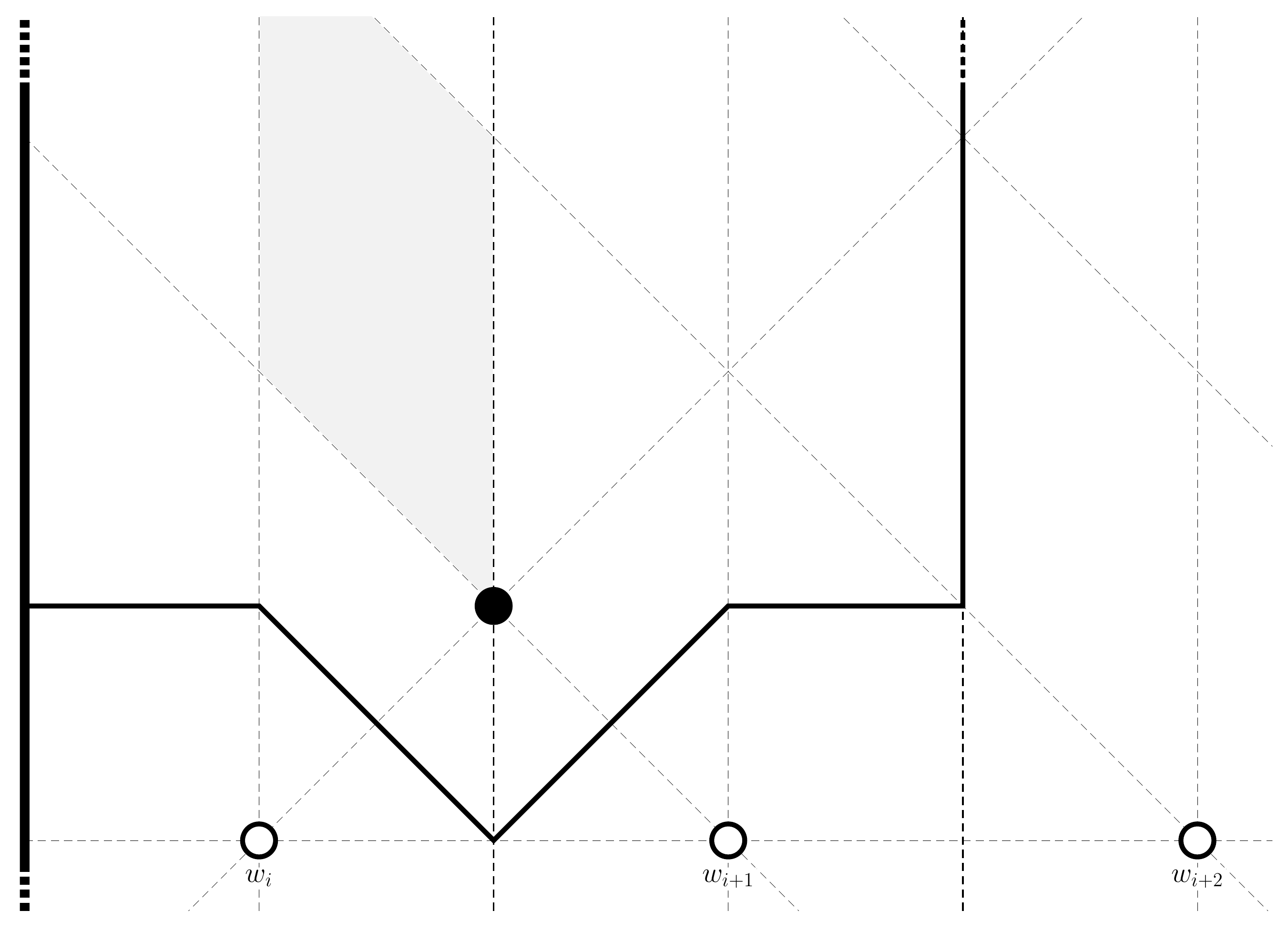}
  \caption{$b_1=(\frac{p}{2n},\frac{(2l-1)p}{2n})$ only if $\frac{(4(l+i)-5)p}{2n} \geq q$.}
%  \caption{$Area(V^+((\frac{p}{2n},\frac{(2l-1)p}{2n}))) =\frac{(l+i)pq}{2n} - \frac{(6l^2 + 2i^2 + 4il - 3l - 3i)p^2}{8n^2}$ only if $\frac{(4(l+i)-5)p}{2n} \geq q$.}
  \label{fig:RowOptimalLeftTouchingIIIa}
\end{subfigure}
\caption{Maximal area Voronoi cells $V^+(b_1)$ for $b_1$ within Section $2l+1$ touching the leftmost vertical edge of $\mathcal{P}$.}
\end{figure}

It is interesting to note that the structures of $V^+(b_1)$ for $b_1$ in Section $2l+1$ and $2(l+1)$ are identical, owing to the fact that the partitioning line $\mathcal{CC}^1(w_{i-l})$ which would normally divide the two does not exist, simply because $w_{i-l}$ does not exist for $l \geq i$ (which our values of $l$ satisfy). We can verify that the areas already found are in fact identical for these two sections. We will use this idea to greatly simplify our work in the following subsection.% \todo{Check? -- I hope it's true!}

However, before we do, let us compare the optimal locations of $b_1$ found in Section $2l+1$ and Section $2(l+1)$. In both of our calculations, the optimum was never found to be within the sections themselves so the boundary cases had to be explored. The optimum over Section $2(l+1)$ was found to be on the boundary $y=x + \frac{lp}{n}$, which is shared with Section $2l+1$, whilst the optimum over Section $2l+1$ was found to be on the boundary $x=\frac{p}{2n}$. Therefore the optimum over Section $2l+1$ and $2(l+1)$ is found on the $x=\frac{p}{2n}$ boundary, as described in the Section $2l+1$ workings.

It is important to note that this comparison is between Sections $2l+1$ and $2(l+1)$ for $l=i,...,n-i-1$, so for Section $2l$ where $l=i$ (the lowest possible value of $l$) there does not exist a Section $2l-1$ within which the Voronoi cell $V^+(b_1)$ touches the leftmost boundary of $\mathcal{P}$, so we must remember to use the Section $2l$ results for Section $2i$, as depicted in Figures~\ref{fig:RowOptimalLeftTouchingIVc}, \ref{fig:RowOptimalLeftTouchingIVb}, and \ref{fig:RowOptimalLeftTouchingIVa}.

\begin{figure}[!ht]\ContinuedFloat
\begin{subfigure}{.5\textwidth}
  \centering
  \includegraphics[width=0.9\textwidth]{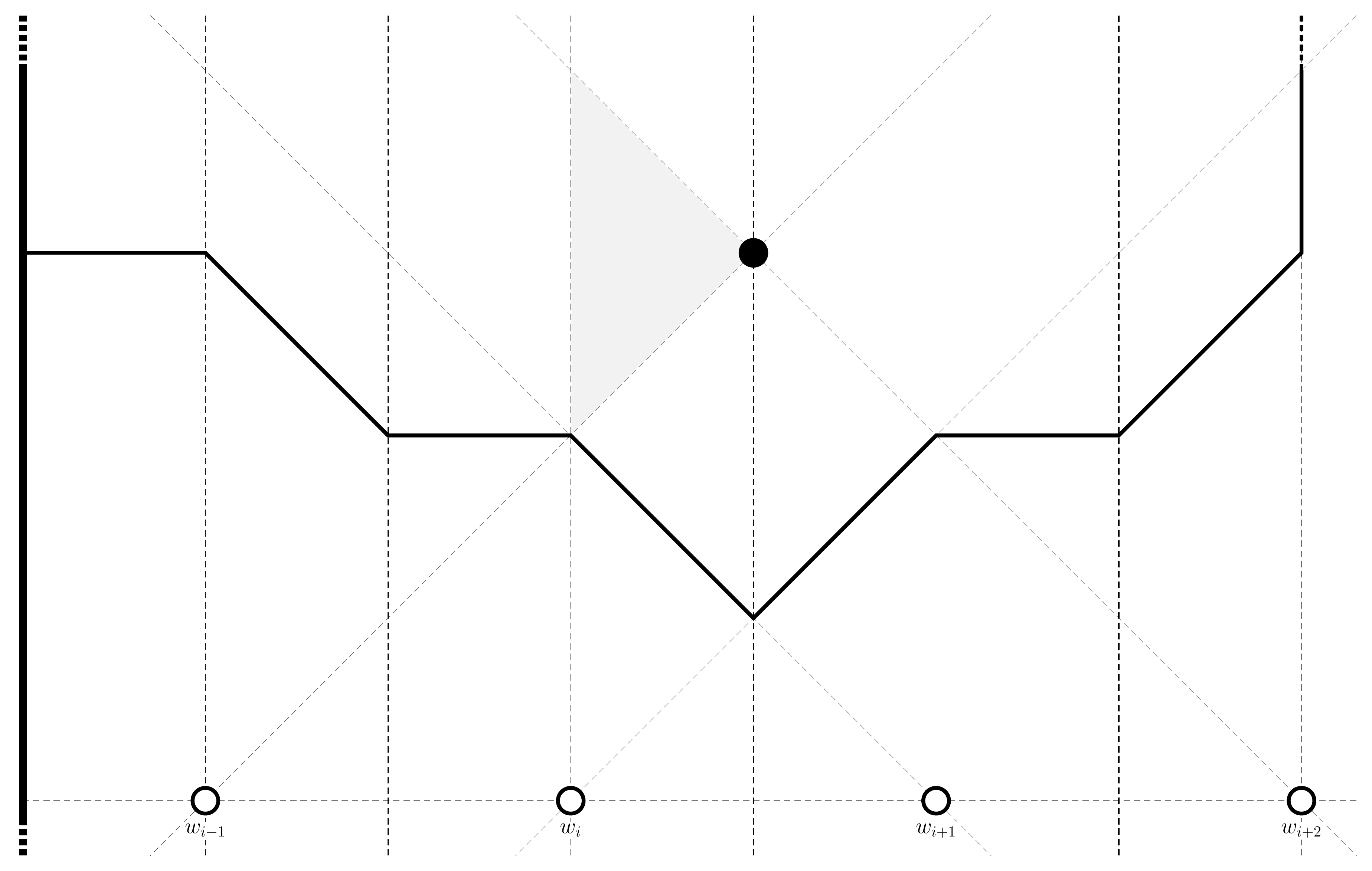}
  \caption{$b_1=(\frac{p}{2n},\frac{(2i-1)p}{2n})$ only if $\frac{(2(i+i)-1)p}{n} \leq q$.\\ \,}
%  \caption{$Area(V^+((\frac{p}{2n},\frac{(2i-1)p}{2n})))= \frac{(4i-1)pq}{4n} - \frac{(6i^2-5i+1)p^2}{4n^2}$ only if $\frac{(2(i+i)-1)p}{n} \leq q$.}
  \label{fig:RowOptimalLeftTouchingIVc}
\end{subfigure}%
\begin{subfigure}{.5\textwidth}
  \centering
  \includegraphics[width=0.9\textwidth]{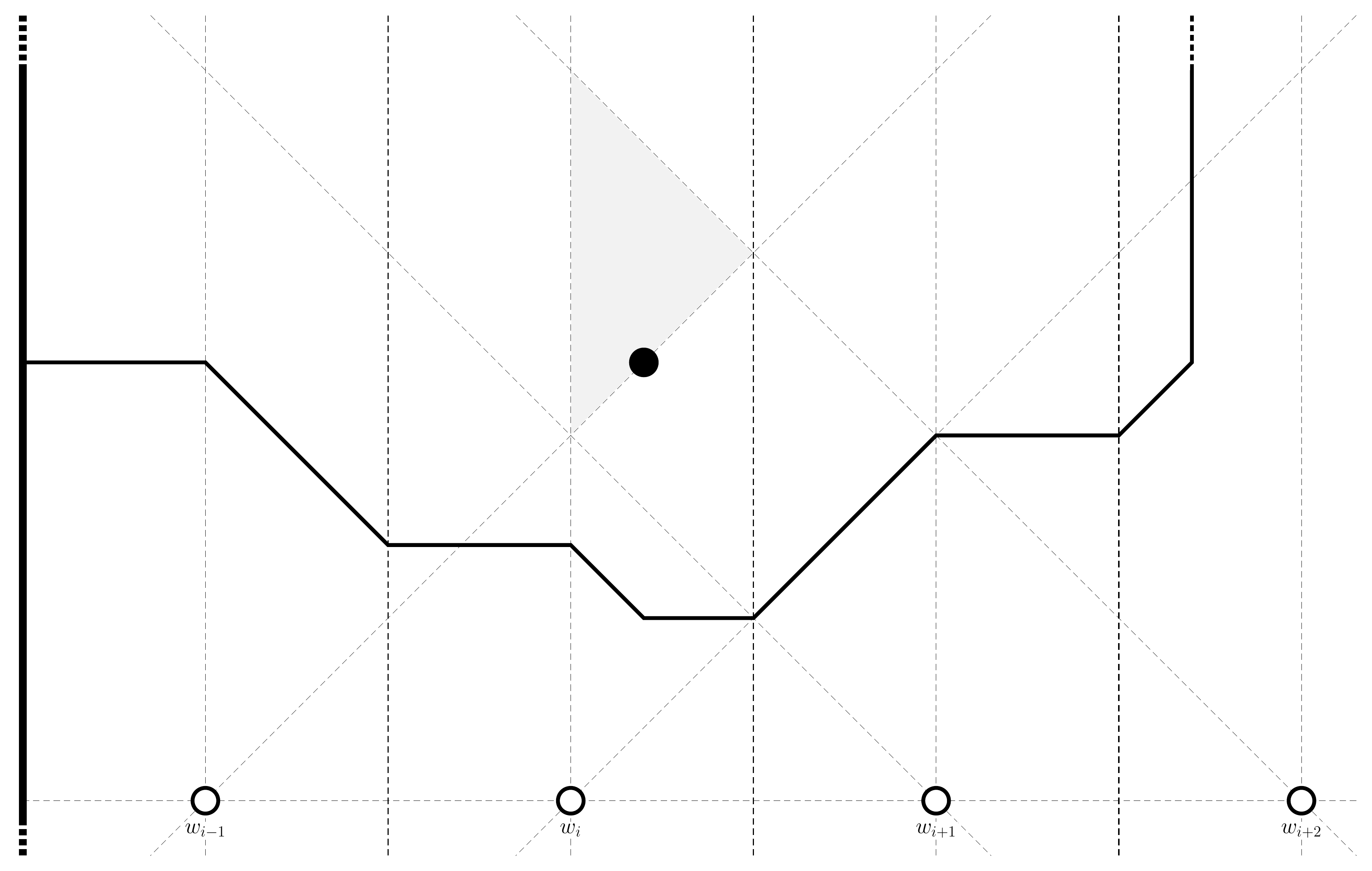}
  \caption{$b_1=(-\frac{(4i-3)p}{4n} + \frac{q}{4}, -\frac{p}{4n} + \frac{q}{4})$ only if $\frac{(4i-3)p}{n} \leq q \leq \frac{(4i-1)p}{n}$.}
%  \caption{$Area(V^+((-\frac{(4i-3)p}{4n} + \frac{q}{4}, -\frac{p}{4n} + \frac{q}{4})))=\frac{(4i-1)pq}{8n} - \frac{(8i^2-12i+3)p^2}{16n^2} + \frac{q^2}{16}$ only if $\frac{(4i-3)p}{n} \leq q \leq \frac{(4i-1)p}{n}$.}
  \label{fig:RowOptimalLeftTouchingIVb}
\end{subfigure}

\begin{subfigure}{1.0\textwidth}
  \centering
  \includegraphics[width=0.45\textwidth]{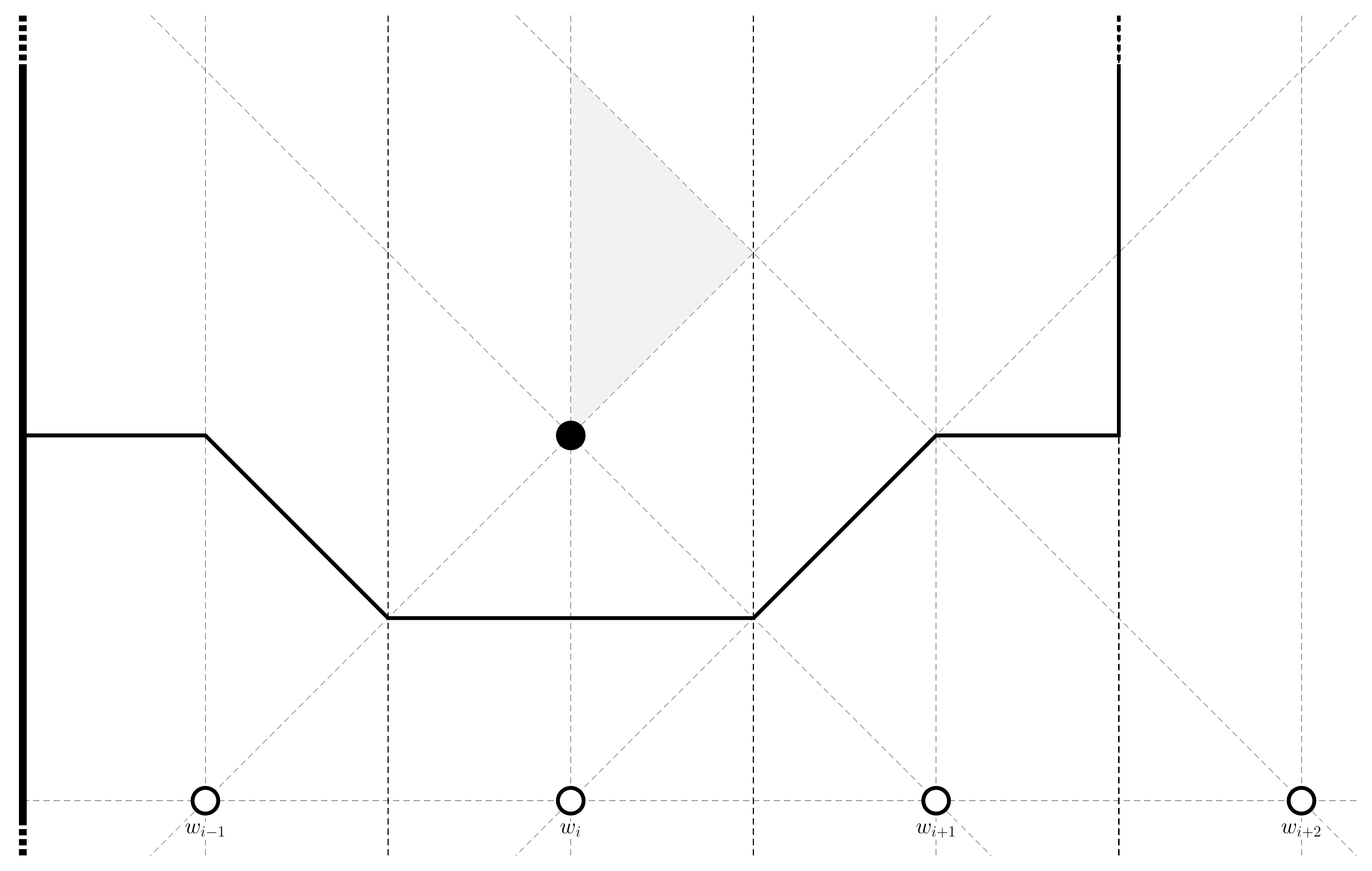}
  \caption{$b_1=(0,\frac{(i-1)p}{n})$ only if $\frac{(4i-3)p}{n} \geq q$.}
%  \caption{$Area(V^+((0,\frac{(i-1)p}{n})))=\frac{(2i-1)pq}{2n} - \frac{(6i^2 - 9i + 3)p^2}{4n^2}$ only if $\frac{(4i-3)p}{n} \geq q$.}
  \label{fig:RowOptimalLeftTouchingIVa}
\end{subfigure}
\caption{Maximal area Voronoi cells $V^+(b_1)$ for $b_1$ within Section $2i$ touching the leftmost vertical edge of $\mathcal{P}$.}
\end{figure}

\subsection{\texorpdfstring{$V^+(b_1)$}{V+(b1)} touching only the rightmost vertical edge of \texorpdfstring{$\mathcal{P}$}{P}}

Naturally the next avenue to explore is that of points $b_1$ which intersect the rightmost and not the leftmost boundary of $\mathcal{P}$. If $i > \frac{n}{2}$ then it can be the case that $V^+(b_1)$ intersects only the rightmost boundary (and not the leftmost boundary) of $\mathcal{P}$. For this, $b_1$ will have to be contained in Sections $2(n-i)+1$ to $2i-1$ (for $n-i>0$ and note that if $b_1$ is in an even section then we require $i > \frac{n}{2}$, otherwise we require $i \geq \frac{n}{2}$). For these sections, as described above for the case on intersecting the leftmost vertical edge of $\mathcal{P}$, the partitioning lines between Sections $2l$ and $2l+1$ no longer exist; they would be $\mathcal{CC}^3(w_{i+l})$ but $w_{i+l}$ does not exist. Therefore $V^+(b_1)$ takes the same form for $b_1$ in Sections $2l$ and $2l+1$ and we can explore them together. However, we must still check the first section (Section $2(n-i)+1$), for which (as described for Section $2i$ for $V^+(b_1)$ touching the leftmost vertical edge of $\mathcal{P}$) there is no Section $2l$ with which it can be paired.

If $b_1$ is in Section $2(n-i)+1$ (so $2l+1$ where $l=n-i$) or in Section $2l$ or Section $2l+1$ for $l=n-i+1,...,i-1$ then, making use of our calculations for $V^+(b_1)$ not touching either vertical boundary of $\mathcal{P}$, for $i<n$,
\begin{equation*}
    \begin{split}
        Area(V^+(b_1)) &= Area(V^+(b_1) \cap V^\circ(w_{i-l})) + \sum_{j=i-l+1}^{i-1}Area(V^+(b_1) \cap V^\circ(w_j)) \\ &\qquad+ Area(V^+(b_1) \cap V^\circ(w_i)) + \sum_{j=i+1}^{n}Area(V^+(b_1) \cap V^\circ(w_j)) + Area(V^+(b_1) \cap V^\circ(w_{i+l})) \\
        &= - \frac{x^2}{4} - \frac{3y^2}{4} + (-\frac{lp}{2n} + \frac{q}{2})y + \frac{lpq}{2n} - \frac{(l-1)lp^2}{4n^2} %original calculation:: -(l^2 p^2)/(4 n^2) + (l p^2)/(4 n^2) + (l p q)/(2 n) - (l p y)/(2 n) + (q y)/2 - x^2/4 - (3 y^2)/4
        - \sum_{j=n+1}^{i+l-1}``Area(V^+(b_1) \cap V^\circ(w_j))" \\ &\qquad- ``Area(V^+(b_1) \cap V^\circ(w_{i+l}))" \\
        &= - \frac{x^2}{4} - \frac{3y^2}{4} + (-\frac{lp}{2n} + \frac{q}{2})y + \frac{lpq}{2n} - \frac{(l-1)lp^2}{4n^2}
        - \sum_{j=n+1}^{i+l-1}(\frac{p}{2n}x - \frac{p}{2n}y + \frac{pq}{2n} - \frac{(4(j-i)-1)p^2}{8n^2}) \\
        &\qquad- (\frac{x^2}{8} -\frac{3y^2}{8} -\frac{xy}{4} + (\frac{q}{4} - \frac{((i+l)-i-1)p}{4n})x + (\frac{((i+l)-i-1)p}{4n} + \frac{q}{4})y \\
        &\qquad- \frac{((i+l)-i-1)pq}{4n} + \frac{((i+l)-i-1)^2p^2}{8n^2}) \\
        &= - \frac{3x^2}{8} - \frac{3y^2}{8} + \frac{xy}{4} + (\frac{(l-1)p}{4n} - \frac{q}{4})x + (- \frac{(3l-1)p}{4n} + \frac{q}{4})y \\
        &\qquad+ \frac{(3l-1)pq}{4n} - \frac{(3l-1)(l-1)p^2}{8n^2} - (i+l-1-n)(\frac{p}{2n}x - \frac{p}{2n}y + \frac{pq}{2n} + \frac{p^2}{8n^2}) \\
        &\qquad+ \frac{p^2}{2n^2}\sum_{j=n+1}^{i+l-1}(j-i) \\
        &= - \frac{3x^2}{8} - \frac{3y^2}{8} + \frac{xy}{4} - (\frac{(l-2n+2i-1)p}{4n} + \frac{q}{4})x + (- \frac{(l+2n-2i+1)p}{4n} + \frac{q}{4})y \\
        &\qquad+ \frac{(l+2n-2i+1)pq}{4n} - \frac{(l^2 - l + 2n^2 + n - 4in + 2i^2 - i)p^2}{8n^2} \\ % - (3 x^2)/8 - (3 y^2)/8 + (x y)/4 - (((a - 2 n + 2 i - 1) p)/(4 n) + q/4) x + (-((a + 2 n - 2 i + 1) p)/(4 n) + q/4) y + ((a + 2 n - 2 i + 1) p q)/(4 n) - ((a^2 - a + 2 n^2 + n - 4 i n + 2 i^2 - i) p^2)/(8 n^2)
    \end{split}
\end{equation*}
or, for $i=n$, 
\begin{equation*}
    \begin{split}
        Area(V^+(b_1)) &= Area(V^+(b_1) \cap V^\circ(w_{n-l})) + \sum_{j=n-l+1}^{n-1}Area(V^+(b_1) \cap V^\circ(w_j)) \\ &\qquad+ Area(V^+(b_1) \cap V^\circ(w_n)) \\
        &= \frac{x^2}{8} - \frac{3y^2}{8} + \frac{xy}{4} + (\frac{(n-(n-l)-1)p}{4n} - \frac{q}{4})x + (\frac{(n-(n-l)-1)p}{4n} + \frac{q}{4})y \\ 
        &\qquad - \frac{(n-(n-l)-1)pq}{4n} + \frac{(n-(n-l)-1)^2p^2}{8n^2} \\
        &\qquad+ \sum_{j=n-l+1}^{n-1}(- \frac{p}{2n}x - \frac{p}{2n}y + \frac{pq}{2n} - \frac{(4(n-j)-1)p^2}{8n^2}) - \frac{x^2}{2} - \frac{p}{2n}y + \frac{pq}{2n} \\
        &= - \frac{3x^2}{8} - \frac{3y^2}{8} + \frac{xy}{4} + (\frac{(l-1)p}{4n} - \frac{q}{4})x + (\frac{(l-3)p}{4n} + \frac{q}{4})y - \frac{(l-3)pq}{4n} + \frac{(l-1)^2p^2}{8n^2} \\
        &\qquad+ (n-1-(n-l))(- \frac{p}{2n}x - \frac{p}{2n}y + \frac{pq}{2n} + \frac{p^2}{8n^2}) - \frac{p^2}{2n^2}\sum_{j=n-l+1}^{n-1}(n-j) \\
        &= - \frac{3x^2}{8} - \frac{3y^2}{8} + \frac{xy}{4} - (\frac{(l-1)p}{4n} + \frac{q}{4})x + (-\frac{(l+1)p}{4n} + \frac{q}{4})y + \frac{(l+1)pq}{4n} - \frac{(l-1)lp^2}{8n^2} \, . %-((a - 1) a p^2)/(8 n^2) - x (((a - 1) p)/(4 n) + q/4) + y (q/4 - ((a + 1) p)/(4 n)) + ((a + 1) p q)/(4 n) - (3 x^2)/8 + (x y)/4 - (3 y^2)/8
    \end{split}
\end{equation*}
As before, this is identical to the representation found by substituting $i=n$ into the previous area formula, so it is this former formula that we use for our studies.

The area has partial derivatives
\begin{equation*}
\begin{split}
\frac{\delta A}{\delta x}&= - \frac{3x}{4} + \frac{y}{4} - \frac{(l-2n+2i-1)p}{4n} - \frac{q}{4} \\
\frac{\delta A}{\delta y}&= - \frac{3y}{4} + \frac{x}{4} - \frac{(l+2n-2i+1)p}{4n} + \frac{q}{4} \\
&\Rightarrow - 2x^* - \frac{(2l-2n+2i-1)p}{2n} - \frac{q}{2} = 0 \Rightarrow x^* =  - \frac{(2l-2(n-i)-1)p}{4n} - \frac{q}{4} \\
\text{and }&\Rightarrow - 2y^* - \frac{(2l+2n-2i+1)p}{2n} + \frac{q}{2} = 0 \Rightarrow y^* = - \frac{(2l+2(n-i)+1)p}{4n} + \frac{q}{4} \\
\end{split}
\end{equation*}
but $x^* = \frac{(2(n-i-l)+1)p}{4n} - \frac{q}{4} \leq \frac{1}{4}(\frac{p}{n} - q) < 0$ so we are required to investigate when $b_1$ is placed on the boundaries of its respective section -- noting that the optimum will never lie on a non-endpoint of $x=\frac{p}{2n}$ because $x^*<0$.

Within Section $2(n-i)+1$, for which $Area(V^+(b_1))= - \frac{3x^2}{8} - \frac{3y^2}{8} + \frac{xy}{4} + (\frac{(n-i+1)p}{4n} - \frac{q}{4})x + (- \frac{(3(n-i)+1)p}{4n} + \frac{q}{4})y + \frac{(3(n-i)+1)pq}{4n} - \frac{3(n-i)^2p^2}{8n^2}$, we produce the following calculations.%(-3 (-i + n)^2 p^2)/(8 n^2) + ((1 + 3 (-i + n)) p q)/(4 n) + (((1 - i + n) p)/(4 n) - q/4) x - (3 x^2)/8 + (-((1 + 3 (-i + n)) p)/(4 n) + q/4) y + (x y)/4 - (3 y^2)/8
\begin{itemize}
    \item Upon boundary $y=\frac{(n-i)p}{n}-x$ we have $Area(V^+((x,\frac{(n-i)p}{n}-x)))= - x^2 + (\frac{(4(n-i)+1)p}{2n} - \frac{q}{2})x + \frac{(4(n-i)+1)pq}{4n} - \frac{(6(n-i)^2 + (n-i))p^2}{4n^2}$, maximised by $x^*=\frac{(4(n-i)+1)p}{4n} - \frac{q}{4}$ giving $Area(V^+(\\(\frac{(4(n-i)+1)p}{4n} - \frac{q}{4},-\frac{p}{4n} + \frac{q}{4})))= \frac{(4(n-i) + 1)pq}{8n} - \frac{(8(n-i)^2 - 4(n-i) - 1)p^2}{16n^2} + \frac{q^2}{16}$. However, for $0 \leq \frac{(4(n-i)+1)p}{4n} - \frac{q}{4} \leq \frac{p}{2n}$ to be true we require $\frac{(4(n-i)-1)p}{n} \leq q \leq \frac{(4(n-i)+1)p}{n}$. If $q \leq \frac{(4(n-i)-1)p}{n}$ then the optimum lies on the endpoint $(\frac{p}{2n},\frac{(2(n-i)-1)p}{2n})$ giving $Area(V^+((\frac{p}{2n},\frac{(2(n-i)-1)p}{2n}))) \\= \frac{(n-i)pq}{n} - \frac{(6(n-i)^2 - 3(n-i))p^2}{4n^2}$. If $\frac{(4(n-i)+1)p}{n} \leq q$ then the optimum lies on the endpoint $(0,\frac{(n-i)p}{n})$ giving $Area(V^+((0,\frac{(n-i)p}{n}))) = \frac{(4(n-i)+1)pq}{4n} - \frac{(6(n-i)^2 + (n-i))p^2}{4n^2}$.
    \item Upon boundary $y=x+\frac{(n-i)p}{n}$ we have $Area(V^+((x,x+\frac{(n-i)p}{n})))= - \frac{x^2}{2} - \frac{(n-i)p}{n}x + \frac{(4(n-i)+1)pq}{4n} - \frac{(6(n-i)^2 + n-i)p^2}{4n^2}$, maximised by $x^* = -\frac{(n-i)p}{n} < 0$ so the optimum is achieved at $x=0$, the value of which has been found above.
\end{itemize}
Thus the optimum lies on the boundary $y=\frac{(n-i)p}{n}-x$ with maximal areas $Area(V^+((0,\frac{(n-i)p}{n}))) = \frac{(4(n-i)+1)pq}{4n} - \frac{(6(n-i)^2 + (n-i))p^2}{4n^2}$ if $\frac{(4(n-i)+1)p}{n} \leq q$ as depicted in Figure~\ref{fig:RowOptimalRightTouchingIIIc}, $Area(V^+((\frac{(4(n-i)+1)p}{4n} \\- \frac{q}{4},-\frac{p}{4n} + \frac{q}{4})))= \frac{(4(n-i) + 1)pq}{8n} - \frac{(8(n-i)^2 - 4(n-i) - 1)p^2}{16n^2} + \frac{q^2}{16}$ if $\frac{(4(n-i)-1)p}{n} \leq q \leq \frac{(4(n-i)+1)p}{n}$ as depicted in Figure~\ref{fig:RowOptimalRightTouchingIIIb}, and $Area(V^+((\frac{p}{2n},\frac{(2(n-i)-1)p}{2n}))) = \frac{(n-i)pq}{n} - \frac{(6(n-i)^2 - 3(n-i))p^2}{4n^2}$ if $q \leq \frac{(4(n-i)-1)p}{n}$ as depicted in Figure~\ref{fig:RowOptimalRightTouchingIIIa}.

\begin{figure}[!ht]\ContinuedFloat
\begin{subfigure}{.5\textwidth}
  \centering
  \includegraphics[width=0.9\textwidth]{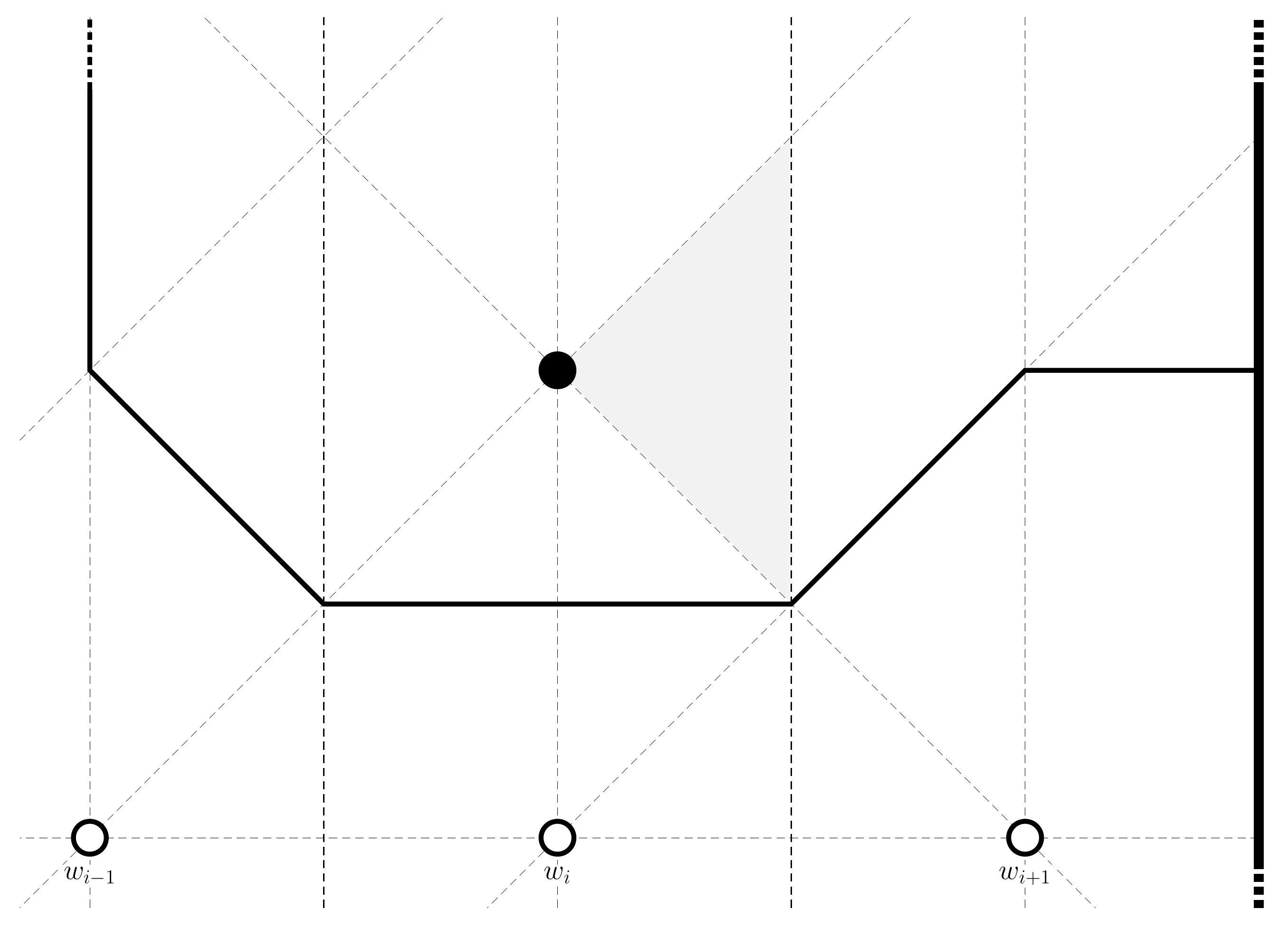}
  \caption{$b_1=(0,\frac{(n-i)p}{n})$ only if $\frac{(4(n-i)+1)p}{n} \leq q$. \\ \,}
%  \caption{$Area(V^+((0,\frac{(n-i)p}{n}))) = \frac{(4(n-i)+1)pq}{4n} - \frac{(6(n-i)^2 + (n-i))p^2}{4n^2}$ only if $\frac{(4(n-i)+1)p}{n} \leq q$.}
  \label{fig:RowOptimalRightTouchingIIIc}
\end{subfigure}%
\begin{subfigure}{.5\textwidth}
  \centering
  \includegraphics[width=0.9\textwidth]{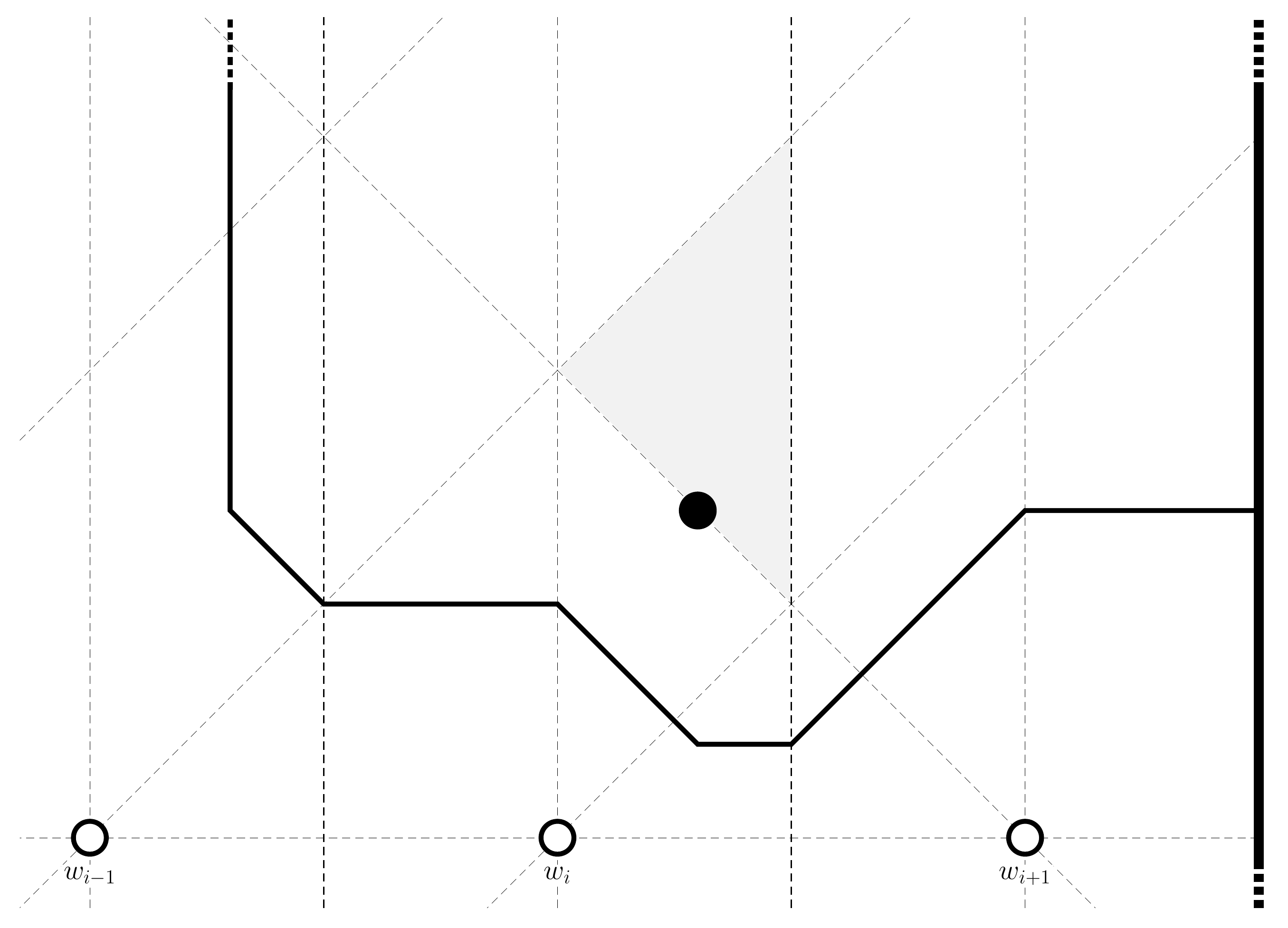}
  \caption{$b_1=(\frac{(4(n-i)+1)p}{4n} - \frac{q}{4},-\frac{p}{4n} + \frac{q}{4})$ only if $\frac{(4(n-i)-1)p}{n} \leq q \leq \frac{(4(n-i)+1)p}{n}$.}
%  \caption{$Area(V^+((\frac{(4(n-i)+1)p}{4n} - \frac{q}{4},-\frac{p}{4n} + \frac{q}{4})))= \frac{(4(n-i) + 1)pq}{8n} - \frac{(8(n-i)^2 - 4(n-i) - 1)p^2}{16n^2} + \frac{q^2}{16}$ only if $\frac{(4(n-i)-1)p}{n} \leq q \leq \frac{(4(n-i)+1)p}{n}$.}
  \label{fig:RowOptimalRightTouchingIIIb}
\end{subfigure}

\begin{subfigure}{1.0\textwidth}
  \centering
  \includegraphics[width=0.45\textwidth]{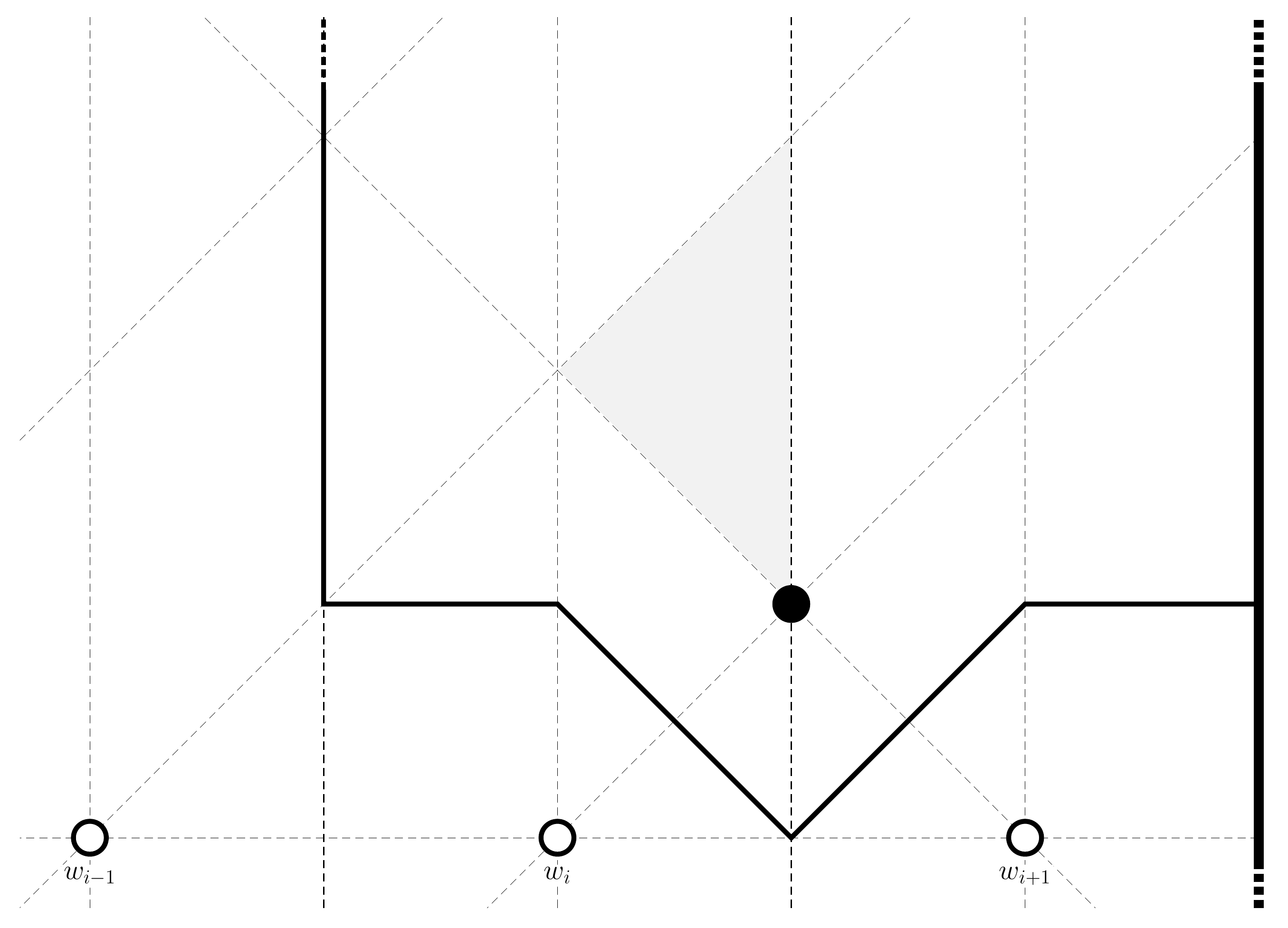}
  \caption{$b_1=(\frac{p}{2n},\frac{(2(n-i)-1)p}{2n})$ only if $\frac{(4(n-i)-1)p}{n} \geq q$.}
%  \caption{$Area(V^+((\frac{p}{2n},\frac{(2(n-i)-1)p}{2n}))) = \frac{(n-i)pq}{n} - \frac{(6(n-i)^2 - 3(n-i))p^2}{4n^2}$ only if $\frac{(4(n-i)-1)p}{n} \geq q$.}
  \label{fig:RowOptimalRightTouchingIIIa}
\end{subfigure}
\caption{Maximal area Voronoi cells $V^+(b_1)$ for $b_1$ within Section $2(n-i)+1$ touching the rightmost vertical edge of $\mathcal{P}$.}
\end{figure}

Alternatively, consider placing on the boundary of Section $2l$ and Section $2l+1$ (i.e. upon the boundaries $x=0$, $y=x+\frac{(l-1)p}{n}$, $x=\frac{p}{2n}$, and $y=x+\frac{lp}{n}$).
%Since $y^* = - \frac{(2l+2(n-i)+1)p}{4n} + \frac{q}{4} \leq - \frac{(2(n-i+1)+2(n-i)+1)p}{4n} + \frac{q}{4} = - \frac{(4(n-i)+3)p}{4n} + \frac{q}{4} ..?$ we will not find the optimum on the boundary $y=l\frac{p}{n}-x$ outside its endpoints.
\begin{itemize}
    \item Upon boundary $x=0$ we have $Area(V^+((0,y)))= - \frac{3y^2}{8} + (- \frac{(l+2n-2i+1)p}{4n} + \frac{q}{4})y + \frac{(l+2n-2i+1)pq}{4n} - \frac{(l^2 - l + 2n^2 + n - 4in + 2i^2 - i)p^2}{8n^2}$, maximised by $y^*= - \frac{(l+2n-2i+1)p}{3n} + \frac{q}{3}$ giving $Area(V^+(0, - \frac{(l+2n-2i+1)p}{3n} + \frac{q}{3}))=\frac{(l+2n-2i+1)pq}{6n} - \frac{(2(l-(n-i))^2 - 5l - (n-i) - 1)p^2}{24n^2} + \frac{q^2}{24}$. However, for $\frac{(l-1)p}{n} \leq - \frac{(l+2n-2i+1)p}{3n} + \frac{q}{3} \leq \frac{lp}{n}$ to be true we require $\frac{(4l+2(n-i)-2)p}{n} \leq q \leq \frac{(4l+2(n-i)+1)p}{n}$. If $\frac{(4l+2(n-i)-2)p}{n} \geq q$ then the optimum lies on the endpoint $(0,\frac{(l-1)p}{n})$ giving $Area(V^+((0,\frac{(l-1)p}{n})))= \frac{(l+n-i)pq}{2n} - \frac{(6l^2 + 4l(n-i) + 2(n-i)^2 - 7l - 3(n-i) + 1)p^2}{4n^2}$, and if $\frac{(4l+2(n-i)+1)p}{n} \leq q$ then $Area(V^+((0,\frac{lp}{n})))= \frac{(2l+2(n-i) + 1)pq}{4n} - \frac{(6l^2 + 4l(n-i) + 2(n-i)^2 + l + (n-i))p^2}{8n^2}$.
    \item Upon boundary $y=x+(l-1)\frac{p}{n}$ we have $Area(V^+((x,x+\frac{(l-1)p}{n})))= - \frac{x^2}{2} - \frac{(2l-1)p}{2n}x + \frac{(l+n-i)pq}{2n} - \frac{(6l^2 + 4l(n-i) + 2(n-i)^2 - 7l - 3(n-i) + 1)p^2}{8n^2}$, maximised by $x^* = - \frac{(2l-1)p}{2n} < 0$ so the optimum is achieved at $x=0$, the value of which has been found above.% \todo{These did not match originally -- coeff of $l^2$ was $-3$ in $x=0$ calc (changedwithoutrecalculation) so may want to check this}
    \item Upon boundary $y=x+l\frac{p}{n}$ we have $Area(V^+((x,x+\frac{lp}{n})))= - \frac{x^2}{2} - \frac{lp}{n}x + \frac{(2l+2(n-i)+1)pq}{4n} - \frac{(6l^2 + 4l(n-i) + 2(n-i)^2 + l + n - i)p^2}{8n^2}$, maximised by $x^* = - \frac{lp}{n}< 0$ so the optimum is achieved at $x=0$, the value of which has been found above.
%    \item Upon boundary $y=l\frac{p}{n} - x$ we have $Area(V^+((x,\frac{lp}{n}-x)))= - x^2 + (\frac{(2l+2(n-i)+1)p}{2n} - \frac{q}{2})x + \frac{(2l+2(n-i)+1)pq}{4n} - \frac{(6l^2 + 4l(n-i) + 2(n-i)^2 + l + n - i)p^2}{8n^2}$, maximised by $x^*=\frac{(2l+2(n-i)+1)p}{4n} - \frac{q}{4}$ giving $Area(V^+((\frac{(2l+2(n-i)+1)p}{4n} - \frac{q}{4},\frac{(2l-2(n-i)-1))p}{4n} + \frac{q}{4})))=\frac{(2l+2(n-i)+1)pq}{8n} - \frac{(8l^2 - 2l - 2(n-i) - 1)p^2}{16n^2} + \frac{q^2}{16}$. However for $0 \leq \frac{(2l+2(n-i)+1)p}{4n} - \frac{q}{4} \leq \frac{p}{2n}$ to be true we require $\frac{(2l+2(n-i)-1)p}{n} \leq q \leq \frac{(2l+2(n-i)+1)p}{n}$. If $\frac{(2l+2(n-i)-1)p}{n} \geq q$ then the optimum lies on the endpoint $(\frac{p}{2n},\frac{(2l-1)p}{2n})$ which, by the $y=x+(l-1)\frac{p}{n}$ boundary observations, is never optimal. If $\frac{(2l+2(n-i)+1)p}{n} \leq q$ then the optimum lies on the endpoint $(0,\frac{lp}{n})$, the area of which has already been found.
\end{itemize}
%This brings us to a similar situation as we were in when optimising the area of cells within Section $2l$ which intersect the leftmost boundary of $\mathcal{P}$. We know that if $\frac{(2l+2(n-i)-1)p}{n} \geq q$ or $\frac{(2l+2(n-i)+1)p}{n} \leq q$ then the optimum lies on the boundary $x=0$. What remains to be seen is whether it is otherwise optimal to place on the boundary $y=\frac{lp}{n} - x$ or $x=0$ when $\frac{(2l+2(n-i)-1)p}{n} \leq q \leq \frac{(2l+2(n-i)+1)p}{n}$ so again we are required to compare the optimum values found within these boundaries.

%Now the values for this problem of determining the overall optimum require us to follow a different approach to that for Section $2l$ areas intersecting the leftmost boundary of $\mathcal{P}$. If $\frac{(4l+2(n-i)-2)p}{n} \geq q$ then the optimum on $x=0$ lies on $(0,\frac{lp}{n})$ and so will be improved upon by playing WHAT VALUES? if $\frac{(2l+2(n-i)+1)p}{n} \leq q$ then the optimum on $y=\frac{lp}{n}-x$ lies on the endpoint $(0,\frac{lp}{n})$.

Thus the optimum lies on the boundary $x=0$, with maximal areas $Area(V^+((0,\frac{lp}{n})))= \frac{(2l+2(n-i) + 1)pq}{4n} - \frac{(6l^2 + 4l(n-i) + 2(n-i)^2 + l + (n-i))p^2}{8n^2}$ if $\frac{(4l+2(n-i)+1)p}{n} \leq q$ as depicted in Figure \ref{fig:RowOptimalRightTouchingIVc}, $Area(V^+(0, - \frac{(l+2n-2i+1)p}{3n} + \frac{q}{3}))=\frac{(l+2n-2i+1)pq}{6n} - \frac{(2(l-(n-i))^2 - 5l - (n-i) - 1)p^2}{24n^2} + \frac{q^2}{24}$ if $\frac{(4l+2(n-i)-2)p}{n} \leq q \leq \frac{(4l+2(n-i)+1)p}{n}$ as depicted in Figure~\ref{fig:RowOptimalRightTouchingIVb}, and $Area(V^+((0,\frac{(l-1)p}{n}))) = \frac{(l+n-i)pq}{2n} - \\ \frac{(6l^2 + 4l(n-i) + 2(n-i)^2 - 7l - 3(n-i) + 1)p^2}{4n^2}$ if $\frac{(4l+2(n-i)-2)p}{n} \geq q$ as depicted in Figure \ref{fig:RowOptimalRightTouchingIVa}.

\begin{figure}[!ht]\ContinuedFloat
\begin{subfigure}{.5\textwidth}
  \centering
  \includegraphics[width=0.9\textwidth]{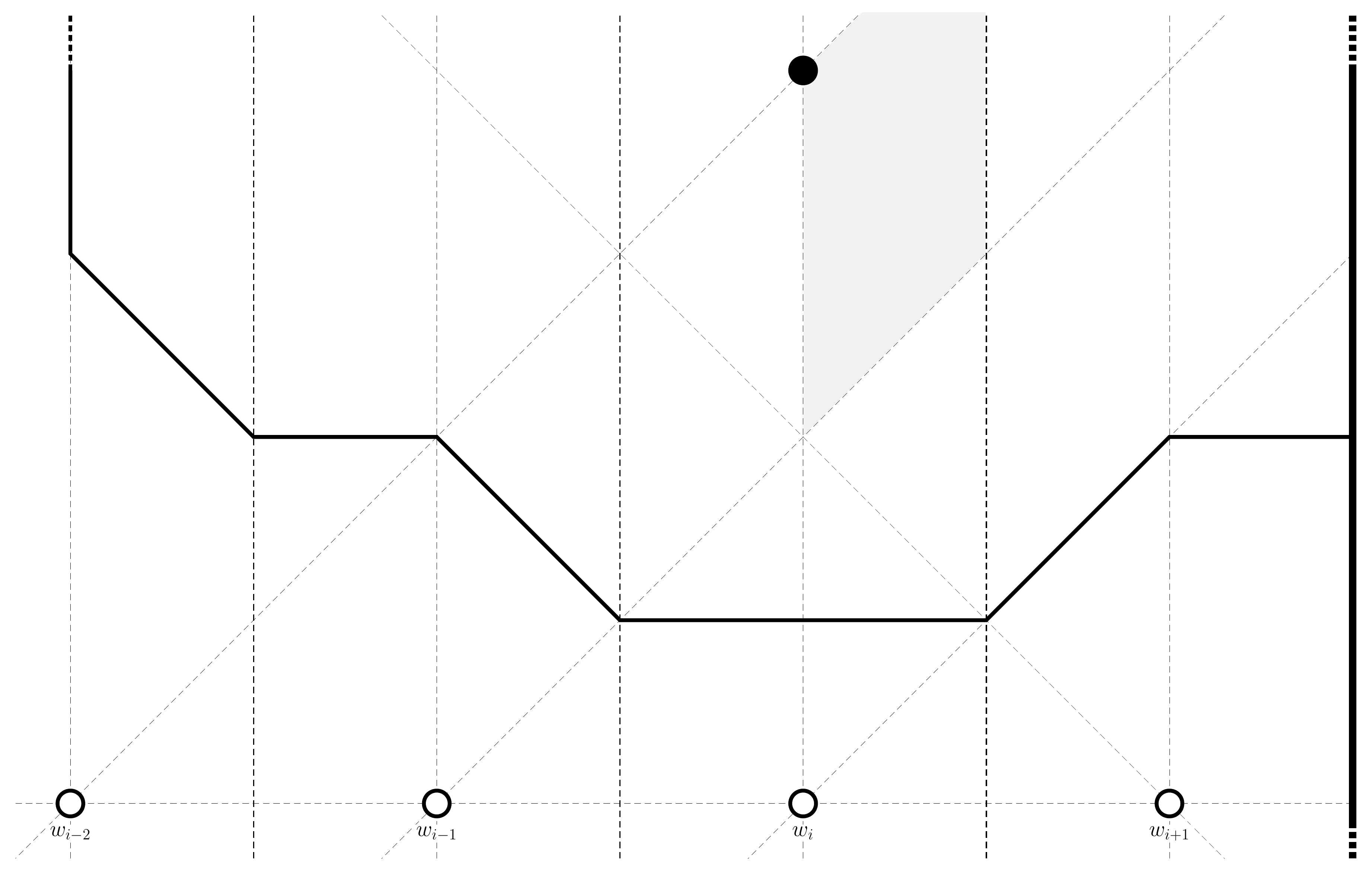}
  \caption{$b_1=(0,\frac{lp}{n})$ only if $\frac{(4l+2(n-i)+1)p}{n} \leq q$. \\ \,}
%  \caption{$Area(V^+((0,\frac{lp}{n})))= \frac{(2l+2(n-i) + 1)pq}{4n} - \frac{(6l^2 + 4l(n-i) + 2(n-i)^2 + l + (n-i))p^2}{8n^2}$ only if $\frac{(4l+2(n-i)+1)p}{n} \leq q$.}
  \label{fig:RowOptimalRightTouchingIVc}
\end{subfigure}%
\begin{subfigure}{.5\textwidth}
  \centering
  \includegraphics[width=0.9\textwidth]{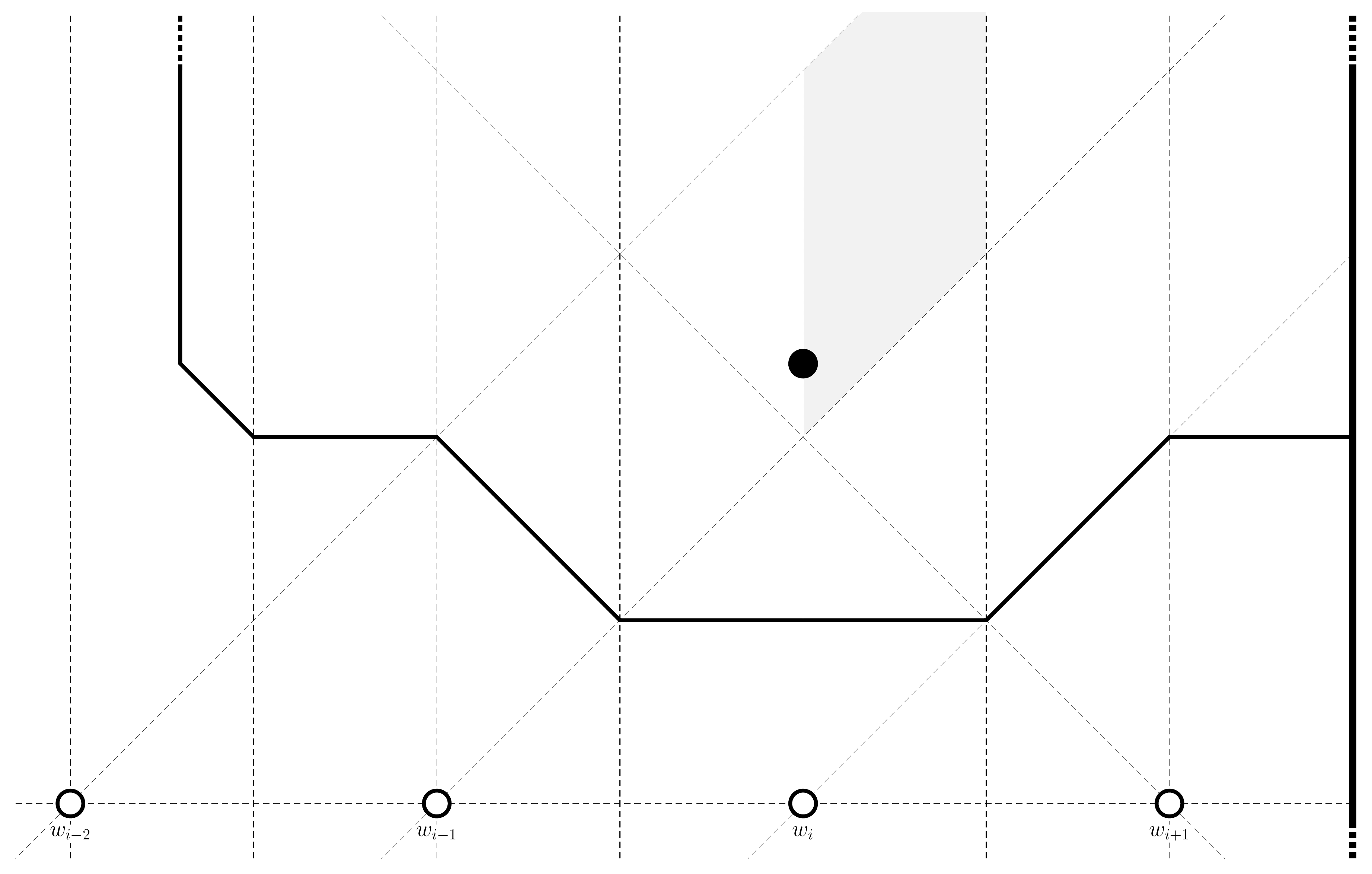}
  \caption{$b_1=(0, - \frac{(l+2n-2i+1)p}{3n} + \frac{q}{3})$ only if $\frac{(4l+2(n-i)-2)p}{n} \leq q \leq \frac{(4l+2(n-i)+1)p}{n}$.}
%  \caption{$Area(V^+(0, - \frac{(l+2n-2i+1)p}{3n} + \frac{q}{3}))=\frac{(l+2n-2i+1)pq}{6n} - \frac{(2(l-(n-i))^2 - 5l - (n-i) - 1)p^2}{24n^2} + \frac{q^2}{24}$ only if $\frac{(4l+2(n-i)-2)p}{n} \leq q \leq \frac{(4l+2(n-i)+1)p}{n}$.}
  \label{fig:RowOptimalRightTouchingIVb}
\end{subfigure}
\end{figure}
\begin{figure}\ContinuedFloat
\begin{subfigure}{1.0\textwidth}
  \centering
  \includegraphics[width=0.45\textwidth]{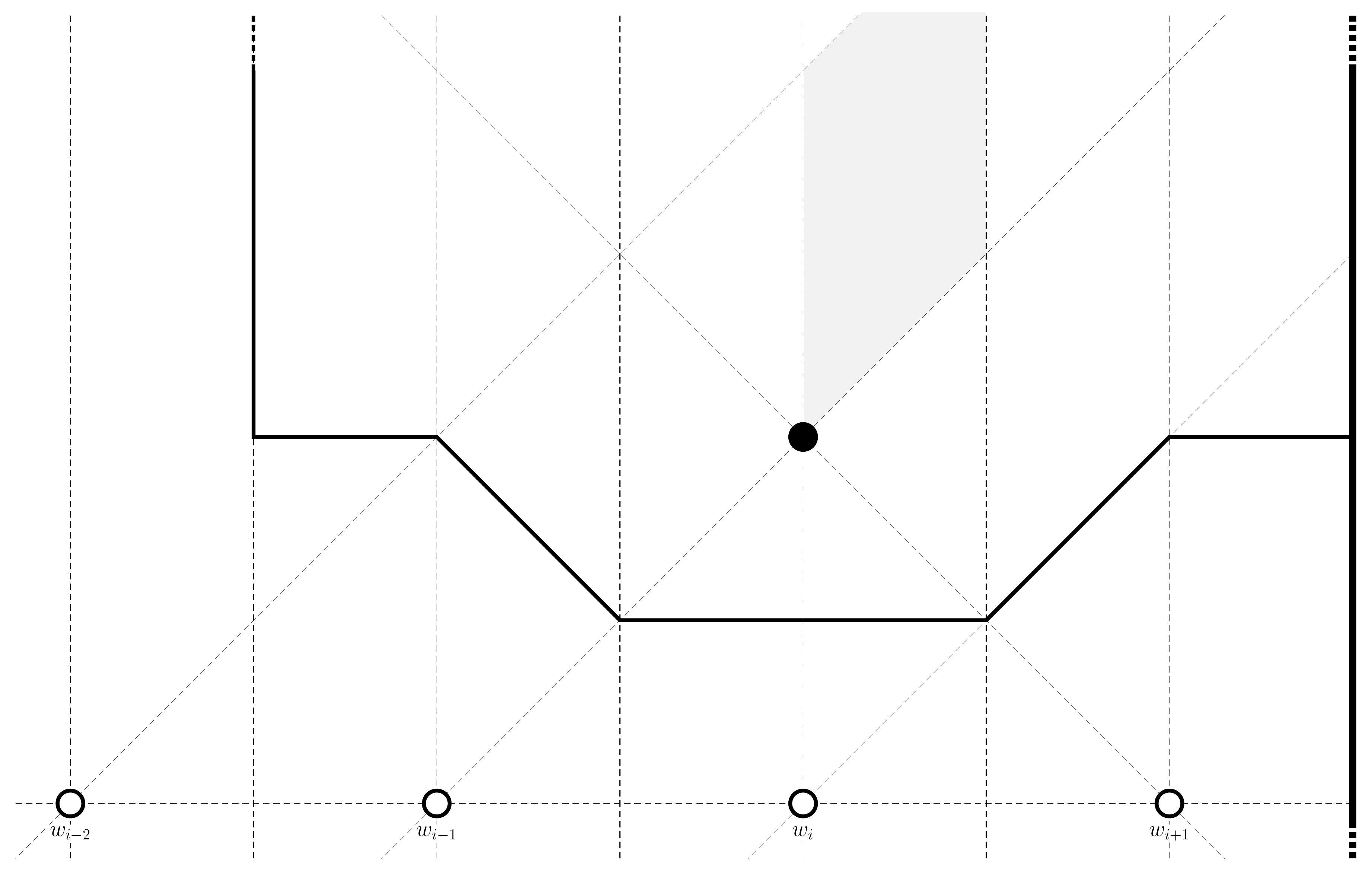}
  \caption{$b_1=(0,\frac{(l-1)p}{n})$ only if $\frac{(4l+2(n-i)-2)p}{n} \geq q$.}
%  \caption{$Area(V^+((0,\frac{(l-1)p}{n})))= \frac{(l+n-i)pq}{2n} - \frac{(6l^2 + 4l(n-i) + 2(n-i)^2 - 7l - 3(n-i) + 1)p^2}{4n^2}$ only if $\frac{(4l+2(n-i)-2)p}{n} \geq q$.}
  \label{fig:RowOptimalRightTouchingIVa}
\end{subfigure}
\caption{Maximal area Voronoi cells $V^+(b_1)$ for $b_1$ within Section $2l$ touching the rightmost vertical edge of $\mathcal{P}$.}
\end{figure}

\pagebreak We should note here that Section $I$ also applies in this case of intersecting the right boundary and not the left boundary of $\mathcal{P}$, though it is plain to see that the optimum for this scenario will lie as close as possible to $(0,0)$ and give an area up to (but not achieving) $\frac{pq}{2n}$.

\subsection{\texorpdfstring{$V^+(b_1)$}{V+(b1)} touching both vertical edges of \texorpdfstring{$\mathcal{P}$}{P}}

Finally we shall investigate the points $b_1$ whose cells $V^+(b_1)$ touch both vertical boundaries of $\mathcal{P}$. These cells are produced for $b_1$ in Section $2i$ and above if $i > \frac{n}{2}$ or Section $2(n-i)+1$ and above if $i \leq \frac{n}{2}$. Importantly, within these sections the structure of $V^+(b_1)$ is identical no matter the section, even or odd. This is because, in actuality, there are no sections beyond Section $\max[2i,2(n-i)+1]$ as it is defined by the edges $x=0$, $x=\frac{p}{2n}$, $\mathcal{CC}^1(w_1)$, and $\mathcal{CC}^3(w_n)$.

Therefore the area for $b_1$ in this region is, for $1 < i < n$,
\begin{equation*}
    \begin{split}
        Area(V^+(b_1)) &= \sum_{j=1}^{i-1}Area(V^+(b_1) \cap V^\circ(w_j)) + Area(V^+(b_1) \cap V^\circ(w_i)) \\ &\qquad+ \sum_{j=i+1}^{n}Area(V^+(b_1) \cap V^\circ(w_j)) \\
        &= \sum_{j=1}^{i-1}(- \frac{p}{2n}x - \frac{p}{2n}y + \frac{pq}{2n} - \frac{(4(i-j)-1)p^2}{8n^2}) - \frac{x^2}{2} - \frac{p}{2n}y + \frac{pq}{2n} \\ &\qquad+ \sum_{j=i+1}^{n}(\frac{p}{2n}x - \frac{p}{2n}y + \frac{pq}{2n} - \frac{(4(j-i)-1)p^2}{8n^2}) \\
        &= (i-1)(- \frac{p}{2n}x - \frac{p}{2n}y + \frac{pq}{2n} + \frac{p^2}{8n^2}) - \frac{p^2}{2n^2}\sum_{j=1}^{i-1}(i-j) - \frac{x^2}{2} - \frac{p}{2n}y + \frac{pq}{2n} \\
        &\qquad+ (n-i)(\frac{p}{2n}x - \frac{p}{2n}y + \frac{pq}{2n} + \frac{p^2}{8n^2}) - \frac{p^2}{2n^2}\sum_{j=i+1}^{n}(j-i) \\
        &= - \frac{x^2}{2} + (n-2i+1)\frac{p}{2n}x - \frac{p}{2}y + \frac{pq}{2} + (n-1)\frac{p^2}{8n^2} - \frac{p^2}{2n^2}\frac{(i-1)i}{2} \\ &\qquad- \frac{p^2}{2n^2}\frac{(n-i)(n-i+1)}{2} \\
        &= - \frac{x^2}{2} + \frac{(n-2i+1)p}{2n}x - \frac{p}{2}y + \frac{pq}{2} - \frac{(2n^2 + 4i^2 - 4in + n - 4i + 1)p^2}{8n^2} \\
    \end{split}
\end{equation*}
or, for $i=1$,
\begin{equation*}
    \begin{split}
        Area(V^+(b_1)) &= Area(V^+(b_1) \cap V^\circ(w_1)) + \sum_{j=2}^{n}Area(V^+(b_1) \cap V^\circ(w_j)) \\
        &= - \frac{x^2}{2} - \frac{p}{2n}y + \frac{pq}{2n} + \sum_{j=2}^{n}(\frac{p}{2n}x - \frac{p}{2n}y + \frac{pq}{2n} - \frac{(4(j-1)-1)p^2}{8n^2}) \\
        &= - \frac{x^2}{2} - \frac{p}{2n}y + \frac{pq}{2n} + (n-1)(\frac{p}{2n}x - \frac{p}{2n}y + \frac{pq}{2n} + \frac{5p^2}{8n^2}) - \frac{p^2}{2n^2}\sum_{j=2}^{n}j \\
        &= - \frac{x^2}{2} + \frac{(n-1)p}{2n}x - \frac{p}{2}y + \frac{pq}{2} - \frac{(2n^2-3n+1)p^2}{8n^2} \\
    \end{split}
\end{equation*}
or, for $i=n$,
\begin{equation*}
    \begin{split}
        Area(V^+(b_1)) &= \sum_{j=1}^{n-1}Area(V^+(b_1) \cap V^\circ(w_j)) + Area(V^+(b_1) \cap V^\circ(w_n)) \\
        &= \sum_{j=1}^{n-1}(- \frac{p}{2n}x - \frac{p}{2n}y + \frac{pq}{2n} - \frac{(4(n-j)-1)p^2}{8n^2}) - \frac{x^2}{2} - \frac{p}{2n}y + \frac{pq}{2n} \\
        &= (n-1)(- \frac{p}{2n}x - \frac{p}{2n}y + \frac{pq}{2n} - \frac{(4n-1)p^2}{8n^2}) + \frac{p^2}{2n^2}\sum_{j=1}^{n-1}{j} - \frac{x^2}{2} - \frac{p}{2n}y + \frac{pq}{2n} \\
        &= - \frac{x^2}{2} - \frac{(n-1)p}{2n}x - \frac{p}{2}y + \frac{pq}{2} - \frac{(2n^2 - 3n + 1)p^2}{8n^2} \, .
    \end{split}
\end{equation*}

All of these areas have partial derivative $\frac{\delta A}{\delta y}= - \frac{p}{2}$ providing, as expected, justification that the area increases as $y$ decreases within the region.

If $1<i<n$ then $$\frac{\delta A}{\delta x}= - x + \frac{(n-2i+1)p}{2n}$$ giving $x^*=\frac{(n-2i+1)p}{2n}$. We have $x^* \geq 0 \Leftrightarrow n-2i+1 \geq 0 \Leftrightarrow \frac{n+1}{2} \geq i$ and $x^* \leq \frac{p}{2n} \Leftrightarrow n-2i+1 \leq 1 \Leftrightarrow \frac{n}{2} \leq i$ so this maximum is only achieved for $i=\ceil{\frac{n}{2}}$. In this case, if $n$ is even then the maximum within Section $n+1$ of $w_{\frac{n}{2}}$ is found at $x^*=\frac{p}{2n}$, and if $n$ is odd then the maximum within Section $n+1$ of $w_{\frac{n+1}{2}}$ is found at $x^*=0$. Before explicitly stating the coordinates of $b^*_1$ for these sections we will explore those values of $i$ which did not satisfy these constraints.

For $1<i<n$ where $i \neq \ceil{\frac{n}{2}}$, $x^*$ is never within this region. Therefore we must explore the boundary of the region; by $\frac{\delta A}{\delta y}$ we need only explore the lower boundary.

Since $x^*<0$ when $i > \frac{n+1}{2}$, for these $i$ the region we are exploring is Section $2i$ and the bottommost point on the lower boundary (satisfying $\frac{\delta A}{\delta y}$) is also the leftmost point (satisfying $\frac{\delta A}{\delta x}$) so this point, $(0,\frac{(i-1)p}{n})$, is our optimum, as well as being the optimum for $i=\frac{n+1}{2}$ when $n$ is odd as found above. This gives $Area(V^+((0,\frac{(i-1)p}{n})))= \frac{pq}{2} - \frac{(2 n^2 + 4 i^2 - 3n - 4 i + 1)p^2}{8n^2}$ as depicted in Figure \ref{fig:RowOptimalBothTouchingIV}.

Since $x^*>\frac{p}{2n}$ when $i<\frac{n}{2}$, for these $i$ the region we are exploring is Section $2(n-i)+1$ and the bottommost point on the lower boundary (satisfying $\frac{\delta A}{\delta y}$) is also the rightmost point (satisfying $\frac{\delta A}{\delta x}$), so this point, $(\frac{p}{2n},\frac{(2(n-i)-1)p}{2n})$, is our optimum, as well as being the optimum for $i=\frac{n}{2}$ when $n$ is even as found above. This gives $Area(V^+((\frac{p}{2n},\frac{(2(n-i)-1)p}{2n})))= \frac{pq}{2} - \frac{(6n^2 + 4i^2 - 8in - 3n)p^2}{8n^2}$ as depicted in Figure \ref{fig:RowOptimalBothTouchingIII}.

\begin{figure}[!ht]\ContinuedFloat
\begin{subfigure}{.5\textwidth}
  \centering
  \includegraphics[width=0.9\textwidth]{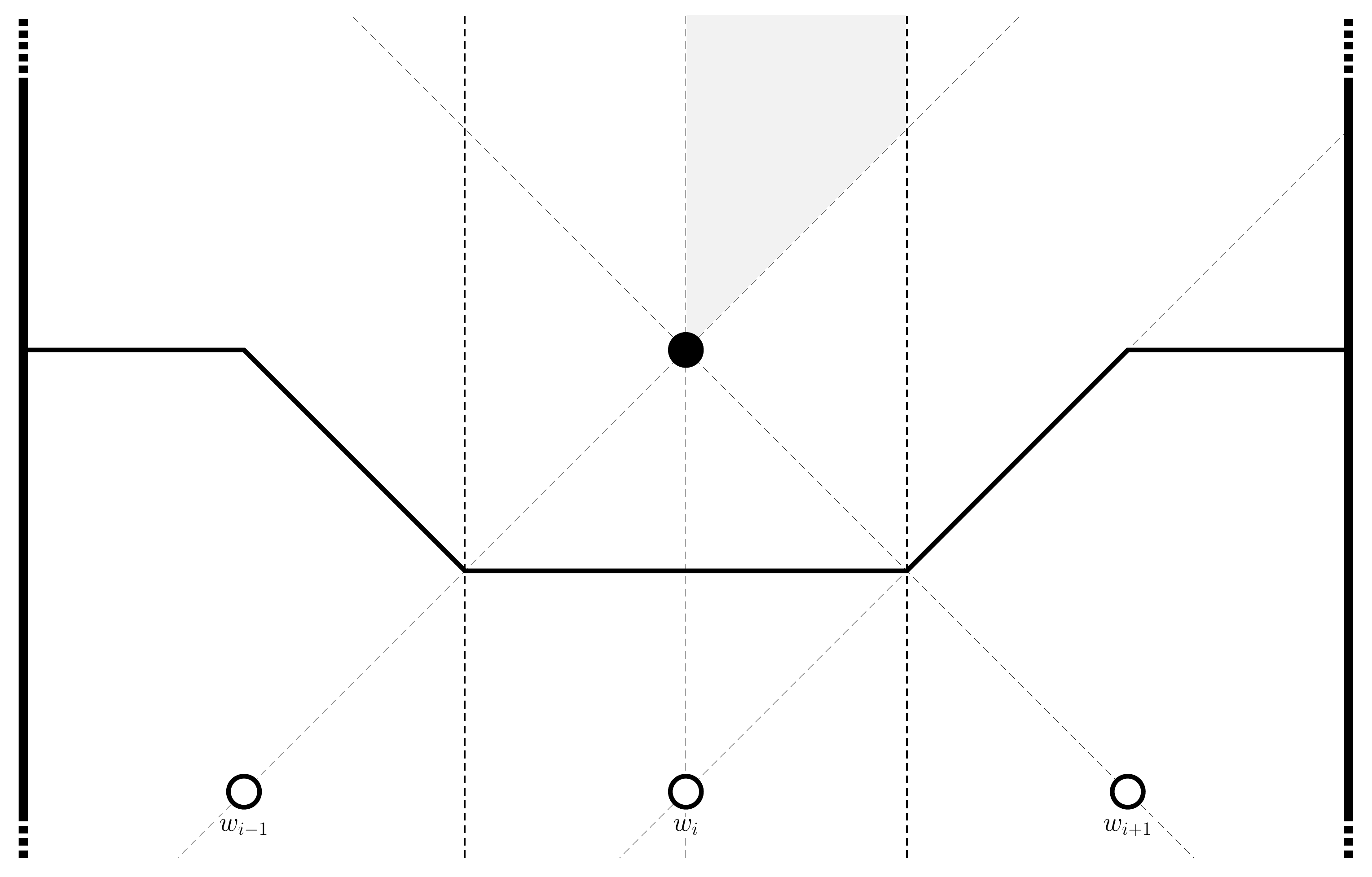}
  \caption{$b_1=(0,\frac{(i-1)p}{n})$ only if $i > \frac{n}{2}$.}
%  \caption{$Area(V^+((0,\frac{(i-1)p}{n})))= \frac{pq}{2} - \frac{(2 n^2 + 4 i^2 - 3n - 4 i + 1)p^2}{8n^2}$ only if $i > \frac{n}{2}$.}
  \label{fig:RowOptimalBothTouchingIV}
\end{subfigure}%
\begin{subfigure}{.5\textwidth}
  \centering
  \includegraphics[width=0.9\textwidth]{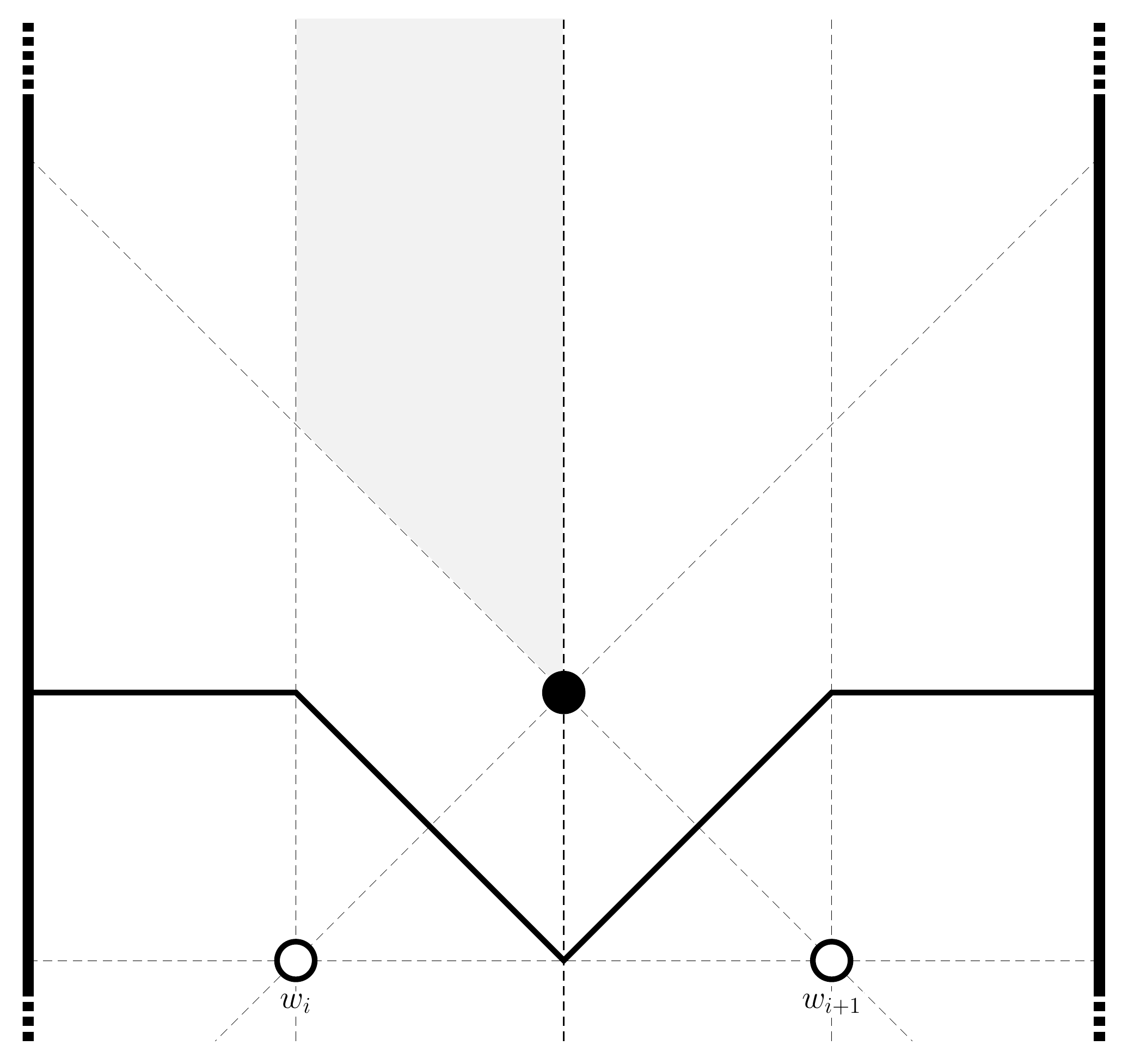}
  \caption{$b_1=(\frac{p}{2n},\frac{(2(n-i)-1)p}{2n})$ only if $i \leq \frac{n}{2}$.}
%  \caption{$Area(V^+((\frac{p}{2n},\frac{(2(n-i)-1)p}{2n})))= \frac{pq}{2} - \frac{(6n^2 + 4i^2 - 8in - 3n)p^2}{8n^2}$ only if $i \leq \frac{n}{2}$.}
  \label{fig:RowOptimalBothTouchingIII}
\end{subfigure}
\caption{Maximal area Voronoi cells $V^+(b_1)$ for $b_1$ within Section $n+1$, touching both vertical edges of $\mathcal{P}$.}
\label{fig:RowOptimals}
\end{figure}

%?\fi

If $i=1$ then $$\frac{\delta A}{\delta x}= - x + \frac{(n-1)p}{2n}$$ giving $x^*=\frac{(n-1)p}{2n} \geq \frac{p}{2n}$. As before, since $i < \frac{n}{2}$ the region is $2(n-1)+1$ so our optimum lies at the bottom rightmost point of Section $2n-1$, $(\frac{p}{2n},\frac{(2n-3)p}{2n})$, giving $Area(V^+((\frac{p}{2n},\frac{(2n-3)p}{2n})))=\frac{pq}{2} - \frac{(6n^2-11n+4)p^2}{8n^2}$ (identical to the above calculation for $i<\frac{n}{2}$ after substituting $i=1$).

Finally, if $i=n$ then $$\frac{\delta A}{\delta x}= - x - \frac{(n-1)p}{2n}$$ giving $x^*=- \frac{(n-1)p}{2n} < 0$. As before, since $i > \frac{n}{2}$ the region is $2n$ so our optimum lies at the bottom leftmost point of Section $2n$, $(0,\frac{(n-1)p}{n})$, giving $Area(V^+((0,\frac{(n-1)p}{n})))=\frac{pq}{2} - \frac{(6 n^2 - 7 n + 1)p^2}{8n^2}$ (identical to the above calculation for $i \geq \frac{n}{2}$ after substituting $i=n$).

This concludes our search for the optimisation of each structure of $V^+(b_1)$ which touches both vertical boundaries of $\mathcal{P}$, and with it our search for the optimisation of every structure of $V^+(b_1)$ given that White plays a $1 \times n$ row.

\section{Black's optimal strategy: White plays a \texorpdfstring{$1 \times n$}{1 x n} row}
\label{sec:BlackRow}

At this stage we have calculated the optimal locations of $b_1$ within every possible partition cell of $\mathcal{P}$ when White plays a $1 \times n$ row. To recap, Figure \ref{fig:RowOptimals} shows all optimal locations of $b_1$ within each section under the certain circumstances discussed above (not depicting the optimum found in Section $I$ since this had location $(0,0)$).

%$b^*_1=(\frac{p}{2n},\frac{q}{3}-\frac{(2l+1)p}{6n})$ and $Area(V^+((\frac{p}{2n},\frac{q}{3}-\frac{(2l+1)p}{6n})))=\frac{(2l+1)pq}{6n} - \frac{(2l^2+2l-1)p^2}{12n^2} + \frac{q^2}{12}$. For $b^*_1$ to lie within Section $2l$ we must have $l\frac{p}{n} - x^* \leq y^* \leq x^* + l\frac{p}{n}$ so we must have $\frac{(4l-1)p}{n} \leq q \leq \frac{(4l+2)p}{n}$.

\subsection{Black's best point}
\label{sec:BlackRowBest}

An obvious question of interest is which point $b_1$ is the best point -- as in, which position of $b_1$ gives the largest area of $V^+(b_1)$? The availability of each section in which to place $b_1$, and the areas of the Voronoi cells $V^+(b_1)$, depend entirely on the relationship between $\frac{p}{n}$ and $q$ so this is not a straightforward question to answer. Nevertheless we shall determine what position $b^*=(x^*,y^*)$ of $b_1$ claims the largest area of $V^+(b_1)$ for which ratios of $\frac{p}{n}$ and $q$.

%Firstly, it is clear from the results of $V^+(b_1)$ touching both vertical edges of $\mathcal{P}$ that if $V^+(b_1)$ touches both vertical edges of $\mathcal{P}$ then the area is maximised by the point
Let us begin by fixing the bottom right corner of $\mathcal{P}$ at $(0,-\frac{q}{2})$ so that $w_i=(\frac{(2i-1)p}{2n},0)$ for $i=1,...,n$. Firstly it is clear to see from Figure~\ref{fig:RowOptimals} that Black's best point $b^*$ will have $x$-coordinate $\frac{kp}{2n}$ for some $k \in \mathbb{N}$. Furthermore, it is never advantageous when seeking to maximise the area of $V^+(b_1)$ to have $V^+(b_1)$ bounded on one side by a vertical edge of $\mathcal{P}$ since this blocks $V^+(b_1)$ from gaining territory on the other side of the boundary of $\mathcal{P}$, which it might be able to do if $b_1$ were moved a horizontal distance of $\frac{p}{n}$ away from this edge. For this reason we can easily claim that the best point $b^*$ has $x$-coordinate $x^* \in \{\frac{p}{2} - \frac{p}{2n}, \frac{p}{2}, \frac{p}{2}+\frac{p}{2n}\}$. By the symmetry of $\mathcal{P}$ and $W$, not only does $x^*=\frac{p}{2}$ produce a Voronoi cell symmetrical about $x=\frac{p}{2}$, but the values $x^*=\frac{p}{2}-\frac{p}{2n}$ and $x^*=\frac{p}{2}+\frac{p}{2n}$ will produce identical Voronoi cells (after reflection in $x^*=\frac{p}{2}$). Therefore we need only consider the location of $b^*$ within the top right quadrant of $V^\circ(w_i)$ for $i = \ceil{\frac{n}{2}}$, though requiring different investigations depending on whether $n$ is even or odd.

Before we delve into the details with respect to the parity of $n$, let us recapitulate the results depicted in Figure~\ref{fig:RowOptimals} in Table~\ref{tab:rowsectionoptimals}.

\begin{table}[!ht]
\centering
\begin{tabular}{c|c|c|c}
Section & Optimum & Area & Condition \\
\hline
$2l$ & $\left(\frac{(2i-1)p}{2n},(l-1)\frac{p}{n}\right)$ & $\frac{(2l-1)pq}{2n} - \frac{3(2l-1)(l-1)p^2}{4n^2}$ & $q \leq \frac{(4l-3)p}{n}$ \\
  & $\left(\frac{(2i-1)p}{2n},\frac{q}{3}-\frac{lp}{3n}\right)$ & $\frac{lpq}{3n} + \frac{(3-2l)lp^2}{12n^2} + \frac{q^2}{12}$ & $\frac{(4l-3)p}{n} \leq q \leq \frac{4lp}{n}$ \\
  & $\left(\frac{(2i-1)p}{2n},l\frac{p}{n}\right)$ & $\frac{lpq}{n} + \frac{(1-6l)lp^2}{4n^2}$ & $\frac{4lp}{n} \leq q$ \\
%  \hline
$2l+1$ & $\left(\frac{ip}{n},\frac{(2l-1)p}{2n}\right)$ & $\frac{lpq}{n} - \frac{(3l-1)lp^2}{2n^2}$ & $q \leq \frac{(4l-1)p}{n}$ \\
  & $\left(\frac{ip}{n},\frac{q}{3}-\frac{(2l+1)p}{6n}\right)$ & $\frac{(2l+1)pq}{6n} - \frac{(2l^2+2l-1)p^2}{12n^2} + \frac{q^2}{12}$ & $\frac{(4l-1)p}{n} \leq q \leq \frac{(4l+2)p}{n}$ \\
  & $\left(\frac{ip}{n},\frac{(2l+1)p}{2n}\right)$ & $\frac{(2l+1)pq}{2n}-\frac{(6l^2+6l+1)p^2}{4n^2}$ & $\frac{(4l+2)p}{n} \leq q$ \\
\end{tabular}
\caption{Optima contained in each section of $V^\circ(w_i)$ assuming that Black's Voronoi cell does not touch either vertical edge of $\mathcal{P}$.}
\label{tab:rowsectionoptimals}
\end{table}
We shall refer to these optima as the \emph{bottom}, \emph{middle}, and \emph{top} optima within each section, listed in the order that they appear in Table~\ref{tab:rowsectionoptimals} with examples depicted in Figures~\ref{fig:RowOptimalNotTouchingIVc} and \ref{fig:RowOptimalNotTouchingIIIc}, Figures~\ref{fig:RowOptimalNotTouchingIVb} and \ref{fig:RowOptimalNotTouchingIIIb}, and Figures~\ref{fig:RowOptimalNotTouchingIVa} and \ref{fig:RowOptimalNotTouchingIIIa} respectively.

We know the optimal positions $b^*_1$ within each of these sections, but we must ask how these optima compare to one another across sections. It is important to realise that some optima within different sections lie on the same point, while capturing different areas (for example the equivalent optima in Figure~\ref{fig:RowOptimalNotTouchingIVc} for Section $VI$ would lie on the same point as shown in Figure~\ref{fig:RowOptimalNotTouchingIVa}). This is due to the fact that many of these positions represent the convergence of $b_1$ to a point, yet these different results are obtained from converging via different paths (i.e. via different sections), choosing different bisectors upon degenerate configuration lines. These are the easiest comparisons to make and can be done by simply referring to graphs of the points as shown, by way of an example, in Figure~\ref{fig:RowOptimalComp}.

\begin{figure}[!ht]
\begin{subfigure}{.5\textwidth}
  \centering
  \includegraphics[width=0.9\textwidth]{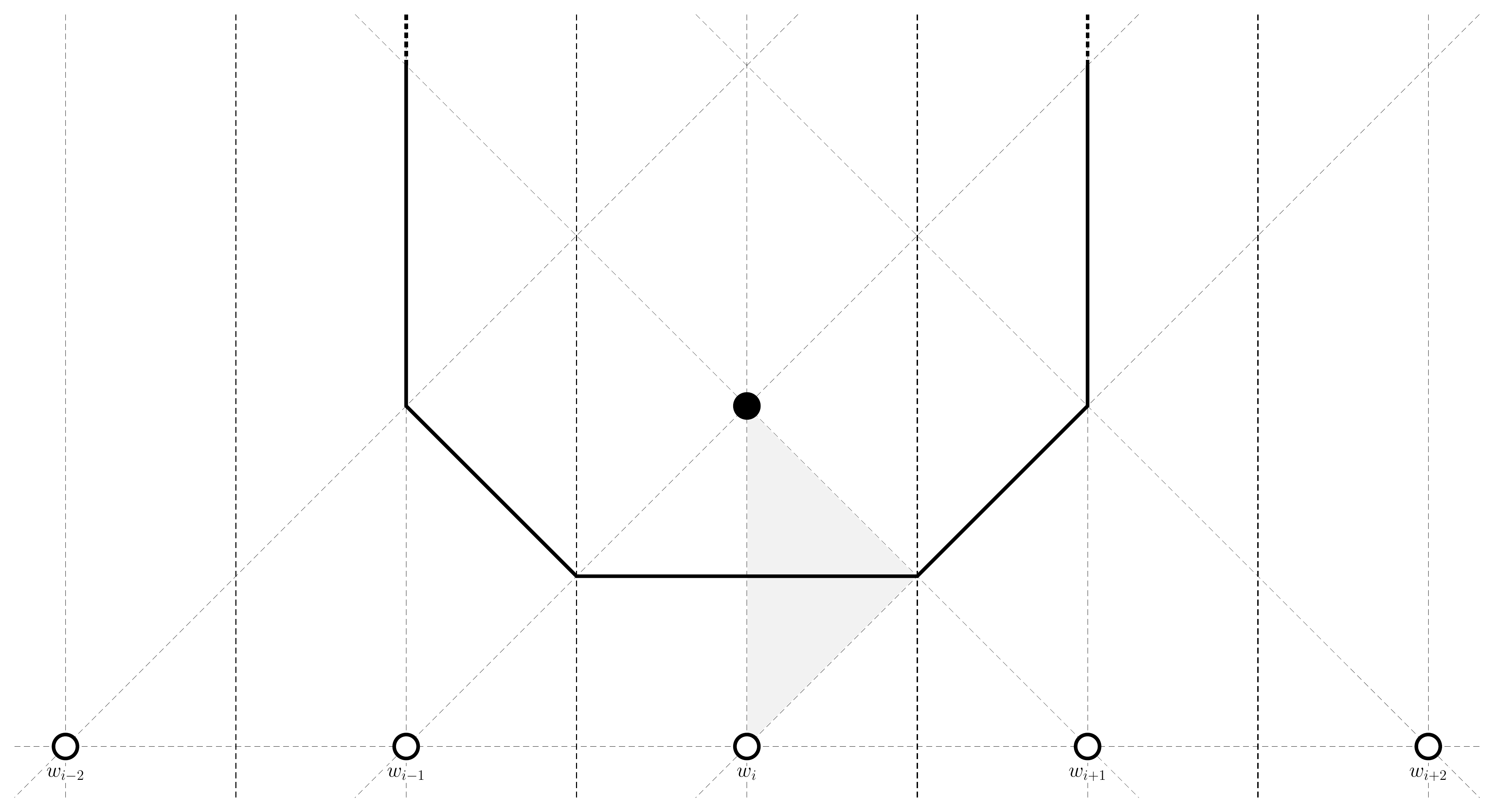}
  \caption{$V^+((\frac{(2i-1)p}{2n},\frac{p}{n}))$ for $(\frac{(2i-1)p}{2n},\frac{p}{n})$ in Section $II$.}
  \label{fig:RowOptimalNotTouchingIIcomp}
\end{subfigure}%
\begin{subfigure}{.5\textwidth}
  \centering
  \includegraphics[width=0.9\textwidth]{VoronoiGameRowPartitionIVOptimalc.pdf}
  \caption{$V^+((\frac{(2i-1)p}{2n},\frac{p}{n}))$ for $(\frac{(2i-1)p}{2n},\frac{p}{n})$ in Section $IV$.}
  \label{fig:RowOptimalNotTouchingIVcomp}
\end{subfigure}
\caption{Comparison of identical optimal positions within different sections (shaded).}
\label{fig:RowOptimalComp}
\end{figure}

From Figure~\ref{fig:RowOptimalComp} it is clear to see that if an optimal point we are comparing is located on the boundary of two sections, the upper section will always outperform the lower section. Therefore the remaining optima to consider are the middle and the bottom optima, as well as the optima for $V^+(b_1)$ touching the appropriate vertical boundaries of $\mathcal{P}$.

Now we ask when, if ever, it is better to locate in Section $k$ as opposed to Section $k+2$ for $k>0$, assuming that the resulting Voronoi cell of Black does not touch either vertical boundary of $\mathcal{P}$.

For even Sections $2l$, we know that the bottom optimum of Section $2(l+1)$ is the optimum over Section $2(l+1)$ if $\frac{(4(l+1)-3)p}{n} = \frac{(4l+1)p}{n} \geq q$, whereas the middle optimum of Section $2l$ is the optimum over Section $2l$ if $\frac{(4l-3)p}{n} \leq q \leq \frac{4lp}{n}$ so we must compare the area that the bottom optimum of Section $2(l+1)$ captures compared to that of the middle optimum of $2l$ when $\frac{(4l-3)p}{n} \leq q \leq \frac{4lp}{n}$:

\begin{equation*}
\begin{split}
    \left( \frac{(2(l+1)-1)pq}{2n} \right. &- \left. \frac{3(2(l+1)-1)((l+1)-1)p^2}{4n^2} \right) - \left(\frac{lpq}{3n} + \frac{(3-2l)lp^2}{12n^2} + \frac{q^2}{12}\right) \\
    &= \frac{(2l+1)pq}{2n} - \frac{3(2l+1)lp^2}{4n^2} - \frac{lpq}{3n} + \frac{(2l-3)lp^2}{12n^2} - \frac{q^2}{12} \\
    &= \frac{(4l+3)pq}{6n} - \frac{(4l+3)lp^2}{3n^2} - \frac{q^2}{12} \geq 0 \\
    &\Leftrightarrow q^2 - \frac{2(4l+3)p}{n}q + \frac{4(4l+3)lp^2}{n^2} \leq 0 \, .
%    &OR
%    &= \frac{(4l+3)pq}{6n} - \frac{(4l+3)lp^2}{3n^2} - \frac{q^2}{12} \\
%    &\geq \frac{(4l+3)pq}{6n} - \frac{(4l+3)lp^2}{3n^2} - \frac{16l^2p^2}{12n^2} \\
%    &= \frac{(4l+3)pq}{6n} - \frac{(8l+3)lp^2}{3n^2} > 0 \\
%    &\Rightarrow q > \frac{2(8l+3)lp}{(4l+3)n} \\
%    &OR \\
%    &= \frac{(4l+3)pq}{6n} - \frac{(4l+3)lp^2}{3n^2} - \frac{q^2}{12} \\
%    &\leq \frac{(4l+3)pq}{6n} - \frac{(4l+3)lp^2}{3n^2} - \frac{q^2}{12} \\
%    &= \frac{(4l+3)p}{3n}(\frac{q}{2} - \frac{lp}{n}) - \frac{q^2}{12} \, .
\end{split}
\end{equation*}
Now it is the case that
\begin{equation*}
\begin{split}
    q^2 - \frac{2(4l+3)p}{n}q + \frac{4(4l+3)lp^2}{n^2} &= (q - \frac{(4l+3)p}{n})^2 + \frac{4(4l+3)lp^2}{n^2} - (\frac{(4l+3)p}{n})^2 \\
    &= (q - \frac{(4l+3)p}{n})^2 - \frac{3(4l+3)p^2}{n^2} \leq 0 \\
    &\Leftrightarrow \frac{(4l+3)p}{n} - \frac{\sqrt{3(4l+3)}p}{n} \leq q \leq \frac{(4l+3)p}{n} + \frac{\sqrt{3(4l+3)}p}{n} \, .
\end{split}
\end{equation*}

Now it remains to find the intersection of $[\frac{(4l-3)p}{n}, \frac{4lp}{n}]$ (the values of $q$ for which the middle optimum is the optimum over Section $2l$) and $[\frac{(4l+3 - \sqrt{3(4l+3)})p}{n}, \frac{(4l+3 + \sqrt{3(4l+3)})p}{n}]$ (the values of $q$ for which the bottom optimum of Section $2(l+1)$ is better than the middle optimum of Section $2l$). It is clear that $\frac{4lp}{n} < \frac{(4l+3 + \sqrt{3(4l+3)})p}{n}$. More involved is the following calculation:
\begin{equation*}
\begin{split}
    \frac{(4l-3)p}{n} - \frac{(4l+3 - \sqrt{3(4l+3)})p}{n} &= \frac{(\sqrt{3(4l+3)} - 6)p}{n} \geq 0 \\
    &\Leftrightarrow \sqrt{3(4l+3)} - 6 \geq 0 \\
    &\Leftrightarrow 3(4l+3) \geq 36 \\
    &\Leftrightarrow l \geq \frac{9}{4} \, .
\end{split}
\end{equation*}

Therefore if $l \leq 2$ then the bottom optimum of Section $2(l+1)$ is better than the middle optimum of Section $2l$ for $\frac{(4l+3 - \sqrt{3(4l+3)})p}{n} \leq q$. Otherwise, if $l \geq 3$, the bottom optimum of Section $2(l+1)$ will always be better than the middle optimum of Section $2l$, and so we must compare the bottom optima of both sections when $\frac{(4l-3)p}{n} \geq q$:
\begin{equation*}
\begin{split}
    \left(\frac{(2(l+1)-1)pq}{2n} \right. &- \left. \frac{3(2(l+1)-1)((l+1)-1)p^2}{4n^2}\right) - \left(\frac{(2l-1)pq}{2n} - \frac{3(2l-1)(l-1)p^2}{4n^2}\right) \\
    &= \frac{(2l+1)pq}{2n} - \frac{3(2l+1)lp^2}{4n^2} - \frac{(2l-1)pq}{2n} + \frac{3(2l-1)(l-1)p^2}{4n^2} \\
    &= \frac{pq}{n} - \frac{3(4l-1)p^2}{4n^2} \geq 0 \\
    &\Leftrightarrow \frac{3(4l-1)p}{4n} \leq q\, .
\end{split}
\end{equation*}
Now $$\frac{3(4l-1)p}{4n} \leq \frac{(4l-3)p}{n} \Leftrightarrow 3(4l-1) \leq 4(4l-3) \Leftrightarrow l \geq \frac{9}{4}$$ so if $l \geq 3$ (the condition which requires us to compare these two local optima) then the bottom optimum of Section $2(l+1)$ is better than the bottom optimum of Section $2l$ for $\frac{3(4l-1)p}{4n} \leq q$.

Let us digest our findings, for which it may be more intuitive to describe the efficacy of each section's optima from Section $II$ upwards. These results are summarised in the following table.

\begin{table}[!ht]
\centering
\begin{tabular}{c|c|c|c}
Section & Optimum & Area & Condition \\
\hline
$II$ & $(x^*,0)$ & $\frac{pq}{2n}$ & $q \leq \frac{p}{n}$ \\
  & $(x^*,\frac{q}{3}-\frac{p}{3n})$ & $\frac{pq}{3n} + \frac{p^2}{12n^2} + \frac{q^2}{12}$ & $\frac{p}{n} \leq q \leq \frac{(7-\sqrt{21})p}{n}$ \\
$IV$ & $(x^*,\frac{p}{n})$ & $\frac{3pq}{2n} - \frac{9p^2}{4n^2}$ & $\frac{(7-\sqrt{21})p}{n} \leq q \leq \frac{5p}{n}$ \\
  & $(x^*,\frac{q}{3}-\frac{2p}{3n})$ & $\frac{2pq}{3n} - \frac{p^2}{6n^2} + \frac{q^2}{12}$ & $\frac{5p}{n} \leq q \leq \frac{(11-\sqrt{33})p}{n}$ \\
$VI$ & $(x^*,\frac{2p}{n})$ & $\frac{5pq}{2n} - \frac{15p^2}{2n^2}$ & $\frac{(11-\sqrt{33})p}{n} \leq q \leq \frac{33p}{4n}$ \\
$2l$ & $(x^*,(l-1)\frac{p}{n})$ & $\frac{(2l-1)pq}{2n} - \frac{3(2l-1)(l-1)p^2}{4n^2}$ & $\frac{3(4l-5)p}{4n} \leq q \leq \frac{3(4l-1)p}{4n}$ \\
% &  &  & $3 < l < \frac{n}{2}$ \\
\end{tabular}
\caption{Optima contained in even sections at $x^*=\frac{(2i-1)p}{2n}$ assuming that Black's Voronoi cell does not touch either vertical edge of $\mathcal{P}$.}
\label{tab:rowevensectionoptimals}
\end{table}

Now we must analogously explore the odd sections $2l+1$. The bottom optimum of Section $2(l+1)+1$ is the optimum of Section $2(l+1)+1$ if $q \leq \frac{(4(l+1)-1)p}{n} = \frac{(4l+3)p}{n}$ whereas the middle optimum of Section $2l+1$ is the optimum over Section $2l+1$ if $\frac{(4l-1)p}{n} \leq q \leq \frac{(4l+2)p}{n}$, so we must compare the area that the bottom optimum of Section $2(l+1)+1$ captures compared to that of the middle optimum of $2l+1$ when $\frac{(4l-1)p}{n} \leq q \leq \frac{(4l+2)p}{n}$:
\begin{equation*}
\begin{split}
    \left( \frac{(l+1)pq}{n} \right. &- \left. \frac{(3(l+1)-1)(l+1)p^2}{2n^2} \right) - \left(\frac{(2l+1)pq}{6n} - \frac{(2l^2+2l-1)p^2}{12n^2} + \frac{q^2}{12}\right) \\
    &= \frac{(l+1)pq}{n} - \frac{(3l+2)(l+1)p^2}{2n^2} - \frac{(2l+1)pq}{6n} + \frac{(2l^2+2l-1)p^2}{12n^2} - \frac{q^2}{12} \\
    &= \frac{(4l+5)pq}{6n} - \frac{(16 l^2 + 28 l + 13)p^2}{12n^2} - \frac{q^2}{12} \geq 0 \\
    &\Leftrightarrow q^2 - \frac{2(4l+5)pq}{n} + \frac{(16 l^2 + 28 l + 13)p^2}{n^2} \leq 0 \, .
\end{split}
\end{equation*}
Now it is the case that
\begin{equation*}
\begin{split}
    q^2 - \frac{2(4l+5)pq}{n} + \frac{(16 l^2 + 28 l + 13)p^2}{n^2} &= \left(q - \frac{(4l+5)p}{n}\right)^2 + \frac{(16 l^2 + 28 l + 13)p^2}{n^2} - \left( \frac{(4l+5)p}{n} \right)^2 \\
    &= \left(q - \frac{(4l+5)p}{n}\right)^2 - \frac{12(l+1)p^2}{n^2} \leq 0 \\
    &\Leftrightarrow \frac{(4l+5)p}{n} - \frac{2\sqrt{3(l+1)}p}{n} \leq q \leq \frac{(4l+5)p}{n} + \frac{2\sqrt{3(l+1)}p}{n} \, .
\end{split}
\end{equation*}

Now it remains to find the intersection of $\left[\frac{(4l-1)p}{n},\frac{(4l+2)p}{n}\right]$ (the values of $q$ for which the middle optimum is the optimum over Section $2l+1$) and $\left[\frac{(4l+5-2\sqrt{3(l+1)})p}{n}, \frac{(4l+5+2\sqrt{3(l+1)})p}{n}\right]$ (the values of $q$ for which the bottom optimum of Section $2(l+1)+1$ is better than the middle optimum of Section $2l+1$). It is clear that $\frac{(4l+2)p}{n} < \frac{(4l+5+2\sqrt{3(l+1)})p}{n}$. More involved is the following calculation:
\begin{equation*}
\begin{split}
    \frac{(4l-1)p}{n} - \frac{(4l+5-2\sqrt{3(l+1)})p}{n} &= \frac{(2\sqrt{3(l+1)}-6)p}{n} \geq 0 \\
    &\Leftrightarrow 2\sqrt{3(l+1)}-6 \geq 0 \\
    &\Leftrightarrow 3(l+1) \geq 9 \\
    &\Leftrightarrow l \geq 2 \, .
\end{split}
\end{equation*}

Therefore if $l = 1$ then the bottom optimum of Section $2(l+1) + 1$ is better than the middle optimum of Section $2l + 1$ for $\frac{(4l+5-2\sqrt{3(l+1)})p}{n} \leq q$. Otherwise, if $l \geq 2$, the bottom optimum of Section $2(l+1)+1$ will always be better than the middle optimum of Section $2l+1$, and so we must compare the bottom optima of both sections when $\frac{(4l-1)p}{n} \geq q$:

\begin{equation*}
\begin{split}
    \left( \frac{(l+1)pq}{n} \right. &- \left. \frac{(3(l+1)-1)(l+1)p^2}{2n^2} \right) - \left( \frac{lpq}{n} - \frac{(3l-1)lp^2}{2n^2} \right) \\
    &= \frac{(l+1)pq}{n} - \frac{(3l+2)(l+1)p^2}{2n^2} - \frac{lpq}{n} + \frac{(3l-1)lp^2}{2n^2} \\
    &= \frac{pq}{n} - \frac{(3l + 1)p^2}{n^2} \geq 0 \\
    &\Leftrightarrow \frac{(3l + 1)p}{n} \leq q\, .
\end{split}
\end{equation*}
Now $$\frac{(3l + 1)p}{n} \leq \frac{(4l-1)p}{n} \Leftrightarrow (3l + 1) \leq (4l-1) \Leftrightarrow l \geq 2$$ so if $l \geq 2$ (the condition which requires us to compare these two local optima) then the bottom optimum of Section $2(l+1)+1$ is better than the bottom optimum of Section $2l+1$ for $\frac{(3l + 1)p}{n} \leq q$.

In contrast to our analysis of the optima in even sections, we must compare Section $III$ with the outlier Section $I$ in order to discern when it is more fruitful to settle with the poor-quality Section $I$ optimum. We can do this simply by comparing the area from the bottom optimum in Section $III$ with the maximum area possible achieved in Section $I$:
$$\left( \frac{pq}{n} - \frac{p^2}{n^2} \right) - \frac{pq}{2n} = \frac{pq}{2n} - \frac{p^2}{n^2} \geq 0 \Leftrightarrow q \geq \frac{2p}{n} \, .$$
Thus, within odd sections we can only do better than $\frac{pq}{2n}$ if $\frac{2p}{n} \leq q$, otherwise it is preferable to locate in Section $I$.

Let us again digest our findings, summarised in Table~\ref{tab:rowoddsectionoptimals}.

\begin{table}[!ht]
\centering
\begin{tabular}{c|c|c|c}
Section & Optimum & Area & Condition \\
\hline
$I$ & $(*,0)$ & $\frac{pq}{2n}$ & $q \leq \frac{2p}{n}$ \\
$III$ & $(x^*,\frac{p}{2n})$ & $\frac{pq}{n} - \frac{p^2}{n^2}$ & $\frac{2p}{n} \leq q \leq \frac{3p}{n}$ \\
  & $(x^*,\frac{q}{3}-\frac{p}{2n})$ & $\frac{pq}{2n} - \frac{p^2}{4n^2} + \frac{q^2}{12}$ & $\frac{3p}{n} \leq q \leq \frac{(9-2\sqrt{6})p}{n}$ \\
$V$ & $(x^*,\frac{3p}{2n})$ & $\frac{2pq}{n} - \frac{5p^2}{n^2}$ & $\frac{(9-2\sqrt{6})p}{n} \leq q \leq \frac{7p}{n}$ \\
$2l+1$ & $(x^*,\frac{(2l-1)p}{2n})$ & $\frac{lpq}{n} - \frac{(3l-1)lp^2}{2n^2}$ & $\frac{(3l - 2)p}{n} \leq q \leq \frac{(3l + 1)p}{n}$ \\
\end{tabular}
\caption{Optima contained in odd sections at $x^*=\frac{ip}{n}$ assuming that Black's Voronoi cell does not touch either vertical edge of $\mathcal{P}$.}
\label{tab:rowoddsectionoptimals}
\end{table}

Having found the optimal locations within the set of even sections and the set of odd sections dependent on the ratio between $\frac{p}{n}$ and $q$, it remains to compare Table~\ref{tab:rowevensectionoptimals} and Table~\ref{tab:rowoddsectionoptimals}. We shall explore the global optima as $q$ increases, starting from the top of the tables and working our way down comparing areas across the tables each time $q$ increases so as to enter a new condition.

Beginning with $q \leq \frac{p}{n}$, both tables give the same maximal area of $\frac{pq}{2n}$ no matter whether locating in Section $I$ and $II$. However, this area can be improved if $\frac{p}{n} \leq q$ by playing the middle optimum of Section $II$ so it is no longer optimal to locate in Section $I$. The next condition occurs when $\frac{2p}{n} \leq q$ so we must compare the middle optimum of Section $II$ with the bottom optimum of Section $III$:

\begin{equation*}
\begin{split}
    \left( \frac{pq}{n} - \frac{p^2}{n^2} \right) - \left( \frac{pq}{3n} + \frac{p^2}{12n^2} + \frac{q^2}{12} \right) &= \frac{pq}{n} - \frac{p^2}{n^2} - \frac{pq}{3n} - \frac{p^2}{12n^2} - \frac{q^2}{12} \\
    &= \frac{2pq}{3n} - \frac{13p^2}{12n^2} - \frac{q^2}{12} \geq 0 \\
    &\Leftrightarrow q^2 - \frac{8pq}{n} + \frac{13p^2}{n^2} \leq 0 \\
%    &\Leftrightarrow \left(q - \frac{4p}{n}\right)^2 + \frac{13p^2}{n^2} - \left( \frac{4p}{n} \right)^2 \leq 0 \\
    &\Leftrightarrow \left(q - \frac{4p}{n}\right)^2 - \frac{3p^2}{n^2} \leq 0 \\
    &\Leftrightarrow \frac{4p}{n} - \frac{\sqrt{3}p}{n} \leq q \leq \frac{4p}{n} + \frac{\sqrt{3}p}{n} \, .
\end{split}
\end{equation*}
Since $\frac{2p}{n} < \frac{(4 - \sqrt{3})p}{n} < \frac{(7 - \sqrt{21})p}{n} < \frac{(4 + \sqrt{3})p}{n}$, the bottom optimum of Section $III$ is better than the middle optimum of Section $II$ for $\frac{(4 - \sqrt{3})p}{n} \leq q$ as long as the middle optimum of Section $II$ is the optimal location within even sections.

The subsequent condition to be met as $q$ increases is when $\frac{(7-\sqrt{21})p}{n} \leq q$ and we must compare the bottom optimum of Section $IV$ to the bottom optimum of Section $III$:
\begin{equation*}
\begin{split}
    \left( \frac{3pq}{2n} - \frac{9p^2}{4n^2} \right) - \left( \frac{pq}{n} - \frac{p^2}{n^2} \right) &= \frac{3pq}{2n} - \frac{9p^2}{4n^2} - \frac{pq}{n} + \frac{p^2}{n^2} \\
    &= \frac{pq}{2n} - \frac{5p^2}{4n^2} \geq 0 \\
    &\Leftrightarrow  q \geq \frac{5p}{2n} \, .
\end{split}
\end{equation*}
The bottom optimum of Section $IV$ is better than the bottom optimum of Section $III$ for $\frac{5p}{2n} \leq q$ and since $\frac{(7-\sqrt{21})p}{n} < \frac{5p}{2n} < \frac{3p}{n}$ it is the global optimum at least until $q = \frac{3p}{n}$.

Subsequently, for $\frac{3p}{n} \leq q \leq \frac{(9-2\sqrt{6})p}{n}$ the optimum in odd sections becomes the middle optimum of Section $III$ so we must compare this to the bottom optimum of Section $IV$:
\begin{equation*}
\begin{split}
    \left( \frac{pq}{2n} - \frac{p^2}{4n^2} + \frac{q^2}{12} \right) - \left( \frac{3pq}{2n} - \frac{9p^2}{4n^2} \right) &= \frac{pq}{2n} - \frac{p^2}{4n^2} + \frac{q^2}{12} - \frac{3pq}{2n} + \frac{9p^2}{4n^2} \\
    &= \frac{q^2}{12} - \frac{pq}{n} + \frac{2p^2}{n^2} \geq 0 \\
    &\Leftrightarrow q^2 - \frac{12pq}{n} + \frac{24p^2}{n^2} \geq 0 \\
    &\Leftrightarrow \left(q - \frac{6p}{n}\right)^2 - \frac{12p^2}{n^2} \geq 0 \\
    &\Leftrightarrow q \leq \frac{6p}{n} - \frac{2\sqrt{3}p}{n} \text{ or } q \geq \frac{6p}{n} + \frac{2\sqrt{3}p}{n} \, .
\end{split}
\end{equation*}
Since $\frac{(6-2\sqrt{3})p}{n} < \frac{3p}{n} < \frac{(9-2\sqrt{6})p}{n} < \frac{(6+2\sqrt{3})p}{n}$ the middle optimum of Section $III$ never beats the bottom optimum of Section $IV$. Given this fact, we know also that the middle optimum of Section $IV$ (which beats the bottom optimum of Section $IV$ at $\frac{5p}{n} \leq q$) beats the middle optimum of Section $III$.

This, when $\frac{(9-2\sqrt{6})p}{n} \leq q$, brings us to the comparison of the bottom optimum of Section $V$ and the middle optimum of Section $IV$:

\begin{equation*}
\begin{split}
    \left( \frac{2pq}{n} - \frac{5p^2}{n^2} \right) - \left( \frac{2pq}{3n} - \frac{p^2}{6n^2} + \frac{q^2}{12} \right) &= \frac{2pq}{n} - \frac{5p^2}{n^2} - \frac{2pq}{3n} + \frac{p^2}{6n^2} - \frac{q^2}{12} \\
    &= \frac{4pq}{3n} - \frac{29p^2}{6n^2} - \frac{q^2}{12} \geq 0 \\
    &\Leftrightarrow q^2 - \frac{16pq}{n} + \frac{58p^2}{n^2} \leq 0 \\
    &\Leftrightarrow \left(q - \frac{8p}{n}\right)^2 - \frac{6p^2}{n^2} \leq 0 \\
    &\Leftrightarrow \frac{8p}{n} - \frac{\sqrt{6}p}{n} \leq q \leq \frac{8p}{n} + \frac{\sqrt{6}p}{n} \, .
\end{split}
\end{equation*}
Since $\frac{(11-\sqrt{33})p}{n}$, the value of $q$ at which the middle optimum of Section $IV$ is no longer optimal for even sections, is less than $\frac{(8-\sqrt{6})p}{n}$, the middle optimum of Section $IV$ is the global optimum for its whole range, and for $\frac{(11-\sqrt{33})p}{n} \leq q$ the bottom optimum of $VI$ (the subsequent optimum in even sections) is the next global optimum. We must therefore compare the bottom optimum of Section $V$ to the bottom optimum of Section $VI$:
\begin{equation*}
\begin{split}
    \left( \frac{2pq}{n} - \frac{5p^2}{n^2} \right) - \left( \frac{5pq}{2n} - \frac{15p^2}{2n^2} \right) &= \frac{2pq}{n} - \frac{5p^2}{n^2} - \frac{5pq}{2n} + \frac{15p^2}{2n^2} \\
    &= \frac{5p^2}{2n^2} - \frac{pq}{2n} \geq 0 \\
    &\Leftrightarrow q \leq \frac{5p}{n} \, .
\end{split}
\end{equation*}
Since $\frac{5p}{n} < \frac{(11-\sqrt{33})p}{n} \leq q$, the bottom optimum of Section $V$ is never better than the bottom optimum of Section $VI$.

At this point, since the condition values are now our general $\frac{3(4l-5)p}{4n}$ and $\frac{(3l - 2)p}{n}$ values (for even and odd sections respectively), we have compared all of the necessary initial areas and can compare the general bottom optima of Section $2l$ and Section $2l+1$ (and of Section $2l+1$ and Section $2(l+1)$) for $l \geq 3$.

At $\frac{(3l - 2)p}{n} \leq q$ we must compare the bottom optima of Section $2l+1$ with that of Section $2l$:
\begin{equation*}
\begin{split}
    \left( \frac{lpq}{n} - \frac{(3l-1)lp^2}{2n^2} \right) &- \left( \frac{(2l-1)pq}{2n} - \frac{3(2l-1)(l-1)p^2}{4n^2} \right) \\
    &= \frac{lpq}{n} - \frac{(3l-1)lp^2}{2n^2} - \frac{(2l-1)pq}{2n} + \frac{3(2l-1)(l-1)p^2}{4n^2} \\
    &= \frac{pq}{2n} - \frac{(7l-3)p^2}{4n^2} \geq 0 \\
    &\Leftrightarrow \frac{(7l-3)p}{2n} \leq q \, .
\end{split}
\end{equation*}
Since $\frac{3(4l-1)p}{4n}$, the value of $q$ at which the bottom optimum of Section $2l$ is bested by the bottom optimum of Section $2(l+1)$, is less than $\frac{(7l-3)p}{2n}$ (because $3(4l-1) \geq 2(7l-3) \Leftrightarrow l \leq \frac{3}{2}$), the bottom optimum of Section $2l+1$ never beats the bottom optimum of Section $2l$ while this is the global optimum, and we must compare the bottom optimum of Section $2l+1$ with the bottom optimum of Section $2(l+1)$ when $\frac{3(4l-1)p}{4n} \leq q$:

\begin{equation*}
\begin{split}
    \left( \frac{lpq}{n} - \frac{(3l-1)lp^2}{2n^2} \right) &- \left( \frac{(2(l+1)-1)pq}{2n} - \frac{3(2(l+1)-1)((l+1)-1)p^2}{4n^2} \right) \\
    &= \frac{lpq}{n} - \frac{(3l-1)lp^2}{2n^2} - \frac{(2l+1)pq}{2n} + \frac{3(2l+1)lp^2}{4n^2} \\
    &= - \frac{pq}{2n} + \frac{5lp^2}{4n^2} \geq 0 \\
    &\Leftrightarrow q \leq \frac{5lp}{2n} \, .
\end{split}
\end{equation*}
However, $\frac{5lp}{2n} < \frac{3(4l-1)p}{4n}$ (because $\frac{5l}{2} \geq \frac{3(4l-1)}{4} \Leftrightarrow l \leq \frac{3}{2}$ so the bottom optimum of Section $2l+1$ is never better than the bottom optimum of Section $2(l+1)$).

Thus we have determined, for every possible proportion of $p$ and $q$, all of the optimal locations of Black's point given that Black's Voronoi cell does not touch either vertical edge of $\mathcal{P}$. These are shown in Table~\ref{tab:rowoptimals}.
\begin{table}[!ht]
\centering
\begin{tabular}{c|c|c|c}
Section & Optimum & Area & Condition \\
\hline
$I$ & $(*,0)$ & $\frac{pq}{2n}$ & $q \leq \frac{p}{n}$ \\
$II$ & $(x^*,\frac{q}{3}-\frac{p}{3n})$ & $\frac{pq}{3n} + \frac{p^2}{12n^2} + \frac{q^2}{12}$ & $\frac{p}{n} \leq q \leq \frac{(4 - \sqrt{3})p}{n}$ \\
$III$ & $(x^*+\frac{p}{2n},\frac{p}{2n})$ & $\frac{pq}{n} - \frac{p^2}{n^2}$ & $\frac{(4 - \sqrt{3})p}{n} \leq q \leq \frac{5p}{2n}$ \\
$IV$ & $(x^*,\frac{p}{n})$ & $\frac{3pq}{2n} - \frac{9p^2}{4n^2}$ & $\frac{5p}{2n} \leq q \leq \frac{5p}{n}$ \\
$IV$ & $(x^*,\frac{q}{3}-\frac{2p}{3n})$ & $\frac{2pq}{3n} - \frac{p^2}{6n^2} + \frac{q^2}{12}$ & $\frac{5p}{n} \leq q \leq \frac{(11-\sqrt{33})p}{n}$ \\
$VI$ & $(x^*,\frac{2p}{n})$ & $\frac{5pq}{2n} - \frac{15p^2}{2n^2}$ & $\frac{(11-\sqrt{33})p}{n} \leq q \leq \frac{33p}{4n}$ \\
$2l$ & $(x^*,(l-1)\frac{p}{n})$ & $\frac{(2l-1)pq}{2n} - \frac{3(2l-1)(l-1)p^2}{4n^2}$ & $\frac{3(4l-5)p}{4n} \leq q \leq \frac{3(4l-1)p}{4n}$ \\
\end{tabular}
\caption{Optima within one quarter cell of $V^\circ(w_i)$ at $x^*=\frac{ip}{2n}$ assuming that Black's Voronoi cell does not touch either vertical edge of $\mathcal{P}$.}
\label{tab:rowoptimals}
\end{table}

Finally we must determine Black's best point in the case that Black's Voronoi cell may touch a vertical edge of $\mathcal{P}$, and for this we will explore the cases of the parity of $n$ separately.

\paragraph{$n$ even}
If $n$ is even then we are investigating the top right quadrant of $V^\circ(w_{\frac{n}{2}})$ and so concern ourselves with $x^* = \frac{(n-1)p}{2n}$ upon which the optima of (all but one) even sections lie, and with $x^* = \frac{p}{2}$ upon which the optima of odd sections lie, and with $y = x - \frac{p}{2n}$ upon which the optima of Section $\frac{n}{2}$ lie (for reference, these are the optima depicted in Figures~\ref{fig:RowOptimalLeftTouchingIVc} to \ref{fig:RowOptimalLeftTouchingIVa}).

On $x^* = \frac{(n-1)p}{2n}$, the Voronoi cells of points in Sections $II$ to $n-2$ will not touch either vertical boundary of $\mathcal{P}$ and Section $n$ will touch the leftmost boundary of $\mathcal{P}$. On $x^* = \frac{p}{2}$, Sections $I$ to $n-1$ will not touch either vertical boundary of $\mathcal{P}$ and Section $n+1$ is the final section, touching both vertical edges of $\mathcal{P}$.

Therefore we need to compare the optima over Section $n$ (shown in Table~\ref{tab:rowsectionnleftoptimals}) and Section $n+1$ (the optimum $(\frac{p}{2n},\frac{(n-1)p}{2n})$ achieves $\frac{pq}{2} - \frac{3(n-1)np^2}{8n^2}$) with appropriate optima in Table~\ref{tab:rowoptimals}.
\begin{table}[!ht]
\centering
\begin{tabular}{c|c|c|c}
Section & Optimum & Area & Condition \\
\hline
$n$ & $(\frac{(n-1)p}{2n},\frac{(n-2)p}{2n})$ & $\frac{(n-1)pq}{2n} - \frac{(3n^2 - 9n + 6)p^2}{8n^2}$ & $q \leq \frac{(2n-3)p}{n}$ \\
 & $(\frac{p}{4n} + \frac{q}{4},\frac{q}{4} - \frac{p}{4n})$ & $\frac{(2n-1)pq}{8n} - \frac{(2n^2-6n+3)p^2}{16n^2} + \frac{q^2}{16}$ & $\frac{(2n-3)p}{n} \leq q \leq \frac{(2n-1)p}{n}$ \\
 & $(\frac{p}{2},\frac{(n-1)p}{2n})$ & $\frac{(2n-1)pq}{4n} - \frac{(3n^2-5n+2)p^2}{8n^2}$ & $\frac{(2n-1)p}{n} \leq q$ \\
\end{tabular}
\caption{Optima within Section $n$ of the top right quadrant of $V^\circ(w_\frac{n}{2})$.}
\label{tab:rowsectionnleftoptimals}
\end{table}

Within the calculations for general Sections $2l$ and $2l+1$ leading to Table~\ref{tab:rowoptimals}, we compared the bottom optima of Section $2l$ to $2l+1$ and of Section $2l+1$ to $2(l+1)$. It is useful to note that while the Voronoi cell of the `bottom' optimum of Section $n$ does share a border with the leftmost boundary of $\mathcal{P}$, this bordering does not remove any area since the boundary of $\mathcal{P}$ is incident on the perimeter of the Voronoi cell (i.e. the point captures the same area as the equivalent point in Section $n$ which does not touch the leftmost boundary of $\mathcal{P}$). Since the `bottom' optimum of Section $n$ acts as if it does not touch either vertical boundary of $\mathcal{P}$, we can use all of the results from Table~\ref{tab:rowoptimals} and check the optima in Section $n$ against that of Section $n+1$ (along with some checks for small $n$).

Firstly comparing the optima of Section $n$ and Section $n+1$, by an identical argument to the one shown in Figure~\ref{fig:RowOptimalComp}, the `top' optimum of Section $n$ is beaten by the optimum in Section $n+1$ (simply because they lie in the same location with Section $n+1$ lying on a more preferential side of $\mathcal{CC}^3(w_n)$). We should, however, compare the optimum of Section $n+1$ with the `middle' optimum of Section $n$:
\begin{equation*}
\begin{split}
    \left( \frac{pq}{2} - \frac{3(n-1)np^2}{8n^2} \right) &- \left( \frac{(2n-1)pq}{8n} - \frac{(2n^2-6n+3)p^2}{16n^2} + \frac{q^2}{16} \right) \\
    &= \frac{pq}{2} - \frac{3(n-1)np^2}{8n^2} - \frac{(2n-1)pq}{8n} + \frac{(2n^2-6n+3)p^2}{16n^2} - \frac{q^2}{16} \\
    &= \frac{(2n+1)pq}{8n} - \frac{(4n^2 - 3)p^2}{16n^2} - \frac{q^2}{16} \geq 0 \\
    &\Leftrightarrow q^2 - \frac{2(2n+1)pq}{n} + \frac{(4n^2 - 3)p^2}{n^2} \leq 0 \\
    &\Leftrightarrow \left(q - \frac{(2n+1)p}{n}\right)^2 - \frac{4(n+1)p^2}{n^2} \leq 0 \\
    &\Leftrightarrow \frac{(2n+1)p}{n} - \frac{2\sqrt{(n+1)}p}{n} \leq q \leq \frac{(2n+1)p}{n} + \frac{2\sqrt{(n+1)}p}{n} \, .
\end{split}
\end{equation*}
Now % \Leftrightarrow 2n-3 < 2n+1 - 2\sqrt{(n+1)}
$$\frac{(2n-3)p}{n} < \frac{(2n+1 - 2\sqrt{(n+1)})p}{n} \Leftrightarrow 2\sqrt{(n+1)} < 4 \Leftrightarrow n < 3$$
so if $n=2$ then the `middle' optimum of Section $n$ is better than the optimum over Section $n+1$ if $q \leq \frac{(2n+1 - 2\sqrt{(n+1)})p}{n}$. In this case, the remaining comparisons are very straightforward as the only sections existing when $n=2$ are Sections $I$, $II$, and $III$ so we can record the optima straightforwardly, as displayed in Table~\ref{tab:rowsectionIIOptimals}.
\begin{table}[!ht]
\centering
\begin{tabular}{c|c|c|c}
Section & $b^*$ & Area & Condition \\
\hline
$I$ & $(*,0)$ & $\frac{pq}{2n}$ & $q \leq \frac{p}{2}$ \\
$II$ & $(\frac{p}{8} + \frac{q}{4},\frac{q}{4} - \frac{p}{8})$ & $\frac{3pq}{16} + \frac{p^2}{64} + \frac{q^2}{16}$ & $\frac{p}{2} \leq q \leq \frac{(5-2\sqrt{3})p}{2}$ \\
$III$ & $(\frac{p}{2},\frac{p}{4})$ & $\frac{pq}{2} - \frac{3p^2}{16}$ & $\frac{(5-2\sqrt{3})p}{2} \leq q$ \\
\end{tabular}
\caption{The best point $b^*$ for $n=2$.}
\label{tab:rowsectionIIOptimals}
\end{table}

If $n \neq 2$ then the optimum in Section $n+1$ is always better than the `middle' optimum of Section $n$. This leads us to the comparison of the `bottom' optimum of Section $n$ and the optimum of Section $n+1$:
\begin{equation*}
\begin{split}
    \left( \frac{pq}{2} - \frac{3(n-1)np^2}{8n^2} \right) &- \left( \frac{(n-1)pq}{2n} - \frac{(3n^2 - 9n + 6)p^2}{8n^2} \right) \\
    &= \frac{pq}{2} - \frac{3(n-1)np^2}{8n^2}-\frac{(n-1)pq}{2n} + \frac{(3n^2 - 9n + 6)p^2}{8n^2} \\
    &= \frac{pq}{2n} - \frac{3(n-1)p^2}{4n^2} \geq 0 \\
    &\Leftrightarrow \frac{3(n-1)p}{2n} \leq q
\end{split}
\end{equation*}
%Now we know from the calculations contributing to Table~\ref{tab:rowoptimals} that the bottom optimum of Section $2l$ becomes preferable to the bottom optimum of Section $2(l-1)$ when $\frac{3(4l-5)p}{4n} \leq q$. Now for $l=\frac{n}{2}$ this is when $\frac{3(2n-5)p}{4n} \leq q$ and $\frac{3(2n-5)p}{4n} < \frac{3(n-1)p}{2n}$ so we need not compare
and, for a sanity check, the comparison of the `bottom' optimum of Section $n$ and the bottom optimum of Section $n-2$ (note that Section $n-2$ may not always exist and we shall discuss these finer details shortly):

\begin{equation*}
\begin{split}
    \left( \frac{(n-1)pq}{2n} \right. &- \left. \frac{(3n^2 - 9n + 6)p^2}{8n^2} \right) - \left( \frac{((n-2)-1)pq}{2n} - \frac{3((n-2)-1)((\frac{n-2}{2}-1)p^2}{4n^2} \right) \\
    &= \frac{(n-1)pq}{2n} - \frac{(3n^2 - 9n + 6)p^2}{8n^2} - \frac{(n-3)pq}{2n} + \frac{3(n-3)(n-4)p^2}{8n^2} \\
    &= \frac{pq}{n} - \frac{3(2 n - 5)p^2}{4n^2} \geq 0 \\
    &\Leftrightarrow \frac{3(2n-5)p}{4n} \leq q
\end{split}
\end{equation*}
(the identical condition for Section $2l$ and $2(l+1)$ where Section $2(l+1)$ does not touch either vertical boundaries of $\mathcal{P}$ as expected). Finally, checking that $\frac{3(2n-5)p}{4n} < \frac{3(n-1)p}{2n}$ confirms that the optimum in Section $n+1$ is better than the `bottom' optimum in Section $n$ when this optimum is better than the bottom optimum of Section $n-2$, so we need not compare the optima of Section $n+1$ and $n-2$.

Now, for $l>1$, the bottom optimum in Section $2l$ was always found to be optimal within some range of $q$ in Table~\ref{tab:rowoptimals} as, since the `bottom' optimum in Section $n$ is identical to the general bottom optimum in Section $2l$, it is simply true that we can use all of the results summarised in Table~\ref{tab:rowoptimals} until Section $n$ at which point we use the results we have most recently found. Hence the best points $b^*$ for every even $n \neq 2$ are recorded in Table~\ref{tab:rowevensectionOptimals}.

\begin{table}[!ht]
\centering
\begin{tabular}{c|c|c|c}
Section & Optimum & Area & Condition \\
\hline
$I$ & $(*,0)$ & $\frac{pq}{2n}$ & $q \leq \frac{p}{n}$ \\
$II$ & $(\frac{(n-1)p}{2n},\frac{q}{3}-\frac{p}{3n})$ & $\frac{pq}{3n} + \frac{p^2}{12n^2} + \frac{q^2}{12}$ & $\frac{p}{n} \leq q \leq \frac{(4 - \sqrt{3})p}{n}$ \\
$III$ & $(\frac{p}{2},\frac{p}{2n})$ & $\frac{pq}{n} - \frac{p^2}{n^2}$ & $\frac{(4 - \sqrt{3})p}{n} \leq q \leq \frac{5p}{2n}$ \\
$IV$ & $(\frac{(n-1)p}{2n},\frac{p}{n})$ & $\frac{3pq}{2n} - \frac{9p^2}{4n^2}$ & $\frac{5p}{2n} \leq q \leq \frac{5p}{n}$ \\
$IV$ & $(\frac{(n-1)p}{2n},\frac{q}{3}-\frac{2p}{3n})$ & $\frac{2pq}{3n} - \frac{p^2}{6n^2} + \frac{q^2}{12}$ & $\frac{5p}{n} \leq q \leq \frac{(11-\sqrt{33})p}{n}$ \\
$VI$ & $(\frac{(n-1)p}{2n},\frac{2p}{n})$ & $\frac{5pq}{2n} - \frac{15p^2}{2n^2}$ & $\frac{(11-\sqrt{33})p}{n} \leq q \leq \frac{33p}{4n}$ \\
$2l$ & $(\frac{(n-1)p}{2n},(l-1)\frac{p}{n})$ & $\frac{(2l-1)pq}{2n} - \frac{3(2l-1)(l-1)p^2}{4n^2}$ & $\frac{3(4l-5)p}{4n} \leq q \leq \frac{3(4l-1)p}{4n}$ \\
$n$ & $(\frac{(n-1)p}{2n},\frac{(n-2)p}{2n})$ & $\frac{(n-1)pq}{2n} - \frac{(3n^2 - 9n + 6)p^2}{8n^2}$ & $\frac{3(2n-5)p}{4n} \leq q \leq \frac{3(n-1)p}{2n}$ \\
$n+1$ & $(\frac{p}{2n},\frac{(n-1)p}{2n})$ & $\frac{pq}{2} - \frac{3(n-1)np^2}{8n^2}$ & $\frac{3(n-1)p}{2n} \leq q$ \\
\end{tabular}
\caption{The best point $b^*$ for even $n \neq 2$.}
\label{tab:rowevensectionOptimals}
\end{table}

\paragraph{$n$ odd}
If $n$ is odd then we are investigating the top right quadrant of $V^\circ(w_{\frac{n+1}{2}})$ and so concern ourselves with $x^* = \frac{p}{2}$ upon which the optima of all even sections lie, and with $x^* = \frac{(n+1)p}{2n}$ upon which the optima of (all but one) odd sections lie, and with $y=\frac{(n-i)p}{n} - x$ upon which the optima of Section $i$ lie (for reference, these are the optima depicted in Figures~\ref{fig:RowOptimalRightTouchingIIIc} to \ref{fig:RowOptimalRightTouchingIIIa}).

On $x^* = \frac{p}{2}$, the Voronoi cells of points in Sections $II$ to $n-1$ will not touch either vertical boundaries of $\mathcal{P}$, and Section $n+1$ is the final section, touching both vertical edges of $\mathcal{P}$. On $x^* = \frac{(n+1)p}{2n}$, the Voronoi cells of points in Sections $I$ to $n-2$ will not touch either vertical boundaries of $\mathcal{P}$ and Section $n$ will touch the rightmost boundary of $\mathcal{P}$.

Therefore, as before, we need to compare the optima over Section $n$ (shown in Table~\ref{tab:rowsectionnrightoptimals}) and Section $n+1$ (the optimum $(\frac{p}{2},\frac{(n-1)p}{2n})$ achieves $\frac{pq}{2} - \frac{3(n-1)np^2}{8n^2}$) with appropriate optima in Table~\ref{tab:rowoptimals}.
\begin{table}[!ht]
\centering
\begin{tabular}{c|c|c|c}
Section & Optimum & Area & Condition \\
\hline
$n$ & $(\frac{(n+1)p}{2n},\frac{(n-2)p}{2n})$ & $\frac{(n-1)pq}{2n} - \frac{3(n-2)(n-1)p^2}{8n^2}$ & $q \leq \frac{(2n-3)p}{n}$ \\
 & $(\frac{(4n-1)p}{4n} - \frac{q}{4},-\frac{p}{4n} + \frac{q}{4})$ & $\frac{(2n-1)pq}{8n} - \frac{(2n^2-6n+3)p^2}{16n^2} + \frac{q^2}{16}$ & $\frac{(2n-3)p}{n} \leq q \leq \frac{(2n-1)p}{n}$ \\
 & $(\frac{p}{2},\frac{(n-1)p}{2n})$ & $\frac{(2n-1)pq}{4n} - \frac{(3n-2)(n-1)p^2}{8n^2}$ & $\frac{(2n-1)p}{n} \leq q$ \\
\end{tabular}
\caption{Optima within Section $n$ of the top right quadrant of $V^\circ(w_\frac{n+1}{2})$.}
\label{tab:rowsectionnrightoptimals}
\end{table}

Now the only optimum within an odd section (ignoring Section $I$) to be a global optimal point for Voronoi cells not touching either vertical boundary of $\mathcal{P}$ is the bottom optimum of Section $III$. Since the areas achieved by the optima in Section $n$ for Voronoi cells touching the rightmost vertical boundary of $\mathcal{P}$ are no greater than the areas achieved by the optima of Section $n$ for Voronoi cells that touch neither vertical boundary of $\mathcal{P}$ and the latter are not global optima unless $n=3$, the optima of our Section $n$ will never be global optima unless $n=3$. Therefore we only need consider the optima in Section $n$ if $n=3$, and by the argument from Figure~\ref{fig:RowOptimalComp} we need never consider the `top' optima of Section $n$. Moreover, we can avoid all further calculations by noting that the bottom optimum of Section $III$ and of Section $IV$ in Table~\ref{tab:rowoptimals} (two consecutive global optima) are, respectively, exactly the same location and achieve exactly the same area as the `bottom' optimum of Section $III$ in the case that the Voronoi cells touch the rightmost boundary of $\mathcal{P}$, and the optimum of Section $IV$ in the case that the Voronoi cells touch both vertical boundaries of $\mathcal{P}$. Therefore we already know our global optima and these are displayed in Table~\ref{tab:rowsectionIIIrightoptimals}.
\begin{table}[!ht]
\centering
\begin{tabular}{c|c|c|c}
Section & Optimum & Area & Condition \\
\hline
$I$ & $(*,0)$ & $\frac{pq}{2n}$ & $q \leq \frac{p}{n}$ \\
$II$ & $(\frac{p}{2},\frac{q}{3}-\frac{p}{3n})$ & $\frac{pq}{3n} + \frac{p^2}{12n^2} + \frac{q^2}{12}$ & $\frac{p}{n} \leq q \leq \frac{(4 - \sqrt{3})p}{n}$ \\
$III$ & $(\frac{(n+1)p}{2n},\frac{p}{2n})$ & $\frac{pq}{n} - \frac{p^2}{n^2}$ & $\frac{(4 - \sqrt{3})p}{n} \leq q \leq \frac{5p}{2n}$ \\
$IV$ & $(\frac{p}{2},\frac{p}{3})$ & $\frac{pq}{2} - \frac{p^2}{4}$ & $\frac{5p}{2n} \leq q$ \\
\end{tabular}
\caption{The best point $b^*$ for $n=3$.}
\label{tab:rowsectionIIIrightoptimals}
\end{table}

\iffalse
Tackling this case first, if $n=3$ then we need to compare the Section $III$ optima (for Voronoi cells touching the rightmost boundary of $\mathcal{P}$) with the optimum for Section $IV$ (for Voronoi cells touching both vertical boundaries of $\mathcal{P}$). Firstly we will look at the optimum of Section $IV$ versus the `middle' optimum of Section $III$:
\begin{equation*}
\begin{split}
    \left( \frac{pq}{2} - \frac{p^2}{4} \right) - \left(\frac{5pq}{24} - \frac{p^2}{48} + \frac{q^2}{16} \right) &= \frac{7pq}{24} - \frac{11p^2}{48} - \frac{q^2}{16} \geq 0 \\
    &\Leftrightarrow q^2 - \frac{14pq}{3} + \frac{11p^2}{3} \leq 0 \\
    &\Leftrightarrow (q - \frac{7p}{3})^2 - \frac{16p^2}{9} \leq 0 \\
    &\Leftrightarrow p \leq q \leq \frac{11p}{3} \, .
\end{split}
\end{equation*}
Now $p = p \left( = \frac{(2n-3)p}{n} \right) < \frac{5p}{3} \left( = \frac{(2n-1)p}{n} \right) < \frac{11p}{3}$ so the optimum of Section $IV$ is preferential to the `middle' optimum of Section $III$ whenever this is the optimum over Section $III$. This brings us to compare the optimum of Section $IV$ with the `bottom' optimum of Section $III$.
\fi

The ideas here can also be transferred to the $n \neq 3$ case. All that remains is to compare the optima found in Table~\ref{tab:rowoptimals} to the optimum of Section $n+1$, yet once we realise that the optimum of Section $n+1$ touching both vertical boundaries of $\mathcal{P}$ is identical in location and area to the bottom optimum of Section $n+1$, touching neither vertical boundary of $\mathcal{P}$ which is the global optimum for cells not touching either vertical boundary, we know that we have already found the hierarchy of optima and we can simply copy the results from Table~\ref{tab:rowoptimals} (and these also hold true for $n=3$). Hence the best points $b^*$ for every odd $n$ are recorded in Table~\ref{tab:rowoddsectionOptimals}.

\begin{table}[!ht]
\centering
\begin{tabular}{c|c|c|c}
Section & Optimum & Area & Condition \\
\hline
$I$ & $(*,0)$ & $\frac{pq}{2n}$ & $q \leq \frac{p}{n}$ \\
$II$ & $(\frac{p}{2},\frac{q}{3}-\frac{p}{3n})$ & $\frac{pq}{3n} + \frac{p^2}{12n^2} + \frac{q^2}{12}$ & $\frac{p}{n} \leq q \leq \frac{(4 - \sqrt{3})p}{n}$ \\
$III$ & $(\frac{(n+1)p}{2n},\frac{p}{2n})$ & $\frac{pq}{n} - \frac{p^2}{n^2}$ & $\frac{(4 - \sqrt{3})p}{n} \leq q \leq \frac{5p}{2n}$ \\
$IV$ & $(\frac{p}{2},\frac{p}{n})$ & $\frac{3pq}{2n} - \frac{9p^2}{4n^2}$ & $\frac{5p}{2n} \leq q \leq \frac{5p}{n}$ \\
$IV$ & $(\frac{p}{2},\frac{q}{3}-\frac{2p}{3n})$ & $\frac{2pq}{3n} - \frac{p^2}{6n^2} + \frac{q^2}{12}$ & $\frac{5p}{n} \leq q \leq \frac{(11-\sqrt{33})p}{n}$ \\
$VI$ & $(\frac{p}{2},\frac{2p}{n})$ & $\frac{5pq}{2n} - \frac{15p^2}{2n^2}$ & $\frac{(11-\sqrt{33})p}{n} \leq q \leq \frac{33p}{4n}$ \\
$2l$ & $(\frac{p}{2},(l-1)\frac{p}{n})$ & $\frac{(2l-1)pq}{2n} - \frac{3(2l-1)(l-1)p^2}{4n^2}$ & $\frac{3(4l-5)p}{4n} \leq q \leq \frac{3(4l-1)p}{4n}$ \\
$n+1$ & $(\frac{p}{2},\frac{(n-1)p}{2n})$ & $\frac{pq}{2} - \frac{3(n-1)np^2}{8n^2}$ & $\frac{3(2n-3)p}{4n} \leq q$ \\
\end{tabular}
\caption{The best point $b^*$ for odd $n$.}
\label{tab:rowoddsectionOptimals}
\end{table}

And thus we have found all of the best points $b^*$ in $\mathcal{P}$ for every combination of $p$, $q$, and $n$.

\subsection{Black's best arrangement}
\label{sec:BlackRowArrangement}

As interesting as Black's best point may be, Black must also consider the placement of their other $n-1$ points, and the best single point may actually be a poor placement when considering where to place Black's remaining points.

On top of the relationship between $\frac{p}{n}$ and $q$ restricting what sections are available within which to place Black's points, the idea of modelling the interaction between different black points and consideration of where a black cell would steal from another black cell, thereby reducing the effectiveness of their placement, fills the writer with fear. However, we can learn something from the optimisation of $b_1$ within every possible partition of $\mathcal{P}$, as depicted in Figure \ref{fig:RowOptimals} and the tables of Black's best single points.

Making use of the work summarised in Table~\ref{tab:rowoptimals} we can achieve crude upper bounds on the area that Black can steal with the na\"{i}ve suggestion that Black manages to steal the area of their best point for each point. This supposition is not so crazy for low values of $q$; it is clear that Black can steal a maximum of $\frac{pq}{2}$ if $\frac{p}{n} \geq q$ by placing their points $b_i$ as close as possible to $w_i$. However, for $\frac{p}{n} \leq q \leq \frac{(4 - \sqrt{3})p}{n}$ for which the middle optimum of Section $II$ is the best point, one can easily see that it is not possible to locate $n$ of these points such that no two of Black's Voronoi cells overlap.

Observe from Figure~\ref{fig:RowCellExamples} that every point $b_1=(x,y)$ placed in Section $III$ and above steals from at least four quarter cells in $\mathcal{VD}(W)$ all of the area from $y$ upwards (or placed in Section $IV$ and above if $b_1$ is placed within a quarter cell sharing a vertical boundary with $\mathcal{P}$). Naturally it would be wasteful if Black were to steal such a portion of a single quarter cell more than once (i.e. by two separate placements $b_1$ and $b_2$). Since there are a total of $4n$ quarter cells to steal from and $n$ Black points to be positioned in order to steal from these quarter cells, an effective position for $b \in B$ would steal as much area as possible from a particular four quarter cells. This suggests that a utilisation of a row of points as depicted in Figure~\ref{fig:RowOptimalNotTouchingIIIc} equally spaced with a horizontal separation of $\frac{2p}{n}$ above and below the white row would work rather effectively. Let us formally describe such an arrangement.

With white points being ordered $w_1$ to $w_n$ left to right where $w_i=(\frac{(2i-1)p}{2n},0)$, this arrangement for Black as described would be $b^u_i=(\frac{ip}{n},\frac{p}{2n})$ and $b^d_i=(\frac{ip}{n},-\frac{p}{2n})$ for $i=1,...,\floor{\frac{n}{2}}$ (being the points played above and below White's row respectively). Of course, if $n$ is odd then we have one remaining point $b_n$ to place, and one Voronoi cell $V^\circ(w_n)$ which remains unchallenged by any of Black's points $b^{\{u,d\}}_i$ so it makes most sense to place $b_n$ as close as possible to $w_n$ in order to steal the most ($\frac{pq}{2n}$) from $V^\circ(w_n)$. These arrangements (for $n$ even and odd) are always possible (i.e. the points $b^u_i$ and $b^d_i$ exist) and are pictured in Figure~\ref{fig:RowOptimalArrangements}.

\begin{figure}[!ht]
\centering
\begin{subfigure}{1.0\textwidth}
  \centering
  \includegraphics[scale=0.17]{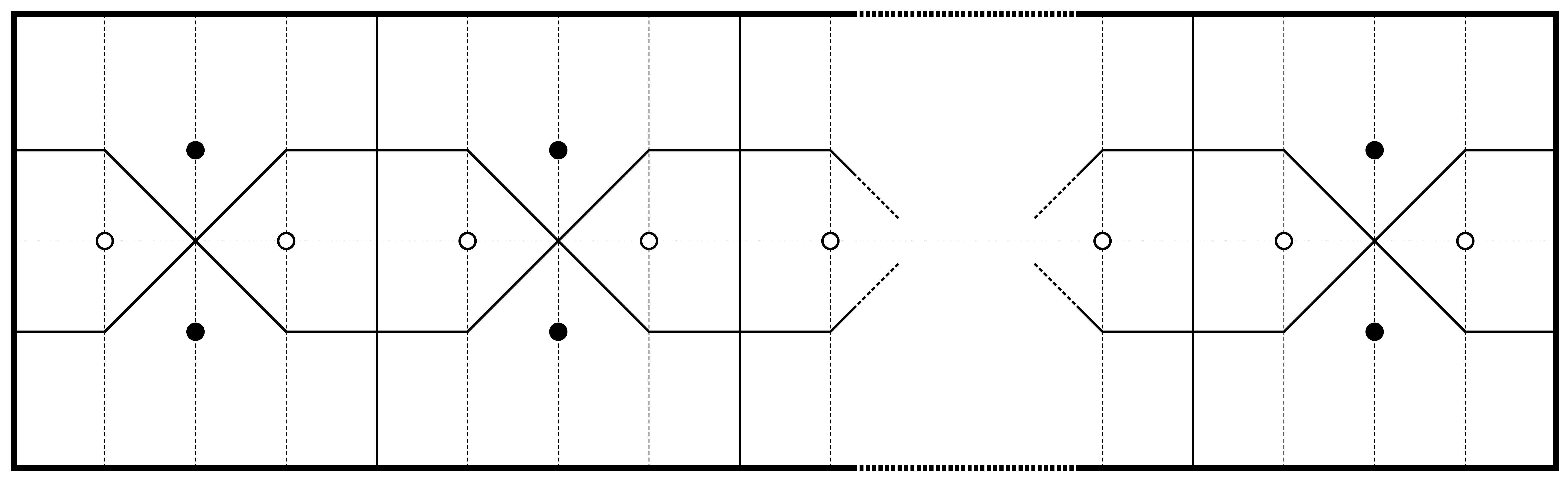}
  \caption{$n$ even.}
  \label{fig:RowOptimalArrangementEven}
\end{subfigure}

\begin{subfigure}{1.0\textwidth}
  \centering
  \includegraphics[scale=0.17]{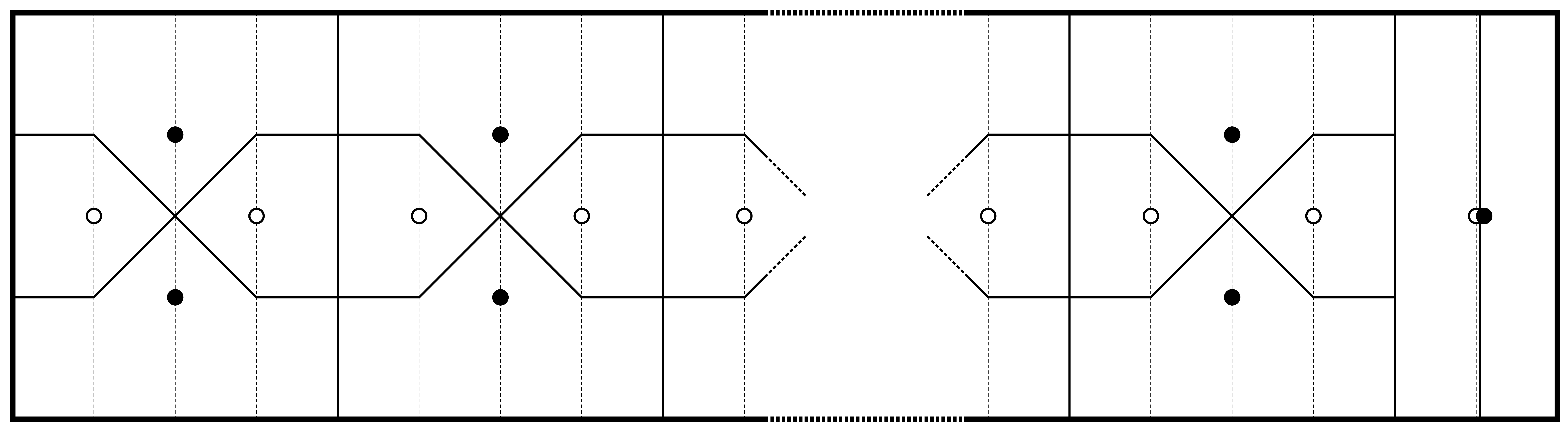}
  \caption{$n$ odd.}
  \label{fig:RowOptimalArrangementOdd}
\end{subfigure}
\caption{Arrangements exploiting the efficacy of the best point in Section $III$.}
\label{fig:RowOptimalArrangements}
\end{figure}

If $n$ is even then this arrangement scores an area of $pq-3n\times\frac{p^2}{4n^2} = pq-\frac{3p^2}{4n}$, capturing all of $\mathcal{P}$ outside $-\frac{p}{2n} < y < \frac{p}{2n}$. We know that this arrangement is optimal for $\frac{(4 - \sqrt{3})p}{n} \leq q \leq \frac{5p}{2n}$ since it is under this condition that the best point $b^*$ is the lower optimum of Section $III$ and so this arrangement is composed entirely of non-overlapping best points $b^*$. 

Furthermore, we hold that this arrangement is optimal for even $n$ even when $\frac{5p}{2n} \leq q$. This is due to the fact that increasing $q$ merely increases the area controlled by Black's points without altering White's area. If $\{b^u_i,b^d_i\}_{i\in[1,\ldots,n]}$ is not the optimal arrangement for Black then there must be another arrangement which controls more area within $-\frac{p}{2n} < y < \frac{p}{2n}$ of $\mathcal{P}$. However, no arrangement containing a point with a $y$-coordinate of absolute value greater than $\frac{5p}{2n}$ can steal more area within $-\frac{p}{2n} < y < \frac{p}{2n}$ of $\mathcal{P}$. Therefore this better arrangement must also exist for $\frac{5p}{2n} \geq q$ and so be the optimal arrangement for some range of $\frac{(4 - \sqrt{3})p}{n} \leq q \leq \frac{5p}{2n}$, providing an obvious contradiction.%\todo{Is this sufficient explanation? I can't think of a strong way to explain this without getting really technical -- which could take another ten pages!}

Thus we have found Black's optimal play for even $n$ and $\frac{(4 - \sqrt{3})p}{n} \leq q$ in response to White playing a row. Results for odd $n$ and $\frac{(4 - \sqrt{3})p}{n} \geq q$ are less obvious, though we suspect that the best point within Section $III$ will be prevalent in optimal arrangements, not least effective arrangements.

%If this arrangement is not optimal for Black then there must exist a point which captures more of the area within $-\frac{p}{2n} < y < \frac{p}{2n}$. From Figure \ref{fig:RowOptimals} we can see that this point must therefore lie in Section $1$ or Section $2$.

%However, for $b_1$ in Section $I$, $V^+(b_1)$ will capture up to $\frac{p^2}{2n^2}$ within $-\frac{p}{2n} < y < \frac{p}{2n}$ ($\frac{p^2}{4n^2}$ more than $b^u_i$ and $b^d_i$) at the expense of only up to $\frac{pq}{2n}-\frac{p^2}{2n^2}$ outside $-\frac{p}{2n} < y < \frac{p}{2n}$ ($\frac{pq}{2n}-\frac{p^2}{2n^2}$ less than $b^u_i$ and $b^d_i$). This is a reduction in captured area by $\frac{pq}{2n}-\frac{p^2}{4n^2}=\frac{p}{2n}(q-\frac{p}{2n})>0$. In order for a Section $1$ point to be a member of the optimal arrangement for Black there must be a position $b_2$ which steals ...

\section{White plays an \texorpdfstring{$a \times b$}{a x b} grid}
\label{sec:WhiteGrid}

Next we shall explore the placement of Black's point $b_1$ within grids with depth greater than one. We assume that the points are positioned in an $a\times b$ grid and without loss of generality let us assume $\frac{p}{a} \geq \frac{q}{b}$. Since White's arrangement is repetitive and has such symmetry, our search for Black's optimal location is greatly simplified as we need only consider the placement of $b_1$ within a small selection of areas of $\mathcal{P}$.

Again we shall investigate the possible Voronoi diagrams $\mathcal{VD}(W\cup b_1)$ (in order to find the placement of $b_1$ so as to maximise $Area(V^+(b_1))$) by partitioning the arena into subsets within which the Voronoi diagram is structurally identical. Since $\mathcal{P}$ is rectangular and all of White's bisectors are vertical and horizontal then, from \citeA{AveBerKalKra15}, the partitioning lines are simply the configuration lines of each of White's points.

In any $a \times b$ grid of points $W$ in $\mathcal{P}$ with $a,b\geq 2$, there exists a $2 \times 2$ subgrid within the arrangement. An example of such a subgrid is shown in Figure~\ref{fig:GridPartitionExample} along with its partition of the space into regions within which the cell $V^+(b_1)$ is structurally identical. This subgrid can be found by choosing any point $w_0 \in W$ and orienting $\mathcal{P}$ so that $w_0$ is the bottom left vertex of a $2 \times 2$ subgrid. Having done this we will label the adjacent point to the right of $w_0$ by $w_{0_{R}}$ and then label every pair of left and right points directly above $w_0$ and $w_{0_{R}}$ by $w_{i_L}$ and $w_{i_{R}}$ respectively and every pair of left and right points directly below $w_0$ and $w_{0_{R}}$ by $w_{{-i}_{L}}$ and $w_{{-i}_{R}}$ respectively, where $i$ marcates these pairs as being the $i$th pairs away from $w_0$ and $w_{0_{R}}$ in their given directions. That is, taking $w_0 = (0,0)$, we have the following expressions: $w_{i_L} = \left( 0 , \frac{iq}{b} \right)$ and $w_{i_{R}} = \left( \frac{p}{a} , \frac{iq}{b} \right)$ for $i \in \mathbb{Z}$ (where $w_{0_L}=w_0$). Similarly define $w_{i_{LL}} = \left( -\frac{p}{a} , \frac{iq}{b} \right)$ for $i \in \mathbb{Z}$ to be the set of points adjacent to the left of $w_{i_L}$. Note that $w_{i_L}$ and $w_{i_{R}}$ may not exist for $i \in \mathbb{Z} \setminus \{0,1\}$, nor may $w_{i_{LL}}$ for any $i \in \mathbb{Z}$, so these are not depicted in Figure~\ref{fig:GridPartitionExample}.

\begin{figure}[!ht]
\centering
\includegraphics[width=0.6\textwidth]{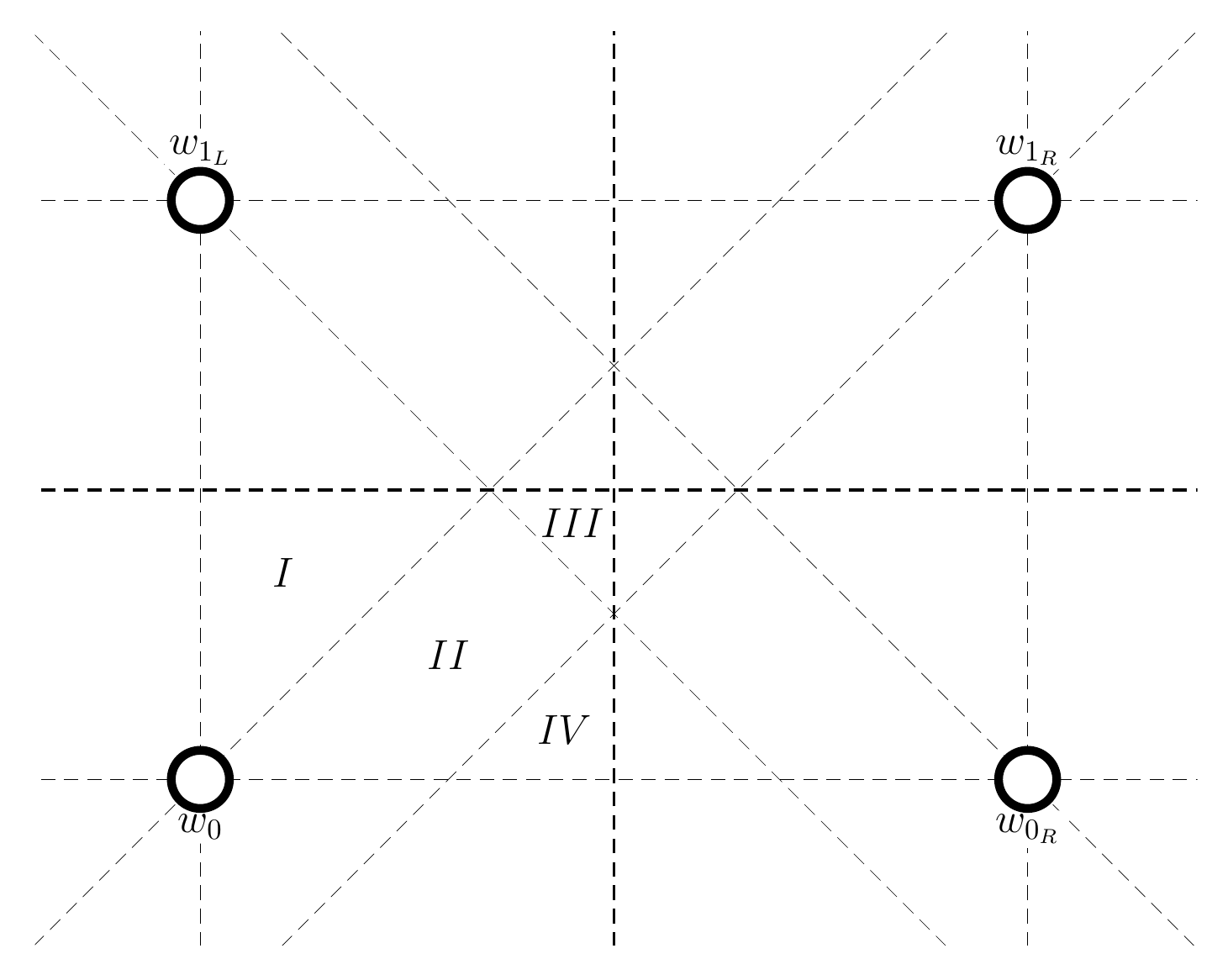}
\caption{The partition of $\mathcal{P}$ for an example $2 \times 2$ subgrid of $W$.}
\label{fig:GridPartitionExample}
\end{figure}

\pagebreak Without loss of generality let us choose to place $b_1$ within the first quadrant of $V^\circ(w_0)$. By the symmetry of this $2 \times 2$ subgrid, every possible cell type of $V^+(b_1)$ can be grown from this placement, adjusting the number of white points with which to generate a bisector outside the $2 \times 2$ subgrid in whichever direction we choose (as will be shown). Note that it is only the points $w_{i_L}$, $w_{i_{R}}$, and $w_{0_{LL}}$ and $w_{1_{LL}}$ that can contribute a bisector part to the perimeter of $V^+(b_1)$ since $\frac{p}{a} \geq \frac{q}{b}$. We can also use this quadrant to investigate the structures that $V^+(b_1)$ can take when placed outside a $2 \times 2$ subgrid (i.e. placed in a quadrant of $V^\circ(w)$ for some $w$ which borders the perimeter of $\mathcal{P}$) by introducing boundaries along the borders of $V^\circ(w_0)$ as required (as will also be shown).

The observant reader may realise that there should perhaps be at least another configuration line contributing to the partition in Figure~\ref{fig:GridPartitionExample}: the potential configuration line $\mathcal{CL}^4(w_{{-1}_{R}})$, for example. While indeed $b_1$ will interact with $w_{{-1}_{R}}$ (if existing) if $b_1 \in \mathcal{CC}^1(w_0)$, this interaction will be identical no matter whether $b_1 \in \mathcal{CC}^3(w_{{-1}_{R}})$ or $b_1 \in \mathcal{CC}^4(w_{{-1}_{R}})$. This is because the only bisector part present in $B(b_1,w_{{-1}_{R}})$ is the diagonal part, identical in its representation for both $\mathcal{CC}^3(w_{{-1}_{R}})$ and $\mathcal{CC}^4(w_{{-1}_{R}})$ bisectors. In order to present a horizontal or vertical bisector part, one of the quadrant lines of $b_1$ must enter the cell $V^\circ(w_{{-1}_{R}})$. Therefore the only configuration lines required are from those points in $W$ lying on the quadrant lines through $w_1$. For that reason we may also ignore the configuration lines $\mathcal{CC}^3(w_{{-i}_{R}})$, $\mathcal{CC}^5(w_{{i}_{R}})$, and $\mathcal{CC}^7(w_{{1}_{LL}})$ for all $i \in \mathbb{N}$. Moreover, since $\frac{p}{a} \geq \frac{q}{b}$, the configuration lines of $w_{0_{R}}$ and $w_{0_{LL}}$ do not enter the first quadrant of $V^\circ(w_0)$, so the only lines contributing to our partition are $\mathcal{CC}^7(w_{{i}_{L}})$ and $\mathcal{CC}^1(w_{{-i}_{L}})$ for $i \in \mathbb{N}$.
%This is due to the fact that if $w_5$ exists, then so will $w_6$, $\frac{q}{b}$ lower than $w_1$ (this is $\frac{p}{a}$ left of $w_5$) and $b_1$ will always lie in $\mathcal{CC}^2(w_6)$.

As $\frac{p}{2a}$ increases in relation to $\frac{q}{2b}$ from the proportions shown in Figure~\ref{fig:GridPartitionExample}, the partitioning lines $\mathcal{CC}^7(w_{i_L})$ and then $\mathcal{CC}^1(w_{{-i}_{L}})$ for $i \in \mathbb{N}$ will begin to contribute to the partition of $\mathcal{P}$. In this way, momentously, the partition confined to the top right quadrant of $V^\circ(w_0)$ is identical to the partition studied for White's row arrangement (shown in Figure~\ref{fig:RowPartition}) but reflected in $y=x$ and with the width $\frac{p}{2n}$ and height $\frac{q}{2}$ of the quadrant being explored replaced by $\frac{q}{2b}$ and $\frac{p}{2a}$ respectively (truncated, instead, by the bisector $B(w_0,w_{0_{R}})$). In this way we have the partition cells as in Figure~\ref{fig:RowPartition} which are Section $I$ ($\mathcal{CC}^2(w_0)$), Section $2l$ ($\mathcal{CC}^1(w_{{1-l}_{L}}) \setminus \mathcal{CC}^8(w_{{l}_{L}})$) for $0 \leq \frac{(l-1)q}{b} \leq \frac{p}{2a}$, and Section $2l+1$ ($\mathcal{CC}^8(w_{{l}_{L}}) \setminus \mathcal{CC}^1(w_{{-l}_{L}})$) for $\frac{q}{2b} \leq \frac{(2l-1)q}{2b} \leq \frac{p}{2a}$, as shown in Figure~\ref{fig:GridPartition}. Note that, since $b$ is finite, Sections $2l$ and $2l+1$ do not always exist. If $w_0$ is on the $i$th row of White points (counting from the bottom of the $a \times b$ grid) then the last possible partition section will be bounded by either $\mathcal{CL}^2(w_{{1-i}_L})$ (so would be Section $2i$) or by $\mathcal{CL}^8(w_{{b-i}_L})$ (so would be Section $2(b-i)+1$) depending on whether $-(1-i) \geq b-i \Leftrightarrow i \geq \frac{b+1}{2}$ or not, respectively.

\begin{figure}[!ht]
\centering
\includegraphics[width=1.0\textwidth]{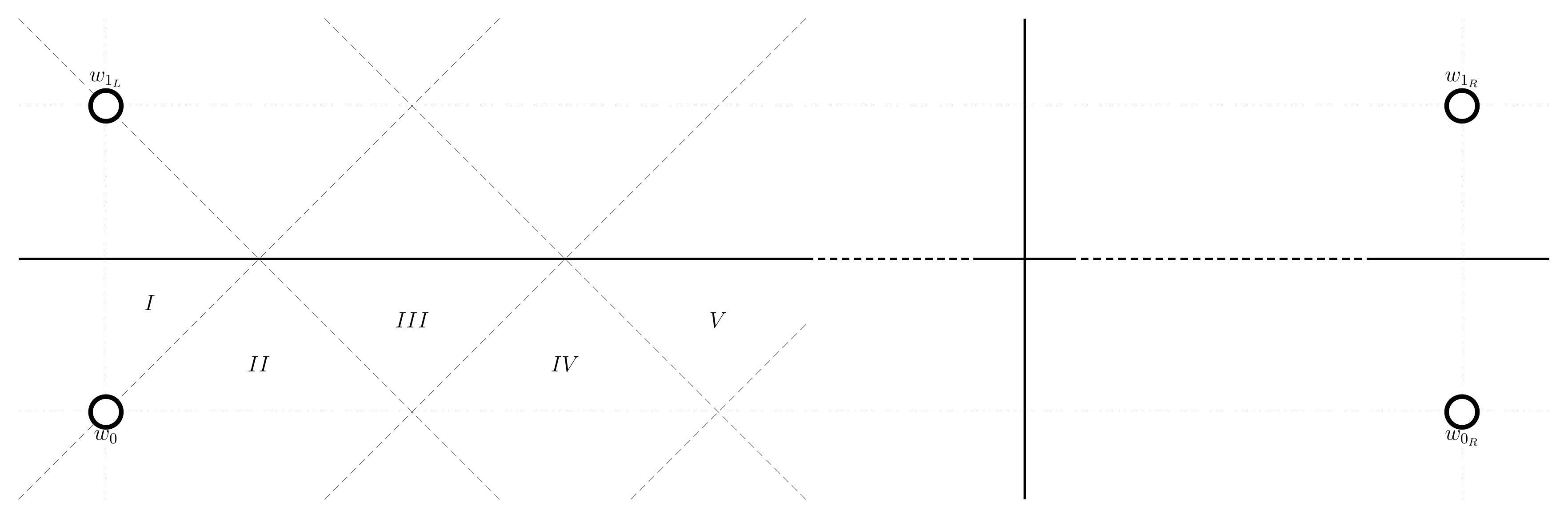}
\caption{The partition of $\mathcal{P}$ for a general $2 \times 2$ subgrid of $W$.}
\label{fig:GridPartition}
\end{figure}

In order to obtain a feel for how $V^+(b_1)$ can appear under White's grid arrangements, in Figure~\ref{fig:GridCellsExamples} we shall draw the first three unique structures that $V^+(b_1)$ can take from the partition displayed in Figure~\ref{fig:GridPartitionExample} (combining $II$ and $IV$ into what we refer to as Section $II$). Whilst they are shown to extend to the furthest that a grid arrangement would allow, one can easily imagine how the cells are truncated if they hit the boundary of $\mathcal{P}$ before closing (for example, if $w_{0_{LL}}$ and $w_{1_{LL}}$ do not exist then $V^+(b_1)$ in Figure~\ref{fig:GridCellI} will be prevented from expanding any further left than the boundary of $\mathcal{P}$ at $x=-\frac{p}{2a}$).

\begin{figure}[!ht]
\centering
\begin{subfigure}{.8\textwidth}
  \centering
  \includegraphics[width=0.9\textwidth]{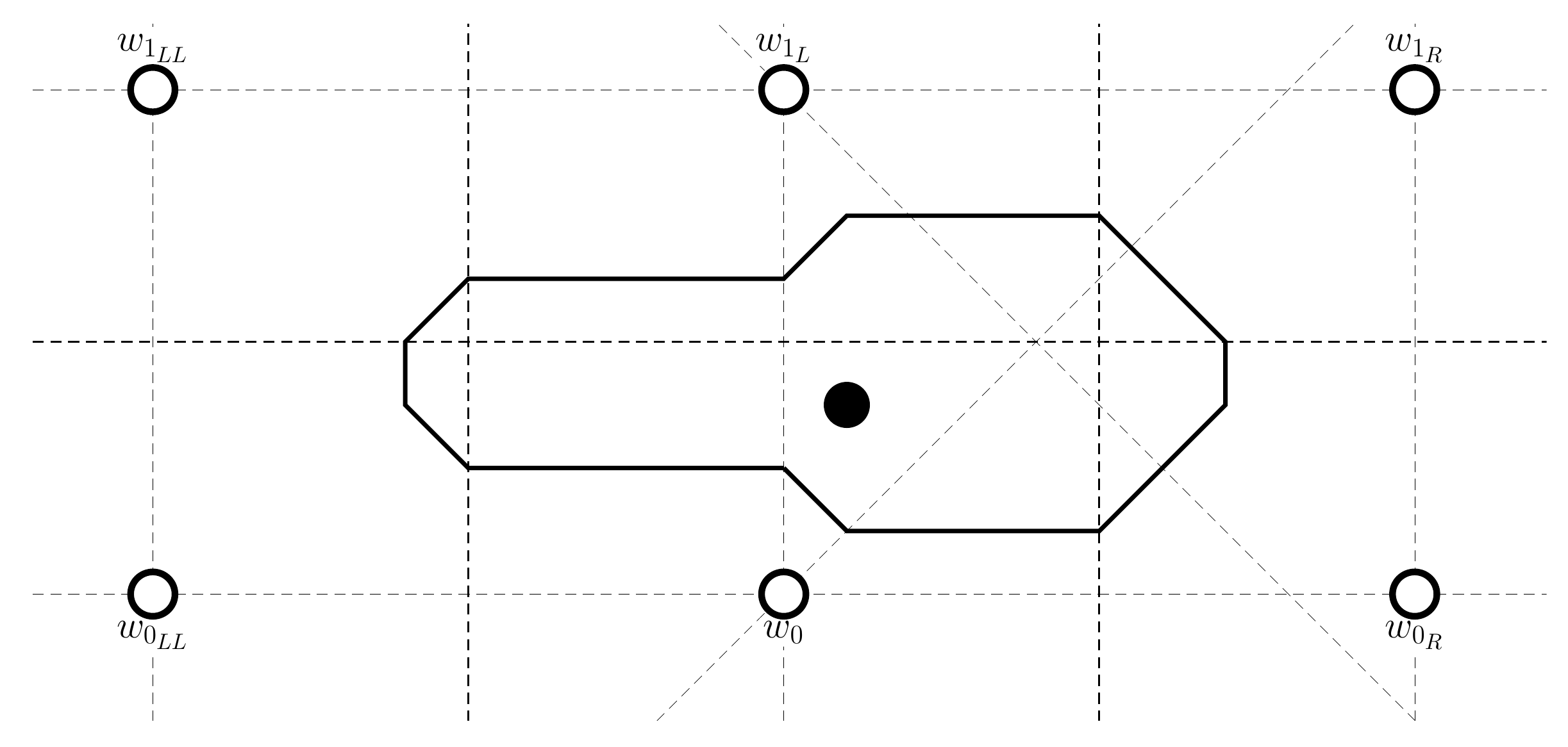}
  \caption{Voronoi cell $V^+(b_1)$ for $b_1$ in Section $I$.}
  \label{fig:GridCellI}
\end{subfigure}

\begin{subfigure}{.5\textwidth}
  \centering
  \includegraphics[width=0.9\textwidth]{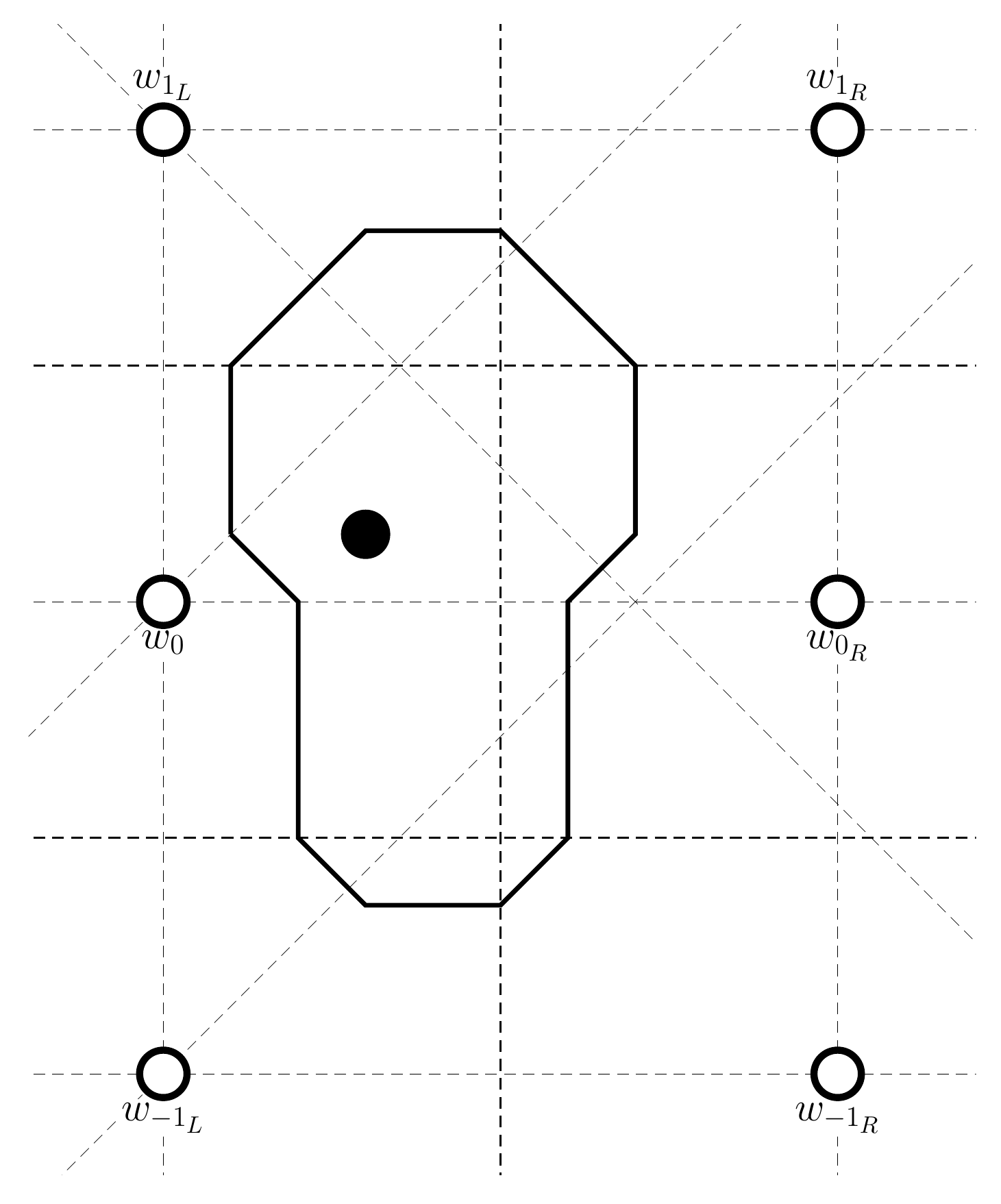}
  \caption{Voronoi cell $V^+(b_1)$ for $b_1$ in Section $II$.}
  \label{fig:GridCellIII}
\end{subfigure}%
\begin{subfigure}{.5\textwidth}
  \centering
  \includegraphics[width=0.9\textwidth]{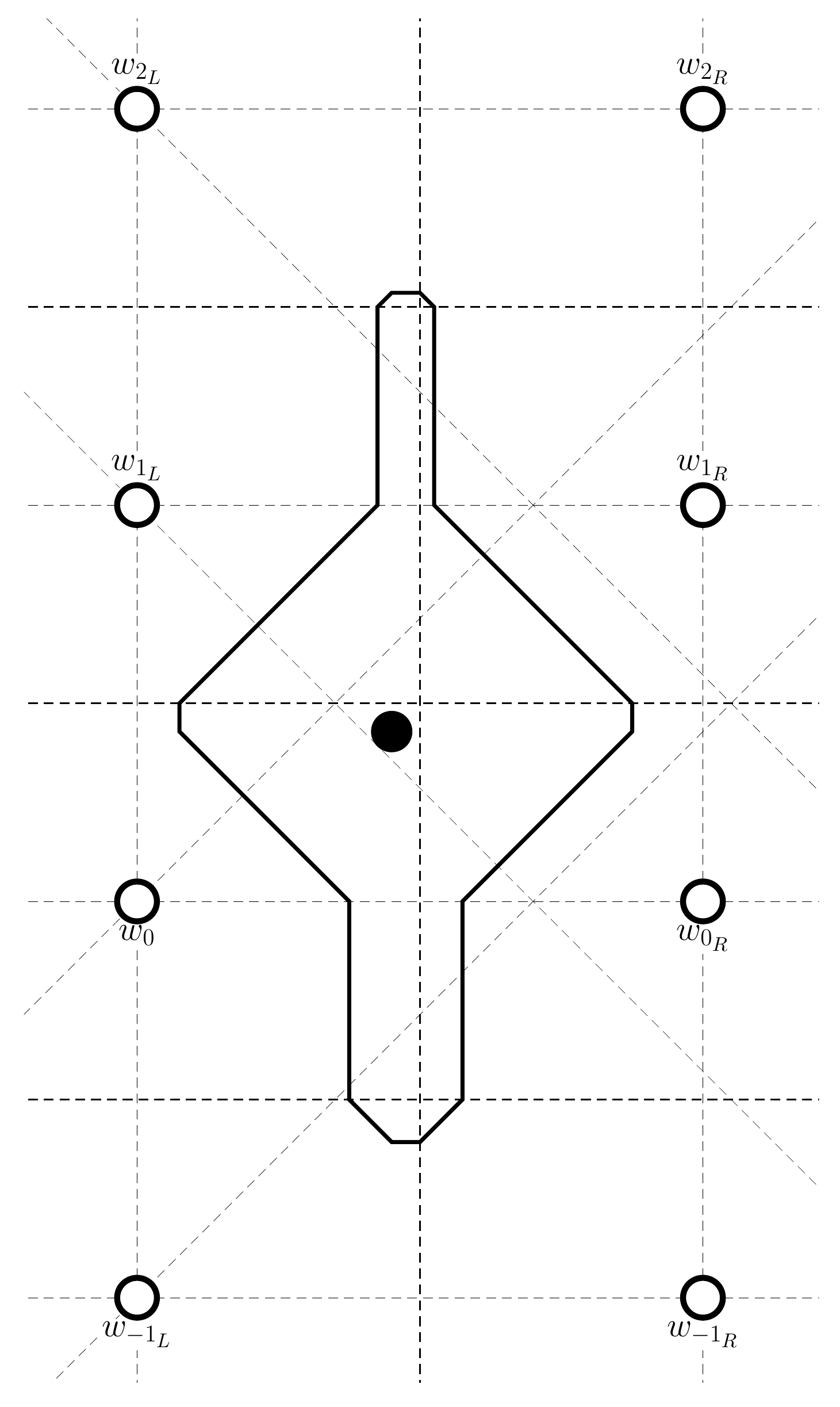}
  \caption{Voronoi cell $V^+(b_1)$ for $b_1$ in Section $III$.}
  \label{fig:GridCellIV}
\end{subfigure}
\caption{Voronoi cells $V^+(b_1)$ for $b_1$ in respective sections according to Figure \ref{fig:GridPartition}.}
\label{fig:GridCellsExamples}
\end{figure}

From these figures we can see that our situation is very similar to the situation we faced in Section~\ref{sec:WhiteRow}. Yet again, if $b_1$ is placed within Section $I$ then $V^+(b_1)$ exhibits a particularly unique structure, whilst Section $2l$ and $2l+1$ Voronoi cells resemble one another fairly closely, with $V^+(b_1)$ entering $V^\circ(w_{{l+1}_L}) \cup V^\circ(w_{{l+1}_R})$ as $b_1$ ventures from Section $2l$ to $2l+1$ and enters into $V^\circ(w_{{-(l+1)}_L}) \cup V^\circ(w_{{-(l+1)}_R})$ as $b_1$ ventures from Section $2l+1$ to $2(l+1)$. Therefore, as before, it will prove useful to analyse the area of $V^+(b_1)$ stolen from each constituent block $V^\circ(w_{{i}_{L}}) \cup V^\circ(w_{{i}_{R}})$ for $i \in \mathbb{Z}$ in the case that $b_1 = (x,y)$ is not in Section $I$.

\paragraph{Theft from $V^\circ(w_{{0}_L}) \cup V^\circ(w_{{0}_R})$}

Firstly we will look at the area of $V^+(b_1)$ intersected with $V^\circ(w_{{0}_L}) \cup V^\circ(w_{{0}_R})$. Since $b_1 \in \mathcal{CC}^1(w_{{0}_L}) \cap \mathcal{CC}^4(w_{{0}_R})$ (as we are not considering $b_1$ in Section $I$) this area always exists, has vertices $(\frac{x-y}{2},\frac{q}{2b})$, $(\frac{x-y}{2},y)$, $(\frac{x+y}{2},0)$, $(\frac{x+y}{2},-\frac{q}{2b})$, $(\frac{p}{2a}+\frac{x-y}{2},-\frac{q}{2b})$, $(\frac{p}{2a}+\frac{x-y}{2},0)$, $(\frac{p}{2a}+\frac{x+y}{2},y)$, and $(\frac{p}{2a}+\frac{x+y}{2},\frac{q}{2b})$, and totals

\begin{equation*}
\begin{split}
    Area(V^+(b_1) \cap (V^\circ(w_{{0}_L}) \cup V^\circ(w_{{0}_R}))) &= (\frac{p}{2a}+\frac{x+y}{2} - \frac{x+y}{2}) \times \frac{q}{b} - y^2 \\
    &= \frac{pq}{2ab} - y^2 \, .
\end{split}
\end{equation*}
It is clear that this area is maximised by $y=0$ and is invariant in the value of $x$.

\paragraph{Theft from $V^\circ(w_{{i}_L}) \cup V^\circ(w_{{i}_R})$ for $b_1 \in \mathcal{CC}^8(w_{{i}_L})$ or $b_1 \in \mathcal{CC}^1(w_{{-i}_L})$}

Now we shall determine the areas of $V^+(b_1)$ intersected with $V^\circ(w_{{i}_L}) \cup V^\circ(w_{{i}_R})$ in which $V^+(b_1)$ steals area at every $y$ value of $V^\circ(w_{{i}_L}) \cup V^\circ(w_{{i}_R})$. That is, if $i>0$, that $V^+(b_1)$ also enters $V^\circ(w_{{i+1}_L}) \cup V^\circ(w_{{i+1}_R})$ (if existing) so $b_1 \in \mathcal{CC}^8(w_{{i}_L})$ (which restricts $b_1$ to also lie within $\mathcal{CC}^5(w_{{i}_R})$) and, if $i<0$, that $V^+(b_1)$ also enters $V^\circ(w_{{i-1}_L}) \cup V^\circ(w_{{i-1}_R})$ (if existing) so $b_1 \in \mathcal{CC}^1(w_{{i}_L})$ (which restricts $b_1$ to also lie within $\mathcal{CC}^4(w_{{i}_R})$). This restricts $b_1$ to lie within Sections $2i+1$ and beyond if $i>0$ or Sections $2(-i+1)$ and beyond if $i<0$. By symmetry these areas (for $i>0$ and $i<0$) have the same structure and, since these structures rely only on the distance between $b_1$ and the generators of the Voronoi cells in question, the representations are nigh identical.

For $i>0$ the area has vertices $(\frac{iq}{2b}+\frac{x-y}{2},\frac{(2i+1)q}{2b})$, $(\frac{iq}{2b}+\frac{x-y}{2},\frac{iq}{b})$, $(\frac{(i-1)q}{2b}+\frac{x-y}{2},\frac{(2i-1)q}{2b})$, $(\frac{p}{2a} - \frac{(i-1)q}{2b} + \frac{x+y}{2}, \frac{(2i-1)q}{2b})$, $(\frac{p}{2a} - \frac{iq}{2b} + \frac{x+y}{2}, \frac{iq}{b})$, and $(\frac{p}{2a} - \frac{iq}{2b} + \frac{x+y}{2}, \frac{(2i+1)q}{2b})$ and totals
\begin{equation*}
\begin{split}
    Area(V^+(b_1) \cap (V^\circ(w_{{i}_L}) \cup V^\circ(w_{{i}_R}))) &= (\frac{p}{2a} - \frac{iq}{2b} + \frac{x+y}{2} - (\frac{iq}{2b}+\frac{x-y}{2})) \times \frac{q}{b} + \left(\frac{q}{2b}\right)^2 \\
    &= (\frac{p}{2a} - \frac{iq}{b} + y) \times \frac{q}{b} + \frac{q^2}{4b^2} \\
    &= \frac{q}{b}y + \frac{pq}{2ab} - \frac{(4i-1)q^2}{4b^2} \, .
\end{split}
\end{equation*}
For $i<0$ the area has vertices $(\frac{x+y}{2}-\frac{(i+1)q}{2b},\frac{(2i+1)q}{2b})$, $(\frac{x+y}{2}-\frac{iq}{2b},\frac{iq}{b})$, $(\frac{x+y}{2}-\frac{iq}{2b},\frac{(2i-1)q}{2b})$, $(\frac{p}{2a} + \frac{iq}{2b} + \frac{x-y}{2}, \frac{(2i-1)q}{2b})$, $(\frac{p}{2a} + \frac{iq}{2b} + \frac{x-y}{2}, \frac{iq}{b})$, and $(\frac{p}{2a} + \frac{(i+1)q}{2b} + \frac{x-y}{2}, \frac{(2i+1)q}{2b})$ and totals
\begin{equation*}
\begin{split}
    Area(V^+(b_1) \cap (V^\circ(w_{{i}_L}) \cup V^\circ(w_{{i}_R}))) &= (\frac{p}{2a} + \frac{iq}{2b} + \frac{x-y}{2} - (\frac{x+y}{2}-\frac{iq}{2b})) \times \frac{q}{b} + \left(\frac{q}{2b}\right)^2 \\
    &= (\frac{p}{2a} + \frac{iq}{b} - y) \times \frac{q}{b} + \frac{q^2}{4b^2} \\
    &= -\frac{q}{b}y + \frac{pq}{2ab} + \frac{(4i+1)q^2}{4b^2} \, .
\end{split}
\end{equation*}
Again, neither area depends on the value of $x$; however, the area is maximised if $i>0$ by $y=\frac{p}{2a}$ and if $i<0$ by $y=0$.

\paragraph{Theft from $V^\circ(w_{{i}_L}) \cup V^\circ(w_{{i}_R})$ for $b_1 \in \mathcal{CC}^7(w_{{i}_L})$ or $b_1 \in \mathcal{CC}^2(w_{{-i}_L})$}

Finally we shall determine the areas of $V^+(b_1)$ intersected with $V^\circ(w_{{i}_L}) \cup V^\circ(w_{{i}_R})$ in which $V^+(b_1)$ steals only at certain values of $y$ in $V^\circ(w_{{i}_L}) \cup V^\circ(w_{{i}_R})$. That is, if $i>0$, that $V^+(b_1)$ does not enter $V^\circ(w_{{i+1}_L}) \cup V^\circ(w_{{i+1}_R})$ (if existing) so $b_1 \in \mathcal{CC}^7(w_{{i}_L})$ (which restricts $b_1$ to also lie within $\mathcal{CC}^6(w_{{i}_R})$) and, if $i<0$, that $V^+(b_1)$ does not enter $V^\circ(w_{{i-1}_L}) \cup V^\circ(w_{{i-1}_R})$ (if existing) so $b_1 \in \mathcal{CC}^2(w_{{i}_L})$ (which restricts $b_1$ to also lie within $\mathcal{CC}^3(w_{{i}_R})$). This restricts $b_1$ to lie within Sections $2i-1$ and $2i$ if $i>0$ (note that this area only holds for $b_1$ within Section $II$ if $i=1$) or Sections $2(-i)$ and $2(-i)+1$ if $i<0$. Analogously to above, both of these areas are the same structures with only subtly different representations.

For $i>0$ the area has vertices $(\frac{p}{2a},\frac{iq}{2b} + \frac{x+y}{2})$, $(x,\frac{iq}{2b} + \frac{x+y}{2})$, $(\frac{(i-1)q}{2b} + \frac{x-y}{2},\frac{(2i-1)q}{2b})$, and $(\frac{p}{2a} - \frac{(i-1)q}{2b} + \frac{x+y}{2},\frac{(2i-1)q}{2b})$ and totals

\begin{equation*}
\begin{split}
    Area(V^+(b_1) \cap (V^\circ(w_{{i}_L}) \cup V^\circ(w_{{i}_R}))) &= \left(\frac{p}{2a} - \left(\frac{(i-1)q}{2b} + \frac{x-y}{2}\right) \right) \times \left(\frac{iq}{2b} + \frac{x+y}{2} - \frac{(2i-1)q}{2b}\right) \\
    &= \left(\frac{p}{2a} - \frac{(i-1)q}{2b} - \frac{x-y}{2} \right) \times \left( - \frac{(i-1)q}{2b} + \frac{x+y}{2}\right) \\
    &= - \frac{x^2}{4} + \frac{y^2}{4} + \frac{p}{4a}x + \left(\frac{p}{4a} - \frac{(i-1)q}{2b}\right)y \\ &\qquad- \frac{(i-1)pq}{4ab} + \frac{(i-1)^2q^2}{4b^2} \, . \\
%    Area(V^+(b_1) \cap (V^\circ(w_{{i}_L}) \cup V^\circ(w_{{i}_R}))) &= \left(\frac{p}{2a} - x\right) \times \left(\frac{iq}{2b} + \frac{x+y}{2} - \frac{(2i-1)q}{2b}\right) \\ &\qquad+ \left(\frac{iq}{2b} + \frac{x+y}{2} - \frac{(2i-1)q}{2b}\right)^2 \\
%    &= \left(\frac{p}{2a} - x\right) \times \left(\frac{x+y}{2} - \frac{(i-1)q}{2b}\right) + \left(\frac{x+y}{2} - \frac{(i-1)q}{2b}\right)^2 \\\
%    &= - \frac{x^2}{4} + \frac{y^2}{4} + \frac{p}{4a}x + (\frac{p}{4a} - \frac{(i-1)q}{2b})y \\ &\qquad- \frac{(i-1)pq}{4ab} + \frac{(i-1)^2q^2}{4b^2} \, .
\end{split}
\end{equation*}
For $i<0$ the area has vertices $(-\frac{(i+1)q}{2b} + \frac{x+y}{2},\frac{(2i+1)q}{2b})$, $(x,\frac{iq}{2b} - \frac{x-y}{2})$, $(\frac{p}{2a},\frac{iq}{2b} - \frac{x-y}{2})$, and $(\frac{p}{2a} + \frac{(i+1)q}{2b} + \frac{x-y}{2}, \frac{(2i+1)q}{2b})$ and totals
\begin{equation*}
\begin{split}
    Area(V^+(b_1) \cap (V^\circ(w_{{i}_L}) \cup V^\circ(w_{{i}_R}))) &= \left(\frac{p}{2a} - \left(-\frac{(i+1)q}{2b} + \frac{x+y}{2}\right) \right) \times \left(\frac{(2i+1)q}{2b} - \left(\frac{iq}{2b} - \frac{x-y}{2}\right)\right) \\
    &= \left(\frac{p}{2a} + \frac{(i+1)q}{2b} - \frac{x+y}{2} \right) \times \left(\frac{(i+1)q}{2b} + \frac{x-y}{2}\right) \\
    &= - \frac{x^2}{4} + \frac{y^2}{4} + \frac{p}{4a}x - \left(\frac{p}{4a} + \frac{(i+1)q}{2b}\right)y \\ &\qquad+ \frac{(i+1)pq}{4ab} + \frac{(i+1)^2q^2}{4b^2} \, .
\end{split}
\end{equation*}
Now both areas are maximised by $x=\frac{p}{2a}$ to give, if $i > 0$, $$Area(V^+(b_1) \cap (V^\circ(w_{{i}_L}) \cup V^\circ(w_{{i}_R}))) = \frac{y^2}{4} + \left(\frac{p}{4a} - \frac{(i-1)q}{2b}\right)y - \frac{(i-1)pq}{4ab} + \frac{p^2}{16a^2} + \frac{(i-1)^2q^2}{4b^2}$$ and, if $i<0$, $$Area(V^+(b_1) \cap (V^\circ(w_{{i}_L}) \cup V^\circ(w_{{i}_R}))) = \frac{y^2}{4} - \left(\frac{p}{4a} + \frac{(i+1)q}{2b}\right)y + \frac{(i+1)pq}{4ab} + \frac{p^2}{16a^2} + \frac{(i+1)^2q^2}{4b^2}$$ (and if this optimum is not achievable then, fixing $y$, the area increases as $x$ moves closer to $\frac{p}{2a}$). If $i>0$ then the maximiser is
%\begin{alignat*}{2}
%    y&=\frac{(i-1)q}{b} - \frac{p}{2a} &&\, \\
%    &=\frac{(2i-3)q}{2b} - \frac{p}{2a} + \frac{q}{2b} \text{ and } &&=\frac{(2i-3)q}{2b} - \frac{p}{2a} + \frac{q}{2b} \\
%    &\leq \frac{q}{2b} &&\leq \frac{q}{2b}\\
%\end{alignat*}
\begin{equation*}
\begin{split}
    y&=\frac{(i-1)q}{b} - \frac{p}{2a} \\
    &=\frac{(2i-3)q}{2b} - \frac{p}{2a} + \frac{q}{2b} \\
    &\leq \frac{q}{2b} \\
\end{split}
\end{equation*}
relying on the fact that, since Section $2i-1$ must exist in some form (for $i \neq 1$) for this area to be formed, it must be the case that $\frac{(2i-3)q}{2b} \leq \frac{p}{2a}$. On the other hand, if $i<0$ then the maximiser is
\begin{equation*}
\begin{split}
    y&=\frac{p}{2a} + \frac{(i+1)q}{b} \\
    &\geq 0 \\
\end{split}
\end{equation*}
relying on the fact that, since Section $2(-i)$ must exist in some form for this area to be formed, it must be the case that $\frac{(-i-1)q}{b} \leq \frac{p}{2a}$.

\medskip

Not only are these calculations useful for formulating the representation of the different areas of $V^+(b_1)$ for $b_1$ contained in different sections within the first quadrant of $V^\circ(w_0)$, but they provide a very strong clue to where we will find the optimum to maximise $Area(V^+(b_1))$. Recall from Section~\ref{sec:BlackRow} (in particular the discussion surrounding Figure~\ref{fig:RowOptimalComp}) that Black will always improve upon their area, when locating upon the boundary of two sections, by choosing to locate in the higher section -- or in our case, the rightmost section. Additionally we have found that the $x$-coordinate of $b_1$ does not affect the area of $V^+(b_1)$ within most Voronoi cells in $\mathcal{VD}(W)$ and, when the value of $x$ does contribute to the representation of $Area(V^+(b_1))$, the optimal direction of movement is rightwards within the first quadrant of $V^\circ(w_0)$. Combining these two properties, we can say that the optimum $b_1$ within Section $II$ and beyond lies on the line $x^*=\frac{p}{2a}$; for any fixed $y$, $Area(V^+(b_1))$ increases as $x$ increases within a section, and will increase upon crossing a configuration line into a section of greater value (increasing $x$) so the best point for fixed $y$ lies at $x=\frac{p}{2a}$.

This remains true no matter whether $V^+(b_1)$ intersects the top or bottom perimeter of $\mathcal{P}$. Therefore, since our interest is Black's best point, in contrast to Section~\ref{sec:WhiteRow} we need not calculate the areas of every possible $V^+(b_1)$ structure and optimise this area over the partition within which this structure is maintained. Now we need only explore Section $I$ and the line $x^*=\frac{p}{2a}$, taking care to remember to check for special cases if $V^+(b_1)$ interacts with the boundary of $\mathcal{P}$.

\section{Black's optimal strategy: White plays an \texorpdfstring{$a \times b$}{a x b} grid}
\label{sec:BlackGrid}

\subsection{Black's best point}
\label{sec:BlackGridBest}

Many ideas from our discussion in Section~\ref{sec:BlackRow} carry over to the case where White plays a grid. We will limit our exploration of Black's best point $b^*$ to \emph{core quadrants}. We will call the first quadrant of $V^\circ(w_0)$ a \emph{core} quadrant if it borders only other Voronoi cells (and not the boundary of $\mathcal{P}$); that is, if $w_{0_R}$ and $w_{1_L}$ both exist. It is only these quadrants that we are interested in because, as explained in Section~\ref{sec:BlackRow}, Black's best point will never be contained in a quadrant bordering $\mathcal{P}$ as long as core quadrants exist. As before, Black's best point will never be located next to the boundary of $\mathcal{P}$. This is simply because, if $V^+(b_1)$ touches one boundary of $\mathcal{P}$, translating the point a distance $\frac{p}{a}$ or $\frac{q}{b}$ perpendicularly away from a vertical or horizontal boundary of $\mathcal{P}$ respectively will allow $V^+(b_1)$ to enter a new uncharted pair of Voronoi cells $V^\circ(w_{{i}_L}) \cup V^\circ(w_{{i}_R})$ (up to orientation), and if $V^+(b_1)$ touches opposite boundaries of $\mathcal{P}$ then, since $b_1$ is more effective at stealing area from Voronoi cells closest to $b_1$, $b_1$ does better when equally distant from both boundaries. This will be explained in greater detail later within this section.

As in Section~\ref{sec:BlackRow}, Figure \ref{fig:GridOptimals} will depict all optimal locations of $b_1$ within each required section under the particular circumstances we will discuss below; again, Section $IV$ and Section $III$ are depicted as the poster children for the general Section $2l$ and Section $2l+1$ results respectively and for clarity these respective sections will be shaded in each figure.

As described above there are two areas of interest within these core quadrants: Section $I$, and the line $x^*=\frac{p}{2a}$.

\paragraph{Section $I$} First we explore $V^+(b_1)$ for $b_1$ in Section $I$. It has vertices, tracing the perimeter clockwise, $(0,\frac{x+y}{2})$, $(x,\frac{y-x}{2})$, $(\frac{p}{2a},\frac{y-x}{2})$, $(\frac{p}{2a}+\frac{x+y}{2},y)$, $(\frac{p}{2a}+\frac{x+y}{2},\frac{q}{2b})$, $(\frac{p}{2a},\frac{q}{2b}+\frac{x+y}{2})$, $(x,\frac{q}{2b}+\frac{x+y}{2})$, $(0,\frac{q}{2b}+\frac{y-x}{2})$, $(-\frac{p}{2a},\frac{q}{2b}+\frac{y-x}{2})$, $(-\frac{p}{2a}-\frac{y-x}{2},\frac{q}{2b})$, $(-\frac{p}{2a}-\frac{y-x}{2},y)$, and $(-\frac{p}{2a},\frac{x+y}{2})$, giving an area
\begin{equation*}
    \begin{split}
        Area(V^+(b_1)) &= (\frac{p}{2a}+\frac{x+y}{2}) \times (\frac{q}{2b}+\frac{x+y}{2}-\frac{y-x}{2}) - x^2 - (\frac{x+y}{2})^2 \\
        &\qquad + (\frac{p}{2a}+\frac{y-x}{2}) \times (\frac{q}{2b}+\frac{y-x}{2} - \frac{x+y}{2}) - (\frac{y-x}{2})^2 \\
        &= \frac{q}{2b} (\frac{p}{a}+y) + x (x) - \frac{3x^2+y^2}{2} \\
        &= - \frac{x^2}{2} - \frac{y^2}{2} + \frac{q}{2b}y +  \frac{pq}{2ab}\\
    \end{split}
\end{equation*}
or, if points $w_{{0}_{LL}}$ and $w_{{1}_{LL}}$ do not exist (i.e. $V^\circ(w_0)$ borders the perimeter of $\mathcal{P}$),

\begin{equation*}
    \begin{split}
        Area(V^+(b_1)) &= (\frac{p}{2a}+\frac{x+y}{2}) \times (\frac{q}{2b}+\frac{x+y}{2}-\frac{y-x}{2}) - x^2 - (\frac{x+y}{2})^2 \\
        &\qquad+ \frac{p}{2a} \times (\frac{q}{2b}+\frac{y-x}{2} - \frac{x+y}{2}) \\
        &= \frac{q}{2b}(\frac{p}{a}+\frac{x+y}{2}) + x (\frac{x+y}{2}) - \frac{5x^2+2xy+y^2}{4} \\
        &= - \frac{3x^2}{4} - \frac{y^2}{4} + \frac{q}{4b}x + \frac{q}{4b}y + \frac{pq}{2ab} \, .
    \end{split}
\end{equation*}

Assuming that points $w_{{0}_{LL}}$ and $w_{{1}_{LL}}$ do exist, to find the maximum of this over Section $I$ we first use gradient methods
\begin{equation*}
\begin{split}
\frac{\delta A}{\delta x}&=- x \\
\frac{\delta A}{\delta y}&=- y + \frac{q}{2b} \\
\end{split}
\end{equation*}
to ascertain that the maximum of $Area(V^+(b_1))$ is found at $b^*_1=(0,\frac{q}{2b})$ (still contained in Section~$I$), giving $Area(V^+(b^*_1))= \frac{pq}{2ab} + \frac{q^2}{8b^2}$. This is depicted in Figure \ref{fig:GridOptimalI}.

Alternatively, if the points $w_{{0}_{LL}}$ and $w_{{1}_{LL}}$ do not exist, we have gradients
\begin{equation*}
\begin{split}
\frac{\delta A}{\delta x}&=- \frac{3x}{2} + \frac{q}{4b} \\
\frac{\delta A}{\delta y}&=- \frac{y}{2} + \frac{q}{4b} \\
\end{split}
\end{equation*}
so the area reaches its maximum at $b^*_1=(\frac{q}{6b},\frac{q}{2b})$ (still contained in Section $I$), giving $Area(V^+(b^*_1))= \frac{pq}{2ab}+\frac{q^2}{12b^2}$. This is depicted in Figure \ref{fig:GridOptimalIRes}.

\begin{figure}[!ht]
\begin{subfigure}{.59\textwidth}
  \centering
  \includegraphics[width=0.9\textwidth]{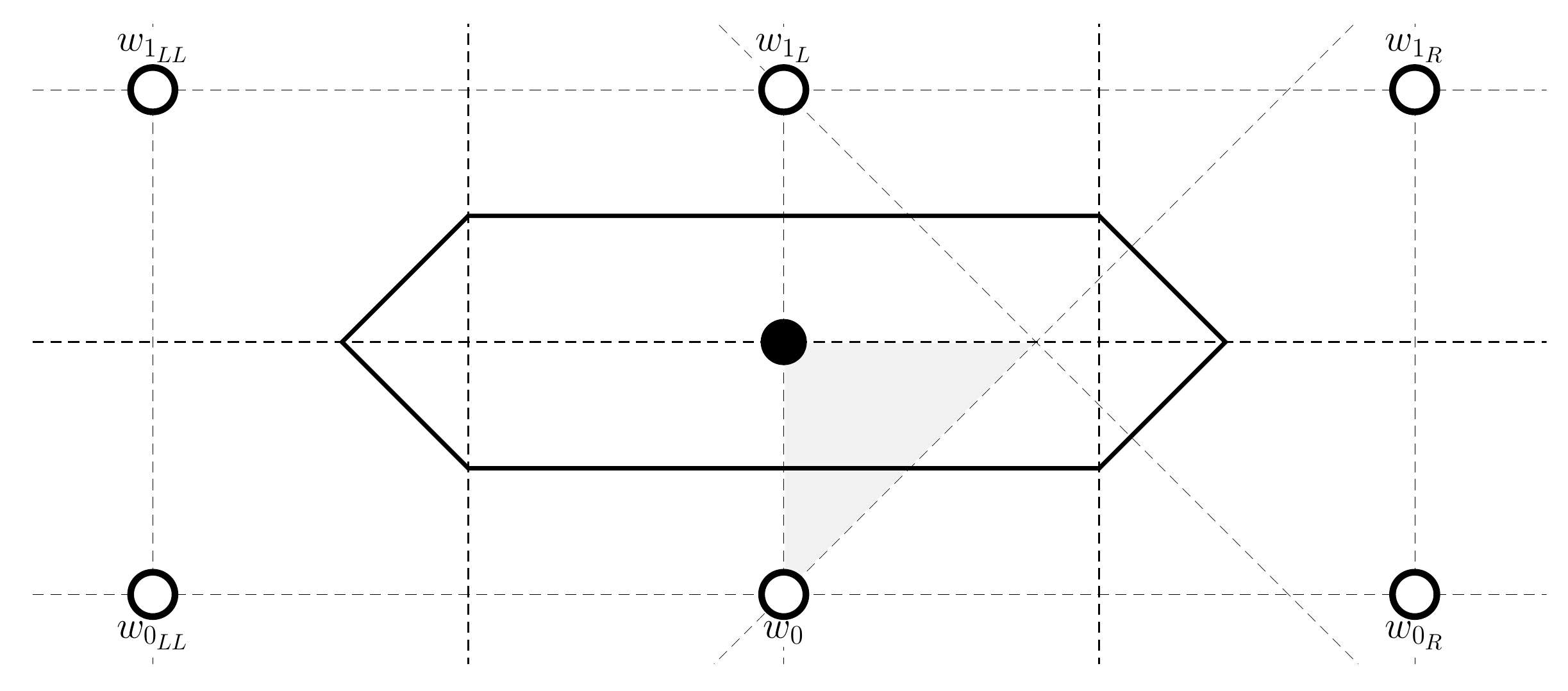}
  \caption{$b_1=(0,\frac{q}{2b})$.}
%  \caption{$Area(V^+((0,\frac{q}{2b})))=\frac{pq}{2ab} + \frac{q^2}{8b^2}$.}
  \label{fig:GridOptimalI}
\end{subfigure}%
\begin{subfigure}{.41\textwidth}
  \centering
  \includegraphics[width=0.9\textwidth]{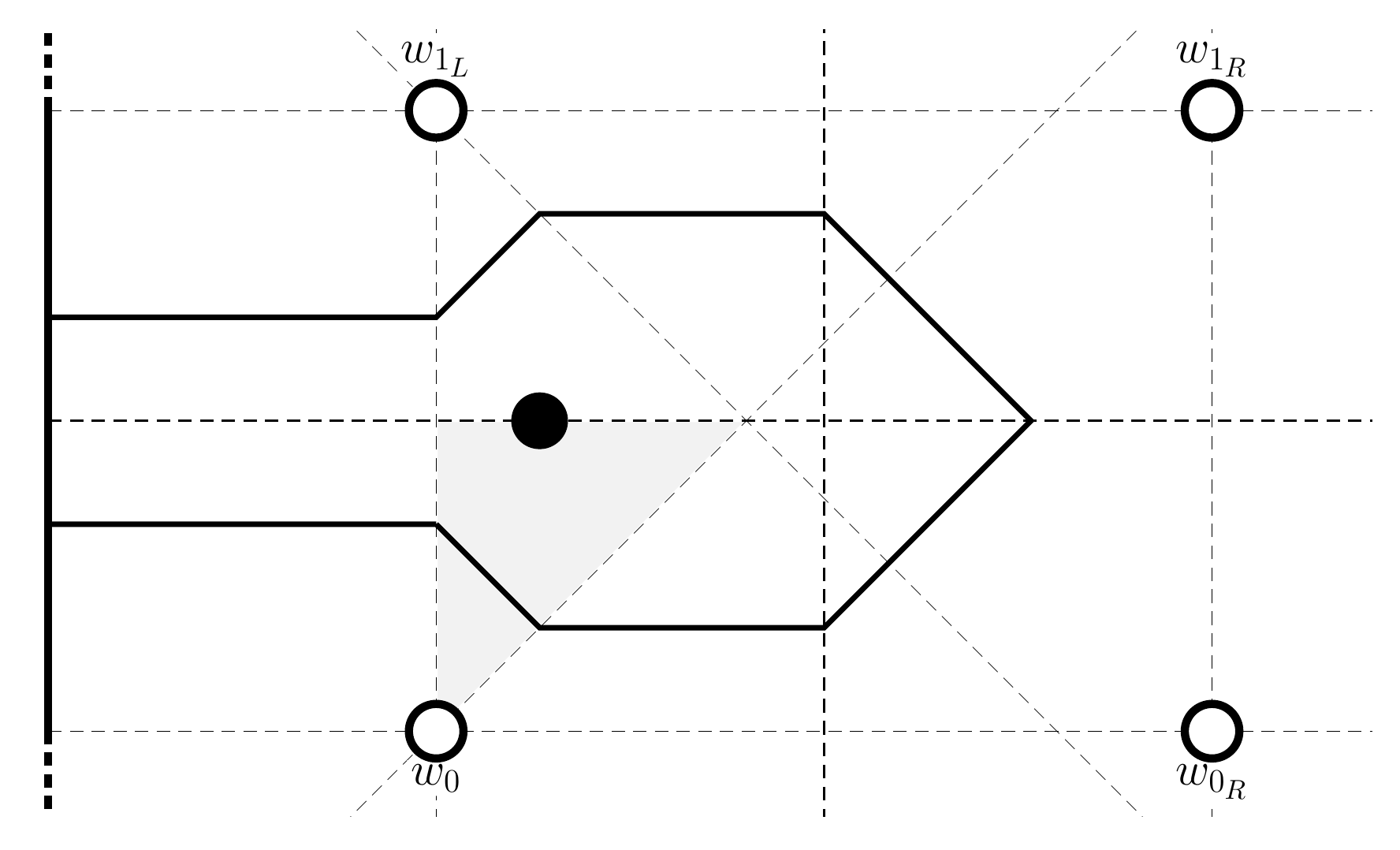}
  \caption{$b_1=(\frac{q}{6b},\frac{q}{2b})$.}
%  \caption{$Area(V^+((\frac{q}{6b},\frac{q}{2b})))=\frac{pq}{2ab}+\frac{q^2}{12b^2}$.}
  \label{fig:GridOptimalIRes}
\end{subfigure}
\caption{Maximal area Voronoi cells $V^+(b_1)$ for $b_1$ within Section $I$.}
\end{figure}

\paragraph{Section $2l$ upon $x^* = \frac{p}{2a}$}
Within the first quadrant of $V^\circ(w_0)$, the line $x^* = \frac{p}{2a}$ can be in an even section, an odd section, or both (entering into Section $2l+1$ from Section $2l$ as $y$ increases from $0$). Therefore we must explore the area formulae for $b_1$ in each section separately. Beginning with the placement of $b_1$ inside the even Section $2l$ (where $\frac{(l-1)q}{b} \leq \frac{p}{2a} \leq \frac{lq}{b} \Rightarrow l = \ceil{\frac{pb}{2qa}}$), $V^+(b_1)$ will extend into Voronoi cells $V^\circ(w_{{i}_L}) \cup V^\circ(w_{{i}_R})$ for $i \in \{ -l,\ldots, l\}$ (if existing).

Now the effect caused by the non-existence of these faraway Voronoi cells (i.e. the area that the boundary of $\mathcal{P}$ cuts off) can greatly diminish the suitability of the placement of $b_1$ if we are in search of Black's best point. It is clear that, if $w_{l_L}$ does not exist while $w_{{-l-1}_L}$ does, simply choosing $w_0$ to be the point directly below it ($w_{{-1}_L}$) will be beneficial to increasing the maximum area that $V^+(b_1)$ can take when locating within the first quadrant of $V^\circ(w_0)$, simply because this translation of our point of reference $\frac{q}{b}$ lower will allow $V^+(b_1)$ not only to keep exactly the same area as before but also to enter another previously untapped Voronoi cell of White's. The same is clearly true for the analogous case where it is $w_{{-l}_L}$ that does not exist while $w_{{l+1}_L}$ does, and these alterations of $w_0$ can of course be repeated until all Voronoi cells $V^\circ(w_{{i}_L}) \cup V^\circ(w_{{i}_R})$ for $i \in \{ -l,\ldots, l\}$ exist (in which case $b > 2l$ and $V^+(b_1)$ does not interact with $\mathcal{P}$, and we may not necessarily have a unique best point of reference for $w_0$), or until both $w_{l_L}$ and $w_{{-l-1}_L}$ or $w_{{-l}_L}$ and $w_{{l+1}_L}$ do not exist.

If both $w_{l_L}$ and $w_{{-l-1}_L}$ or $w_{{-l}_L}$ and $w_{{l+1}_L}$ do not exist while $w_{i_L}$ does for $i \in \{-l, \ldots, l-1\}$ or for $i \in \{-l+1, \ldots, l\}$ respectively then it is the case that $b = 2l$ and $V^+(b_1)$ touches only one bounding edge of $\mathcal{P}$. In this scenario we must decide which area we would prefer: that stolen from $V^\circ(w_{{l}_L}) \cup V^\circ(w_{{l}_R})$ or from $V^\circ(w_{{-l}_L}) \cup V^\circ(w_{{-l}_R})$. Since $b_1$ is being located in the first quadrant of $V^\circ(w_0)$, lying closer to $w_{l_L}$ than $w_{{-l}_L}$, it can steal a larger area from $V^\circ(w_{{l}_L}) \cup V^\circ(w_{{l}_R})$ than it could from $V^\circ(w_{{-l}_L}) \cup V^\circ(w_{{-l}_R})$, so it is favourable for $w_0$ to be chosen to be on the $l$th row of points in $W$ (counting from the bottom of the grid) so that $V^+(b_1)$ consists of areas stolen from $V^\circ(w_{{i}_L}) \cup V^\circ(w_{{i}_R})$ for all $i \in \{ -l+1,\ldots, l\}$ (as opposed to for all $i \in \{ -l,\ldots, l-1\}$).

This idea also applies to areas $V^+(b_1)$ that touch both horizontal bounding edges of $\mathcal{P}$ (so both $w_{l_L}$ and $w_{{-l}_L}$ do not exist, meaning that $b < 2l$) since, for $i \in \mathbb{Z}^+$, the area stolen from $V^\circ(w_{{\pm i}_L}) \cup V^\circ(w_{{\pm i}_R})$ will always be greater than that stolen from $V^\circ(w_{{\pm (i+1)}_L}) \cup V^\circ(w_{{\pm (i+1)}_R})$, and also the area stolen from $V^\circ(w_{{i}_L}) \cup V^\circ(w_{{i}_R})$ will always be greater than or equal to the area stolen from $V^\circ(w_{{-i}_L}) \cup V^\circ(w_{{-i}_R})$. Therefore it is still optimal to choose $w_0$ to be on the $\ceil{\frac{b}{2}}$th row of points in $W$ in order to steal from $V^\circ(w_{{i}_L}) \cup V^\circ(w_{{i}_R})$ for all $i \in \{-\ceil{\frac{b}{2}}+1,\ldots,\floor{\frac{b}{2}}\}$. Note though that, as described in the work preceding Figure~\ref{fig:GridPartition}, the final section possible in the top right quadrant of $V^\circ(w_0)$, where $w_0$ is on the $\ceil{\frac{b}{2}}$th row of points in $W$, is Section $b+1$, no matter whether $b$ is even or odd. Therefore there is only one section within which $V^\circ(w_0)$ touches both horizontal edges of $\mathcal{P}$ and this is Section $b+1$. We will explore this section separately to this investigation, after the Section $2l+1$ material is presented. Thus we shall only consider $b \geq 2l$ here.
%To avoid repeating calculations when exploring Section $2l+1$ cell, we will explore this Section $b+1$ in this  \todo{Do we need to treat this specially?}

Now that we have chosen the optimal $w_0$ and recorded which Voronoi cells $V^\circ(w_{{i}_L}) \cup V^\circ(w_{{i}_R})$ will be entered, we can calculate the areas of the Voronoi cell $V^+(b_1)$ for different values of $b$ and optimise the location of $b_1$ upon $x^*=\frac{p}{2a}$ within Section $2l$. If $b > 2l \left(= 2\ceil{\frac{pb}{2qa}}\right)$ then
\begin{equation*}
\begin{split}
    Area(V^+(b_1))&= Area(V^+(b_1) \cap (V^\circ(w_{{-l}_L}) \cup V^\circ(w_{{-l}_R}))) \\ &\qquad+ \sum_{i=-(l-1)}^{-1}{Area(V^+(b_1) \cap (V^\circ(w_{{i}_L}) \cup V^\circ(w_{{i}_R})))} \\ &\qquad+ Area(V^+(b_1) \cap (V^\circ(w_{{0}_L}) \cup V^\circ(w_{{0}_R}))) \\ &\qquad+ \sum_{i=1}^{l-1}{Area(V^+(b_1) \cap (V^\circ(w_{{i}_L}) \cup V^\circ(w_{{i}_R})))} \\ &\qquad+ Area(V^+(b_1) \cap (V^\circ(w_{{l}_L}) \cup V^\circ(w_{{l}_R}))) \\
    &= \frac{y^2}{4} - \left(\frac{p}{4a} + \frac{(-l+1)q}{2b}\right)y + \frac{(-l+1)pq}{4ab} + \frac{p^2}{16a^2} + \frac{(-l+1)^2q^2}{4b^2} \\ &\qquad+ \sum_{i=-l+1}^{-1}{\left( -\frac{q}{b}y + \frac{pq}{2ab} + \frac{(4i+1)q^2}{4b^2} \right)} \\ &\qquad+ \frac{pq}{2ab} - y^2 \\ &\qquad+ \sum_{i=1}^{l-1}{\left(\frac{q}{b}y + \frac{pq}{2ab} - \frac{(4i-1)q^2}{4b^2} \right)} \\ &\qquad+ \frac{y^2}{4} + (\frac{p}{4a} - \frac{(l-1)q}{2b})y - \frac{(l-1)pq}{4ab} + \frac{p^2}{16a^2} + \frac{(l-1)^2q^2}{4b^2} \\
    &= - \frac{y^2}{2} - \frac{(l-2)pq}{2ab} + \frac{p^2}{8a^2} + \frac{(l-1)^2q^2}{2b^2} + 2\sum_{i=1}^{l-1}{\left(\frac{pq}{2ab} - \frac{(4i-1)q^2}{4b^2} \right)} \\
    &= - \frac{y^2}{2} + \frac{lpq}{2ab} + \frac{p^2}{8a^2} - \frac{(l-1)lq^2}{2b^2} %\, , \\
\end{split}
\end{equation*}
and if $b=2l$ then, adapting this formula,
\begin{equation*}
\begin{split}
    Area(V^+(b_1))&= - \frac{y^2}{2} + \frac{lpq}{2ab} + \frac{p^2}{8a^2} - \frac{(l-1)lq^2}{2b^2} \\
    &\qquad- ``Area(V^+(b_1) \cap (V^\circ(w_{{-l}_L}) \cup V^\circ(w_{{-l}_R})))" \\
    &= - \frac{y^2}{2} + \frac{lpq}{2ab} + \frac{p^2}{8a^2} - \frac{(l-1)lq^2}{2b^2} \\
    &\qquad- \left(\frac{y^2}{4} - \left(\frac{p}{4a} - \frac{(l-1)q}{2b}\right)y - \frac{(l-1)pq}{4ab} + \frac{p^2}{16a^2} + \frac{(l-1)^2q^2}{4b^2}\right) \\
    &= - \frac{3y^2}{4} + \left(\frac{p}{4a} - \frac{(l-1)q}{2b}\right)y + \frac{(3l-1)pq}{4ab} + \frac{p^2}{16a^2} - \frac{(l-1)(3l-1)q^2}{4b^2} %\, . \\
\end{split}
\end{equation*}

It is straightforward to see that $(\frac{p}{2a},0)$ is the optimum if $b>2l$ giving $Area(V^+((\frac{p}{2a},0)))= \frac{lpq}{2ab} + \frac{p^2}{8a^2} - \frac{(l-1)lq^2}{2b^2}$. This is depicted in Figure~\ref{fig:GridOptimalEven}.

For $b=2l$ we have derivative
$$\frac{\delta A}{\delta y} = - \frac{3y}{2} + \frac{p}{4a} - \frac{(l-1)q}{2b}$$
which gives our optimum to be at $y^* = \frac{p}{6a} - \frac{(l-1)q}{3b} = \frac{p}{6a} - \frac{(b-2)q}{6b}$. However, in order for $b_1=\left(\frac{p}{2a},\frac{p}{6a} - \frac{(b-2)q}{6b}\right)$ to lie within Section $b$ it must be the case that, if $\frac{(b-2)q}{2b} \leq \frac{p}{2a} \leq \frac{(b-1)q}{2b}$, $b_1$ lies below $\mathcal{CL}^2(w_{{-\frac{b-2}{2}}_L})$ and, if $\frac{(b-1)q}{2b} \leq \frac{p}{2a} \leq \frac{q}{2}$, $b_1$ lies below $\mathcal{CL}^8(w_{{\frac{b}{2}}_L})$.

Therefore if $\frac{(b-2)q}{2b} \leq \frac{p}{2a} \leq \frac{(b-1)q}{2b}$ then it must be the case that $y \leq x - \frac{(b-2)q}{2b}$ so we require
$$\frac{p}{6a} - \frac{(b-2)q}{6b} \leq \frac{p}{2a} - \frac{(b-2)q}{2b} \Leftrightarrow \frac{(b-2)q}{3b} \leq \frac{2p}{3a} \Leftrightarrow \frac{(b-2)q}{4b} \leq \frac{p}{2a} \, .$$
Hence $b_1$ is the optimum in Section $b$ for all values $\frac{(b-2)q}{2b} \leq \frac{p}{2a} \leq \frac{(b-1)q}{2b}$.

Otherwise, if $\frac{(b-1)q}{2b} \leq \frac{p}{2a} \leq \frac{q}{2}$ then it must be the case that $y \leq \frac{q}{2} - x$ so we require
$$\frac{p}{6a} - \frac{(b-2)q}{6b} \leq \frac{q}{2} - \frac{p}{2a} \Leftrightarrow \frac{2p}{3a} \leq \frac{(2b-1)q}{3b} \Leftrightarrow \frac{p}{2a} \leq \frac{(2b-1)q}{4b} \, .$$
Therefore if $\frac{(b-1)q}{2b} \leq \frac{p}{2a} \leq \frac{(2b-1)q}{4b}$ then $b_1$ is the optimum in Section $b$. Otherwise, if $\frac{(2b-1)q}{4b}\leq \frac{p}{2a} \leq \frac{q}{2}$ then $b_1$ will lie above Section $b$. If this is the case then the optimum over Section $b$ must lie on the boundary between Section $b$ and $b+1$. However, as we saw when exploring Black's best point when White plays a row (see Figure~\ref{fig:RowOptimalComp} for example), any point lying in Section $b$ on the boundary with Section $b+1$ will be dominated by the identical point within Section $b+1$. Therefore Black's best point will not lie in Section $b$ if $\frac{(2b-1)q}{4b}\leq \frac{p}{2a}$.

To summarise, if $b$ is even then: if $\frac{(b-2)q}{2b} \leq \frac{p}{2a} \leq \frac{(2b-1)q}{4b}$ then the optimum in Section $b$ is $(\frac{p}{2a},\frac{p}{6a} - \frac{(b-2)q}{6b})$ giving $Area(V^+((\frac{p}{2a},\frac{p}{6a} - \frac{(b-2)q}{6b}))) = \frac{(2b-1)pq}{6ab} + \frac{p^2}{12a^2} - \frac{(b-2)(4b-1)q^2}{12b^2}$ (depicted in Figure~\ref{fig:GridOptimalEvenRes}), otherwise if $\frac{(2b-1)q}{4b}\leq \frac{p}{2a} \leq \frac{q}{2}$ then the optimum lies on the boundary with Section $b+1$ and is not Black's best point (and so is not drawn).

\begin{figure}[!ht]\ContinuedFloat
\begin{subfigure}{.52\textwidth}
  \centering
  \includegraphics[width=0.9\textwidth]{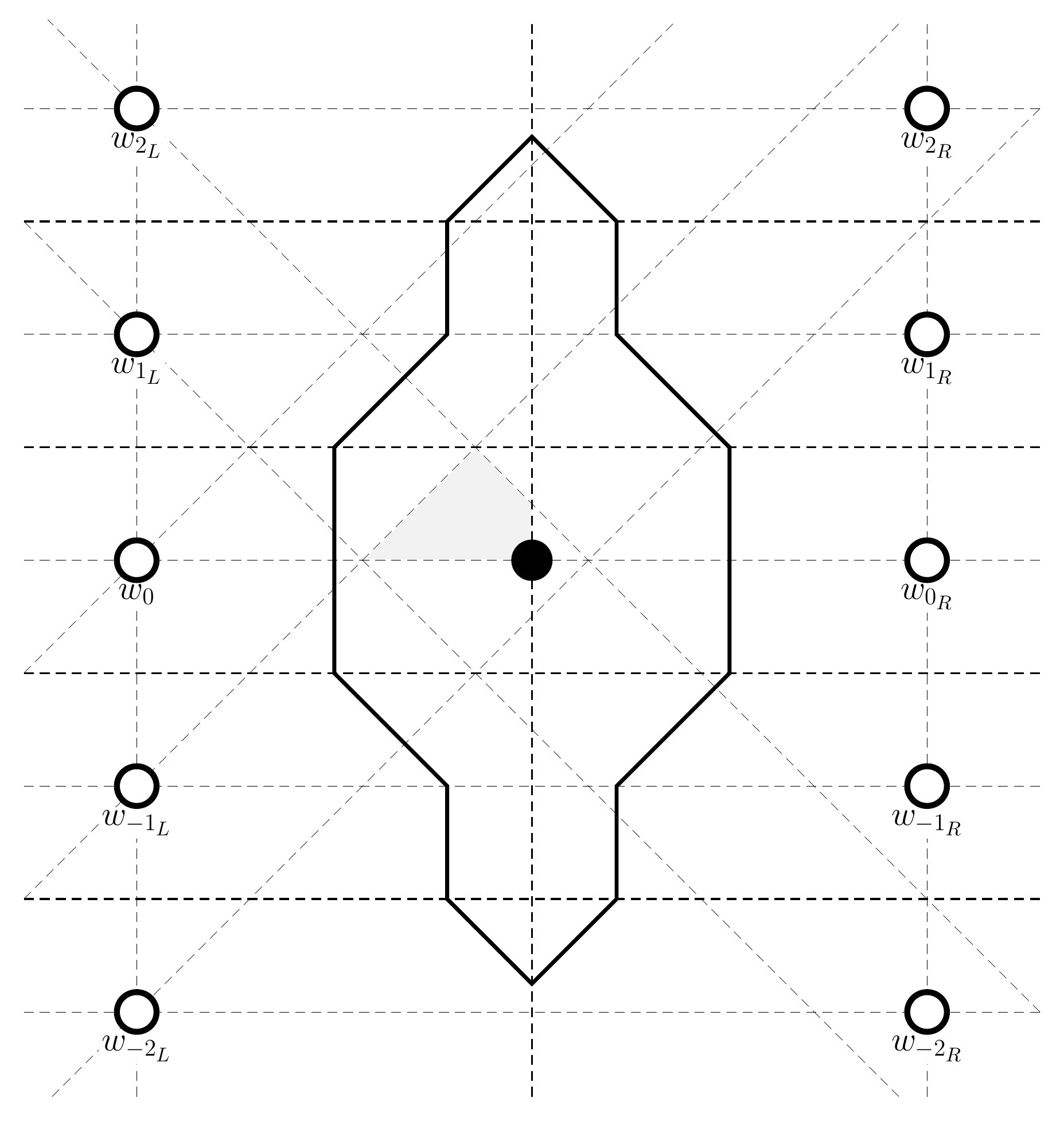}
  \caption{$b_1=(\frac{p}{2a},0)$.\\ \,}
  \label{fig:GridOptimalEven}
\end{subfigure}%
\begin{subfigure}{.48\textwidth}
  \centering
  \includegraphics[width=0.9\textwidth]{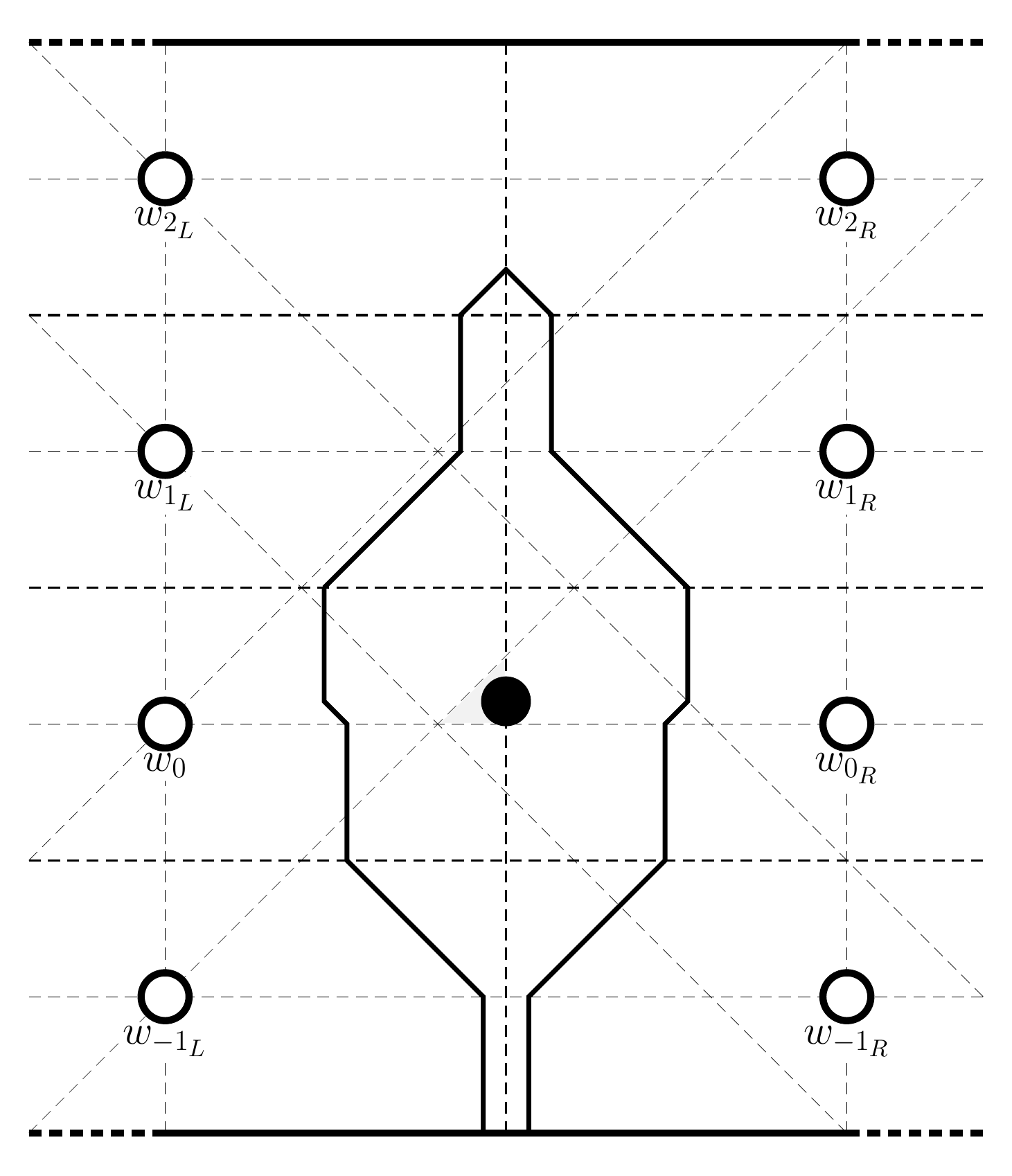}
  \captionsetup{justification=centering}
  \caption{$b_1=(\frac{p}{2a},\frac{p}{6a} - \frac{(b-2)q}{6b})$ only if\\ $\frac{(b-2)q}{2b} \leq \frac{p}{2a} \leq \frac{(2b-1)q}{4b}$ and $b=2l$.}
  \label{fig:GridOptimalEvenRes}
\end{subfigure}
\caption{Maximal area Voronoi cells $V^+(b_1)$ for $b_1$ within Section $2l$ upon $x^*=\frac{p}{2a}$.}
\end{figure}

%%%

\paragraph{Section $2l+1$ upon $x^* = \frac{p}{2a}$}
We now consider the placement of $b_1$ on $x^*=\frac{p}{2a}$ inside the odd Section $2l+1$ (where $\frac{(2l-1)q}{2b} \leq \frac{p}{2a} \leq \frac{(2l+1)q}{2b} \Rightarrow l = \ceil{\frac{pb-qa}{2qa}}$) where $V^+(b_1)$ will extend into Voronoi cells $V^\circ(w_{{i}_L}) \cup V^\circ(w_{{i}_R})$ for $i \in \{ -l,\ldots, l+1\}$ (if existing).

Now, as before, we will explore the effect caused by the non-existence of these faraway Voronoi cells (i.e. what area the boundary of $\mathcal{P}$ cuts off). We can use identical processes to those described in Section $2l$ in order to find the best point in $W$ to assign to be $w_0$. If $b>2l+1$ then we can choose $w_0$ such that all Voronoi cells $V^\circ(w_{{i}_L}) \cup V^\circ(w_{{i}_R})$ for $i \in \{ -l,\ldots, l+1\}$ exist and $V^+(b_1)$ does not interact with $\mathcal{P}$ (again we will not have a unique best point of reference for $w_0$ if $b \neq 2l+2$).

\pagebreak If both $w_{{l+1}_L}$ and $w_{{-l-1}_L}$ or $w_{{-l}_L}$ and $w_{{l+2}_L}$ do not exist while $w_{i_L}$ does for $i \in \{-l, \ldots, l\}$ or for $i \in \{-l+1, \ldots, l+1\}$ respectively then it is the case that $b = 2l+1$ and $V^+(b_1)$ touches only one bounding edge of $\mathcal{P}$. In this scenario we must again decide which area we would prefer: that stolen from $V^\circ(w_{{l+1}_L}) \cup V^\circ(w_{{l+1}_R})$ or from $V^\circ(w_{{-l}_L}) \cup V^\circ(w_{{-l}_R})$. Since $b_1$ is being located in the first quadrant of $V^\circ(w_0)$, lying closer to $w_{{-l}_L}$ than $w_{{l+1}_L}$, it can steal a larger area from $V^\circ(w_{{-l}_L}) \cup V^\circ(w_{{-l}_R})$ than it could from $V^\circ(w_{{l+1}_L}) \cup V^\circ(w_{{l+1}_R})$, so it is favourable for $w_0$ to be chosen to be on the $l+1$th row of points in $W$ (counting from the bottom of the grid) so that $V^+(b_1)$ consists of areas stolen from $V^\circ(w_{{i}_L}) \cup V^\circ(w_{{i}_R})$ for all $i \in \{ -l,\ldots, l\}$ (as opposed to for all $i \in \{ -l+1,\ldots, l+1\}$).

Using an identical argument to that for even sections, for areas $V^+(b_1)$ that touch both horizontal bounding edges of $\mathcal{P}$ (so both $w_{{l+1}_L}$ and $w_{{-l}_L}$ do not exist, meaning that $b < 2l+1$) it is still optimal to choose $w_0$ to be on the $\ceil{\frac{b}{2}}$th row of points in $W$ in order to steal from $V^\circ(w_{{i}_L}) \cup V^\circ(w_{{i}_R})$ for all $i \in \{-\ceil{\frac{b}{2}}+1,\ldots,\floor{\frac{b-1}{2}}\}$. As justified in our analysis of Sections $2l$ it is only within this final Section $b+1$ that both horizontal edges of $\mathcal{P}$ are touched, and we shall explore this section after finishing a full investigation of Sections $2l+1$ for $b \geq 2l+1$.

Now that we have chosen the optimal $w_0$ and recorded which Voronoi cells $V^\circ(w_{{i}_L}) \cup V^\circ(w_{{i}_R})$ will be entered, we can calculate the areas of the Voronoi cell $V^+(b_1)$ for different values of $b$ and optimise the location of $b_1$ upon $x^*=\frac{p}{2a}$ within Section $2l+1$. We calculate these areas by taking the area found in Section $2l$ and adapting it for Section $2l+1$ (noting that a move from Section $2l$ to $2l+1$ means that $V^+(b_1)$ enters $V^\circ(w_{{l+1}_L}) \cup V^\circ(w_{{l+1}_R})$ for the first time). If $b > 2l+1 \left(= 2\ceil{\frac{pb-qa}{2qa}}+1\right)$ then
\begin{equation*}
\begin{split}
    Area(V^+(b_1))&= - \frac{y^2}{2} + \frac{lpq}{2ab} + \frac{p^2}{8a^2} - \frac{(l-1)lq^2}{2b^2} \\ &\qquad- ``Area(V^+(b_1) \cap (V^\circ(w_{{l}_L}) \cup V^\circ(w_{{l}_R})))" \\ &\qquad+ Area(V^+(b_1) \cap (V^\circ(w_{{l}_L}) \cup V^\circ(w_{{l}_R}))) \\ &\qquad+ Area(V^+(b_1) \cap (V^\circ(w_{{l+1}_L}) \cup V^\circ(w_{{l+1}_R}))) \qquad \qquad \\
\end{split}
\end{equation*}
\begin{equation*}
\begin{split}
    \qquad \qquad&= - \frac{y^2}{2} + \frac{lpq}{2ab} + \frac{p^2}{8a^2} - \frac{(l-1)lq^2}{2b^2} \\ &\qquad- \left(\frac{y^2}{4} + (\frac{p}{4a} - \frac{(l-1)q}{2b})y - \frac{(l-1)pq}{4ab} + \frac{p^2}{16a^2} + \frac{(l-1)^2q^2}{4b^2}\right) \\ 
    &\qquad+ \frac{q}{b}y + \frac{pq}{2ab} - \frac{(4l-1)q^2}{4b^2} \\
    &\qquad+ \frac{y^2}{4} + \left(\frac{p}{4a} - \frac{((l+1)-1)q}{2b}\right)y - \frac{((l+1)-1)pq}{4ab} + \frac{p^2}{16a^2} + \frac{((l+1)-1)^2q^2}{4b^2} \\
    &= - \frac{y^2}{2} + \frac{q}{2b}y + \frac{(l+1)pq}{4ab}  + \frac{p^2}{8a^2} - \frac{l^2q^2}{2b^2} \\
\end{split}
\end{equation*}
and if $b=2l+1$ then, adapting this formula,
\begin{equation*}
\begin{split}
    Area(V^+(b_1))&= - \frac{y^2}{2} + \frac{q}{2b}y + \frac{(l+1)pq}{4ab}  + \frac{p^2}{8a^2} - \frac{l^2q^2}{2b^2} \\
    &\qquad- ``Area(V^+(b_1) \cap (V^\circ(w_{{l+1}_L}) \cup V^\circ(w_{{l+1}_R})))" \\
    &= - \frac{y^2}{2} + \frac{q}{2b}y + \frac{(l+1)pq}{4ab}  + \frac{p^2}{8a^2} - \frac{l^2q^2}{2b^2} \\
    &\qquad- \left(\frac{y^2}{4} + \left(\frac{p}{4a} - \frac{lq}{2b}\right)y - \frac{lpq}{4ab} + \frac{p^2}{16a^2} + \frac{l^2q^2}{4b^2}\right) \\
    &= - \frac{3y^2}{4} - \left(\frac{p}{4a} - \frac{(l+1)q}{2b}\right)y + \frac{(2l+1)pq}{4ab}  + \frac{p^2}{16a^2} - \frac{3l^2q^2}{4b^2} \, .
\end{split}
\end{equation*}

Clearly $(\frac{p}{2a},\frac{q}{2b})$ is the optimum if $b > 2l+1$ giving $Area(V^+((\frac{p}{2a},\frac{q}{2b})))=\frac{(4l-1)pq}{4ab} + \frac{p^2}{8a^2} - \frac{(4l^2-1)q^2}{8b^2}$ as depicted in Figure~\ref{fig:GridOptimalOdd}.

For $b=2l+1$ we have derivative
$$\frac{\delta A}{\delta y} = - \frac{3y}{2} - \frac{p}{4a} + \frac{(l+1)q}{2b}$$
which gives our optimum to be at $y^* = - \frac{p}{6a} + \frac{(l+1)q}{3b} = - \frac{p}{6a} + \frac{(b+1)q}{6b}$. However, in order for $b_1=(\frac{p}{2a}, \frac{(b+1)q}{6b} - \frac{p}{6a})$ to lie within Section $b+1$ it must be the case that, if $\frac{(b-2)q}{2b} \leq \frac{p}{2a} \leq \frac{(b-1)q}{2b}$, $b_1$ lies above $\mathcal{CL}^8(w_{{\frac{b-1}{2}}_L})$ and, if $\frac{(b-1)q}{2b} \leq \frac{p}{2a} \leq \frac{q}{2}$, $b_1$ lies above $\mathcal{CL}^2(w_{{-\frac{b-1}{2}}_L})$.

Therefore, if $\frac{(b-2)q}{2b} \leq \frac{p}{2a} \leq \frac{(b-1)q}{2b}$ then it must be the case that $\frac{(b-1)q}{2b} - x \leq y$ so we require
$$\frac{(b-1)q}{2b} - \frac{p}{2a} \leq \frac{(b+1)q}{6b} - \frac{p}{6a} \Leftrightarrow \frac{(b-2)q}{3b} \leq \frac{p}{3a} \Leftrightarrow \frac{(b-2)q}{2b} \leq \frac{p}{2a} \, .$$
Hence $b_1$ is the optimum in Section $b$ for all values $\frac{(b-2)q}{2b} \leq \frac{p}{2a} \leq \frac{(b-1)q}{2b}$.

Otherwise, if $\frac{(b-1)q}{2b} \leq \frac{p}{2a} \leq \frac{q}{2}$ then it must be the case that $x - \frac{(b-1)q}{2b} \leq y$ so we require
$$\frac{p}{2a} - \frac{(b-1)q}{2b} \leq \frac{(b+1)q}{6b} - \frac{p}{6a} \Leftrightarrow \frac{2p}{3a} \leq \frac{(2b-1)q}{3b} \Leftrightarrow \frac{p}{2a} \leq \frac{(2b-1)q}{4b} \, .$$
Therefore, if $\frac{(b-1)q}{2b} \leq \frac{p}{2a} \leq \frac{(2b-1)q}{4b}$ then $b_1$ is the optimum in Section $b$. Otherwise, if $\frac{(2b-1)q}{4b}\leq \frac{p}{2a} \leq \frac{q}{2}$ then $b_1$ will lie below Section $b$. If this is the case then, following identical working as that for even $b$, the optimum over Section $b$ must lie on the boundary between Section $b$ and $b+1$ whereupon it will be dominated by the identical point within Section $b+1$. Therefore Black's best point will not lie in Section $b$ if $\frac{(2b-1)q}{4b}\leq \frac{p}{2a}$.

To summarise, if $b$ is odd then: if $\frac{(b-2)q}{2b} \leq \frac{p}{2a} \leq \frac{(2b-1)q}{4b}$ then the optimum in Section $b$ is $(\frac{p}{2a}, \frac{(b+1)q}{6b} - \frac{p}{6a})$ giving $Area(V^+((\frac{p}{2a}, \frac{(b+1)q}{6b} - \frac{p}{6a}))) = \frac{(5b-1)pq}{24ab} + \frac{p^2}{12a^2} - \frac{(b-2)(2b-1)q^2}{12b^2}$ (depicted in Figure~\ref{fig:GridOptimalOddRes}), otherwise if $\frac{(2b-1)q}{4b}\leq \frac{p}{2a} \leq \frac{q}{2}$ then the optimum lies on the boundary with Section $b+1$ and is not Black's best point (and so is not drawn).

\begin{figure}[!ht]\ContinuedFloat
\begin{subfigure}{.53\textwidth}
  \centering
  \includegraphics[width=0.9\textwidth]{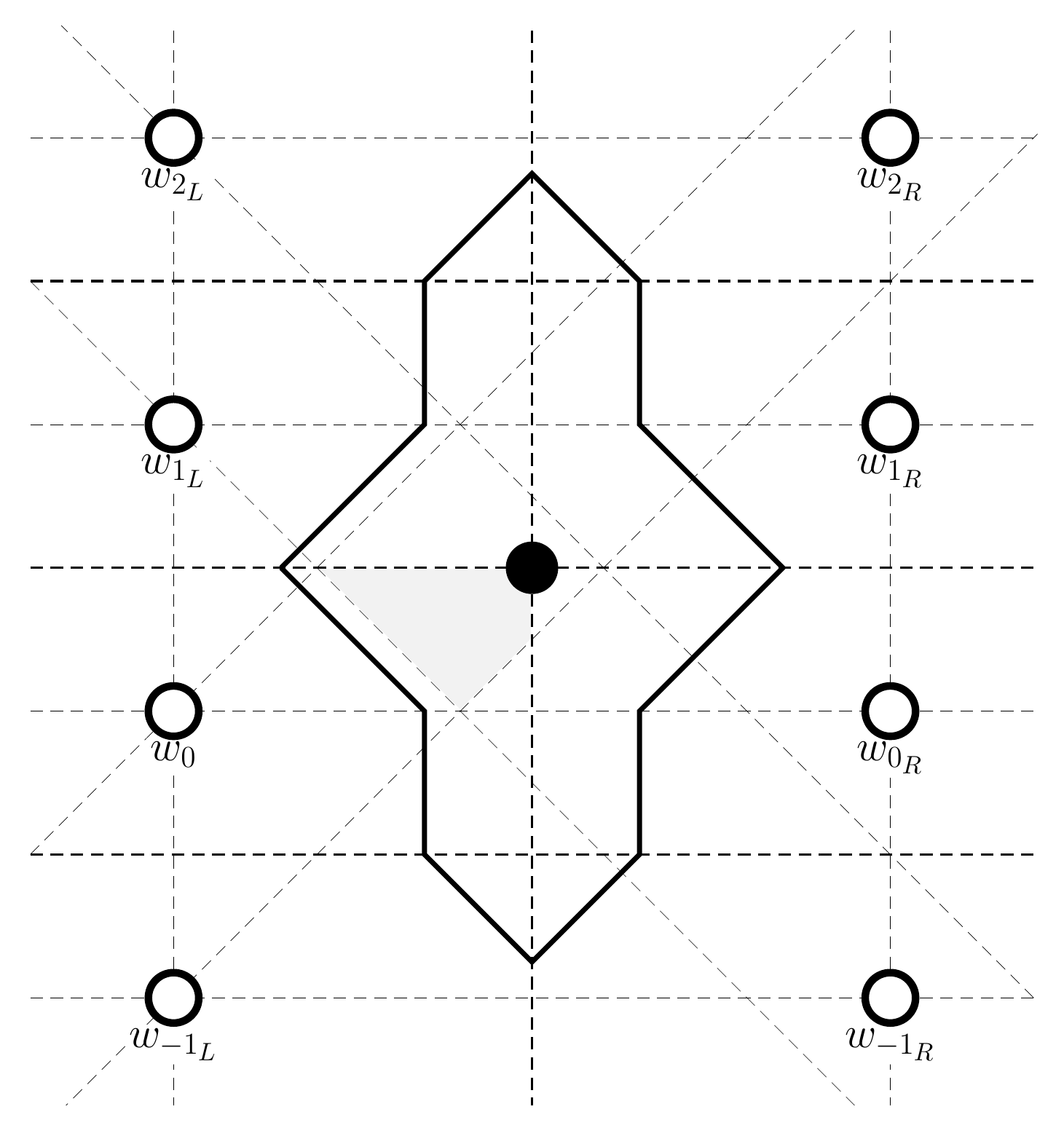}
  \caption{$b_1=(\frac{p}{2a},\frac{q}{2b})$.\\ \,}
  \label{fig:GridOptimalOdd}
\end{subfigure}%
\begin{subfigure}{.47\textwidth}
  \centering
  \includegraphics[width=0.9\textwidth]{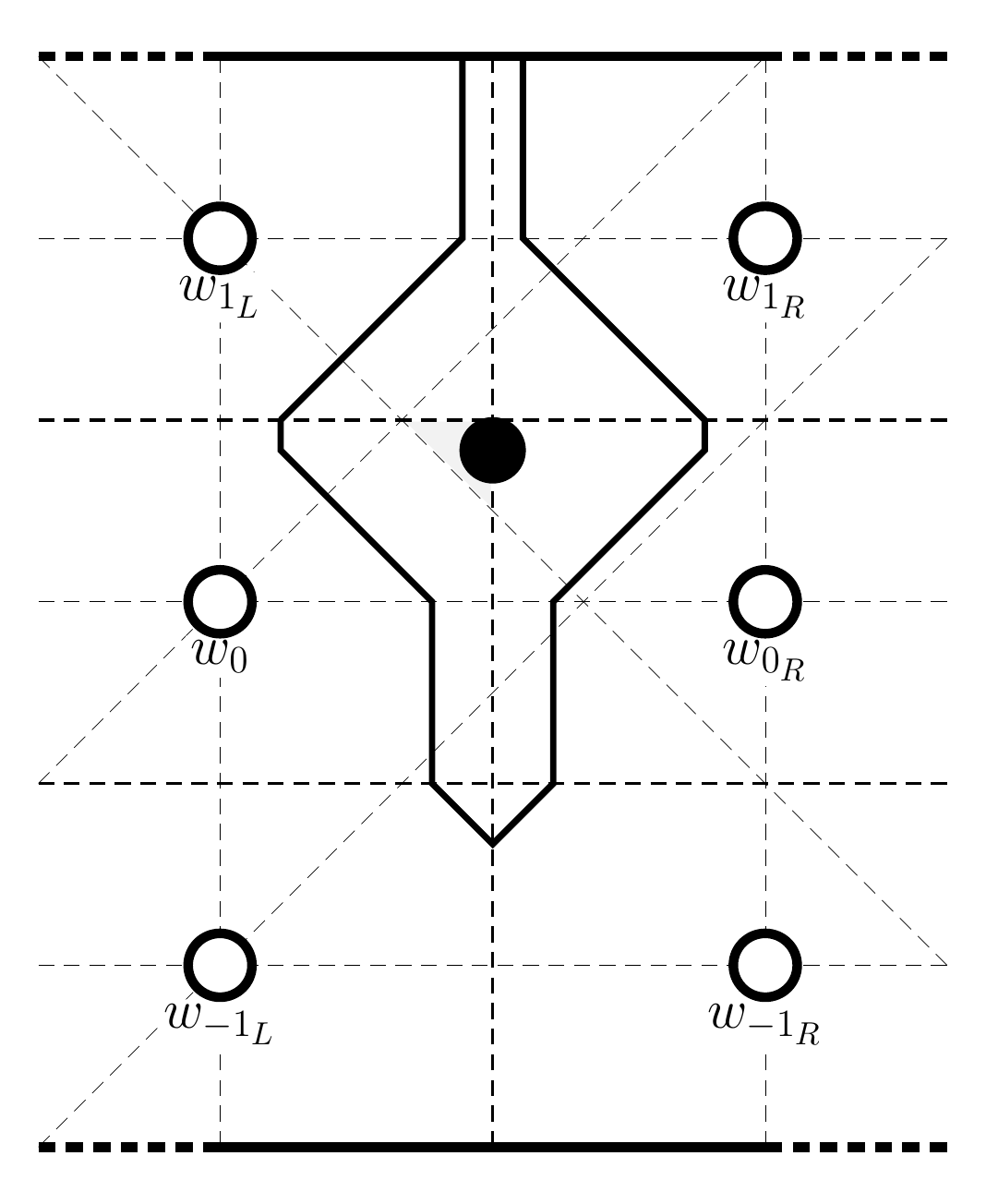}
  \captionsetup{justification=centering}
  \caption{$b_1 = (\frac{p}{2a}, \frac{(b+1)q}{6b} - \frac{p}{6a})$ only if \, \, \, \,\, \\ $\frac{(b-2)q}{2b} \leq \frac{p}{2a} \leq \frac{(2b-1)q}{4b}$ and $b=2l+1$.}
  \label{fig:GridOptimalOddRes}
\end{subfigure}
\caption{Maximal area Voronoi cells $V^+(b_1)$ for $b_1$ within Section $2l+1$ upon $x^*=\frac{p}{2a}$.}
\end{figure}

\pagebreak \paragraph{Section $b+1$ upon $x^* = \frac{p}{2a}$}

Finally we explore Section $b+1$, the last possible section, where $w_0$ is chosen to be on the $\ceil{\frac{b}{2}}$th row. If $b_1$ is placed within this section then $V^+(b_1)$ touches both horizontal boundaries of $\mathcal{P}$. This simply has the areas, if $b$ is even,

\begin{equation*}
\begin{split}
    Area(V^+(b_1))&= \sum_{i=-\frac{b-2}{2}}^{-1}{Area(V^+(b_1) \cap (V^\circ(w_{{i}_L}) \cup V^\circ(w_{{i}_R})))} \\ &\qquad+ Area(V^+(b_1) \cap (V^\circ(w_{{0}_L}) \cup V^\circ(w_{{0}_R}))) \\ &\qquad+ \sum_{i=1}^{\frac{b}{2}}{Area(V^+(b_1) \cap (V^\circ(w_{{i}_L}) \cup V^\circ(w_{{i}_R})))} \\
    &= \sum_{i=-\frac{b-2}{2}}^{-1}{\left( -\frac{q}{b}y + \frac{pq}{2ab} + \frac{(4i+1)q^2}{4b^2} \right)} + \frac{pq}{2ab} - y^2 \\ &\qquad+ \sum_{i=1}^{\frac{b}{2}}{\left(\frac{q}{b}y + \frac{pq}{2ab} - \frac{(4i-1)q^2}{4b^2} \right)} \\
    &= - y^2 + \frac{q}{b}y + \frac{pq}{2a} - \frac{(b^2 - b + 1)q^2}{4b^2} \\
\end{split}
\end{equation*}
and, if $b$ is odd,
\begin{equation*}
\begin{split}
    Area(V^+(b_1))&= \sum_{i=-\frac{b-1}{2}}^{-1}{Area(V^+(b_1) \cap (V^\circ(w_{{i}_L}) \cup V^\circ(w_{{i}_R})))} \\ &\qquad+ Area(V^+(b_1) \cap (V^\circ(w_{{0}_L}) \cup V^\circ(w_{{0}_R}))) \\ &\qquad+ \sum_{i=1}^{\frac{b-1}{2}}{Area(V^+(b_1) \cap (V^\circ(w_{{i}_L}) \cup V^\circ(w_{{i}_R})))} \\
\end{split}
\end{equation*}
\begin{equation*}
\begin{split}
    \qquad \qquad&= \sum_{i=-\frac{b-1}{2}}^{-1}{\left( -\frac{q}{b}y + \frac{pq}{2ab} + \frac{(4i+1)q^2}{4b^2} \right)} + \frac{pq}{2ab} - y^2 \\ &\qquad+ \sum_{i=1}^{\frac{b-1}{2}}{\left(\frac{q}{b}y + \frac{pq}{2ab} - \frac{(4i-1)q^2}{4b^2} \right)} \\
    &= - y^2 + \frac{pq}{2a} - \frac{(b-1)q^2}{4b} \, . \\
\end{split}
\end{equation*}

It is clear that $(\frac{p}{2a},\frac{q}{2b})$ and $(\frac{p}{2a},0)$ are the optima for $b$ is even and $b$ is odd, and these are both in Section $b+1$ for $b$ even and odd respectively. We are certainly pleased to see this result since, as we might expect, both of these points lie on the horizontal line of symmetry of $\mathcal{P}$ and can be considered to be the centre of $\mathcal{P}$ which we would presume to be an effective placement. This gives us areas $Area(V^+((\frac{p}{2a},\frac{q}{2b})))=\frac{pq}{2a} - \frac{(b-1)q^2}{4b}$ and $Area(V^+((\frac{p}{2a},0)))=\frac{pq}{2a} - \frac{(b-1)q^2}{4b}$ (interestingly, identical to each other) as depicted in Figures~\ref{fig:GridOptimalbeven} and \ref{fig:GridOptimalbodd} respectively.

\begin{figure}[!ht]\ContinuedFloat
\begin{subfigure}{.5\textwidth}
  \centering
  \includegraphics[width=0.9\textwidth]{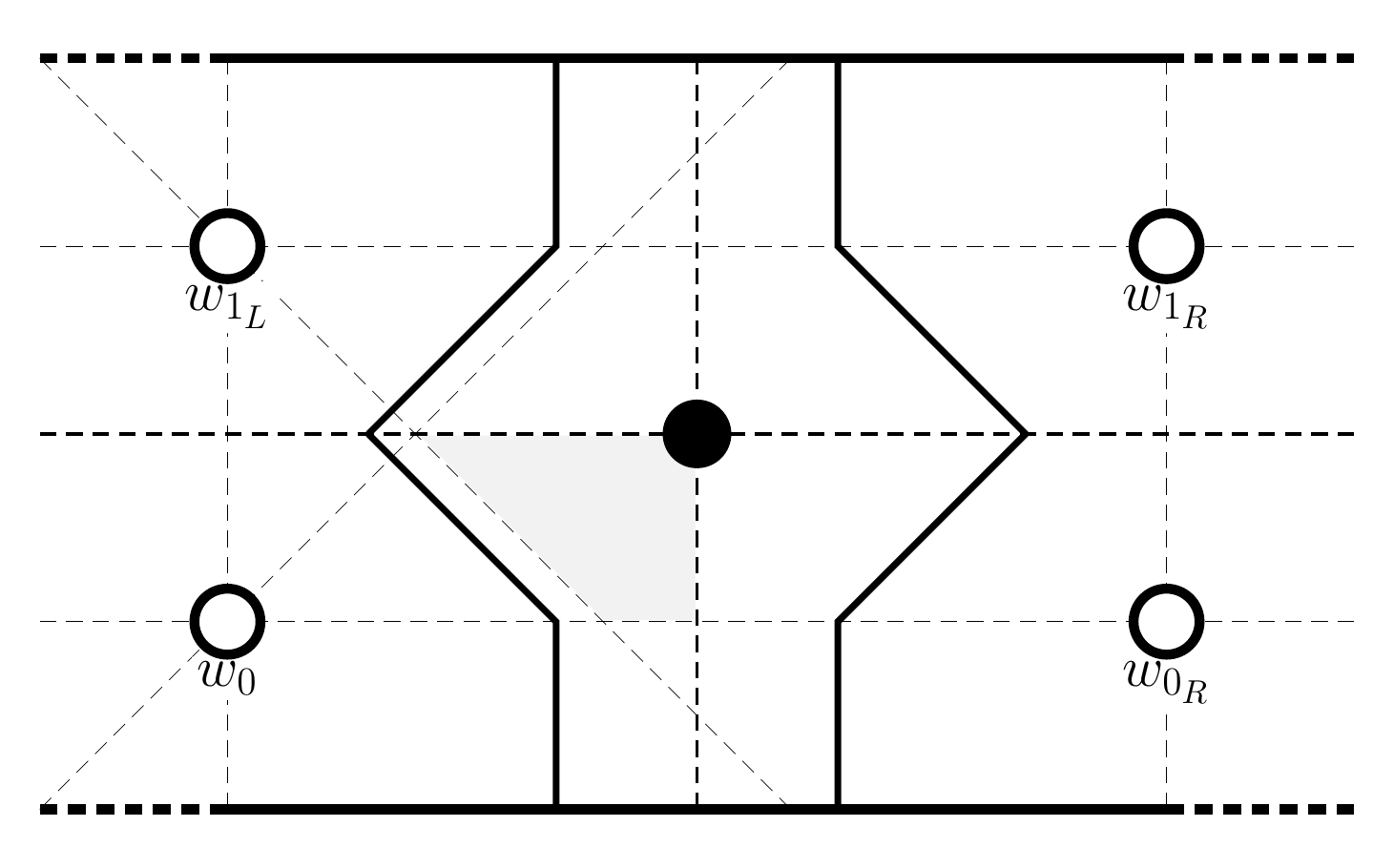}
  \caption{$b_1=(\frac{p}{2a},\frac{q}{2b})$.}
  \label{fig:GridOptimalbeven}
\end{subfigure}%
\begin{subfigure}{.5\textwidth}
  \centering
  \includegraphics[width=0.9\textwidth]{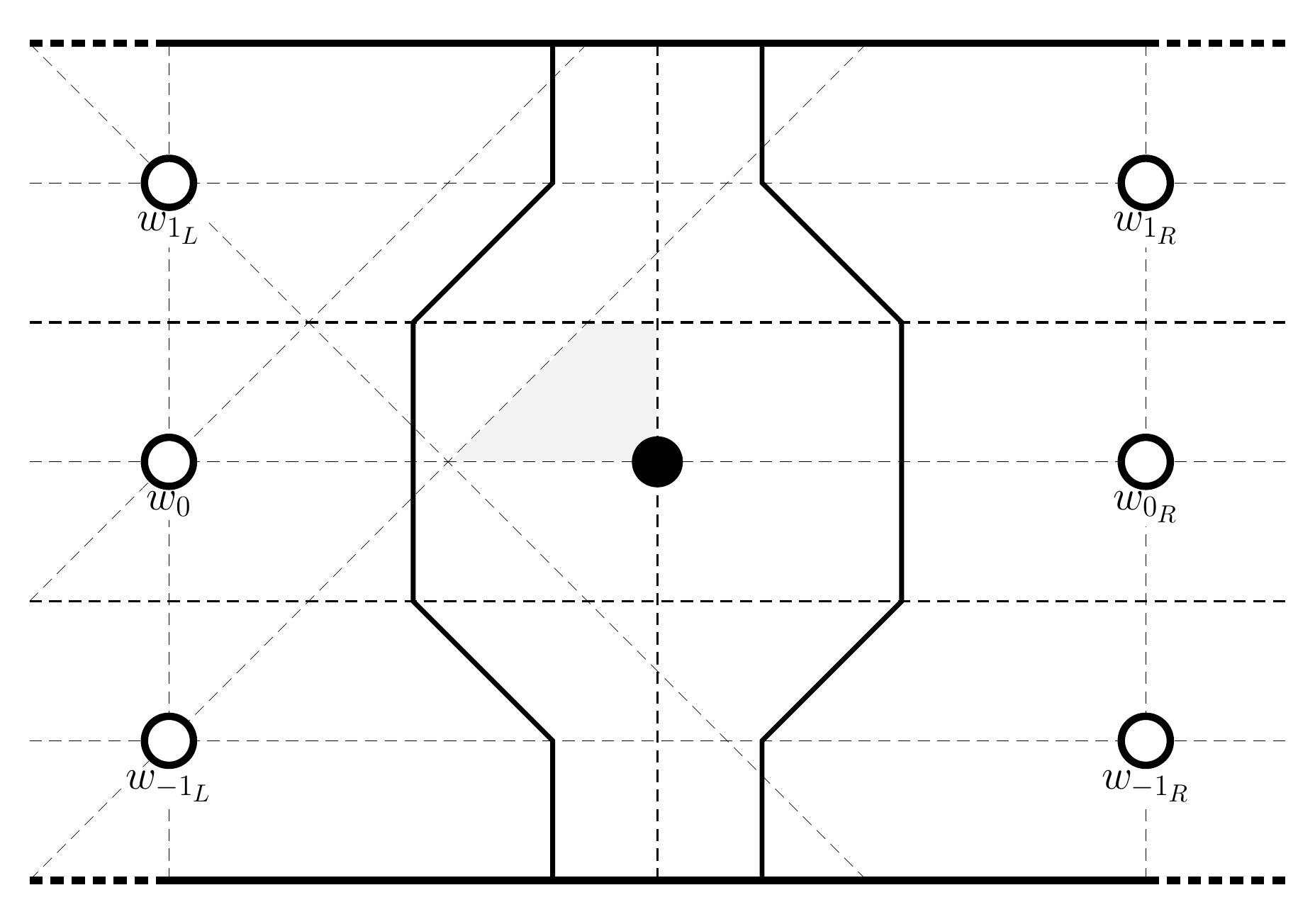}
  \caption{$b_1 = (\frac{p}{2a},0)$.}
  \label{fig:GridOptimalbodd}
\end{subfigure}
\caption{Maximal area Voronoi cells $V^+(b_1)$ for $b_1$ within Section $b+1$ upon $x^*=\frac{p}{2a}$.}
\label{fig:GridOptimals}
\end{figure}

%Finally, if $b<2l$ and $b$ is even then we have $y^*=\frac{q}{2b}$. If $(\frac{p}{2a},\frac{q}{2b})$ does not lie in Section $2l$ then, as explained above, the local optimum will be dominated by the identical point within Section $2l+1$ and  $(\frac{p}{2a},\frac{q}{2b})$ will only be contained in Section $2l$ if $\frac{p}{2a} = \frac{(2l-1)q}{2b}$. However, even this does not give the point a fighting to be the best point as this vertex of Section $2l$ also lies on the configuration line delineating Section $2l$ from Section $2l+1$, so this will never be the best point $b^*$.
\bigskip

And thus we have found every optimal location within every possible partition of $\mathcal{P}$ that is a candidate for Black's best point $b^*$. To recap, Figure~\ref{fig:GridOptimals} shows all of the potential candidates for $b^*$ within each appropriate section. Following our discussion of the choice of $w_0$ (and the fact that the best choice of $w_0$ for $b_1$ in Section $I$ only requires that the first quadrant of $V^+(b_1)$ is a core quadrant and, if possible, that $w_{{-1}_L}$ exists) we can say with confidence that, without loss of generality, the best point $b^*$ lies in the first quadrant of the $(\ceil{\frac{a}{2}},\ceil{\frac{b}{2}})$th point in $W$ (where the $(i,j)$th point in $W$ is the point $w \in W$ which is in the $i$th column (counting from the left) and $j$th row (counting from the bottom)). This quadrant is the unique (or one of two or four identical) most central quadrant in $\mathcal{P}$, thus furthest from the boundaries of $\mathcal{P}$. Hence the best point $b^*$ will lie in this quadrant and in Section $I$ or on the line $x^*=\frac{p}{2a}$ and we must determine which optimum within which of these areas gives the best point depending on the relationship between $p$, $q$, $a$, and $b$.

Fortunately the nature of our investigation into $x^*=\frac{p}{2a}$ allows us to fairly easily compare optima upon this line where we have relatively restrictive conditions on which sections contain $x^*=\frac{p}{2a}$. For $\frac{(2l-1)q}{2b} \leq \frac{p}{2a} \leq \frac{lq}{b}$ and assuming that $2l+1<b$ (we will assess $2l+1 \geq b$ later), $x^*=\frac{p}{2a}$ enters Sections $2l$ and $2l+1$ so we shall compare the optima within these sections for this condition. In Section $2l$, $Area(V^+((\frac{p}{2a},0)))= \frac{lpq}{2ab} + \frac{p^2}{8a^2} - \frac{(l-1)lq^2}{2b^2}$ and in Section $2l+1$, $Area(V^+((\frac{p}{2a},\frac{q}{2b})))=\frac{(4l-1)pq}{4ab} + \frac{p^2}{8a^2} - \frac{(4l^2-1)q^2}{8b^2}$. The optimum in Section $2l+1$ is better than that in Section $2l$ if

\begin{equation*}
\begin{split}
    \left(\frac{(4l-1)pq}{4ab} + \frac{p^2}{8a^2} \right. &- \left. \frac{(4l^2-1)q^2}{8b^2}\right) - \left(\frac{lpq}{2ab} + \frac{p^2}{8a^2} - \frac{(l-1)lq^2}{2b^2}\right) \\
    &= \frac{(4l-1)pq}{4ab} + \frac{p^2}{8a^2} - \frac{(4l^2-1)q^2}{8b^2} - \frac{lpq}{2ab} - \frac{p^2}{8a^2} + \frac{(l-1)lq^2}{2b^2} \\
    &= \frac{(2l-1)pq}{4ab} - \frac{(4l-1)q^2}{8b^2} \geq 0 \\
    &\Leftrightarrow \frac{p}{2a} \geq \frac{(4l-1)q}{4(2l-1)b} \, .\\
\end{split}
\end{equation*}
Now $$\frac{(2l-1)q}{2b} \leq \frac{(4l-1)q}{4(2l-1)b} \, \Leftrightarrow \, 8l^2 - 12l + 3 \leq 0 \, \Leftrightarrow \, \frac{3 - \sqrt{3}}{4} \leq l \leq \frac{3 + \sqrt{3}}{4}$$ so if $l>1$ then the optimum in Section $2l+1$ is always better than the optimum in Section $2l$. Otherwise, if $l=1$ then Section $2l$ ($II$) is better than Section $2l+1$ ($III$) for $\left( \frac{q}{2b} =\right) \frac{(2l-1)q}{2b} \leq \frac{p}{2a} \leq \frac{(4l-1)q}{4(2l-1)b} \left(= \frac{3q}{4b}\right)$. %\qquad\qquad (2l-1)^2 \leq \frac{3}{2} \Leftrightarrow \frac{2-\sqrt{6}}{4} \leq l \leq \frac{2+\sqrt{6}}{4}

For $\frac{lq}{b} \leq \frac{p}{2a} \leq \frac{(2l+1)q}{2b}$ and assuming that $2(l+1)<b$ (we will assess $2(l+1) \geq b$ soon), $x^*=\frac{p}{2a}$ enters Sections $2l+1$ and $2(l+1)$ (note that $x^*=\frac{p}{2a}$ will never enter Section $I$ because $\frac{p}{2a} \geq \frac{q}{2b}$). In Section $2l+1$, $Area(V^+((\frac{p}{2a},\frac{q}{2b})))=\frac{(4l-1)pq}{4ab} + \frac{p^2}{8a^2} - \frac{(4l^2-1)q^2}{8b^2}$ and in Section $2(l+1)$, $Area(V^+((\frac{p}{2a},0)))= \frac{(l+1)pq}{2ab} + \frac{p^2}{8a^2} - \frac{(l+1)lq^2}{2b^2}$ so the optimum in Section $2(l+1)$ is better than that in Section $2l+1$ if
\begin{equation*}
\begin{split}
    \left(\frac{(l+1)pq}{2ab} + \frac{p^2}{8a^2} \right. &- \left. \frac{(l+1)lq^2}{2b^2}\right) - \left(\frac{(4l-1)pq}{4ab} + \frac{p^2}{8a^2} - \frac{(4l^2-1)q^2}{8b^2}\right) \\
%    &= \frac{(l+1)pq}{2ab} + \frac{p^2}{8a^2} - \frac{(l+1)lq^2}{2b^2} - \frac{(4l-1)pq}{4ab} - \frac{p^2}{8a^2} + \frac{(4l^2-1)q^2}{8b^2} \\
    &= - \frac{(2l-3)pq}{4ab} - \frac{(4l+1)q^2}{8b^2} \geq 0 \\
    &\Leftrightarrow \frac{(3-2l)p}{2a} \geq \frac{(4l+1)q}{4b} \, . \\
%    &\Leftrightarrow  \frac{p}{2a} \geq \frac{(4l+1)q}{4(3-2l)b} \, .\\
\end{split}
\end{equation*}
Now if $(3-2l)\leq 0$ (i.e. $l \geq \frac{3}{2}$) then $0 \geq \frac{(3-2l)p}{2a} \geq \frac{(4l+1)q}{4b}$ so the optimum in Section $2(l+1)$ is never better than the optimum in Section $2l+1$. Otherwise if $l=1$ then $\frac{p}{2a} \geq \frac{(4l+1)q}{4(3-2l)b} = \frac{5q}{4b} > \frac{q}{b} = \frac{lq}{b}$, so the optimum in Section $2(l+1)$ ($IV$) is better than the optimum in Section $2l+1$ ($III$) for $\frac{5q}{4b} \leq \frac{p}{2a} \leq \frac{3q}{2b} = \frac{(2l+1)q}{2b}$.

Upon $x^*=\frac{p}{2a}$ we have the possibility of two special cases with regard to the area of $V^+(b_1)$ to which we must give careful consideration: Section $b$ (within which $b_1$ produces a Voronoi cell touching exactly one horizontal boundary of $\mathcal{P}$) and Section $b+1$ (within which $b_1$ produces a Voronoi cell touching both horizontal boundaries). Therefore we must compare the areas of Section $b-1$ with Section $b$ as well as the areas of Section $b$ with Section $b+1$.

Firstly, suppose $b$ is even. For $\frac{(b-2)q}{2b} \leq \frac{p}{2a} \leq \frac{(b-1)q}{2b}$ and assuming $b>2$ (since $\frac{p}{2a} \geq \frac{q}{2b}$), $x^*=\frac{p}{2a}$ enters Sections $b-1$ and $b$ whose maximal areas are $Area(V^+((\frac{p}{2a},\frac{q}{2b})))=\frac{(2b-5)pq}{4ab} + \frac{p^2}{8a^2} - \frac{(b^2-4b+3)q^2}{8b^2}$ and $Area(V^+((\frac{p}{2a},\frac{p}{6a} - \frac{(b-2)q}{6b}))) = \frac{(2b-1)pq}{6ab} + \frac{p^2}{12a^2} - \frac{(b-2)(4b-1)q^2}{12b^2}$ respectively. The optimum in Section $b$ is better than the optimum in Section $b-1$ if
\begin{equation*}
\begin{split}
    \left(\frac{(2b-1)pq}{6ab} \right. &+ \left. \frac{p^2}{12a^2} - \frac{(b-2)(4b-1)q^2}{12b^2}\right) - \left(\frac{(2b-5)pq}{4ab} + \frac{p^2}{8a^2} - \frac{(b^2-4b+3)q^2}{8b^2}\right) \\
%    &= \frac{(2b-1)pq}{6ab} + \frac{p^2}{12a^2} - \frac{(b-2)(4b-1)q^2}{12b^2} - \frac{(2b-5)pq}{4ab} - \frac{p^2}{8a^2} + \frac{(b^2-4b+3)q^2}{8b^2} \\
    &= - \frac{(2b-13)pq}{12ab} - \frac{p^2}{24a^2} - \frac{(5b^2-6b-5)q^2}{24b^2} \geq 0 \\
    &\Leftrightarrow \frac{p^2}{a^2} + \frac{2(2b-13)pq}{ab} + \frac{(5b^2-6b-5)q^2}{b^2} \leq 0 \\
%    &\Leftrightarrow \left(\frac{p}{a} + \frac{(2b-13)q}{b}\right)^2 - \frac{(2b-13)^2q^2}{b^2} + \frac{(5b^2-6b-5)q^2}{b^2} \leq 0 \\
\end{split}
\end{equation*}
\begin{equation*}
\begin{split}
    \qquad \qquad&\Leftrightarrow \left(\frac{p}{a} + \frac{(2b-13)q}{b}\right)^2 + \frac{(b^2+46b-174)q^2}{b^2} \leq 0 \\
    &\Leftrightarrow \frac{(13-2b-\sqrt{-(b^2+46b-174)})q}{2b} \leq \frac{p}{2a} \leq \frac{(13-2b+\sqrt{-(b^2+46b-174)})q}{2b} \, . \\
%    &\Leftrightarrow  \frac{p}{2a} \geq \frac{(4l+1)q}{4(3-2l)b} \, .\\
\end{split}
\end{equation*}
Now this condition only holds if $b^2+46b-174 \leq 0 \Leftrightarrow -23-\sqrt{703} \leq b \leq -23 + \sqrt{703} \approx 3.5$ so $b$ must be either $2$ or $3$, and since $b$ is even it must be the case that $b=2$ which contradicts our assumption. So the optimum in Section $b-1$ is better than the optimum in Section $b$ for all $b \neq 2$. This result is actually as expected: the optimum in Section $2(l+1)$ is dominated by the optimum in Section $2l+1$ and the structure in Section $b$ has a lesser area than the structure in a general Section $2(l+1)$ would have for a value of $l=\frac{b-2}{2}$; and this paragraph simply serves as a sanity check.

For $\frac{(b-1)q}{2b} \leq \frac{p}{2a} \leq \frac{q}{2}$, $x^*=\frac{p}{2a}$ enters Sections $b$ and $b+1$ whose maximal areas are $Area(V^+((\frac{p}{2a},\frac{p}{6a} - \frac{(b-2)q}{6b}))) = \frac{(2b-1)pq}{6ab} + \frac{p^2}{12a^2} - \frac{(b-2)(4b-1)q^2}{12b^2}$ (only for $\frac{(b-2)q}{2b} \leq \frac{p}{2a} \leq \frac{(2b-1)q}{4b}$ so we now only consider the interval $\frac{(b-1)q}{2b} \leq \frac{p}{2a} \leq \frac{(2b-1)q}{4b}$) and $Area(V^+((\frac{p}{2a},\frac{q}{2b})))=\frac{pq}{2a} - \frac{(b-1)q^2}{4b}$ respectively. The optimum in Section $b+1$ is better than the optimum in Section $b$ if
\begin{equation*}
\begin{split}
    \left(\frac{pq}{2a} \right. &- \left. \frac{(b-1)q^2}{4b}\right) - \left(\frac{(2b-1)pq}{6ab} + \frac{p^2}{12a^2} - \frac{(b-2)(4b-1)q^2}{12b^2}\right) \\
%    &=\frac{pq}{2a} - \frac{(b-1)q^2}{4b} - \frac{(2b-1)pq}{6ab} - \frac{p^2}{12a^2} + \frac{(b-2)(4b-1)q^2}{12b^2} \\
    &=\frac{(b+1)pq}{6ab} - \frac{p^2}{12a^2} + \frac{(b^2-6b+2)q^2}{12b^2} \geq 0 \\
    &\Leftrightarrow \frac{p^2}{a^2} - \frac{2(b+1)pq}{ab} - \frac{(b^2-6b+2)q^2}{b^2} \leq 0 \\
    &\Leftrightarrow \left(\frac{p}{a} - \frac{(b+1)q}{b}\right)^2 - \frac{(2b^2-4b+3)q^2}{b^2} \leq 0 \\
    &\Leftrightarrow \frac{(b+1-\sqrt{2b^2-4b+3})q}{2b} \leq \frac{p}{2a} \leq \frac{(b+1+\sqrt{2b^2-4b+3})q}{2b} \, . \\
\end{split}
\end{equation*}
Now $2b^2-4b+3 = 2(b-1)^2 + 1 > 0$ so the condition holds for all $b$. Comparing the limits of both conditions,
\begin{equation*}
\begin{split}
    \frac{(b-1)q}{2b} - \frac{(b+1-\sqrt{2b^2-4b+3})q}{2b} &= \frac{(-2+\sqrt{2(b-1)^2 + 1}))q}{2b} \geq 0 \\
    &\Leftrightarrow \sqrt{2(b-1)^2 + 1} \geq 2 \\
    &\Leftrightarrow 2(b-1)^2 \geq 3 \\
    &\Leftrightarrow b \leq \frac{2-\sqrt{6}}{2} \text{ or } b \geq \frac{2+\sqrt{6}}{2} \\
\end{split}
\end{equation*}
and
\begin{equation*}
    \frac{(b+1+\sqrt{2b^2-4b+3})q}{2b} - \frac{(2b-1)q}{4b} = \frac{(3+2\sqrt{2(b-1)^2 + 1})q}{4b} \geq 0 \, .
\end{equation*}
Hence the optimum in Section $b+1$ is better than the optimum in Section $b$ unless $b=2$ (since $b \geq \frac{2+\sqrt{6}}{2}$) in which case the optimum in Section $b$ ($II$) is better than the optimum in Section $b+1$ ($III$) for $\left(\frac{(b-1)q}{2b} =\right) \frac{q}{4} \leq \frac{p}{2a} \leq \frac{(3+\sqrt{3})q}{4} \left(= \frac{(b+1+\sqrt{2b^2-4b+3})q}{2b}\right)$.

Finally, suppose instead that $b$ is odd. For $\frac{(b-2)q}{2b} \leq \frac{p}{2a} \leq \frac{(b-1)q}{2b}$, $x^*=\frac{p}{2a}$ enters Sections $b-1$ and $b$ whose maximal areas are $Area(V^+((\frac{p}{2a},0)))=\frac{(b-1)pq}{4ab} + \frac{p^2}{8a^2} - \frac{(b-3)(b-1)q^2}{8b^2}$ and $Area(V^+((\frac{p}{2a}, \frac{(b+1)q}{6b} - \frac{p}{6a}))) = \frac{(5b-1)pq}{24ab} + \frac{p^2}{12a^2} - \frac{(b-2)(2b-1)q^2}{12b^2}$ respectively. The optimum in Section $b$ is better than the optimum in Section $b-1$ if
\begin{equation*}
\begin{split}
    \left(\frac{(5b-1)pq}{24ab} \right. &+ \left. \frac{p^2}{12a^2} - \frac{(b-2)(2b-1)q^2}{12b^2}\right) - \left( \frac{(b-1)pq}{4ab} + \frac{p^2}{8a^2} - \frac{(b-3)(b-1)q^2}{8b^2} \right) \\
    &= - \frac{(b-5)pq}{24ab} - \frac{p^2}{24a^2} - \frac{(b^2+2b-5)q^2}{24b^2} \geq 0 \\
\end{split}
\end{equation*}
\begin{equation*}
\begin{split}
    \qquad \qquad&= \frac{p^2}{a^2} + \frac{(b-5)pq}{ab} + \frac{(b^2+2b-5)q^2}{b^2} \leq 0 \\
%    &= (\frac{p}{a} + \frac{(b-5)q}{2b})^2 - \frac{(b-5)^2q^2}{4b^2} + \frac{(b^2+2b-5)q^2}{b^2} \leq 0 \\
    &= \left(\frac{p}{a} + \frac{(b-5)q}{2b}\right)^2 + \frac{(3b^2+18b-45)q^2}{4b^2} \leq 0 \\
    &= \frac{(b-5-\sqrt{-3(b^2+6b-15)})q}{4b} \leq \frac{p}{2a} \leq \frac{(b-5+\sqrt{-3(b^2+6b-15)})q}{4b} \, . \\
\end{split}
\end{equation*}
Now this condition only holds if $b^2+6b-15 = (b+3)^2 - 24 \leq 0 \Leftrightarrow -3-\sqrt{24} \leq b \leq -3+2\sqrt{6} \approx 1.9$ so it holds for no value of $b$. Hence the optimum in Section $b-1$ is better than the optimum in Section $b$.

For $\frac{(b-1)q}{2b} \leq \frac{p}{2a} \leq \frac{q}{2}$, $x^*=\frac{p}{2a}$ enters Sections $b$ and $b+1$ whose maximal areas are $Area(V^+((\frac{p}{2a}, \frac{(b+1)q}{6b} - \frac{p}{6a}))) = \frac{(5b-1)pq}{24ab} + \frac{p^2}{12a^2} - \frac{(b-2)(2b-1)q^2}{12b^2}$ (only for $\frac{(b-2)q}{2b} \leq \frac{p}{2a} \leq \frac{(2b-1)q}{4b}$ so we now only consider the interval $\frac{(b-1)q}{2b} \leq \frac{p}{2a} \leq \frac{(2b-1)q}{4b}$ as before) and $Area(V^+((\frac{p}{2a},0)))=\frac{pq}{2a} - \frac{(b-1)q^2}{4b}$ respectively. The optimum in Section $b+1$ is better than the optimum in Section $b$ if
\begin{equation*}
\begin{split}
    \left( \frac{pq}{2a} \right. &- \left. \frac{(b-1)q^2}{4b}\right) - \left(\frac{(5b-1)pq}{24ab} + \frac{p^2}{12a^2} - \frac{(b-2)(2b-1)q^2}{12b^2}\right) \\
%    &=\frac{pq}{2a} - \frac{(b-1)q^2}{4b} - \frac{(5b-1)pq}{24ab} - \frac{p^2}{12a^2} + \frac{(b-2)(2b-1)q^2}{12b^2} \\
    &=\frac{(7b+1)pq}{24ab} - \frac{p^2}{12a^2} - \frac{(b^2 + 2 b - 2)q^2}{12b^2} \geq 0 \\
    &\Leftrightarrow \frac{p^2}{a^2} - \frac{(7b+1)pq}{2ab} + \frac{(b^2+2b-2)q^2}{b^2} \leq 0 \\
    &\Leftrightarrow \left(\frac{p}{a} - \frac{(7b+1)q}{4b}\right)^2 - \frac{(33b^2-18b+33)q^2}{16b^2} \leq 0 \\
    &\Leftrightarrow \frac{(7b+1-\sqrt{3(11b^2-6b+11)})q}{8b} \leq \frac{p}{2a} \leq \frac{(7b+1+\sqrt{3(11b^2-6b+11)})q}{8b} \, . \\
\end{split}
\end{equation*}
Comparing the limits of both conditions,
\begin{equation*}
\begin{split}
    \frac{(7b+1-\sqrt{3(11b^2-6b+11)})q}{8b}-\frac{(b-1)q}{2b} &= \frac{(3b+5-\sqrt{3(11b^2-6b+11)})q}{8b} \geq 0 \\
    &\Leftrightarrow 3b+5 \geq \sqrt{3(11b^2-6b+11)} \\
%    &\Leftrightarrow 9b^2+30b+25 \geq 33b^2-18b+33 \\
    &\Leftrightarrow 24b^2-48b+8 = 8( 3b^2-6b+1 ) \leq 0 \\
    &\Leftrightarrow \frac{3-\sqrt{6}}{3} \leq b \leq \frac{3+\sqrt{6}}{3} \approx 1.8 \\
\end{split}
\end{equation*}
and $$\frac{(7b+1+\sqrt{3(11b^2-6b+11)})q}{8b}-\frac{(2b-1)q}{4b} = \frac{(3b+3+\sqrt{3(11(b-\frac{3}{11})^2+\frac{112}{11})}))q}{8b} \geq 0 \, .$$
Hence the optimum in Section $b+1$ is better than the optimum in Section $b$.

Thus we have analysed all possible solutions upon $x^*=\frac{p}{2a}$ and discerned the best possible location on $x^*=\frac{p}{2a}$ for every combination of $p$, $q$, $a$, and $b$. These are summarised in Table~\ref{tab:gridedgeoptimals}.

%%%%%%%%%%%%%%%%%%%%%%%%%%%%%%%%%%%%%%%%%%%%%%%%%%%%%%%%%%%%%%%%%%%%%%%%%%%%%%%%%%%%%%%%%%%%%%%%

\begin{table}[!ht]
\centering
\begin{tabular}{c|c|c|c}
Section & Optimum & Area & Condition \\
\hline
%$I$ & $(*,0)$ & $\frac{pq}{2n}$ & $q \leq \frac{p}{n}$ \\
$II$ & $(\frac{p}{2a},0)$ & $\frac{pq}{2ab} + \frac{p^2}{8a^2}$ & $\frac{q}{2b} \leq \frac{p}{2a} \leq \frac{3q}{4b}$ \\ %$Area(V^+((\frac{p}{2a},0)))= \frac{lpq}{2ab} + \frac{p^2}{8a^2} - \frac{(l-1)lq^2}{2b^2}$
$III$ & $(\frac{p}{2a},\frac{q}{2b})$ & $\frac{3pq}{4ab} + \frac{p^2}{8a^2} - \frac{3q^2}{8b^2}$ & $\frac{3q}{4b} \leq \frac{p}{2a} \leq \frac{5q}{4b}$ \\
$IV$ & $(\frac{p}{2a},0)$ & $\frac{pq}{ab} + \frac{p^2}{8a^2} - \frac{q^2}{b^2}$ & $\frac{5q}{4b} \leq \frac{p}{2a} \leq \frac{3q}{2b}$ \\
%%%%%%
$2l+1$ & $(\frac{p}{2a},\frac{q}{2b})$ & $\frac{(4l-1)pq}{4ab} + \frac{p^2}{8a^2} - \frac{(4l^2-1)q^2}{8b^2}$ & $\frac{(2l-1)q}{2b} \leq \frac{p}{2a} \leq \frac{(2l+1)q}{2b}$ \\
%%%%%%
$b = 2$ & $(\frac{p}{2a},\frac{p}{6a})$ & $\frac{pq}{4a} + \frac{p^2}{12a^2}$ & $\frac{q}{4} \leq \frac{p}{2a} \leq \frac{(3+\sqrt{3})q}{4}$ \\ %$Area(V^+((\frac{p}{2a},\frac{p}{6a} - \frac{(b-2)q}{6b}))) = \frac{(2b-1)pq}{6ab} + \frac{p^2}{12a^2} - \frac{(b-2)(4b-1)q^2}{12b^2}$
$b + 1 = 3$ & $(\frac{p}{2a},\frac{q}{4})$ & $\frac{pq}{2a} - \frac{q^2}{8}$ & $\frac{(3+\sqrt{3})q}{4} \leq \frac{p}{2a}$ \\ %$Area(V^+((\frac{p}{2a},\frac{q}{2b})))=\frac{pq}{2a} - \frac{(b-1)q^2}{4b}$
%%%%%%
$b - 1$ even & $(\frac{p}{2a},0)$ & $\frac{(b-1)pq}{4ab} + \frac{p^2}{8a^2} - \frac{(b-3)(b-1)q^2}{8b^2}$ & $\frac{(b-2)q}{2b} \leq \frac{p}{2a} \leq \frac{(b-1)q}{2b}$ \\
$b+1$ even & $(\frac{p}{2a},0)$ & $\frac{pq}{2a} - \frac{(b-1)q^2}{4b}$ & $\frac{(b-1)q}{2b} \leq \frac{p}{2a}$ \\
%%%%%%
$b + 1$ odd & $(\frac{p}{2a},\frac{q}{2b})$ & $\frac{pq}{2a} - \frac{(b-1)q^2}{4b}$ & $\frac{(b-1)q}{2b} \leq \frac{p}{2a}$ \\
\end{tabular}
\caption{Optima upon $x^* = \frac{p}{2a}$.}
\label{tab:gridedgeoptimals}
\end{table}

Now all that remains is to compare the optima upon $x^*=\frac{p}{2a}$ with the optima in Section $I$ according to the conditions in Table~\ref{tab:gridedgeoptimals} as well as the existence of $w_{{0}_{LL}}$. Since we have chosen $w_0$ to be the $(\ceil{\frac{a}{2}},\ceil{\frac{b}{2}})$th point in $W$, the condition that $w_{{0}_{LL}}$ does not exist amounts to $a=2$ (importantly, this condition has no effect upon the optima upon $x^*=\frac{p}{2a}$).

Recalling our earlier results, the maximal area in Section $I$ is $Area(V^+((0,\frac{q}{2b})))= \frac{pq}{2ab} + \frac{q^2}{8b^2}$ unless $a=2$ in which case the maximal area is $Area(V^+((\frac{q}{6b},\frac{q}{2b})))= \frac{pq}{2ab}+\frac{q^2}{12b^2}$. Studying the areas claimed by the optima within Table~\ref{tab:gridedgeoptimals} we can see that, if all other variables are fixed, the value of every optimal area increases with $p$ at a rate of at least $\frac{q}{2ab}$ (which is the rate of increase of the optimal values for Section $I$). Moreover, the value of the optimum upon $x^*=\frac{p}{2a}$ further increases with $\frac{p}{2a}$ whenever the optimal solution is replaced by the next (better) solution. Therefore, for $b \neq 2$, if the optimum in Section $I$ is better than the optimum upon $x^*=\frac{p}{2a}$ at $\frac{p}{2a}=X$ then the optimum in Section $I$ is better than the optimum upon $x^*=\frac{p}{2a}$ for all $\frac{p}{2a} \leq X$, and similarly if the optimum upon $x^*=\frac{p}{2a}$ is better than the optimum in Section $I$ at $\frac{p}{2a}=X$ then the optimum upon $x^*=\frac{p}{2a}$ is better than the optimum in Section $I$ for all $\frac{p}{2a} \geq X$. Hence there exists a value $X$ of $\frac{p}{2a}$ at which point the values of the optima in Section $I$ and upon $x^*=\frac{p}{2a}$ are equal, and this determines our best point $b^*$. Additionally, since the value of the optimum in Section $I$ is reduced if $a=2$, the value $X$ for $a > 2$ will be more than the analogous value $X'$ for $a=2$.

If $a > 2$ and $b > 2$ then we can simply compare the area within Section $I$ with the smallest possible maximal area upon $x^*=\frac{p}{2a}$ (Section $II$):
$$\left(\frac{pq}{2ab} + \frac{q^2}{8b^2}\right) - \left(\frac{pq}{2ab} + \frac{p^2}{8a^2}\right) = \frac{q^2}{8b^2} - \frac{p^2}{8a^2} \geq 0 \Leftrightarrow \frac{q}{b} \geq \frac{p}{a}$$ but $\frac{p}{a} \geq \frac{q}{b}$ so if $a > 2$ and $b > 2$ then the optimum lies upon $x^*=\frac{p}{2a}$. The same is indeed true for $a=2$ as described above. Therefore we have the best points $b^*$ as outlined in Table~\ref{tab:gridOptimals}.

\begin{table}[!ht]
\centering
\begin{tabular}{c|c|c|c}
Section & Optimum & Area & Condition \\
\hline
%$I$ & $(*,0)$ & $\frac{pq}{2n}$ & $q \leq \frac{p}{n}$ \\
$II$ & $(\frac{p}{2a},0)$ & $\frac{pq}{2ab} + \frac{p^2}{8a^2}$ & $\frac{q}{2b} \leq \frac{p}{2a} \leq \frac{3q}{4b}$ \\ %$Area(V^+((\frac{p}{2a},0)))= \frac{lpq}{2ab} + \frac{p^2}{8a^2} - \frac{(l-1)lq^2}{2b^2}$
$III$ & $(\frac{p}{2a},\frac{q}{2b})$ & $\frac{3pq}{4ab} + \frac{p^2}{8a^2} - \frac{3q^2}{8b^2}$ & $\frac{3q}{4b} \leq \frac{p}{2a} \leq \frac{5q}{4b}$ \\
$IV$ & $(\frac{p}{2a},0)$ & $\frac{pq}{ab} + \frac{p^2}{8a^2} - \frac{q^2}{b^2}$ & $\frac{5q}{4b} \leq \frac{p}{2a} \leq \frac{3q}{2b}$ \\
%%%%%%
$2l+1$ & $(\frac{p}{2a},\frac{q}{2b})$ & $\frac{(4l-1)pq}{4ab} + \frac{p^2}{8a^2} - \frac{(4l^2-1)q^2}{8b^2}$ & $\frac{(2l-1)q}{2b} \leq \frac{p}{2a} \leq \frac{(2l+1)q}{2b}$ \\
%%%%%%
$b - 1$ even & $(\frac{p}{2a},0)$ & $\frac{(b-1)pq}{4ab} + \frac{p^2}{8a^2} - \frac{(b-3)(b-1)q^2}{8b^2}$ & $\frac{(b-2)q}{2b} \leq \frac{p}{2a} \leq \frac{(b-1)q}{2b}$ \\
$b+1$ even & $(\frac{p}{2a},0)$ & $\frac{pq}{2a} - \frac{(b-1)q^2}{4b}$ & $\frac{(b-1)q}{2b} \leq \frac{p}{2a}$ \\
%%%%%%
$b + 1$ odd & $(\frac{p}{2a},\frac{q}{2b})$ & $\frac{pq}{2a} - \frac{(b-1)q^2}{4b}$ & $\frac{(b-1)q}{2b} \leq \frac{p}{2a}$ \\
\end{tabular}
\caption{The best point $b^*$ for $b \neq 2$.}
\label{tab:gridOptimals}
\end{table}

Alternatively, if $a > 2$ and $b=2$ then we must compare the maximal areas within Section $I$ and Section $b$:
$$\left(\frac{pq}{4a} + \frac{q^2}{32}\right) - \left(\frac{pq}{4a} + \frac{p^2}{12a^2}\right) = \frac{q^2}{32} - \frac{p^2}{12a^2} \geq 0 \Leftrightarrow \frac{p}{2a} \leq \frac{\sqrt{6}q}{8} \, .$$
The optimum in Section $b$ is valid for $\frac{q}{4} \leq \frac{p}{2a} \leq \frac{(3+\sqrt{3})q}{4}$ so the optimum within Section $I$ is $b^*$ for $\frac{q}{4} \leq \frac{p}{2a} \leq \frac{\sqrt{6}q}{8}$ but we need not compare the maximal area in Section $I$ any further. This gives the best points $b^*$ as displayed in Table~\ref{tab:gridOptimalsb=2}.

\begin{table}[!ht]
\centering
\begin{tabular}{c|c|c|c}
Section & Optimum & Area & Condition \\
\hline
$I$ & $(0,\frac{q}{2b})$ & $\frac{pq}{4a} + \frac{q^2}{32}$ & $\frac{q}{4} \leq \frac{p}{2a} \leq \frac{\sqrt{6}q}{8}$ \\ 
$b$ & $(\frac{p}{2a},\frac{p}{6a})$ & $\frac{pq}{4a} + \frac{p^2}{12a^2}$ & $\frac{\sqrt{6}q}{8} \leq \frac{p}{2a} \leq \frac{(3+\sqrt{3})q}{4}$ \\ 
$b + 1$ & $(\frac{p}{2a},\frac{q}{4})$ & $\frac{pq}{2a} - \frac{q^2}{8}$ & $\frac{(3+\sqrt{3})q}{4} \leq \frac{p}{2a}$ \\
\end{tabular}
\caption{The best point $b^*$ for $b = 2$ and $a \neq 2$.}
\label{tab:gridOptimalsb=2}
\end{table}

\pagebreak Finally if $a=2$ and $b=2$ then, comparing the maximal areas within Section $I$ and Section $b$,
$$\left(\frac{pq}{8}+\frac{q^2}{48}\right) - \left(\frac{pq}{8} + \frac{p^2}{48}\right) = \frac{q^2}{48} - \frac{p^2}{48} \geq 0 \Leftrightarrow \frac{q}{b} \geq \frac{p}{a} \, .$$ Therefore, since $\frac{p}{a} \geq \frac{q}{b}$, the optimum in Section $b$ is always better than the optimum in Section $I$. This gives the best points $b^*$ as displayed in Table~\ref{tab:gridOptimalsb=2a=2}.

\begin{table}[!ht]
\centering
\begin{tabular}{c|c|c|c}
Section & Optimum & Area & Condition \\
\hline
$b$ & $(\frac{p}{2a},\frac{p}{6a})$ & $\frac{pq}{4a} + \frac{p^2}{12a^2}$ & $\frac{q}{4} \leq \frac{p}{2a} \leq \frac{(3+\sqrt{3})q}{4}$ \\
$b + 1$ & $(\frac{p}{2a},\frac{q}{4})$ & $\frac{pq}{2a} - \frac{q^2}{8}$ & $\frac{(3+\sqrt{3})q}{4} \leq \frac{p}{2a}$ \\
\end{tabular}
\caption{The best point $b^*$ for $b = 2$ and $a = 2$.}
\label{tab:gridOptimalsb=2a=2}
\end{table}

And thus, we have found the best points $b^*$ in response to White playing an $a \times b$ grid.

\bigskip

\subsection{Black's best arrangement}
\label{sec:BlackGridArrangement}

We have found Black's best point $b^*$ but, as we have seen in Section~\ref{sec:BlackRow}, these are often not useful points to play when considering a whole arrangement. As we saw towards the end of Section~\ref{sec:BlackRow}, a good point for Black to play within an arrangement steals the best proportion of two halves of White's cells in $\mathcal{VD}(W)$. However, as Black's points venture further away from $w_0$ they steal less and less from the the two Voronoi cells they steal the most from, sacrificing this area in order to steal more area from a greater number of White's Voronoi cells. Therefore it may be useful to explore how Black performs playing closer to White's points, and we shall investigate Black's possible placements within Sections $I$, $II$, and $III$.

\subsubsection{Core quadrants}

We have already investigated the core quadrants in our search for Black's best point so there is little extra work we need do within these quadrants.

Section $I$ was given a full exploration in our search for $b^*$ so we will simply refer the reader to the results summarised in Figures~\ref{fig:GridOptimalI} and \ref{fig:GridOptimalIRes}. On the other hand, only the areas upon $x^*=\frac{p}{2a}$ were optimised within Sections $II$ and $III$. It requires little effort, however, to extend the results already investigated to cases where $x^*=\frac{p}{2a}$ does not lie within the section in question.

Each constituent area component $V^+(b_1) \cap (V^\circ(w_{i_L}) \cup V^\circ(w_{i_R}))$ calculated in Section~\ref{sec:WhiteGrid} was found to be either independent of the value of $x$, or maximised by choosing $x$ as close as possible to $\frac{p}{2a}$ (which leads us to the result that $b^*$ must lie on $x^*=\frac{p}{2a}$ or Section $I$). In Sections $II$ and $IV$, these maximum values of $x$ are $\frac{q}{b}$ and $\frac{3q}{2b}$ (at $y=0$ and $y=\frac{q}{2b}$), so we must consider the optimal solutions within these sections if $\frac{p}{2a} > \frac{q}{b}$ and $\frac{p}{2a} > \frac{3q}{2b}$ respectively. Therefore, in order to confirm that the maximisation of $y$ in our previous work does not conflict with this maximisation of $x$ that we must now consider, we need only to check that the optima obtained through our previous discoveries are $\left(\frac{q}{b},0\right)$ and $\left(\frac{3q}{2b},\frac{q}{2b}\right)$ respectively (i.e. $x$ is maximal at $\frac{p}{2a} = \frac{q}{b}$ and $\frac{p}{2a} = \frac{3q}{2b}$).

\pagebreak \paragraph{Section $II$} For Section $II$, assuming $V^+(b_1)$ does not touch the boundary of $\mathcal{P}$, the maximal area when $\frac{p}{2a} \leq \frac{q}{b}$ is $Area(V^+((\frac{p}{2a},0)))= \frac{pq}{2ab} + \frac{p^2}{8a^2}$ as depicted in Figure~\ref{fig:GridOptimalII}. If $\frac{p}{2a} > \frac{q}{b}$ then, since the optimum at $\frac{p}{2a} = \frac{q}{b}$ is $(\frac{q}{b},0)$ as required, our previous results give that the optimum is at $b_1=\left(\frac{q}{b},0\right)$ with
\begin{equation*}
\begin{split}
    Area(V^+(b_1)) &= Area(V^+(b_1) \cap (V^\circ(w_{1_L}) \cup V^\circ(w_{1_R}))) \\ &\qquad+ Area(V^+(b_1) \cap (V^\circ(w_{0_L}) \cup V^\circ(w_{0_R}))) \\ &\qquad+ Area(V^+(b_1) \cap (V^\circ(w_{{-1}_L}) \cup V^\circ(w_{{-1}_R}))) \\
    &= \left(- \frac{1}{4}\left(\frac{q}{b}\right)^2 + \frac{0^2}{4} + \frac{p}{4a}\frac{q}{b} + \left(\frac{p}{4a} - \frac{(1-1)q}{2b}\right)\times 0 - \frac{(1-1)pq}{4ab} + \frac{(1-1)^2q^2}{4b^2}\right) \\ &\qquad+ \left(\frac{pq}{2ab} - 0^2\right) + \left(- \frac{1}{4}\left(\frac{q}{b}\right)^2 + \frac{0^2}{4} + \frac{p}{4a}\frac{q}{b} - \left(\frac{p}{4a} + \frac{(-1+1)q}{2b}\right) \times 0 \right. \\  &\left.\qquad+ \frac{(-1+1)pq}{4ab} + \frac{(-1+1)^2q^2}{4b^2}\right) \\
    &= \frac{pq}{ab} - \frac{q^2}{2b^2}  \\
\end{split}
\end{equation*}
as depicted in Figure~\ref{fig:GridOptimalIILong}.

\begin{figure}[!ht]
\begin{subfigure}{.414\textwidth}
  \centering
  \includegraphics[width=0.99\textwidth]{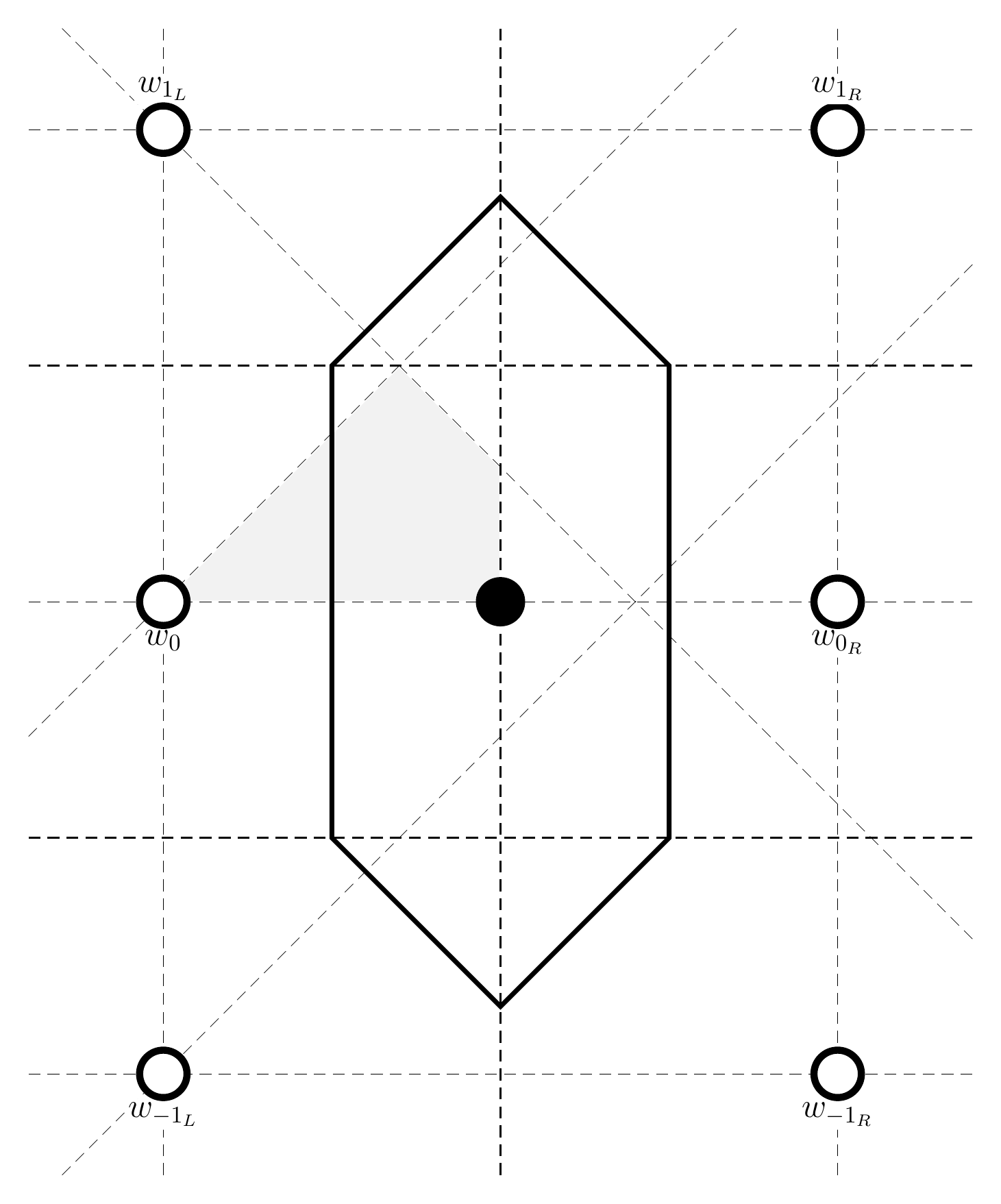}
  \caption{$b_1=(\frac{p}{2a},0)$ only if $\frac{p}{2a} \leq \frac{q}{b}$.}
  \label{fig:GridOptimalII}
\end{subfigure}%
\begin{subfigure}{.586\textwidth}
  \centering
  \includegraphics[width=0.99\textwidth]{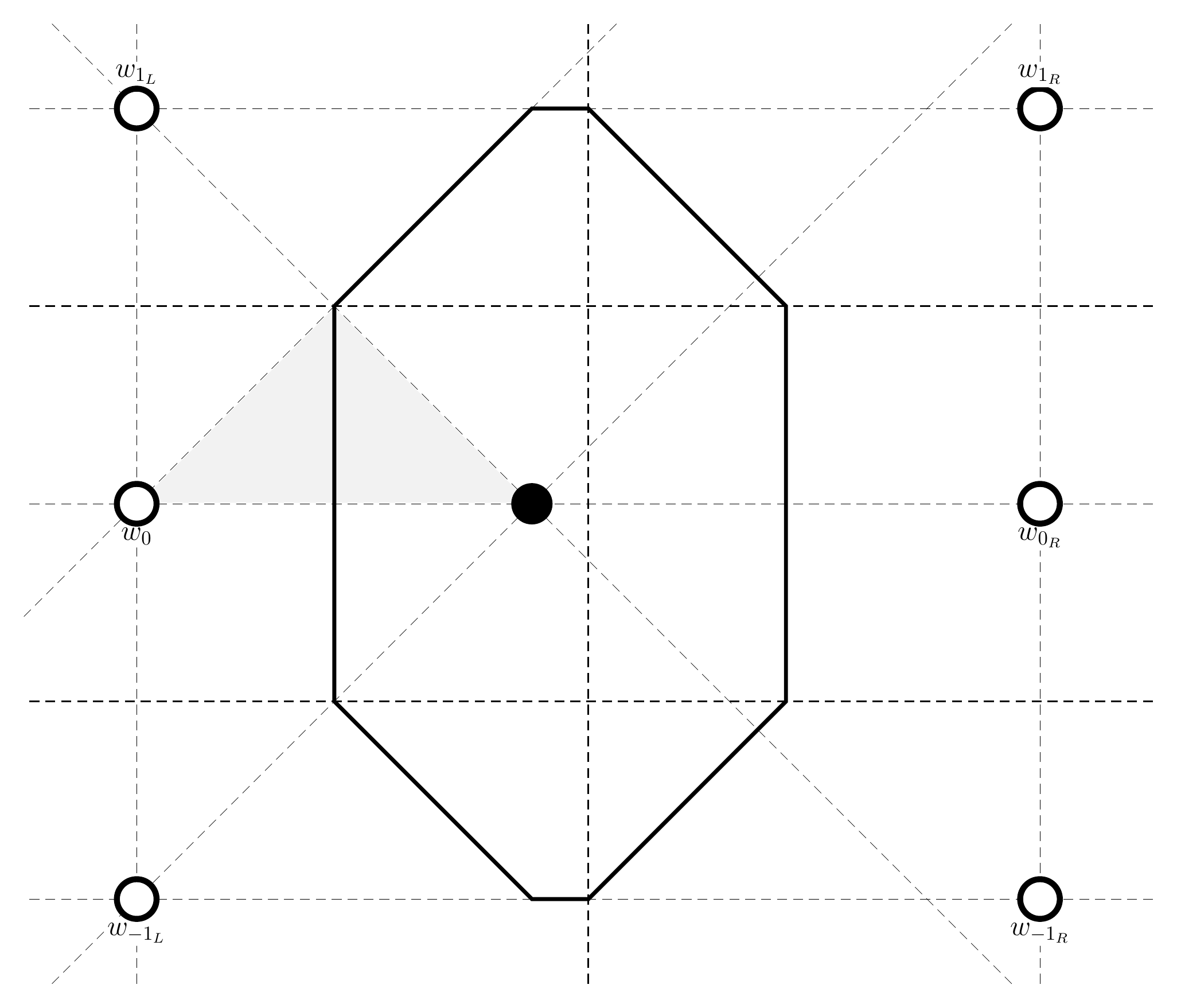}
  \caption{$b_1=(\frac{q}{b},0)$ only if $\frac{p}{2a} \geq \frac{q}{b}$.}
  \label{fig:GridOptimalIILong}
\end{subfigure}
\caption{Maximal area Voronoi cells $V^+(b_1)$ for $b_1$ within Section $II$ not touching the horizontal edges of $\mathcal{P}$.}
\end{figure}

Otherwise if $w_{{-1}_L}$ does not exist then, using the results of Section $b$ from above: if $\frac{p}{2a} \leq \frac{3q}{4b}$ then the maximal area is $Area(V^+((\frac{p}{2a},\frac{p}{6a}))) = \frac{pq}{2ab} + \frac{p^2}{12a^2}$ as depicted in Figure~\ref{fig:GridOptimalIIResa}; if $\frac{3q}{4b} \leq \frac{p}{2a} \leq \frac{q}{b}$ then the optimum is $(\frac{p}{2a}, \frac{q}{b} - \frac{p}{2a})$ (upon the boundary between Section $II$ and Section $III$) giving
\begin{equation*}
\begin{split}
     Area(V^+(b_1)) &= - \frac{3}{4}\left(\frac{q}{b} - \frac{p}{2a}\right)^2 + \left(\frac{p}{4a} - \frac{(1-1)q}{2b}\right)\left(\frac{q}{b} - \frac{p}{2a}\right) \\ &\qquad+ \frac{(3(1)-1)pq}{4ab} + \frac{p^2}{16a^2} - \frac{(1-1)(3(1)-1)q^2}{4b^2} \\
     &= \frac{3pq}{2ab} - \frac{p^2}{4a^2} - \frac{3q^2}{4b^2} \\
\end{split}
\end{equation*}
as depicted in Figure~\ref{fig:GridOptimalIIResb}; and if $\frac{p}{2a} \geq \frac{q}{b}$, since the optimum at $\frac{p}{2a} = \frac{q}{b}$ is $(\frac{q}{b},0)$, our previous results give that the optimum is at $b_1=(\frac{q}{b},0)$ with
\begin{equation*}
\begin{split}
    Area(V^+(b_1)) &= Area(V^+(b_1) \cap (V^\circ(w_{1_L}) \cup V^\circ(w_{1_R}))) \\ &\qquad+ Area(V^+(b_1) \cap (V^\circ(w_{0_L}) \cup V^\circ(w_{0_R}))) \\
    &= - \frac{1}{4}\left(\frac{q}{b}\right)^2 + \frac{0^2}{4} + \frac{p}{4a}\frac{q}{b} + \left(\frac{p}{4a} - \frac{(1-1)q}{2b}\right)\times 0 - \frac{(1-1)pq}{4ab} + \frac{(1-1)^2q^2}{4b^2} \\ &\qquad+ \frac{pq}{2ab} - 0^2 \\
    &= \frac{3pq}{4ab} - \frac{q^2}{4b^2} \\
\end{split}
\end{equation*}
as depicted in Figure~\ref{fig:GridOptimalIIResLong}.

\begin{figure}[!ht]
\begin{subfigure}{.469\textwidth}
  \centering
  \includegraphics[width=0.99\textwidth]{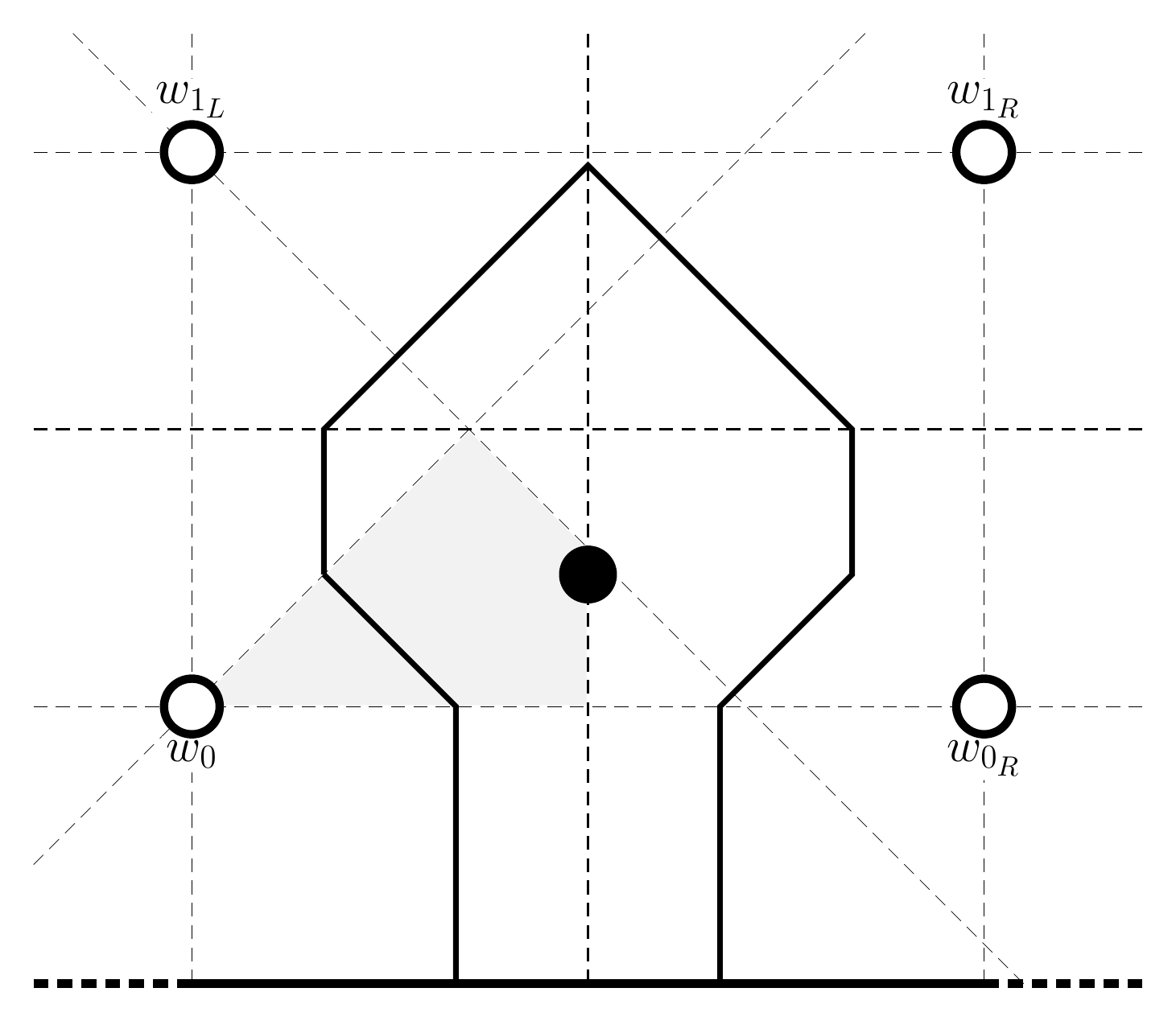}
  \caption{$b_1=(\frac{p}{2a},\frac{p}{6a})$ only if $\frac{p}{2a} \leq \frac{3q}{4b}$.}
  \label{fig:GridOptimalIIResa}
\end{subfigure}%
\begin{subfigure}{.531\textwidth}
  \centering
  \includegraphics[width=0.99\textwidth]{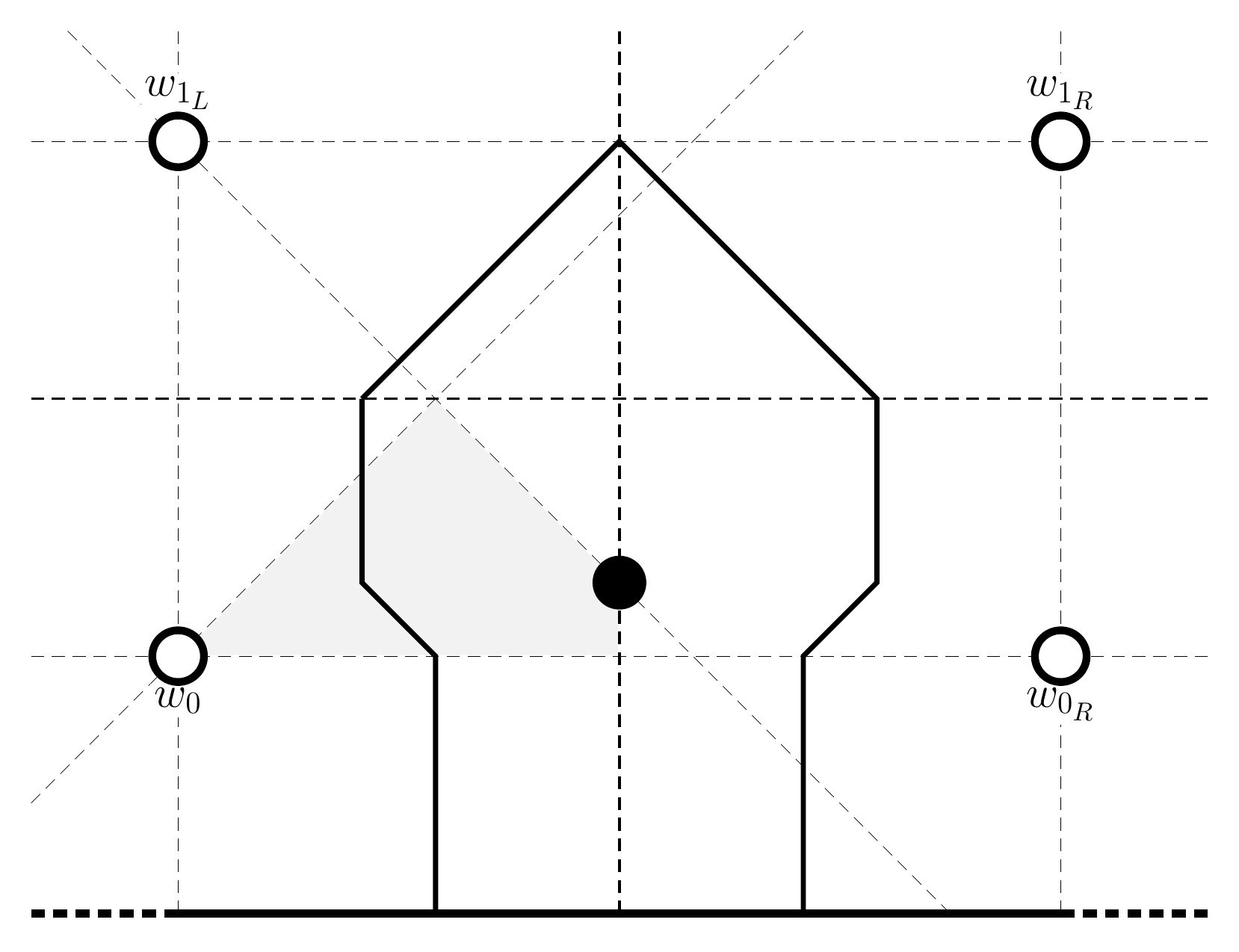}
  \caption{$b_1=(\frac{p}{2a}, \frac{q}{b} - \frac{p}{2a})$ only if $\frac{3q}{4b} \leq \frac{p}{2a} \leq \frac{q}{b}$.}
  \label{fig:GridOptimalIIResb}
\end{subfigure}

\centering
\begin{subfigure}{.65\textwidth}
  \centering
  \includegraphics[width=0.99\textwidth]{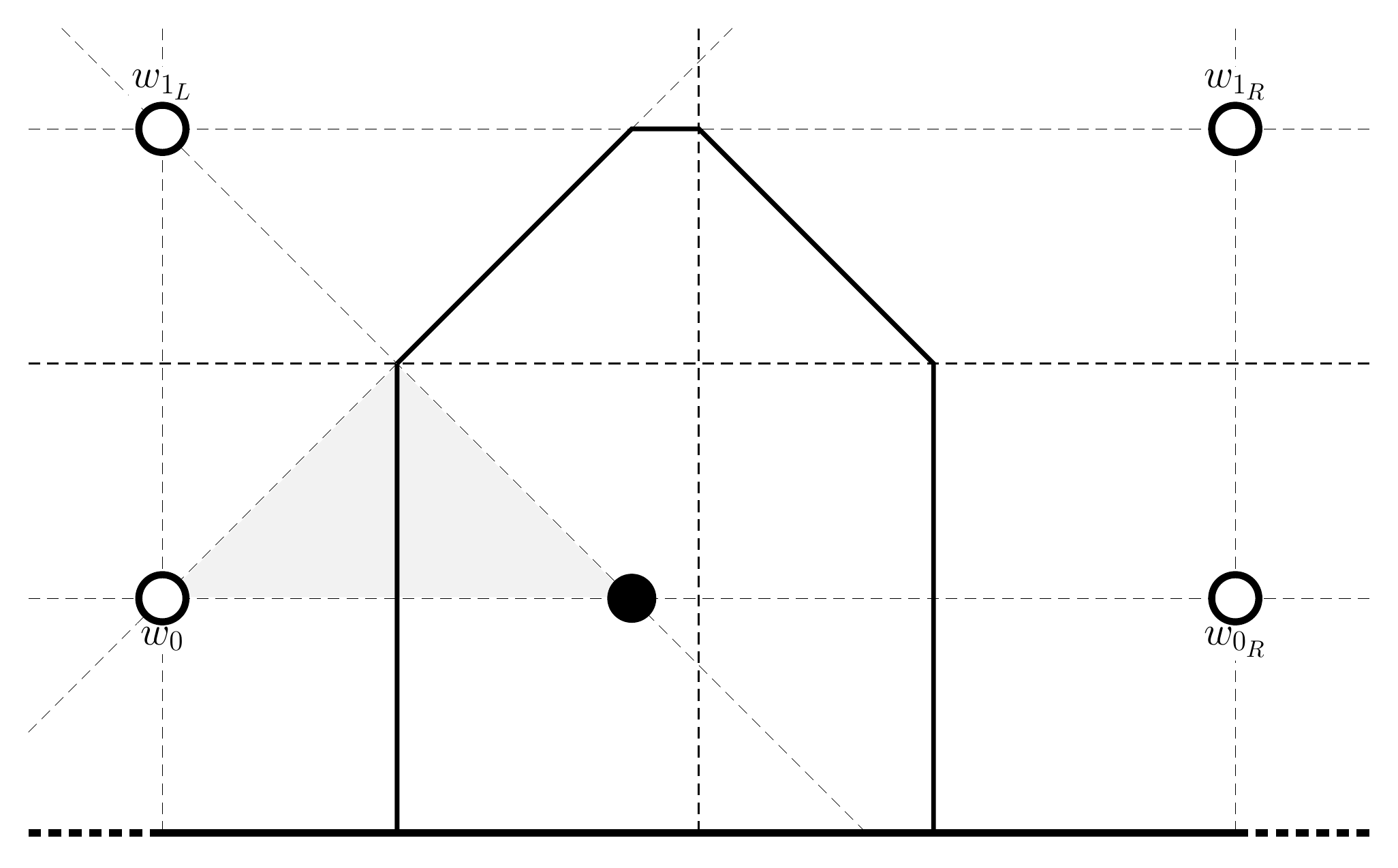}
  \caption{$b_1=(\frac{q}{b},0)$ only if $\frac{p}{2a} \geq \frac{q}{b}$.}
  \label{fig:GridOptimalIIResLong}
\end{subfigure}
\caption{Maximal area Voronoi cells $V^+(b_1)$ for $b_1$ within Section $II$ touching the bottommost horizontal edge of $\mathcal{P}$.}
\end{figure}

\paragraph{Section $III$} For Section $III$, assuming $V^+(b_1)$ does not touch the boundary of $\mathcal{P}$, the maximal area when $\frac{p}{2a} \leq \frac{3q}{2b}$ is $Area(V^+((\frac{p}{2a},\frac{q}{2b})))= \frac{3pq}{4ab} + \frac{p^2}{8a^2} - \frac{3q^2}{8b^2}$ as depicted in Figure~\ref{fig:GridOptimalIII}. If $\frac{p}{2a} \geq \frac{3q}{2b}$ then, since the optimum at $\frac{p}{2a} = \frac{3q}{2b}$ is $\left(\frac{3q}{2b},\frac{q}{2b}\right)$ as required, our previous results give that the optimum is at $b_1=\left(\frac{3q}{2b},\frac{q}{2b}\right)$ with

\begin{equation*}
\begin{split}
    Area(V^+(b_1)) &= Area(V^+(b_1) \cap (V^\circ(w_{2_L}) \cup V^\circ(w_{2_R})) \\ &\qquad+ Area(V^+(b_1) \cap (V^\circ(w_{1_L}) \cup V^\circ(w_{1_R})) \\ &\qquad+ Area(V^+(b_1) \cap (V^\circ(w_{0_L}) \cup V^\circ(w_{0_R})) \\ &\qquad+ Area(V^+(b_1) \cap (V^\circ(w_{{-1}_L}) \cup V^\circ(w_{{-1}_R})) \\
    &= \left(- \frac{1}{4}\left(\frac{3q}{2b}\right)^2 + \frac{1}{4}\left(\frac{q}{2b}\right)^2 + \frac{p}{4a}\frac{3q}{2b} + \left(\frac{p}{4a} - \frac{(2-1)q}{2b}\right)\frac{q}{2b} - \frac{(2-1)pq}{4ab} + \frac{(2-1)^2q^2}{4b^2}\right) \\
    &\qquad+ \left(\frac{q}{b}\frac{q}{2b} + \frac{pq}{2ab} - \frac{(4(1)-1)q^2}{4b^2}\right) + \left(\frac{pq}{2ab} - \left(\frac{q}{2b}\right)^2\right) + \left(- \frac{1}{4}\left(\frac{3q}{2b}\right)^2 + \frac{1}{4}\left(\frac{q}{2b}\right)^2 \right. \\ &\left.\qquad+ \frac{p}{4a}\frac{3q}{2b} - \left(\frac{p}{4a} + \frac{(-1+1)q}{2b}\right)\frac{q}{2b} + \frac{(-1+1)pq}{4ab} + \frac{(-1+1)^2q^2}{4b^2}\right) \\
    &= \frac{3pq}{2ab} - \frac{3q^2}{2b^2} \\
    %&= - \frac{x^2}{4} + \frac{y^2}{4} + \frac{p}{4a}x + \left(\frac{p}{4a} - \frac{(i-1)q}{2b}\right)y - \frac{(i-1)pq}{4ab} + \frac{(i-1)^2q^2}{4b^2} \\
    %&\qquad+ \frac{q}{b}y + \frac{pq}{2ab} - \frac{(4i-1)q^2}{4b^2} \\
    %&\qquad+ \frac{pq}{2ab} - y^2 \\
    %&\qquad - \frac{x^2}{4} + \frac{y^2}{4} + \frac{p}{4a}x - \left(\frac{p}{4a} + \frac{(i+1)q}{2b}\right)y + \frac{(i+1)pq}{4ab} + \frac{(i+1)^2q^2}{4b^2} \\
\end{split}
\end{equation*} %\frac{q}{b}y + \frac{pq}{2ab} - \frac{(4i-1)q^2}{4b^2}
as depicted in Figure~\ref{fig:GridOptimalIIILong}.

\begin{figure}[!ht]\ContinuedFloat
\begin{subfigure}{.357\textwidth}
  \centering
  \includegraphics[width=0.99\textwidth]{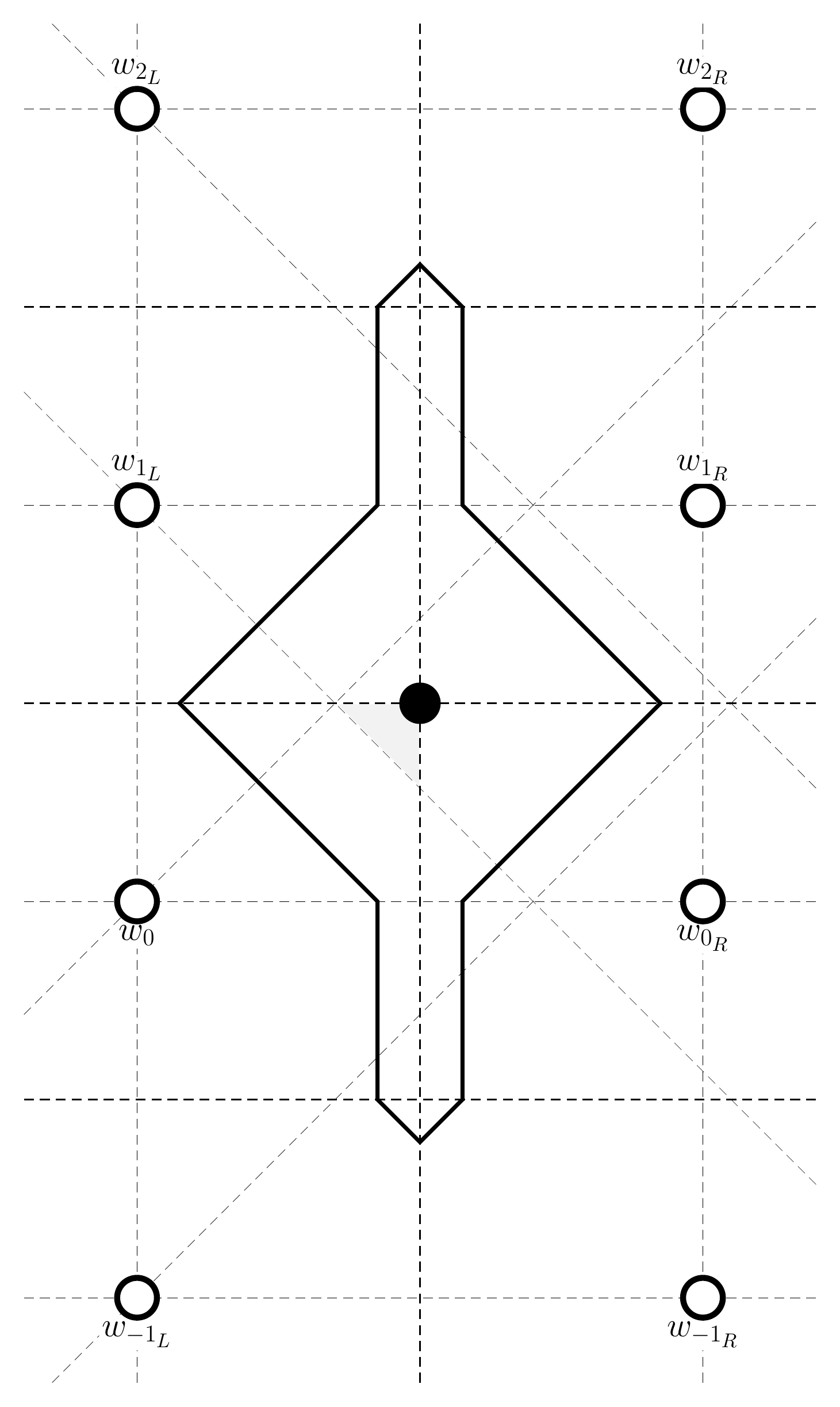}
  \caption{$b_1=(\frac{p}{2a},\frac{q}{2b})$ only if $\frac{p}{2a} \leq \frac{3q}{2b}$.}
  \label{fig:GridOptimalIII}
\end{subfigure}%
\begin{subfigure}{.643\textwidth}
  \centering
  \includegraphics[width=0.99\textwidth]{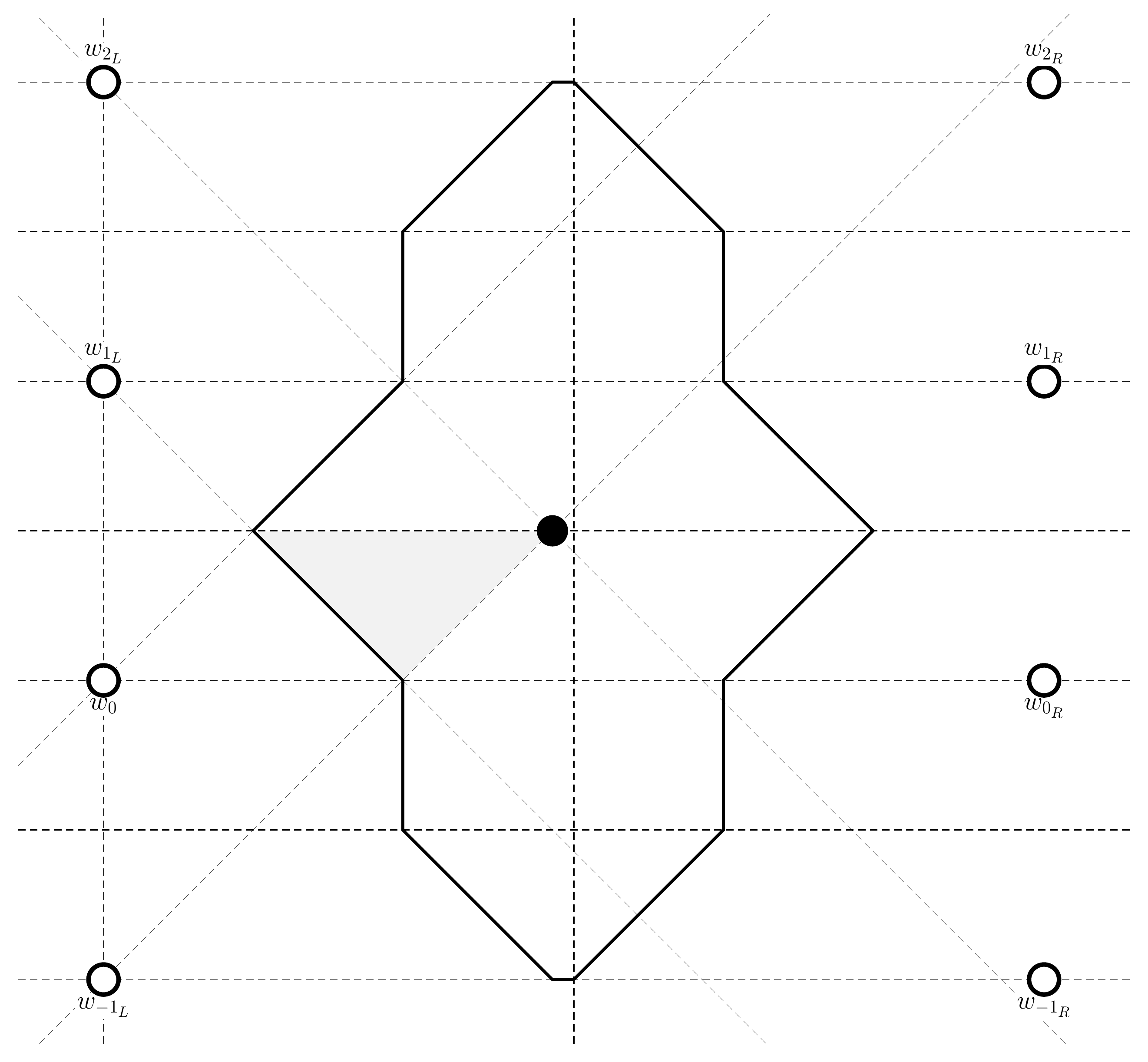}
  \caption{$b_1=(\frac{3q}{2b},\frac{q}{2b})$ only if $\frac{p}{2a} \geq \frac{3q}{2b}$.}
  \label{fig:GridOptimalIIILong}
\end{subfigure}
\caption{Maximal area Voronoi cells $V^+(b_1)$ for $b_1$ within Section $III$ not touching the horizontal edges of $\mathcal{P}$.}
\end{figure}

Otherwise, if $w_{{-1}_L}$ does not exist (a situation previously not necessary to study) then
\begin{equation*}
\begin{split}
    Area(V^+(b_1)) &= Area(V^+(b_1) \cap (V^\circ(w_{2_L}) \cup V^\circ(w_{2_R})) \\ &\qquad+ Area(V^+(b_1) \cap (V^\circ(w_{1_L}) \cup V^\circ(w_{1_R})) \\ &\qquad+ Area(V^+(b_1) \cap (V^\circ(w_{0_L}) \cup V^\circ(w_{0_R})) \\
    &= \left(- \frac{x^2}{4} + \frac{y^2}{4} + \frac{p}{4a}x + \left(\frac{p}{4a} - \frac{(2-1)q}{2b}\right)y - \frac{(2-1)pq}{4ab} + \frac{(2-1)^2q^2}{4b^2}\right) \\
    &\qquad+ \left(\frac{q}{b}y + \frac{pq}{2ab} - \frac{(4(1)-1)q^2}{4b^2}\right) + \left(\frac{pq}{2ab} - y^2\right) \\
\end{split}
\end{equation*}
\begin{equation*}
\begin{split}
    &= - \frac{x^2}{4} - \frac{3y^2}{4} + \frac{p}{4a}x + \left(\frac{p}{4a} + \frac{q}{2b}\right)y + \frac{3pq}{4ab} - \frac{q^2}{2b^2} \qquad \, \\
\end{split}
\end{equation*}
gives partial derivatives
\begin{equation*}
\begin{split}
\frac{\delta A}{\delta x}&= - \frac{x}{2} + \frac{p}{4a} \\
\frac{\delta A}{\delta y}&= - \frac{3y}{2} + \frac{p}{4a}+\frac{q}{2b} \\
\end{split}
\end{equation*}
which gives optimal $y^*=\frac{p}{6a}+\frac{q}{3b} \geq \frac{q}{6b}+\frac{q}{3b} = \frac{q}{2b}$. Therefore, since the area is maximised by choosing $x$ as close as possible to $\frac{p}{2a}$ and choosing $y$ as close as possible to $y^* \geq \frac{q}{2b}$, the maximal areas within Section $III$ are, for $\frac{p}{2a} \leq \frac{3q}{2b}$, $Area(V^+((\frac{p}{2a},\frac{q}{2b})))= \frac{7pq}{8ab} + \frac{p^2}{16a^2} - \frac{7q^2}{16b^2}$ as depicted in Figure~\ref{fig:GridOptimalIIIRes1} and, for $\frac{p}{2a} \geq \frac{3q}{2b}$, $Area(V^+(((\frac{3q}{2b},\frac{q}{2b})))=\frac{5pq}{4ab} - \frac{q^2}{b^2}$ as depicted in Figure~\ref{fig:GridOptimalIIIRes1Long}.

\begin{figure}[!ht]\ContinuedFloat
\begin{subfigure}{.358\textwidth}
  \centering
  \includegraphics[width=0.99\textwidth]{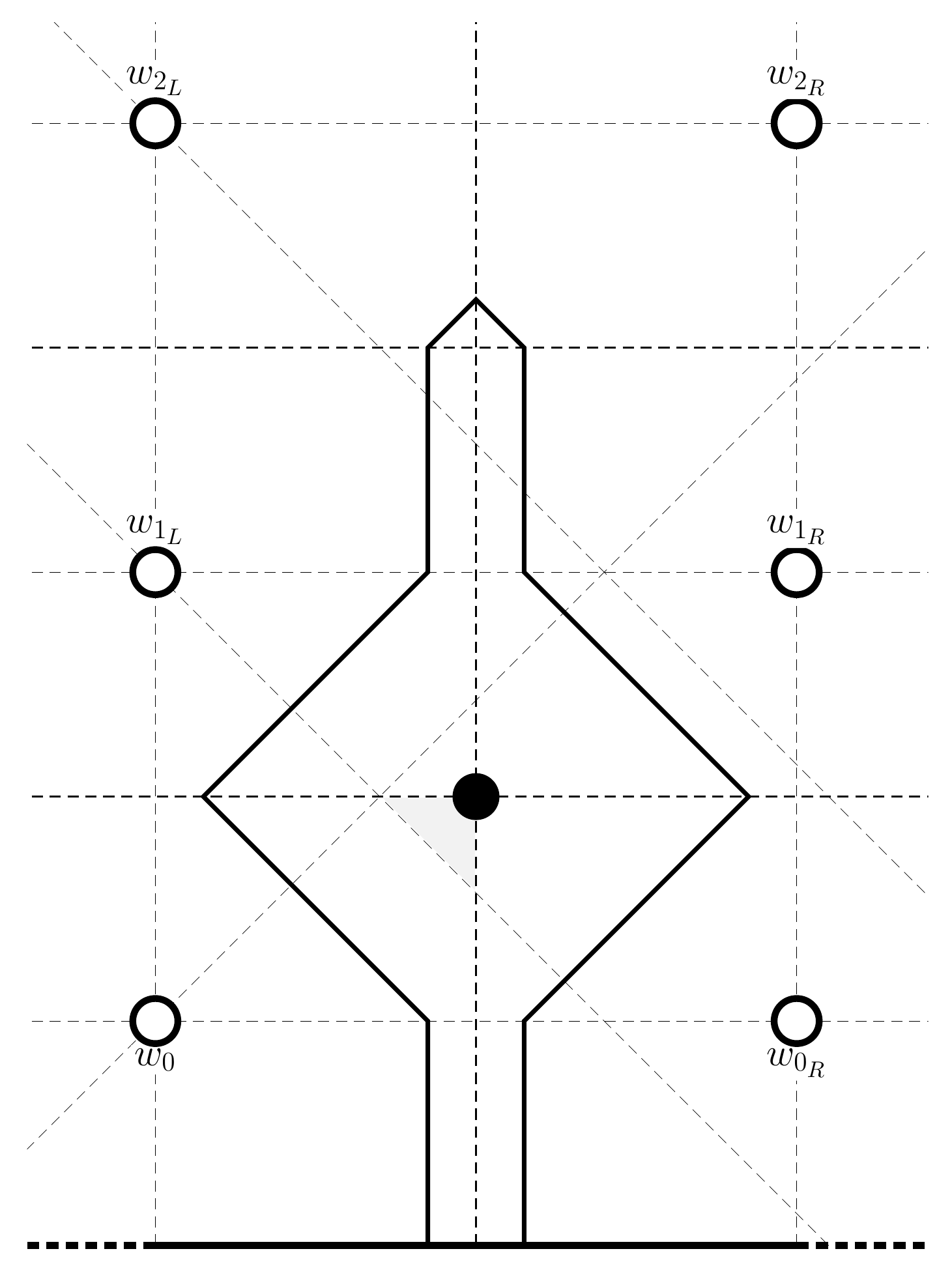}
  \caption{$b_1=(\frac{p}{2a},\frac{q}{2b})$ only if $\frac{p}{2a} \leq \frac{3q}{2b}$.}
  \label{fig:GridOptimalIIIRes1}
\end{subfigure}%
\begin{subfigure}{.642\textwidth}
  \centering
  \includegraphics[width=0.99\textwidth]{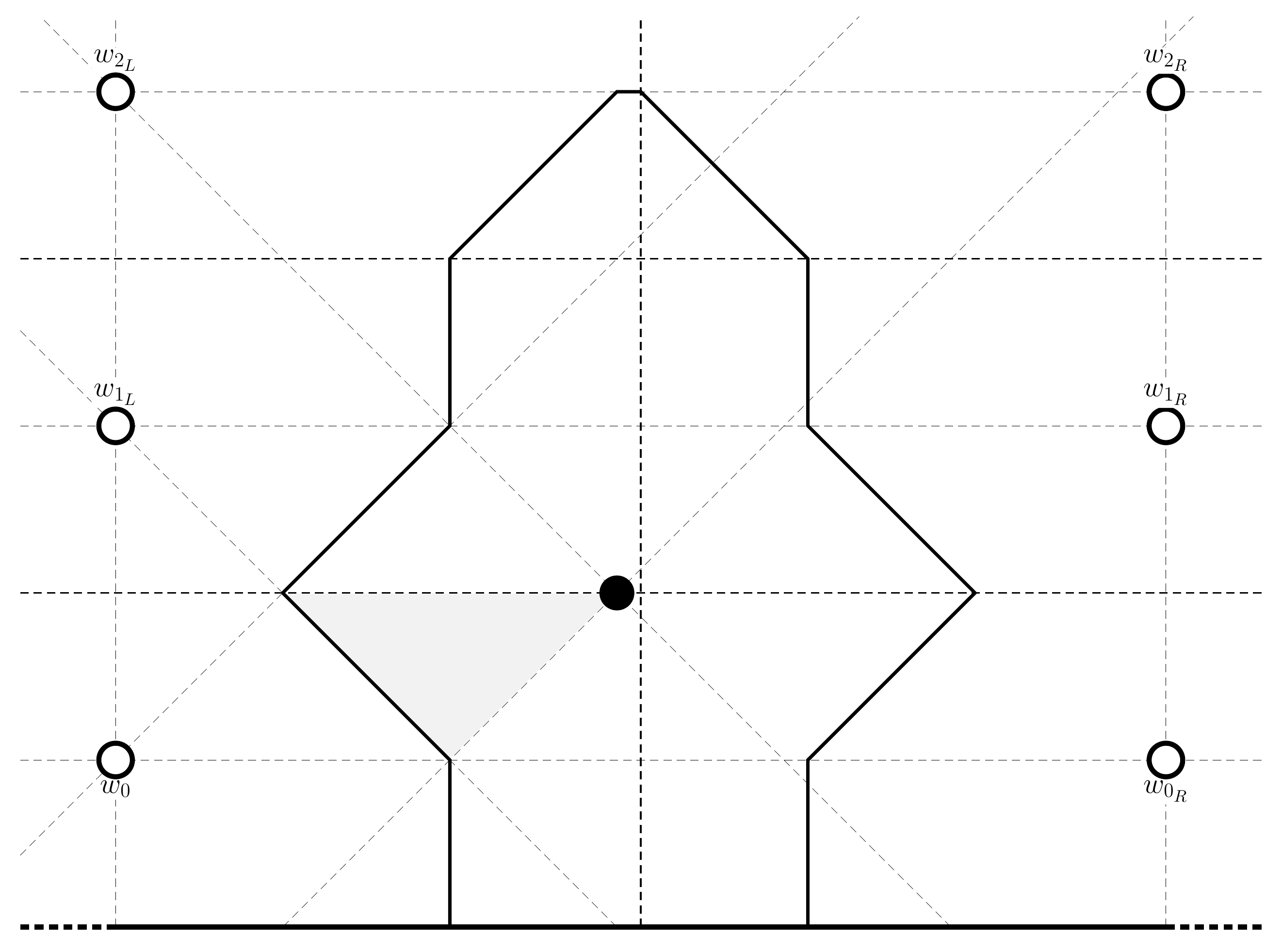}
  \caption{$b_1=(\frac{3q}{2b},\frac{q}{2b})$ only if $\frac{p}{2a} \geq \frac{3q}{2b}$.}
  \label{fig:GridOptimalIIIRes1Long}
\end{subfigure}
\caption{Maximal area Voronoi cells $V^+(b_1)$ for $b_1$ within Section $III$ touching the bottommost horizontal edge of $\mathcal{P}$.}
\end{figure}

Otherwise, if $w_{{2}_L}$ does not exist then, using the results of Section $b$ from above: if $\frac{p}{2a} \leq \frac{5q}{4b}$ then the maximal area is $Area(V^+((\frac{p}{2a}, \frac{2q}{3b} - \frac{p}{6a}))) = \frac{7pq}{12ab} + \frac{p^2}{12a^2} - \frac{5q^2}{12b^2}$ as depicted in Figure~\ref{fig:GridOptimalIIIRes2a}; if $\frac{5q}{4b} \leq \frac{p}{2a} \leq \frac{3q}{2b}$ then the optimum is $(\frac{p}{2a}, \frac{p}{2a} - \frac{q}{b})$ (upon the boundary between Section $III$ and Section $IV$) giving
\begin{equation*}
\begin{split}
    Area(V^+(b_1))&= - \frac{3}{4}\left(\frac{p}{2a} - \frac{q}{b}\right)^2 - \left(\frac{p}{4a} - \frac{(1+1)q}{2b}\right)\left(\frac{p}{2a} - \frac{q}{b}\right) \\ &\qquad+ \frac{(2(1)+1)pq}{4ab}  + \frac{p^2}{16a^2} - \frac{3(1)^2q^2}{4b^2} \\
%    &= - \frac{3}{4}\left(\frac{p}{2a} - \frac{q}{b}\right)^2 - \left(\frac{p}{4a} - \frac{q}{b}\right)\left(\frac{p}{2a} - \frac{q}{b}\right) + \frac{3pq}{4ab}  + \frac{p^2}{16a^2} - \frac{3q^2}{4b^2} \\
    &= \frac{9pq}{4ab} - \frac{p^2}{4a^2} - \frac{5q^2}{2b^2}
\end{split}
\end{equation*}
as depicted in Figure~\ref{fig:GridOptimalIIIRes2b}; and if $\frac{p}{2a} \geq \frac{3q}{2b}$, since the optimum at $\frac{p}{2a}=\frac{3q}{2b}$ is $(\frac{3q}{2b},\frac{q}{2b})$, our previous results give that the optimum is at $b_1=(\frac{3q}{2b},\frac{q}{2b})$ with

\begin{equation*}
\begin{split}
    Area(V^+(b_1)) &= Area(V^+(b_1) \cap (V^\circ(w_{1_L}) \cup V^\circ(w_{1_R}))) + Area(V^+(b_1) \cap (V^\circ(w_{0_L}) \cup V^\circ(w_{0_R}))) \\ &\qquad+ Area(V^+(b_1) \cap (V^\circ(w_{{-1}_L}) \cup V^\circ(w_{{-1}_R}))) \\
    &= \left(\frac{q}{b}\frac{q}{2b} + \frac{pq}{2ab} - \frac{(4(1)-1)q^2}{4b^2}\right) + \left(\frac{pq}{2ab} - \left(\frac{q}{2b}\right)^2\right) + \left(- \frac{1}{4}\left(\frac{3q}{2b}\right)^2 + \frac{1}{4}\left(\frac{q}{2b}\right)^2 \right. \\ &\left.\qquad+ \frac{p}{4a}\frac{3q}{2b} - \left(\frac{p}{4a} + \frac{(-1+1)q}{2b}\right)\frac{q}{2b} + \frac{(-1+1)pq}{4ab} + \frac{(-1+1)^2q^2}{4b^2}\right) \\
    &= \frac{5pq}{4ab} - \frac{q^2}{b^2}
\end{split}
\end{equation*}
as depicted in Figure~\ref{fig:GridOptimalIIIRes2Long}.

\begin{figure}[!ht]\ContinuedFloat
\begin{subfigure}{.394\textwidth}
  \centering
  \includegraphics[width=0.99\textwidth]{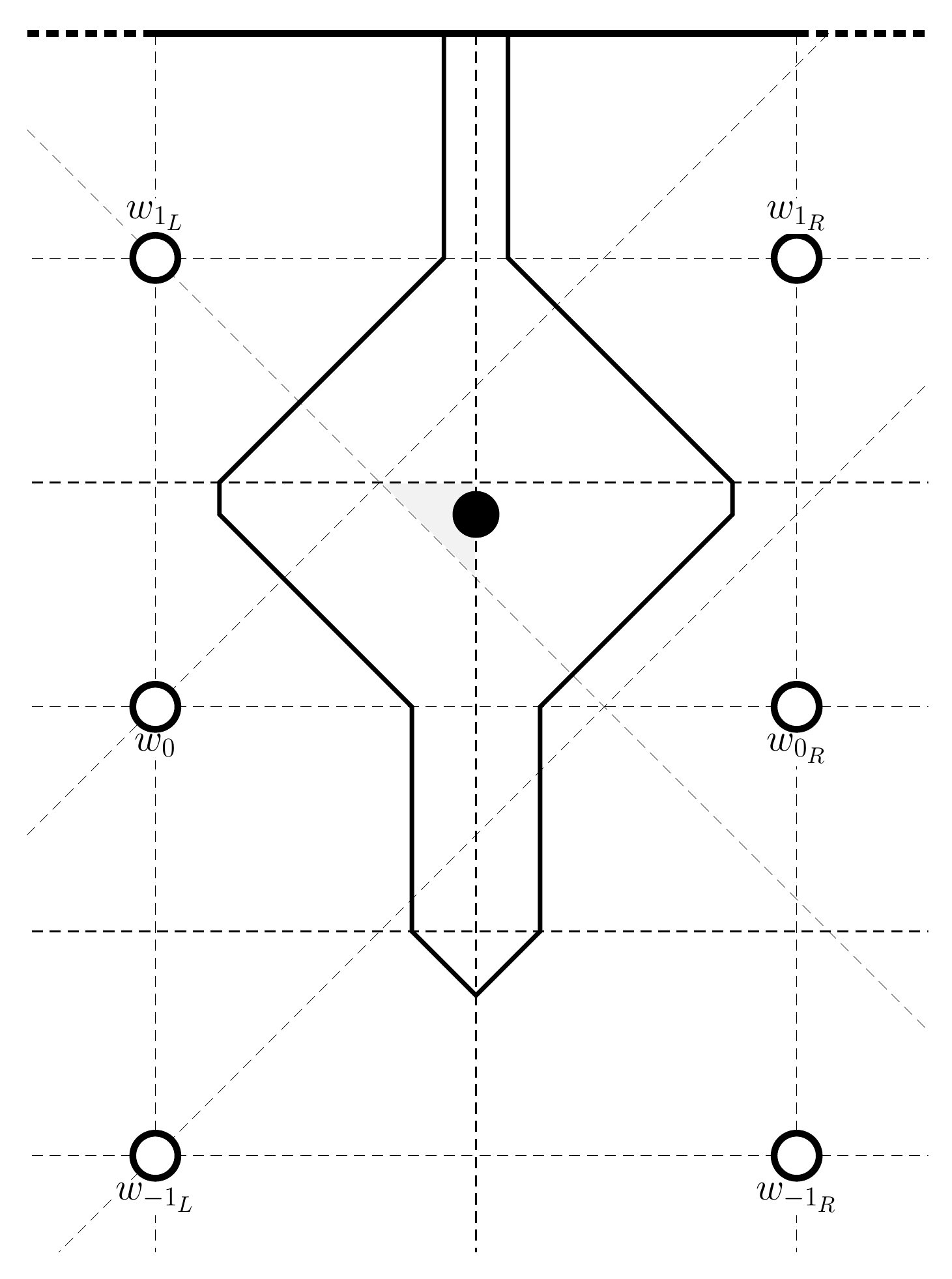}
  \caption{$b_1=(\frac{p}{2a}, \frac{2q}{3b} - \frac{p}{6a})$ only if $\frac{p}{2a} \leq \frac{5q}{4b}$.}
  \label{fig:GridOptimalIIIRes2a}
\end{subfigure}%
\begin{subfigure}{.606\textwidth}
  \centering
  \includegraphics[width=0.99\textwidth]{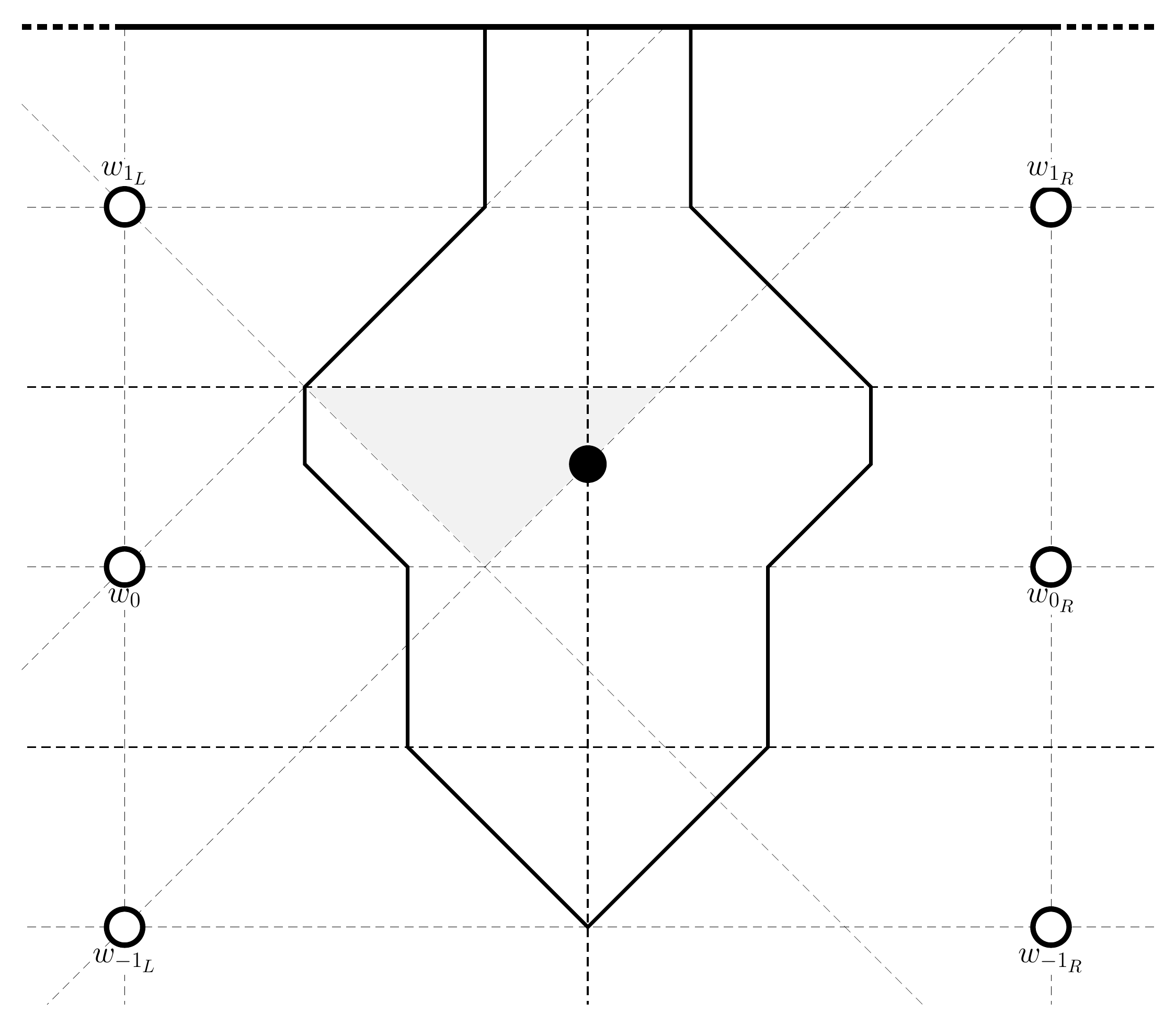}
  \caption{$b_1=(\frac{p}{2a}, \frac{p}{2a} - \frac{q}{b})$ only if $\frac{5q}{4b} \leq \frac{p}{2a} \leq \frac{3q}{2b}$.}
  \label{fig:GridOptimalIIIRes2b}
\end{subfigure}

\centering
\begin{subfigure}{.77\textwidth}
  \centering
  \includegraphics[width=0.99\textwidth]{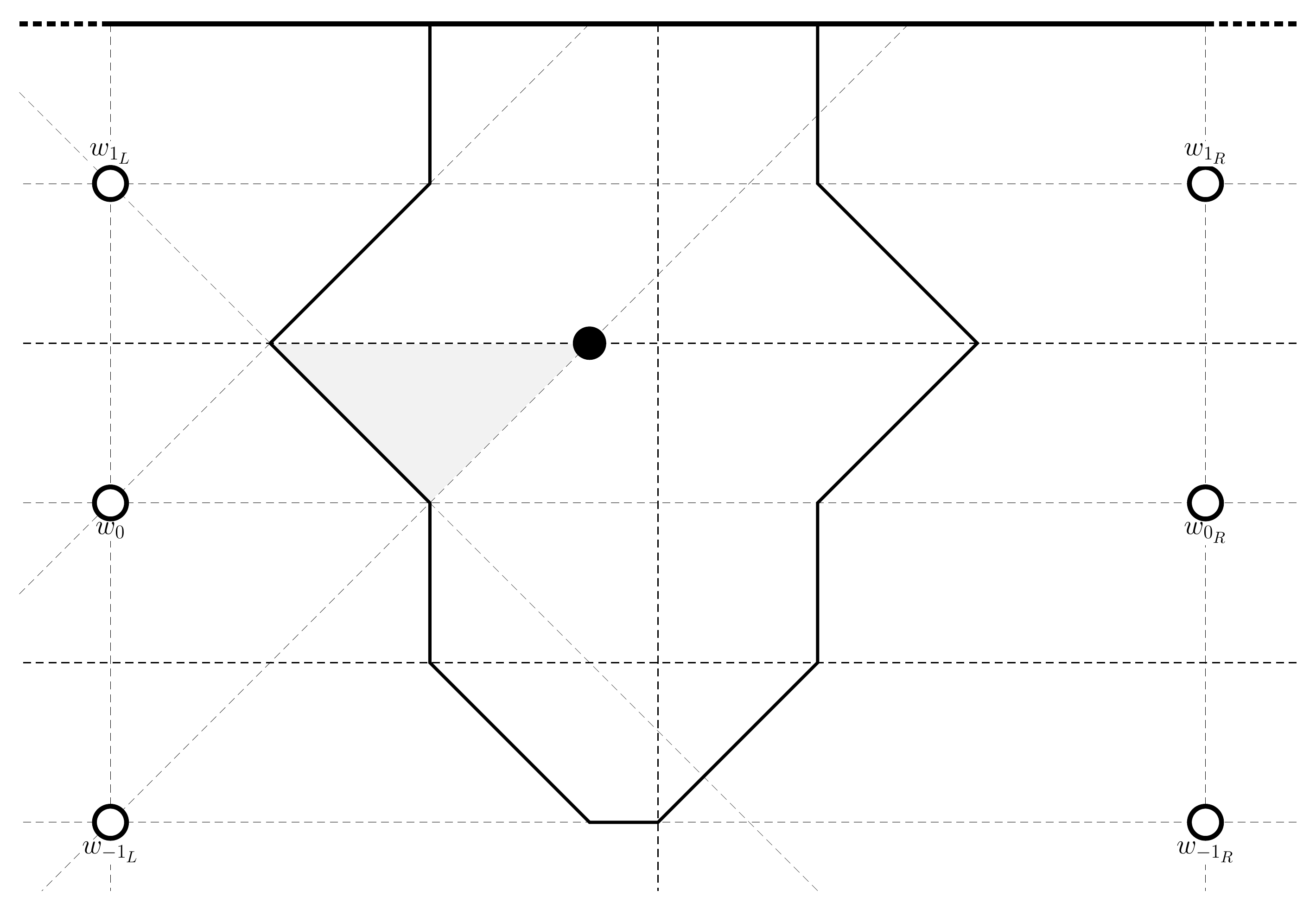}
  \caption{$b_1=(\frac{3q}{2b},\frac{q}{2b})$ only if $\frac{p}{2a} \geq \frac{3q}{2b}$.}
  \label{fig:GridOptimalIIIRes2Long}
\end{subfigure}
\caption{Maximal area Voronoi cells $V^+(b_1)$ for $b_1$ within Section $III$ touching the topmost horizontal edge of $\mathcal{P}$.}
\end{figure}

% then the optimum in Section $b$ is $(\frac{p}{2a}, \frac{(b+1)q}{6b} - \frac{p}{6a})$ giving $Area(V^+((\frac{p}{2a}, \frac{(b+1)q}{6b} - \frac{p}{6a}))) = \frac{(5b-1)pq}{24ab} + \frac{p^2}{12a^2} - \frac{(b-2)(2b-1)q^2}{12b^2}$ (depicted in Figure~\ref{fig:GridOptimalOddRes}), otherwise if $\frac{(2b-1)q}{4b}\leq \frac{p}{2a} \leq \frac{q}{2}$ then the optimum lies on the boundary with Section $b+1$ and is not Black's best point (and so is not drawn).

Finally, if neither $w_{{-1}_L}$ nor $w_{2_L}$ exists then
\begin{equation*}
    \begin{split}
        Area(V^+(b_1)) &= Area(V^+(b_1) \cap (V^\circ(w_{1_L}) \cup V^\circ(w_{1_R}))) \\ &\qquad+ Area(V^+(b_1) \cap (V^\circ(w_{0_L}) \cup V^\circ(w_{0_R}))) \\
        &= - \frac{x^2}{4} - \frac{3y^2}{4} + \frac{p}{4a}x + ( - \frac{p}{4a} + \frac{q}{b})y  + \frac{pq}{ab} - \frac{3q^2}{4b^2} - (\frac{p}{2a} - \frac{x+y}{2}) \times \frac{x-y}{2} \\
        &= - y^2 + \frac{q}{b}y  + \frac{pq}{ab} - \frac{3q^2}{4b^2}
    \end{split}
\end{equation*}
is maximised by $y=\frac{q}{2b}$, irrespective of the value of $x$. Therefore the optimum is $(x,\frac{q}{2b})$ as depicted in Figure~\ref{fig:GridOptimalIIIRes12}.

\begin{figure}[!ht]\ContinuedFloat
\centering
\begin{subfigure}{.55\textwidth}
  \centering
  \includegraphics[width=0.99\textwidth]{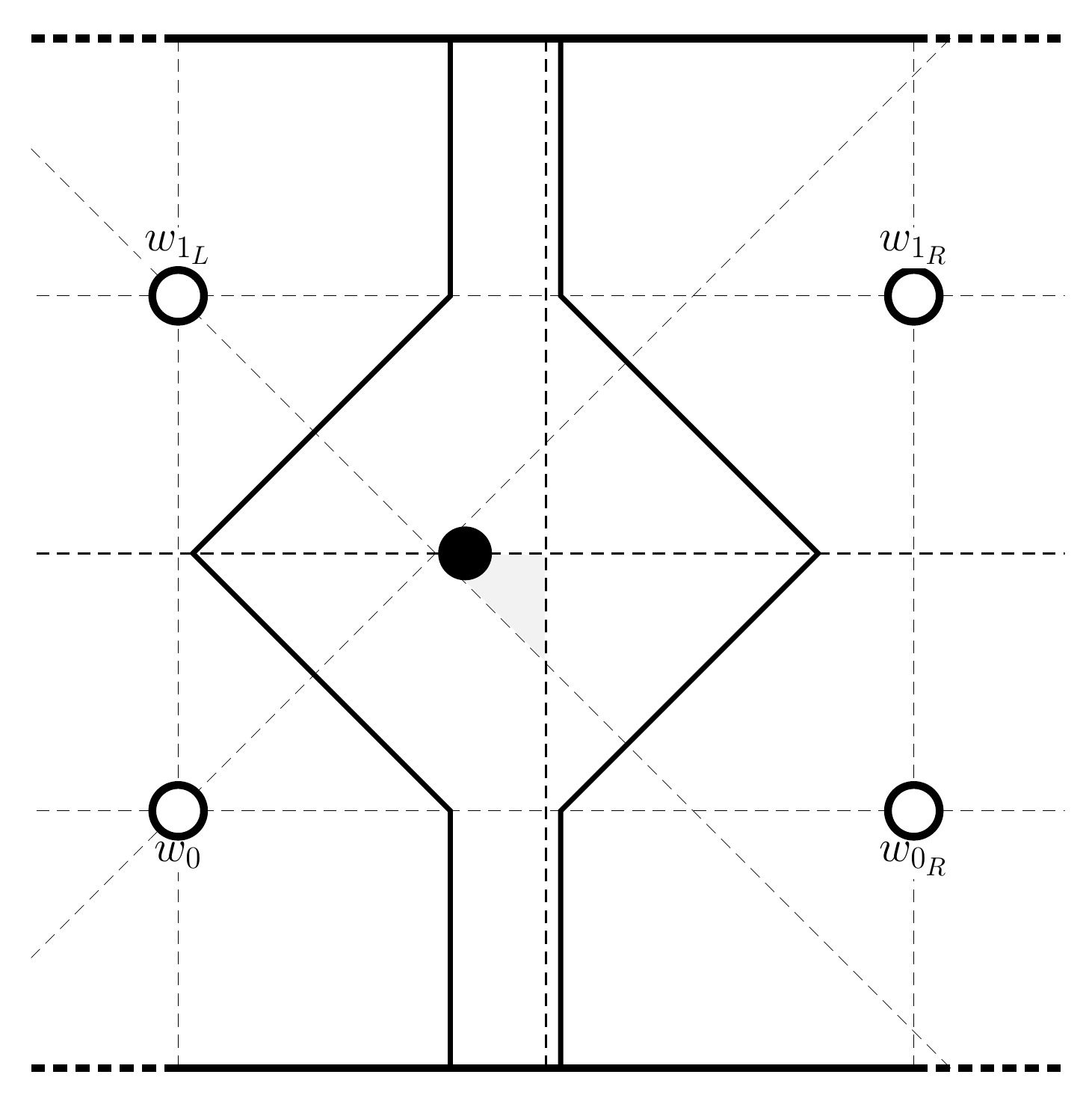}
  \caption{$b_1=(x,\frac{q}{2b})$.}
  \label{fig:GridOptimalIIIRes12}
\end{subfigure}
\caption{Maximal area Voronoi cells $V^+(b_1)$ for $b_1$ within Section $III$ touching both horizontal edges of $\mathcal{P}$.}
\end{figure}

\subsubsection{Edge quadrants}

Now, to consider the placement of $b_1$ in a quadrant of $V^\circ(w)$ which borders the perimeter of $\mathcal{P}$, we shall use the structures explored above and determine all possible cells $V^\circ(b_1)$ in the presence of one boundary of $\mathcal{P}$.

Firstly let us imagine that the quadrant of the cell containing $b_1$ touches $\mathcal{P}$ but does not contain a corner of $\mathcal{P}$ -- i.e. it borders exactly one of the edges of $\mathcal{P}$, so exactly one of $w_{{1}_{L}}$ or $w_{{0}_{R}}$ does not exist. We shall refer to this type of quadrant as an \emph{edge quadrant}. In Figure \ref{fig:GridCellsExamples} this would amount to discarding all area either above $y=\frac{q}{2b}$ or to the right of $x=\frac{p}{2a}$ and we must consider both cuts. However, before we dive into our calculations, let us notice that a vertical cut at $x=\frac{p}{2a}$ would produce Voronoi cells $V^+(b_1)$ for $b_1$ in Section $II$ and beyond exactly resembling those studied in Section~\ref{sec:WhiteRow}, reflected in $y=x$. It should be a great relief to spot this as it saves us having to repeat our calculations since we can simply take our results from Section~\ref{sec:WhiteRow}, remembering to exchange $\frac{p}{n}$ and $q$ with $\frac{q}{b}$ and $\frac{p}{a}$ respectively.

Therefore we need only explore Section $I$ with the vertical cut, and then all sections with the horizontal cut. Another important point to note is that, with a horizontal cut, the partitioning lines donated by points $w_{{i}_L}$ for $i>0$ no longer exist (since the points $w_{{i}_L}$ no longer exist). This means that there is no distinction between Voronoi cells of points in Section $2l$ and Section $2l+1$, so we can explore these together.

But, despite the fear of holding the reader back from diving into the analysis, we can (and will) say still more. The structures from Section $I$ with a vertical cut and Section $II$ (and $III$) with a horizontal cut are identical up to the reflection in $x=y$ as already described. Therefore we need only investigate Section $I$ and transfer the representation symmetrically to Section $II$ (and $III$).

\paragraph{Section $I$} Finding the area of the cell obtained through a horizontal cut according to the Section $I$ structure (with vertices $(0,\frac{x+y}{2})$, $(x,\frac{y-x}{2})$, $(\frac{p}{2a},\frac{y-x}{2})$, $(\frac{p}{2a}+\frac{x+y}{2},y)$, $(\frac{p}{2a}+\frac{x+y}{2},\frac{q}{2b})$, $(-\frac{p}{2a}-\frac{y-x}{2},\frac{q}{2b})$, $(-\frac{p}{2a}-\frac{y-x}{2},y)$, and $(-\frac{p}{2a},\frac{x+y}{2})$) to be

\begin{equation*}
    \begin{split}
        Area(V^+(b_1))&=(\frac{p}{2a}+\frac{x+y}{2}) \times (\frac{q}{2b} - \frac{y-x}{2}) - \frac{1}{2}x^2 -\frac{1}{2}(\frac{x+y}{2})^2 \\
        &+ (\frac{p}{2a}+\frac{y-x}{2}) \times (\frac{q}{2b} - \frac{x+y}{2}) - \frac{1}{2}(\frac{y-x}{2})^2 \\
        &= - \frac{x^2}{4} - \frac{3y^2}{4} + (- \frac{p}{2a} + \frac{q}{2b})y + \frac{pq}{2ab} \, ,  \\
    \end{split}
\end{equation*}
or, if $w_{{0}_{LL}}$ does not exist (i.e. $V^\circ(w_0)$ also touches the perimeter of $\mathcal{P}$ on its left edge),
\begin{equation*}
    \begin{split}
        Area(V^+(b_1))&=(\frac{p}{2a}+\frac{x+y}{2}) \times (\frac{q}{2b} - \frac{y-x}{2}) - \frac{1}{2}x^2 -\frac{1}{2}(\frac{x+y}{2})^2 + (\frac{p}{2a}) \times (\frac{q}{2b} - \frac{x+y}{2}) \\
        &= - \frac{3x^2}{8} - \frac{3y^2}{8} - \frac{xy}{4} + \frac{q}{4b}x + (- \frac{p}{2a} + \frac{q}{4b})y + \frac{pq}{2ab}  \\
    \end{split}
\end{equation*}
gives partial derivatives
\begin{equation*}
\begin{split}
\frac{\delta A}{\delta x}&= - \frac{x}{2} \\
\frac{\delta A}{\delta y}&= - \frac{3y}{2} - \frac{p}{2a} + \frac{q}{2b}
\end{split}
\end{equation*}
enforcing $y^*=- \frac{p}{3a} + \frac{q}{3b} \ngtr 0$, or, if $w_{{0}_{LL}}$ does not exist, gives partial derivatives
\begin{equation*}
\begin{split}
\frac{\delta A}{\delta x}&= - \frac{3x}{4} - \frac{y}{4} + \frac{q}{4b} \\
\frac{\delta A}{\delta y}&= - \frac{3y}{4} - \frac{x}{4} - \frac{p}{2a} + \frac{q}{4b} \\
&\Rightarrow 2x^* - \frac{p}{2a} - \frac{q}{2b} = 0 \Rightarrow x^*= \frac{p}{4a} + \frac{q}{4b}, y^*=- \frac{3p}{4a} + \frac{q}{4b}
\end{split}
\end{equation*}
where $y^*=- \frac{3p}{4a} + \frac{q}{4b} \ngtr 0$ since $\frac{p}{a} \geq \frac{q}{b}$. Therefore for both instances we must explore the boundary of Section $I$ for an optimal location of $b_1$.
\begin{itemize}
    \item Upon boundary $x=0$ we have $Area(V^+((0,y)))= - \frac{3y^2}{4} + (- \frac{p}{2a} + \frac{q}{2b})y + \frac{pq}{2ab}$, maximised by $y^*= - \frac{p}{3a} + \frac{q}{3b} \ngtr 0$. Therefore the maximum will be at $b^*_1=(0,0)$ where $Area(V^+(b^*_1))=\frac{pq}{2ab}$. Alternatively, if $(-\frac{p}{a},0)$ does not exist then $Area(V^+((0,y)))= - \frac{3y^2}{8} + (- \frac{p}{2a} + \frac{q}{4b})y + \frac{pq}{2ab}$, maximised by $y^*= - \frac{2p}{3a} + \frac{q}{3b} \ngtr 0$. Therefore the maximum will also be at $b^*_1=(0,0)$ where $Area(V^+(b^*_1))=\frac{pq}{2ab}$.
    \item Upon boundary $x=y$ we have $Area(V^+((x,x)))= - x^2 + (- \frac{p}{2a} + \frac{q}{2b})x + \frac{pq}{2ab}$ which is maximised by $x^*=- \frac{p}{4a} + \frac{q}{4b} \ngtr 0$, so again the maximum is found at $b^*_1=(0,0)$. Alternatively, if $(-\frac{p}{a},0)$ does not exist then $Area(V^+((x,x)))= - x^2 + (- \frac{p}{2a} + \frac{q}{2b})x + \frac{pq}{2ab}$ is maximised by $x^*= - \frac{p}{4a} + \frac{q}{4b} \ngtr 0$ so again the maximum lies at $b^*_1=(0,0)$.
    \item Upon boundary $y=\frac{q}{2b}$, since its endpoints are shared with endpoints of the other two boundaries which were found not to be optimal over those boundaries, any area will be less than the maximised area already found, and thus the optimum will not exist on this boundary.
\end{itemize}
Therefore the optimal location of $b_1$ in Section $I$ of an edge quadrant for a horizontal edge is to place as close as possible to White's point, and since this technique is summarised in Lemma~\ref{steal} %4.1.1
we do not depict this in Figure \ref{fig:GridOptimals}.

Now, the cell obtained through a vertical cut according to the Section $I$ structure has vertices $(0,\frac{x+y}{2})$, $(x,\frac{y-x}{2})$, $(\frac{p}{2a},\frac{y-x}{2})$, $(\frac{p}{2a},\frac{q}{2b}+\frac{x+y}{2})$, $(x,\frac{q}{2b}+\frac{x+y}{2})$, $(0,\frac{q}{2b}+\frac{y-x}{2})$, $(-\frac{p}{2a},\frac{q}{2b}+\frac{y-x}{2})$, $(-\frac{p}{2a}-\frac{y-x}{2},\frac{q}{2b})$, $(-\frac{p}{2a}-\frac{y-x}{2},y)$, and $(-\frac{p}{2a},\frac{x+y}{2})$ and area

\begin{equation*}
    \begin{split}
        Area(V^+(b_1)) &= - \frac{x^2}{2} - \frac{y^2}{2} + \frac{q}{2b}y +  \frac{pq}{2ab} - (\frac{x+y}{2} \times (\frac{q}{2b} - y) + (\frac{x+y}{2})^2) \\
        &= - \frac{3x^2}{4} - \frac{y^2}{4} - \frac{q}{4b}x + \frac{q}{4b}y +  \frac{pq}{2ab} \, .\\
    \end{split}
\end{equation*}
Note that since $a>1$, the points $w_{{0}_{LL}}$ and $w_{{1}_{LL}}$ always exist. The partial derivatives
\begin{equation*}
\begin{split}
\frac{\delta A}{\delta x}&= - \frac{3x}{2} - \frac{q}{4b} \\
\frac{\delta A}{\delta y}&= - \frac{y}{2} + \frac{q}{4b} \\
\end{split}
\end{equation*}
give $b^*_1=(-\frac{q}{6b},\frac{q}{2b})$ which is not contained in Section $I$ so again we must explore the boundary of Section $I$.
\begin{itemize}
    \item Upon boundary $x=0$ we have $Area(V^+((0,y)))= - \frac{y^2}{4} + \frac{q}{4b}y +  \frac{pq}{2ab}$ which is maximised by $y^*=\frac{q}{2b}$ to give $Area(V^+((0,\frac{q}{2b}))=\frac{pq}{2ab} + \frac{q^2}{16b^2}$.
    \item Upon boundary $x=y$ we have $Area(V^+((x,x)))= - x^2 +  \frac{pq}{2ab}$ which is maximised by $x^*=0$ to give $Area(V^+((0,0)))=\frac{pq}{2ab}$.
    \item Upon boundary $y=\frac{q}{2b}$, since its endpoints are shared with endpoints of the other two boundaries which were found not to be optimal over those boundaries, any area will be less than the maximised area already found, and thus the optimum will not exist on this boundary.
\end{itemize}
%Since the area formula is separable in $x$ and $y$ and we know that $y^*=\frac{q}{2b}$, we need only check the two intersections of $y=\frac{q}{2b}$ with the boundary of Section $I$. Simply, $Area(V^+((0,\frac{q}{2b}))=\frac{pq}{2ab} + \frac{q^2}{16b^2}$ and $Area(V^+((0,\frac{q}{2b}))=\frac{pq}{2ab} - \frac{q^2}{4b^2}$,
Therefore, for the vertical cut, the optimal location in Section $I$ is $b^*_1=(0,\frac{q}{2b})$ giving $Area(V^+((0,\frac{q}{2b})) = \frac{pq}{2ab} + \frac{q^2}{16b^2}$. This is depicted in Figure \ref{fig:GridOptimalIEdgeVert}.

\begin{figure}[!ht]\ContinuedFloat
\centering
\begin{subfigure}{.5\textwidth}
  \centering
  \includegraphics[width=0.9\textwidth]{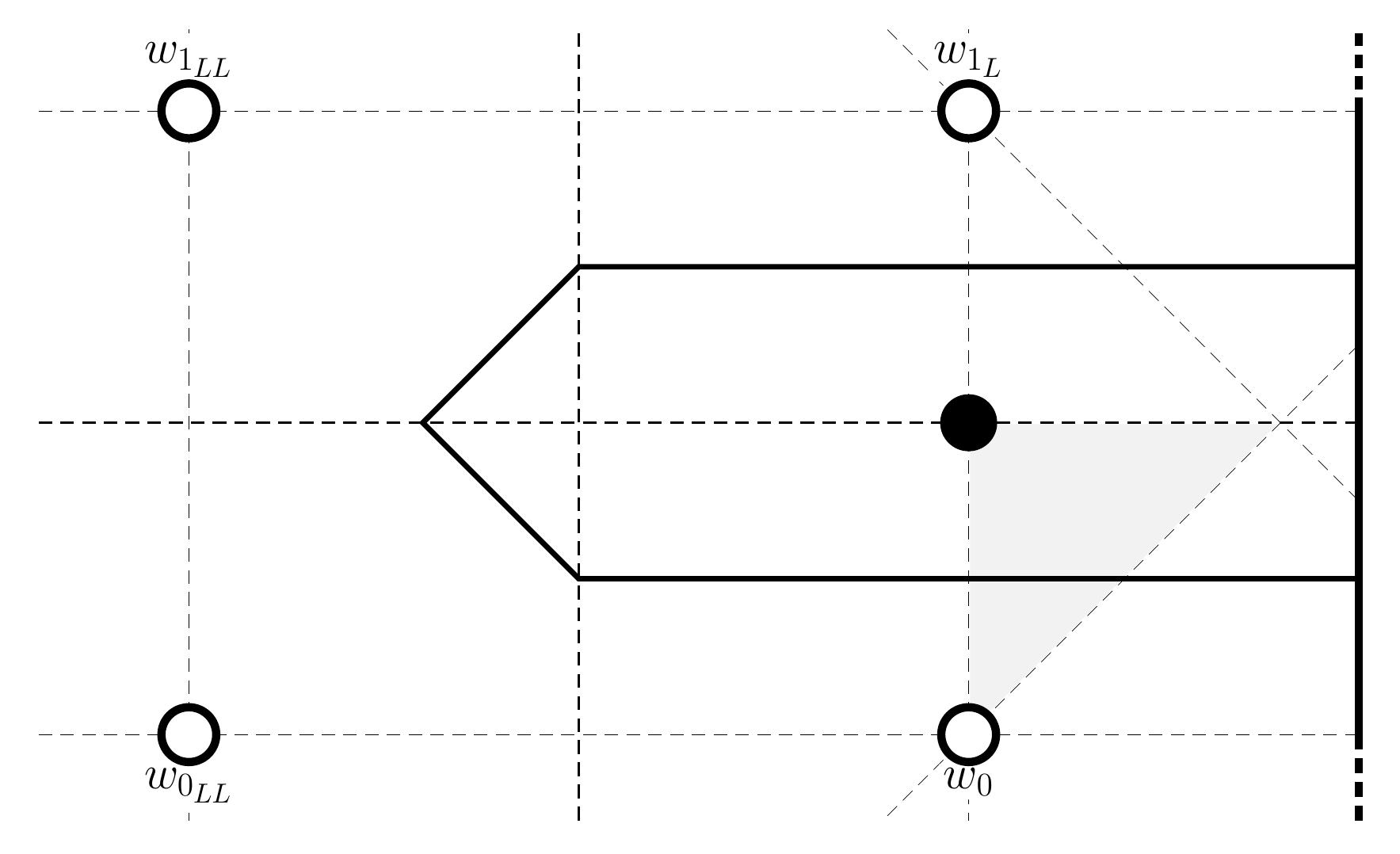}
  \caption{$Area(V^+((0,\frac{q}{2b})))=\frac{pq}{2ab} + \frac{q^2}{16b^2}$.}
  \label{fig:GridOptimalIEdgeVert}
\end{subfigure}
\caption{Maximal area Voronoi cell $V^+(b_1)$ touching a vertical boundary for $b_1$ within Section $I$.}
\end{figure}

\paragraph{Section $II$ (and $III$)} So what does this tell us about Section $II$ and Section $III$?
%Retracing our steps in the symmetrical framework,it is now the case, for the vertical edge situation (taking results from the horizontal edge situation studied for Section $I$), that the global optimum $(\frac{p}{3a}-\frac{q}{3b},0)$ is feasible ($y^*$ previously lay outside Section $I$), giving a total area $Area(V^+(\frac{p}{3a}-\frac{q}{3b},0))=\frac{pq}{3ab} + \frac{p^2}{12a^2} + \frac{q^2}{12b^2}$. This is depicted in Figure \ref{fig:GridOptimalIIEdgeVert}.
%
%However, the global optimum for the section\todo{This was cell} where $(0,-\frac{q}{b})$ does not exist is now $(\frac{p}{4a}-\frac{3q}{4b},\frac{p}{4a}+\frac{q}{4b})$ which lies the wrong side of $x=y$ to fall inside Section $II$. For this case we must study the boundaries of Section $II$ which we can take directly from the results found in Section $I$ for when $(-\frac{p}{a},0)$ does not exist; the optimal location of $b_1$ is as close as possible to White's point in order to achieve an area no greater than $\frac{pq}{2ab}$. Again, this is not depicted in Figure~\ref{fig:GridOptimals}.

For the horizontal edge situation in Sections $II$ and $III$, the results can be transformed directly from the vertical edge situation studied for Section $I$ since this work did not rely on any relationship between the sizes of $\frac{p}{a}$ and $\frac{q}{b}$. Therefore, retracing our steps, the area $$Area(V^+(b_1)) = - \frac{x^2}{4} - \frac{3y^2}{4} + \frac{p}{4a}x - \frac{p}{4a}y + \frac{pq}{2ab}$$ has partial derivatives
\begin{equation*}
\begin{split}
\frac{\delta A}{\delta x}&= - \frac{x}{2} + \frac{p}{4a} \\
\frac{\delta A}{\delta y}&= - \frac{3y}{2} - \frac{p}{4a} \\
\end{split}
\end{equation*}
which give an optimum $b^*_1=(\frac{p}{2a},-\frac{p}{6a})$ which is obviously not within Section $II$ or Section $III$. Therefore the optimum will lie upon either one of the boundaries $y=0$ or $y=x-\frac{q}{b}$ of Section $II$ and Section $III$:
\begin{itemize}
    \item Upon the boundary $y=0$, the optimum will clearly lie at $(\frac{p}{2a},0)$ giving $Area(V^+((\frac{p}{2a},0))) = \frac{pq}{2ab}+\frac{p^2}{16a^2}$. However, if $\frac{p}{2a} \geq \frac{q}{b}$ then this point will not lie within Section $II$ or $III$ so the optimum will instead be $(\frac{q}{b},0)$, giving $Area(V^+((\frac{q}{b},0)))=\frac{3pq}{4ab} - \frac{q^2}{4b^2}$.
    
    \item Upon the boundary $y=x-\frac{q}{b}$ we have $Area(V^+(x,x-\frac{q}{b}))= - x^2 + \frac{3q}{2b}x + \frac{3pq}{4ab} - \frac{3q^2}{4b^2}$ which is maximised by $x^*=\frac{3q}{4b}$. Since this is not within Section $II$ or $III$ the optimum must lie at the closest endpoint, the intersection of $y=0$ and $y=x-\frac{q}{b}$ already studied.
\end{itemize}
Hence $b_1=(\frac{p}{2a},0)$ gives the largest possible $Area(V^+((\frac{p}{2a},0))) = \frac{pq}{2ab}+\frac{p^2}{16a^2}$ as depicted in Figure~\ref{fig:GridOptimalIIEdgeHoria}, unless $\frac{p}{2a} > \frac{q}{b}$ in which case $b_1=(\frac{q}{b},0)$ gives the largest possible $Area(V^+((\frac{q}{b},0)))=\frac{3pq}{4ab} - \frac{q^2}{4b^2}$ as depicted in Figure \ref{fig:GridOptimalIIEdgeHorib}.

\begin{figure}[!ht]\ContinuedFloat
\centering
\begin{subfigure}{.5\textwidth}
  \centering
  \includegraphics[width=0.9\textwidth]{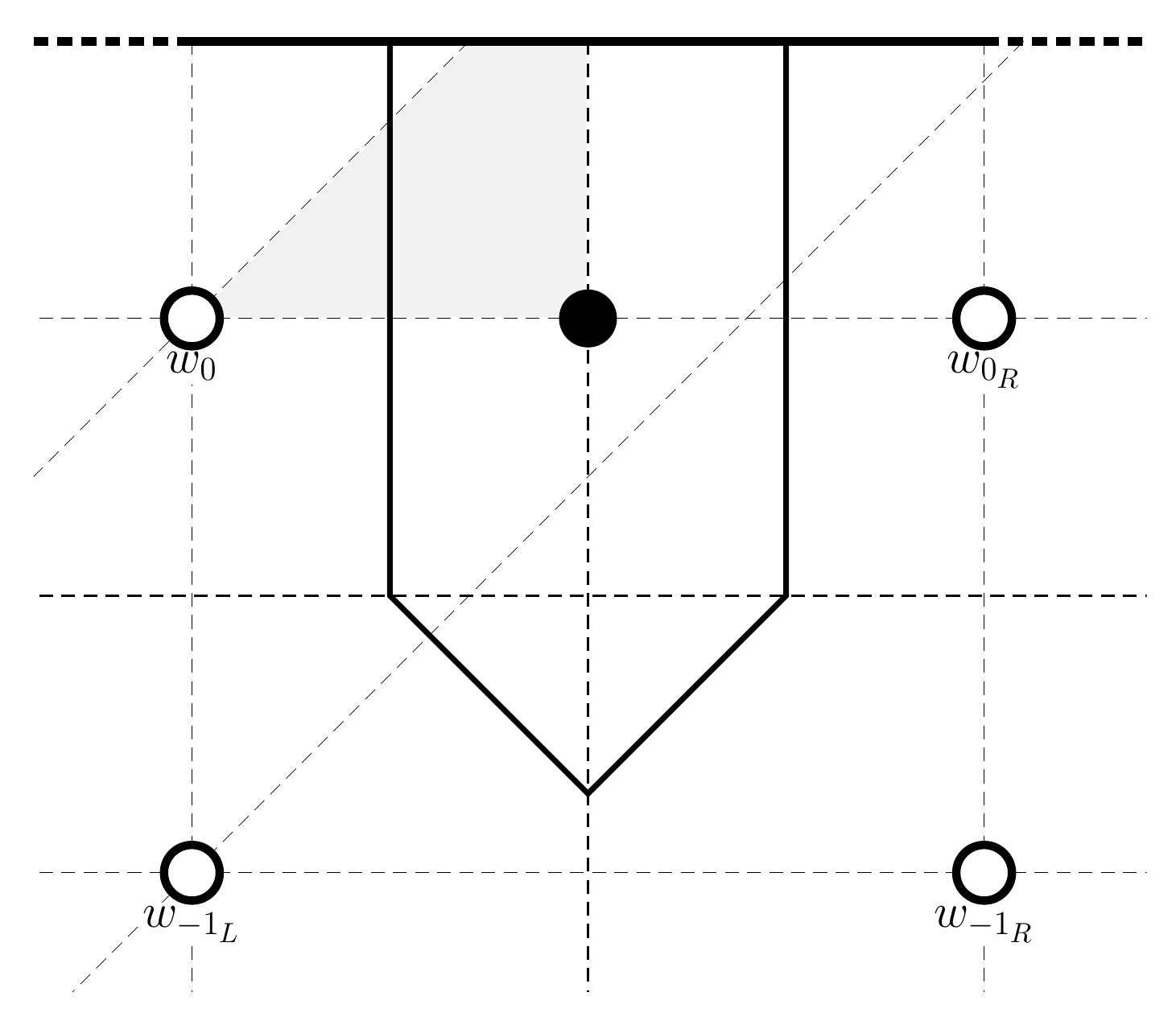}
  \caption{$b_1=(\frac{p}{2a},0)$ only if $\frac{p}{2a} \leq \frac{q}{b}$.}
  \label{fig:GridOptimalIIEdgeHoria}
\end{subfigure}
\begin{subfigure}{.9\textwidth}
  \centering
  \includegraphics[width=0.9\textwidth]{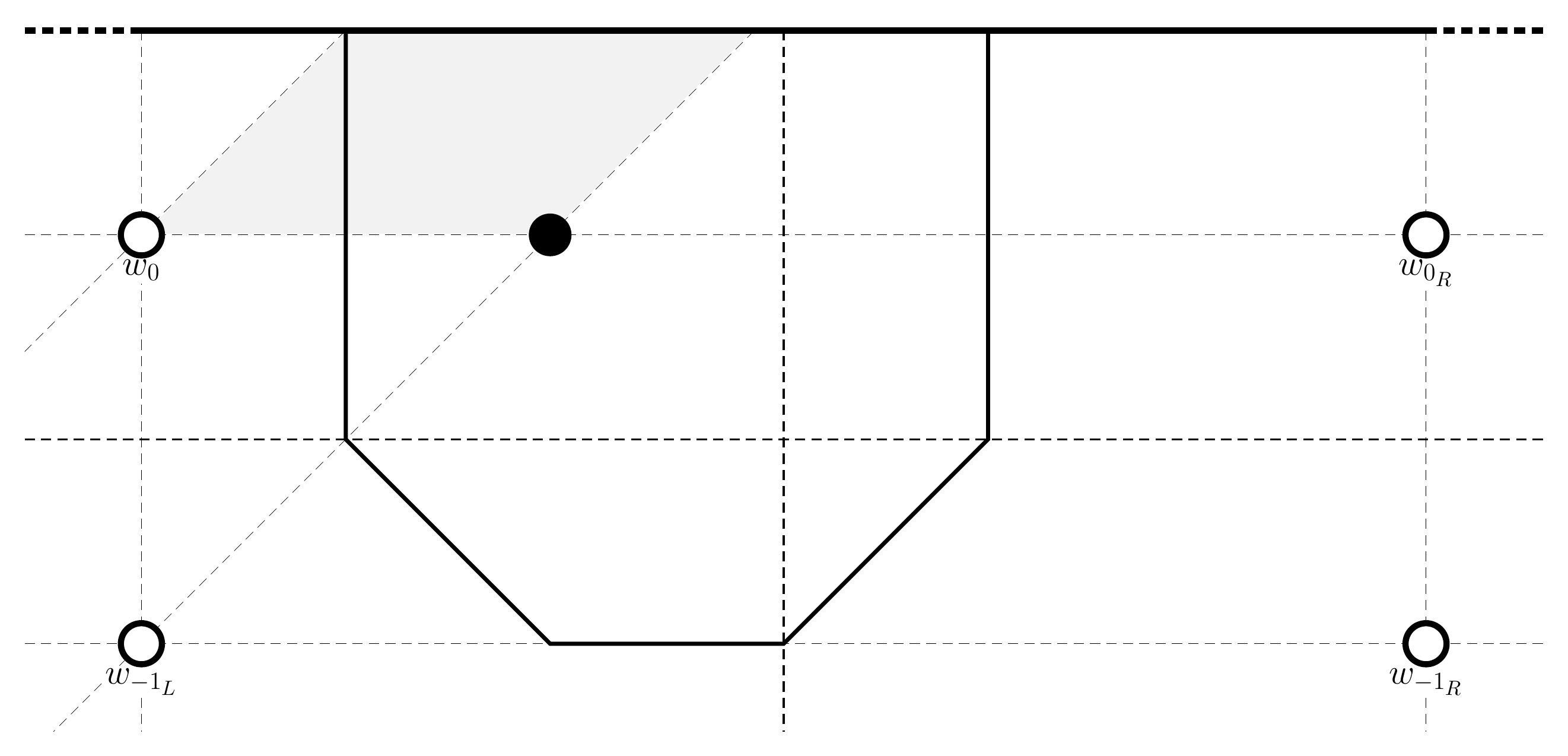}
  \caption{$b_1=(\frac{q}{b},0)$ only if $\frac{p}{2a} > \frac{q}{b}$.}
  \label{fig:GridOptimalIIEdgeHorib}
\end{subfigure}
\caption{Maximal area Voronoi cells $V^+(b_1)$ touching a horizontal boundary for $b_1$ within Sections $II$ and $III$.}
\label{fig:GridOptimalIandIIandIII}
\end{figure}

\subsubsection{Corner quadrants}

Now let us imagine that the quadrant of the cell containing $b_1$ contains a corner of $\mathcal{P}$ -- i.e. both outside boundaries of this quadrant are on the perimeter of $\mathcal{P}$ (neither $w_{1_L}$ nor $w_{0_R}$ exists). We shall refer to this type of quadrant as a \emph{corner quadrant}. In Figure \ref{fig:GridCellsExamples} this would amount to discarding all areas above $y=\frac{q}{2b}$ and to the right of $x=\frac{p}{2a}$. Fortuitously, as with many cases for edge quadrants, Section $II$ and beyond in corner quadrants have already been covered in Section~\ref{sec:WhiteRow} since the areas are identical after a reflection in $y=x$ and suitable rescaling of $p$ and $q$. To our delight, the same is also true for areas within Section $I$, even without the described reflection, since these cells would only require bisectors from one row of points in $W$. Therefore with one fell swoop we can discard having to explore any corner quadrants as the results are contained in the work in Section~\ref{sec:WhiteRow}.
%We can see that we have two distinct possible structures (Section $II$ and $III$ cells are identical) and, again, even that these are identical up to the reflection in $x=y$.

\bigskip

Having found all local optima within the sections closest to each $w_i$ we can experiment with different combinations of these points in order to form reasonable arrangements for Black. However, while we were able to make use of the best point $b^*$ within the first sections when White plays a row, we find that the points drawn in Figure~\ref{fig:GridOptimalIandIIandIII} make pretty lousy team players. This is due to the fact that without any existing black points the placement of $b_1$ within any section, with the exception of Section $I$, is always improved by locating closer to the line $x=\frac{p}{2a}$ in order to make the most of the alluring area left to capture from the Voronoi cells $V^\circ(w_{i_R})$. In this way $V^+(b_1)$ steals a mediocre amount from a large number of Voronoi cells, instead of a large amount from a few cells which we were hoping for (points which steal efficiently from fewer cells provide less risk of overlapping with other black points and so will work well within an arrangement).

It is for this reason that it is probably a more fruitful approach to consider the best points from Black's row strategy for candidates within an arrangement for the grid scenario. Since, when White plays a row, $\mathcal{P}$ provides a physical cap bounding the area available to capture in the opposite direction of the generator of the cell within which Black is locating, the optimal locations better reflect the attempt to steal as much as possible from fewer cells (at least for smaller sections). One such effective arrangement is the grid adaptation of the optimal row arrangement found in Section~\ref{sec:BlackRow} (see Figure~\ref{fig:RowOptimalArrangements} for a reminder) whereby, at least for the preferable even $b$ scenario, each column $w_{i_X}$ of white points for $i=1,\ldots,a$ can be sandwiched between points of Black to create $a$ columns of the arrangement in Figure~\ref{fig:RowOptimalArrangementEven}, rotated by $90^\circ$, where $n=b$. Similarly for a row arrangement $W$, this response from Black does seem, at least for certain values of $\frac{p}{a}$ and $\frac{q}{b}$, to be a good arrangement. However, if the challenges in proving the optimality of an arrangement presented in Section~\ref{sec:BlackRow} left us hiding behind our cushions then we should certainly avert our eyes from the problem of finding optimal arrangements in response to a grid, since the supreme difficulties exhibited here far overshadow what we have already seen. This is in part due to the fact that a point which steals from $V^\circ(w_0)$ will encroach on the thefts from cells below, and above, and to the left or right of $V^\circ(w_0)$ -- a whole new direction to consider in contrast to the two-dimensional reasoning exploited in Section~\ref{sec:BlackRow}.

For this reason we should feel content at having found the best points in response to both a general row and a general grid arrangement, as well as an optimal black arrangement for the row case, and reward ourselves with a forage among the bounties that the non-grid world may furnish.

\bibliographystyle{apacite}
\bibliography{literature}

\end{document}